\documentclass[a4paper,12pt,reqno]{amsart}
\pdfoutput=1
\usepackage{amsaddr}
\usepackage{amsmath,amssymb,amsthm}
\usepackage{amscd}
\usepackage{mathrsfs}
\usepackage[usenames,dvipsnames]{color}
\usepackage{graphicx}
\usepackage[all]{xy}
\usepackage{eucal}
\usepackage{extarrows} 
\usepackage{txfonts}      
\usepackage{tabularx}
\usepackage{tikz}
\usepackage{xparse}
\usepackage{colordvi}
\usepackage{multicol}
\usepackage{multirow}
\usepackage{enumitem}
\usepackage[normalem]{ulem}
\usepackage{epsfig}
\usepackage{xspace}
\usepackage{stmaryrd}
\usepackage{tikz-cd}
\usepackage{xfrac}
\usepackage{lipsum}
\usepackage{booktabs}
\usepackage{float}
\usepackage{tikzit}
\tikzstyle{pointoperator}=[fill=black, draw=none, shape=circle, minimum size=.1 cm, inner sep=0 pt]

\tikzstyle{dashedline}=[dashed, -, thick]
\tikzstyle{blueline}=[-, draw=blue, thick]
\tikzstyle{arrowline}=[<-, thick]
\tikzstyle{normalline}=[-, thick]
\tikzstyle{fillline}=[-, thick, fill=red, fill opacity=.5]
\tikzstyle{bluefillline}=[-, fill opacity=.5, fill=blue, thick]
\tikzstyle{purplefillline}=[-, fill opacity=.5, fill=purple, thick]
\tikzstyle{bluedashedline}=[-, draw=blue, thick, dashed]
\tikzstyle{bluearrowline}=[->, draw=blue, thick]

\usepackage{simpler-wick}
\usepackage{dsfont}
\usepackage{halloweenmath}
\usepackage{mathtools}
\usepackage{cite}
\newcolumntype{H}{>{\setbox0=\hbox\bgroup}c<{\egroup}@{}}
\newcounter{magicrownumbers}

\usepackage{ifpdf}
\ifpdf
 \usepackage[colorlinks,final,backref=page,hyperindex]{hyperref}
\else
 \usepackage[colorlinks,final,backref=page,hyperindex,hypertex]{hyperref}
\fi

\usepackage{graphicx}
\usepackage{epstopdf}
\usepackage{epsfig}
\usepackage{pdfsync}    

\usetikzlibrary{arrows}  
\usetikzlibrary{decorations.pathmorphing}  

\def\XXint#1#2#3{{\setbox0=\hbox{$#1{#2#3}{\int}$ }
\vcenter{\hbox{$#2#3$ }}\kern-.58\wd0}}

\def\XXsum#1#2#3{{\setbox0=\hbox{$#1{#2#3}{\sum}$ }
\vcenter{\hbox{$#2#3$ }}\kern-.51\wd0}}

\begin{document}

\newcommand\cutoffint{\mathop{-\hskip -4mm\int}\limits}
\newcommand\cutoffsum{\mathop{-\hskip -4mm\sum}\limits}
\newcommand\cutoffzeta{-\hskip -1.7mm\zeta}
\newcommand{\goth}[1]{\ensuremath{\mathfrak{#1}}}
\newcommand{\bbox}{\normalsize {}%
        \nolinebreak \hfill $\blacksquare$ \medbreak \par}
\newcommand{\simall}[2]{\underset{#1\rightarrow#2}{\sim}}

\newtheorem{theorem}{Theorem}[section]
\newtheorem{prop}[theorem]{Proposition}
\newtheorem{lemdefn}[theorem]{Lemma-Definition}
\newtheorem{propdefn}[theorem]{Proposition-Definition}
\newtheorem{lem}[theorem]{Lemma}
\newtheorem{lemma}[theorem]{Lemma}
\newtheorem{thm}[theorem]{Theorem}
\newtheorem{coro}[theorem]{Corollary}
\newtheorem{claim}[theorem]{Claim}
\newtheorem{outline}[theorem]{Outline}
\newtheorem{question}[theorem]{Problem}
\newtheorem{Goal}[theorem]{Goal}
\theoremstyle{definition}
\newtheorem{defn}[theorem]{Definition}
\newtheorem{rk}[theorem]{Remark}
\newtheorem{remark}[theorem]{Remark}
\newtheorem{ex}[theorem]{Example}
\newtheorem{coex}[theorem]{Counterexample}
\newtheorem{conj}[theorem]{Conjecture}

\renewcommand{\theenumi}{{\it\roman{enumi}}}
\renewcommand{\theenumii}{\alpha{enumii}}

\newenvironment{thmenumerate}{\leavevmode\begin{enumerate}[leftmargin=1.5em]}{\end{enumerate}}

\setcounter{MaxMatrixCols}{20}

\newcommand{\nc}{\newcommand}

\nc{\delete}[1]{{}}

\nc{\mlabel}[1]{\label{#1}} 
\nc{\mcite}[1]{\cite{#1}}  
\nc{\mref}[1]{\ref{#1}}  
\nc{\meqref}[1]{\eqref{#1}}
\nc{\mbibitem}[1]{\bibitem{#1}}

\delete{
\nc{\mcite}[1]{\cite{#1}{{\bf{{\ }(#1)}}}}  
\nc{\mlabel}[1]{\label{#1}  
	{\hfill \hspace{1cm}{\bf{{\ }\hfill(#1)}}}}
\nc{\mref}[1]{\ref{#1}{{\bf{{\ }(#1)}}}}  
\nc{\meqref}[1]{\eqref{#1}{{\bf{{\ }(#1)}}}} 
\nc{\mbibitem}[1]{\bibitem[\bf #1]{#1}} 
}

%%%%%% new symbols---------------------------------------------------

\nc{\mseva}{{\cale^Q_{\rm MS}}}

\nc{\wvec}[2]{{\scriptsize{\Big [ \!\!\begin{array}{c} #1 \\ #2 \end{array} \!\! \Big ]}}}

\nc{\name}[1]{{\bf #1}}

\nc{\scs}[1]{\scriptstyle{#1}}
\newfont{\scyr}{wncyr10 scaled 550}
\nc{\ssha}{\,\mbox{\bf \scyr X}\,}
\newfont{\bcyr}{wncyr10 scaled 1000}

\newcommand{\bottop}{\top\hspace{-0.8em}\bot}
\newcommand{\bbB}{\mathbb{B}}
\newcommand {\bbC}{\mathbb{C}}
\newcommand {\bbF}{{\mathbb{F}}}
\newcommand{\bbG}{\mathbb{G}}
\newcommand{\K}{\mathbb{K}}
\newcommand{\N}{\mathbb{N}}
\newcommand {\bbP}{{\mathbb{P}}}
\newcommand{\PP}{\mathbb{P}}
\newcommand{\Q}{\mathbb{Q}}
\newcommand{\R}{\mathbb{R}}
\newcommand {\bbR}{{\mathbb{R}}}
\newcommand{\T}{\mathbb{T}}
\newcommand{\bbU}{\mathbb{U}}
\newcommand{\W}{\mathbb{W}}
\newcommand {\bbW}{{\mathbb{W}}}
\newcommand{\Z}{\mathbb{Z}}
%%%---------------

\def\AA{\textsl{A}}
\def\BB{\textsl{B}}
\def\CC{\textsl{C}}
\def\DD{\textsl{D}}
\def\EE{\textsl{E}}
\def\FF{\textsl{F}}
\def\GG{\textsl{G}}
\def\SL{\textsl{SL}}
\def\GL{\textsl{GL}}
\def\ch{\mathrm{ch}}
\def\V{\mathcal{V}}
\def\Ex{\mathrm{Ex}}
\def\Cat{\mathcal{C}}
\def\Com{\mathrm{Com}}
\def\Vec{\textsl{Vec}}
\def\Rep{\textsl{Rep}}

\def\checkmark{\tikz\fill[scale=0.4](0,.35) -- (.25,0) -- (1,.7) -- (.25,.15) -- cycle;}

%temp-------------------------
\nc{\deff}{K}
\nc{\RR}{{\mathbb R}}
\nc{\ZZ}{{\mathbb Z}}
\nc{\QQ}{{\mathbb Q}}
\nc{\BBC}{{\mathbb C}}

\nc{\mrm}[1]{{\rm #1}}
\nc{\Aut}{\mathrm{Aut}}
\nc{\depth}{{\mrm d}}
\nc{\id}{\mrm{id}}
\nc{\Id}{\mathrm{Id}}
\nc{\Irr}{\mathrm{Irr}}
\nc{\Span}{\mathrm{span}}
\nc{\mapped}{operated\xspace}
\nc{\Mapped}{Operated\xspace}
\newcommand{\redtext}[1]{{\textcolor{red}{#1}}}
\newcommand{\Hol}{\text{Hol}}
\newcommand{\Mer}{\text{Mer}}
\newcommand{\lin}{\text{lin}}
\nc{\ot}{\otimes}
\nc{\Hom}{\mathrm{Hom}}
\nc{\Gen}{\mathrm{Gen}}
\nc{\CS }{\mathcal{CS}}
\nc{\bfk}{{K}}
\nc{\lwords}{\calw}
\nc{\ltrees}{\calf}
\nc{\lpltrees}{\calp}
\nc{\Map}{\mathrm{Map}}
\nc{\rep}{\beta}
\nc{\free}[1]{\bar{#1}}
\nc{\OS}{\mathbf{OS}}
\nc{\OM}{\mathbf{OM}}
\nc{\OA}{\mathbf{OA}}
\nc{\based}{based\xspace}
\nc{\tforall}{\text{ for all }}
\nc{\hwp}{\widehat{P}^\calw}
\nc{\sha}{{\mbox{\cyr X}}}
\font\cyr=wncyr10 \font\cyrs=wncyr7
\nc{\Mor}{\mathrm{Mor}}
\def\lc{\lfloor}
\def\rc{\rfloor}
\nc{\oF}{{\overline{F}}}
\nc{\mge}{_{bu}\!\!\!\!{}}
\newcommand {\bfc}{{\bf {C}}}
\nc {\conefamilyc}{{\underline{{C}}}}
\newcommand{\loc}{locality\xspace}
\newcommand{\Loc}{Locality\xspace}

\nc{\supp}{\mathrm{Dep }}
\nc {\ordcone} {ordered cone\xspace}
\nc {\ordcones} {ordered cones\xspace}
\nc {\simplex}{simplex\xspace}
\nc{\QsubS}{\QQ\Pi^Q(\mti{\cals})}
\nc{\subS}{\Pi^Q(\mti{\cals})}
\nc{\wt}{weight\xspace}
\nc{\wts}{weights\xspace}

\NewDocumentCommand{\A}{o o}{%
  \IfNoValueTF{#1}
    {\mathsf{A}}
    {\IfNoValueTF{#2}
       {\mathsf{A}_{#1}}
       {\mathsf{A}_{#1}^{#2}}
    }
}
\NewDocumentCommand{\B}{o o}{%
  \IfNoValueTF{#1}
    {\mathsf{B}}
    {\IfNoValueTF{#2}
       {\mathsf{B}_{#1}}
       {\mathsf{B}_{#1}^{#2}}
    }
}
\NewDocumentCommand{\C}{o o}{%
  \IfNoValueTF{#1}
    {\mathsf{C}}
    {\IfNoValueTF{#2}
       {\mathsf{C}_{#1}}
       {\mathsf{C}_{#1}^{#2}}
    }
}
\NewDocumentCommand{\D}{o o}{%
  \IfNoValueTF{#1}
    {\mathsf{D}}
    {\IfNoValueTF{#2}
       {\mathsf{D}_{#1}}
       {\mathsf{D}_{#1}^{#2}}
    }
}
\NewDocumentCommand{\E}{o o}{%
  \IfNoValueTF{#1}
    {\mathsf{E}}
    {\IfNoValueTF{#2}
       {\mathsf{E}_{#1}}
       {\mathsf{E}_{#1}^{#2}}
    }
}
\NewDocumentCommand{\F}{o o}{%
  \IfNoValueTF{#1}
    {\mathsf{F}}
    {\IfNoValueTF{#2}
       {\mathsf{F}_{#1}}
       {\mathsf{F}_{#1}^{#2}}
    }
}
\NewDocumentCommand{\G}{o o}{%
  \IfNoValueTF{#1}
    {\mathsf{G}}
    {\IfNoValueTF{#2}
       {\mathsf{G}_{#1}}
       {\mathsf{G}_{#1}^{#2}}
    }
}

\title[]{Bosonic Rational Conformal Field Theories in Small Genera, \\ Chiral Fermionization, and Symmetry/Subalgebra Duality}

    \author{\textsc{Brandon C.\ Rayhaun}}

\address{C.~N.~Yang Institute for Theoretical Physics,
	Stony Brook University, Stony Brook, USA}
\email{brandon.rayhaun@stonybrook.edu}

\begin{abstract}  
A (1+1)D unitary bosonic rational conformal field theory (RCFT) may be organized according to its genus, a tuple $(c,\Cat)$ consisting of its central charge $c$ and a unitary modular tensor category $\Cat$ which describes the (2+1)D topological quantum field theory (TQFT) for which its maximally extended chiral algebra forms a holomorphic boundary condition. We establish a number of results pertaining to RCFTs in ``small'' genera, by which we informally mean genera with the central charge $c$ and the number of primary operators $\mathrm{rank}(\Cat)$ both not too large. We start by completely solving the modular bootstrap problem for theories with at most four primary operators. In particular, we characterize, and provide an algorithm which efficiently computes, the function spaces to which the partition function of any bosonic RCFT with $\mathrm{rank}(\Cat)\leq 4$ must belong. Using this result, and leveraging relationships between RCFTs and holomorphic vertex operator algebras which come from ``gluing'' and cosets, we rigorously enumerate all bosonic theories in $95$ of the $105$ genera $(c,\Cat)$ with $c\leq 24$ and $\mathrm{rank}(\Cat)\leq 4$. This includes as (new) special cases the classification of chiral algebras with three primaries and $c<\sfrac{120}{7}\sim 17.14$, and the classification of chiral algebras with four primaries and $c<\sfrac{62}{3}\sim 20.67$. We then study two applications of our classification. First, by making use of chiral versions of bosonization and fermionization, we obtain the complete list of purely left-moving fermionic RCFTs with $c<23$ as a corollary of the results of the previous paragraph. Second, using a (conjectural) concept which we call ``symmetry/subalgebra duality,'' we precisely relate our bosonic classification to the problem of determining certain generalized global symmetries of holomorphic vertex operator algebras.
\end{abstract}

\begingroup
\def\uppercasenonmath#1{} 
\let\MakeUppercase\relax
\maketitle
\endgroup

\clearpage 

\vspace{-1cm}
\tableofcontents

\vspace{-1cm}

\allowdisplaybreaks

\section{Overview}\label{sec:overview}

This paper deals with the classification of two-dimensional unitary, rational conformal field theories --- both bosonic and fermionic --- as well as their discrete (generalized) global symmetries. This class of theories has enjoyed innumerable applications in physics and mathematics, and their classification is a  decades-old problem, so we hope that the reader is sufficiently motivated to read on without an extended introduction. In this section, which is the only ``non-appendix'' section of the paper, we give a gentle summary of our main results. All of the technical details and calculations can be found in the appendices.

There are two parts to studying a two-dimensional unitary, bosonic RCFT. The first part is the specification of a chiral algebra $\V$ --- known in mathematics as a vertex operator algebra (VOA) --- to serve as the sector of holomorphic local operators of the theory. The chiral algebra of an RCFT has two important invariants associated with it: its central charge $c$, and an abstract mathematical object $\Cat\cong \Rep(\V)$ known as a unitary modular tensor category (UMTC) \cite{Moore:1988qv,Huang:2005gs} which encodes various coarse features of the representation theory of $\V$, like its fusion rules and the $\SL_2(\ZZ)$ representation which governs the transformation properties of its characters (see \S\ref{subsec:overview:characters}).\footnote{More precisely, the UMTC $\Cat$ encodes the modular representation up to a small ambiguity which can be fixed by specifying the central charge modulo $24$.} Alternatively, using the well-known relationship between UMTCs and three-dimensional topological quantum field theories (TQFTs) \cite{Witten:1988hf,Elitzur:1989nr,Reshetikhin:1990pr,Reshetikhin:1991tc,turaev2020quantum,bakalov2001lectures}, one can more physically think of $\Cat$ as representing the bulk (2+1)D TQFT which supports $\V$ on its boundary. Together, these two invariants $(c,\Cat)$ are known as the \emph{genus} of $\V$, generalizing an analogous invariant of lattices \cite{Hohn:2002dm}.\footnote{See \cite{Moriwaki:2020ktv} for another generalization of the notion of the genus of a lattice, which is inequivalent to the one which is used in this paper.} We sometimes use the notation $\Gen(c,\Cat)$ to refer to the set of (isomorphism classes) of chiral algebras in the genus $(c,\Cat)$.

The second part is the specification of a modular invariant, which organizes the Hilbert space of the RCFT into irreducible modules of $\V$ and its complex conjugate.\footnote{For simplicity of exposition, we restrict ourselves to theories whose left-moving and right-moving chiral algebras are the same, though everything we say has a ``heterotic'' generalization.} It turns out that $\Cat$ distills just enough information from $\V$ that the problem of finding a physically consistent modular invariant can be phrased completely in terms of finding a special kind of ``algebra object'' in the abstract category $\Cat$, and otherwise does not depend on the precise choice of chiral algebra $\V$ which is used \cite{Fuchs:2002cm}. Equivalently, a modular invariant is defined by the choice of a topological surface operator in the TQFT corresponding to $\mathcal{C}$ \cite{Kapustin:2010if}, or, by folding, a topological boundary condition in the doubled theory $\mathcal{C}\boxtimes \overline{\mathcal{C}}$ (see Figure \ref{fig:KS}). The upshot is that these two parts may be cleanly separated from one another, and studied independently.

\begin{figure}
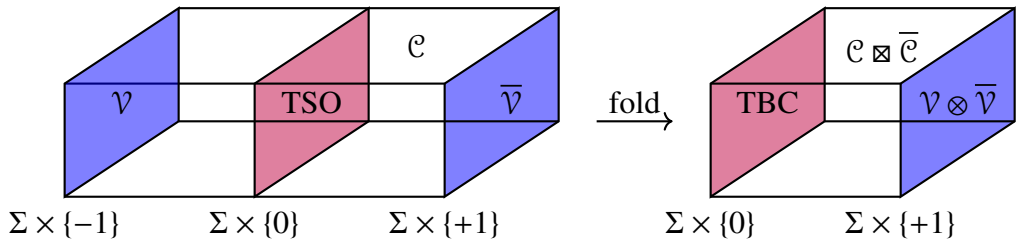

\ctikzfig{Figures/KS}
\caption{A (1+1)D RCFT on a Riemann surface $\Sigma$ may be thought of as a (2+1)D TQFT $\Cat$ on $\Sigma\times [-1,1]$ with a boundary condition corresponding to the chiral algebra $\V$ (resp.\ $\overline{\V}$) placed on $\Sigma\times \{-1\}$ (resp.\ $\Sigma \times \{+1\})$. The choice of a modular invariant corresponds to the choice of a topological surface operator (TSO) inserted at $\Sigma \times \{0\}$. Alternatively, this setup may be folded so that the topological surface operator corresponds to a topological boundary condition (TBC) in the doubled TQFT $\mathcal{C}\boxtimes \overline{\mathcal{C}}$.}\label{fig:KS}
\end{figure}

Now, fix $c^\ast$ to be a positive number, and $p^\ast$ to be a positive integer. In light of the discussion above, it is natural to organize the problem of classifying bosonic RCFTs as follows. \\

\noindent\textbf{Bosonic $(c^\ast,p^\ast)$-Classification Problem}: \emph{Classify all two-dimensional, unitary, bosonic rational conformal field theories whose (maximally extended) chiral algebra $\V$ belongs to a genus $(c,\Cat)$ with $c\leq  c^\ast$ and $\mathrm{rank}(\Cat)\leq p^\ast$.}\\

\noindent We will often abbreviate ``Bosonic $(c^\ast,p^\ast)$-Classification'' simply to ``$(c^\ast,p^\ast)$-Classification'' for the sake of brevity.

Attempting to classify (not necessarily rational) compact conformal field theories with $c\leq c^\ast$ is particularly natural in light of the monotonicity of the renormalization group flow \cite{Zamolodchikov:1986gt}: indeed, a successful classification of such theories would give one a ``menu'' of options for the IR fate of a UV CFT which is perturbed by relevant operators, so long as $c_{\mathrm{UV}}\leq c^\ast$. However, without further qualifiers, this problem is too challenging, in part because there are infinitely many (even continuously many) theories with $c\leq c^\ast$ whenever $c^\ast\geq 1$.\footnote{If $c^\ast<1$, there are only finitely many theories by the classification of minimal models \cite{belavin1984infinite,kac1978highest,feigin1984verma,friedan1984conformal,cappelli1987modular,cappelli1987ade}. However as soon as $c^\ast\geq 1$, there are continuously many theories, due e.g.\ to the existence of the $c=1$ Narain moduli space \cite{Ginsparg:1987eb}.} This is why we constrain $\mathrm{rank}(\Cat)$: physically, it is the number of (multiplets of) primary operators, and demanding that it be less than or equal to some fixed $p^\ast$ can (conjecturally, cf.\ Conjecture 3.5 of \cite{Hohn:2002dm}) be viewed as regulating the classification question so that its answer is a finite number of theories. 

In practice, there are at least two obstacles which limit how much progress one can make on the Bosonic $(c^\ast,p^\ast)$-Classification.
\begin{enumerate}
    \item The $(c^\ast,p^\ast)$-Classification is at least as difficult as classifying UMTCs with $\mathrm{rank}(\Cat)\leq p^\ast$. The current state of the art as far as UMTCs are concerned is $p^\ast=5$ \cite{rowell2009classification,hong2010classification,bruillard2016classification,ng2022reconstruction} (cf.\ Table \ref{tab:generaclassification} and Table \ref{tab:rank5UMTCs}), so that addressing higher values of $p^\ast$ will likely need to wait until there are more advancements in category theory/TQFT.
    \item The number of theories once one exceeds $c^\ast \sim 24$ grows quite dramatically. For example, a refinement \cite{kingmass} of the Smith--Minkowski--Siegel mass formula shows that there are more than a billion theories when $(c^\ast,p^\ast)=(32,1)$, making explicit enumeration a hopeless endeavor.\footnote{For theories with $\mathrm{rank}(\Cat)=1$, the central charge must be a multiple of $8$, which is why we have jumped from $c^\ast=24$ to $c^\ast=32$.} 
\end{enumerate}
We use these estimates to define the scope of our work. More specifically, we restrict ourselves in this paper to rational conformal field theories which belong to \emph{small genera}. That is, we make use of tools which are most effective when $c^\ast$ and $p^\ast$ are small: roughly, $c^\ast\sim 24$ and $p^\ast \sim 5$. When $p^\ast$ is this small, it is almost trivial to enumerate modular invariants, and thus, nearly all of the complexity of the Bosonic $(c^\ast,p^\ast)$-Classification lies in determining the possible choices of chiral algebra $\V$.

\begin{table}
\begin{small}
\begin{center}
% \begin{tabular}{c c}
\begin{tabular}[t]{c|c|c|c|c|c}
$p$ & $\Cat$ & $c~\mathrm{mod}~8$ & $0<c\leq 8$ $(\V_\Cat)$ & $8<c\leq 16$ & $16<c\leq 24$ \\\toprule
$1$ & $\Vec$ & $0$ & $\E[8,1]$ & $2$ & $(71)$  \\\midrule
$2$ & \hyperref[app:(A1,1)]{$(\AA_1,1)$} & $1$ & $\A[1,1]$  & $\A[1,1]\E[8,1]$ & $4$ \\
& \hyperref[app:(E7,1)]{$(\EE_7,1)$} & $7$ & $\E[7,1]$ & $2$ & $72$ \\
& \hyperref[app:(G2,1)]{$(\GG_2,1)$} & $\sfrac{14}5$ & $\G[2,1]$ & $\G[2,1]\E[8,1]$ & $8$\\
& \hyperref[app:(F4,1)]{$(\FF_4,1)$} & $\sfrac{26}5$ & $\F[4,1]$ & $2$ & $29$ \\\midrule
$3$ & \hyperref[app:(A2,1)]{$(\AA_2,1)$} & $2$ & $\A[2,1]$ & $\A[2,1]\E[8,1]$ & $6$ \\
& \hyperref[app:(E6,1)]{$(\EE_6,1)$} & $6$ & $\E[6,1]$ & $2$ & $53$ \\
& \hyperref[app:(B8,1)]{$(\BB_8,1)$} & $\sfrac12$ & $L_{\sfrac12}$ & $2$ & $9$ \\
& \hyperref[app:(A1,2)]{$(\AA_1,2)$} & $\sfrac32$ & $\A[1,2]$ & $2$ & $11$ \\
& \hyperref[app:(B2,1)]{$(\BB_2,1)$} & $\sfrac52$ & $\B[2,1]$ & $2$ & $16$ \\
& \hyperref[app:(B3,1)]{$(\BB_3,1)$} & $\sfrac72$ & $\B[3,1]$ & $2$ & $23$ \\
& \hyperref[app:(B4,1)]{$(\BB_4,1)$} & $\sfrac92$ & $\B[4,1]$ & $3$ & $42$ \\
& \hyperref[app:(B5,1)]{$(\BB_5,1)$} & $\sfrac{11}2$ & $\B[5,1]$ & $3$ & $72$ \\
& \hyperref[app:(B6,1)]{$(\BB_6,1)$} & $\sfrac{13}2$ & $\B[6,1]$ & $4$ & $158$ \\
& \hyperref[app:(B7,1)]{$(\BB_7,1)$} & $\sfrac{15}2$ & $\B[7,1]$ & $6$ & $(401)^\ast$ \\
& \hyperref[app:(A1,5)half]{$(\AA_1,5)_{\sfrac12}$} & $\sfrac{48}7$ & $\mathrm{Ex}(\E[6,1]L_{\sfrac67})$ & $3$ & $\skull$ \\
& \hyperref[app:C(A1,5)half]{$\overline{(\AA_1,5)}_{\sfrac12}$} & $\sfrac{8}{7}$ & $-$ & $\mathrm{Ex}(\A[1,5]\E[7,1])$ & $\skull$ \\\midrule
$4$ & \hyperref[app:(U1,4)]{$(\textsl{U}_1,4)$} & $1$ & $V_{2\mathbb{Z}}$ & $2$ & $10$  \\
& \hyperref[app:(A1,1)^2]{$(\AA_1,1)^{\boxtimes 2}$} & $2$ & $\A[1,1]^2$ & $2$ & $15$  \\ 
& \hyperref[app:(A3,1)]{$(\AA_3,1)$} & $3$ & $\A[3,1]$ & $2$ & $19$ \\
& \hyperref[app:(D4,1)]{$(\DD_4,1)$} & $4$ & $\D[4,1]$ & $2$ & $25$\\
& \hyperref[app:(D5,1)]{$(\DD_5,1)$} & $5$ & $\D[5,1]$ & $3$ & $62$\\
& \hyperref[app:(D6,1)]{$(\DD_6,1)$} & $6$ & $\D[6,1]$ & $4$ & $128$\\
& \hyperref[app:(D7,1)]{$(\DD_7,1)$} & $7$ & $\D[7,1]$ & $5$ & $265^\ast$ \\
& \hyperref[app:(D8,1)]{$(\DD_8,1)$} & $0$ & $\D[8,1]$ & $7$ & $(969)^\ast$\\
& \hyperref[app:(A1,3)]{$(\AA_1,3)$} & $\sfrac95$ & $\A[1,3]$ & $2$ & $13$\\
& \hyperref[app:(C3,1)]{$(\CC_3,1)$} & $\sfrac{21}5$ & $\C[3,1]$ & $3$ & $45$ \\
& \hyperref[app:(G2,2)]{$(\GG_2,2)$} & $\sfrac{14}3$ & $\G[2,2]$ & $3$ & $\skull$\\
& \hyperref[app:(A1,7)half]{$(\AA_1,7)_{\sfrac12}$} & $\sfrac{10}3$ & $\Ex(\A[1,1] \A[1,7])$ & $3$ & $52$ \\
& \hyperref[app:(G2,1)^2]{$(\GG_2,1)^{\boxtimes 2}$} & $\sfrac{28}5$ & $\G[2,1]^2$ & $4$ & $\skull$\\
& \hyperref[app:(F4,1)^2]{$(\FF_4,1)^{\boxtimes 2}$} & $\sfrac{12}5$ & $\mathrm{Ex}(\A[1,8])$ & $3$ & $26$ \\
& \hyperref[app:(A1,1)x(E7,1)]{$(\AA_1,1)\boxtimes (\EE_7,1)$} & $0$ & $\A[1,1]\E[7,1]$ & $5$ & $\skull$ \\
& \hyperref[app:(A1,1)x(G2,1)]{$(\AA_1,1)\boxtimes (\GG_2,1)$} & $\sfrac{19}5$ & $\A[1,1]\G[2,1]$ & $2$ & $43$\\
& \hyperref[app:(A1,1)x(F4,1)]{$(\AA_1,1)\boxtimes (\FF_4,1)$} & $\sfrac{31}5$ & $\A[1,1]\F[4,1]$ & $4$ & $\skull$ \\
& \hyperref[app:(G2,1)x(F4,1)]{$(\GG_2,1)\boxtimes (\FF_4,1)$} & $0$ & $\G[2,1]\F[4,1]$ & $8$ & $\skull$
\\\bottomrule
\end{tabular}
\caption{\small The complete list of $35$ unitary modular tensor categories $\Cat$ with $\mathrm{rank}(\Cat)\leq 4$. See Example \ref{ex:modularcategoriesrankleq4} for an explanation of the notation. For each $\Cat$, there are three genera $(c,\Cat)$ with $0<c\leq 24$, leading to a total of $105=3\times 35$ genera. When all chiral algebras in the genus $(c,\Cat)$ have been fully classified, we either report their number $|\mathrm{Gen}(c,\Cat)|$, or simply give the unique chiral algebra if $|\mathrm{Gen}(c,\Cat)|=1$. The notation $\mathrm{Ex}(\mathcal{V})$ means a conformal extension of $\mathcal{V}$, and $V_{2\ZZ}$ is the lattice VOA associated to the one-dimensional lattice generated by a vector with length-squared 4. A skull means that only a partial list of theories has been obtained. Parentheses $()$ indicate that the number we have reported depends on the conjecture that the monster CFT is the unique $c=24$ chiral CFT without any spin-1 currents. An asterisk $\ast$ denotes a case which was not considered in this paper, but which can be obtained by bosonizing the results of \cite{HohnMoller}. Each modular category $\Cat$ has a clickable hyperlink to the location in Appendix \ref{app:data} where its corresponding chiral algebras are listed.}\label{tab:generaclassification}
\end{center}
\end{small}
\end{table}

Of course,  progress has been made in this program already (see e.g.\ \cite{Schellekens:1992db,vanEkeren:2017scl,van2020dimension,vanEkeren:2020rnz,Hohn:2020xfe,Betsumiya:2022avv,Moller:2021clp,Moller:2019tlx,Mathur:1988na,Mason:2021xfs,Mukhi:2022bte,Grady:2020kls,Kaidi:2021ent,Bae:2021mej,Das:2022uoe,Das:2022slz,Das:2021uvd,Mukhi:2020gnj,Mukhi:2019xjy,Chandra:2018pjq,Hampapura:2016mmz,Gaberdiel:2016zke,Hampapura:2015cea} for some examples). The solution to the $(24,1)$-Classification was famously predicted by Schellekens in \cite{Schellekens:1992db}, and settled in more recent papers \cite{vanEkeren:2017scl,van2020dimension,vanEkeren:2020rnz,Hohn:2020xfe,Betsumiya:2022avv,Moller:2021clp,Moller:2019tlx}, modulo proving that the monster CFT is the unique VOA with one simple module, central charge $c=24$, and no continuous global symmetries. There is also the seminal work of Mathur--Mukhi--Sen \cite{Mathur:1988na} who, among other things, obtained the complete list of theories  in the $(8,2)$-Classification (see \cite{Mason:2021xfs} for the proof). More recently, the $(24,2)$-Classification was solved by the author and Sunil Mukhi in \cite{Mukhi:2022bte}. In this paper, we elaborate on the techniques of op.\ cit.\ to go most of the way towards solving the $(24,4)$-Classification.\footnote{In principle, the techniques we employ in this paper can be effectively pushed to $\mathrm{rank}(\Cat)=5$, and to even higher rank as well if one is satisfied with working genus-by-genus. We leave more detailed consideration to the future, see e.g.\ \S\ref{subsec:overview:future}.}\\

\noindent\textbf{Main result.} The tables in Appendix \ref{app:data} give the complete classification of bosonic, unitary chiral algebras in 95 of the 105 genera $(c,\Cat)$ with $c\leq 24$ and $\mathrm{rank}(\Cat)\leq 4$. As special cases, we obtain the classification of theories with three primaries and $c<\sfrac{120}7\sim 17.14$, and the classification of theories with four primaries and $c<\sfrac{62}3\sim 20.67$. See Table \ref{tab:generaclassification} for precisely which genera we classify as well as a consolidated presentation of the $c\leq 8$ classification, and Tables \ref{tab:clt16classificationpt0}--\ref{tab:clt16classificationpt3} for the classification of theories with $8<c\leq 16$.\\

\noindent We sketch the proof in \S\ref{subsec:overview:bosonic} and give the details in Appendix \ref{app:classification:clt8} and \ref{app:classification:clt24}.

This main result features a number of improvements to much of the literature on theories with multiple primary operators. First, we genuinely take up the problem of enumerating rational conformal field theories, as opposed to just candidate partition functions (though en route, we completely solve the modular bootstrap problem for theories with at most four primaries, see \S\ref{subsec:overview:characters}, Appendix \ref{app:vvmfs}, and Appendix \ref{app:data}). Second, we do not need to make any kind of assumption on the smallness of the Wronskian index defined in Equation \eqref{eqn:Wronskianindex}.\footnote{However, Equation \eqref{eqn:Wronskianindex} implies that chiral algebras with Wronskian index larger than $30$ do not arise in the $(24,4)$-Classification.} As we will see, freeing ourselves of this technical limitation dramatically increases the number of theories which are visible to our methods; in particular, we find more than $10^3$ chiral algebras in the 95 genera which we study!

With the main result in hand, we turn to applications. One qualitatively new feature which we encounter in the $(c^\ast,p^\ast)$-Classification when $p^\ast\geq 3$ is the ability to use fermionization techniques to make contact with fermionic CFTs. In particular, we explain in \S\ref{subsec:overview:fermions} how having control over bosonic chiral algebras with up to four primary operators translates to control over chiral fermionic RCFTs, i.e.\ left-moving fermionic theories which are physically consistent without the need for incorporating any right-moving degrees of freedom.\footnote{In mathematics, such theories are known as self-dual vertex operator superalgebras, but we opt to call them self-dual fermionic vertex operator algebras in order to avoid overloading the word ``super."\label{footnote:superalgebra}} Previous work in this direction includes a conjectural list of theories with central charge $c=16$ \cite{Kawai:1986vd}, a conjectural list of theories with $c\leq 15\sfrac12$ \cite{hohn1996self}, and a proven classification of theories with $c\leq 12$ \cite{Creutzig:2017fuk}. We extend these results as follows.  \\

\noindent \textbf{First application of the main result.} Appendix \ref{app:data:chiralfermionicRCFTs} gives the complete classification of unitary chiral fermionic RCFTs with central charge $c<23$. \\

\noindent We sketch the main ideas in \S\ref{subsec:overview:fermions}, and provide the full proof in Appendix \ref{app:classification:fermionic}.

The second application concerns generalized global symmetries. The last several years have witnessed a renaissance in our understanding of symmetry in quantum field theory, starting with the work of \cite{Gaiotto:2014kfa}. This program has arguably crystalized most clearly in the arena of (1+1)D QFTs where, in addition to ordinary group-like symmetries, one encounters also non-invertible topological defect lines whose mathematical structure is encoded in an object called a fusion category, see e.g.\ \cite{Chang:2018iay,Bhardwaj:2017xup} for more recent summaries.\footnote{In general, one may also encounter topological local operators corresponding to one-form symmetries, but these do not arise in CFTs with a unique vacuum. We make this standard assumption throughout the paper.\label{footnote:oneform}}  Still, even in this relatively simple and well-understood setting, there are only a handful of general constructions (e.g.\ Verlinde lines \cite{Verlinde:1988sn} and duality/Kramers--Wannier defects \cite{Frohlich:2004ef,Frohlich:2006ch}) and almost no CFTs for which the full structure of global symmetries is entirely known (setting aside the minimal models). Thus, it is of interest to develop general tools to systematically study the topological defect lines of two-dimensional conformal field theories. 

To this end, we emphasize in \S\ref{subsec:overview:symmetries} how the classification of chiral algebras by genus is related to the problem of classifying global symmetries in holomorphic VOAs. (A holomorphic VOA is essentially a completely chiral bosonic RCFT, see Definition \ref{defn:holomorphicVOA} for the precise statement.) This is achieved by using a concept which we call symmetry/subalgebra duality, which allows one to constrain the symmetries of a theory by bootstrapping the subalgebras of operators they preserve (see \S\ref{subsec:overview:symmetries} and Appendix  \ref{app:classification:symmetries} for more).

The following gives a taste of the kinds of results that we are able to obtain. Recall \cite{ostrik2003fusion} that there are exactly three unitary fusion categories of rank two: $\Vec_{\ZZ_2}$ (non-anomalous $\ZZ_2$ symmetry), $\Vec_{\ZZ_2}^\omega$ (anomalous $\ZZ_2$ symmetry), and $\textsl{Fib}$ (non-invertible Fibonacci symmetry with composition rule $\Phi\times\Phi=1+\Phi$).\\

\noindent \textbf{Second application of the main result.} Each of the three rank-2 unitary fusion categories --- $\Vec_{\ZZ_2}$, $\Vec_{\ZZ_2}^\omega$, and $\textsl{Fib}$ --- acts in exactly one way (up to conjugation by invertible symmetries) on the holomorphic VOA $\E[8,1]$. Furthermore, every modular category $\Cat$ with $\mathrm{rank}(\Cat)\leq 4$, except for possibly $\Cat\cong (\AA_1,5)_{\sfrac12}$ or $\overline{(\AA_1,5)}_{\sfrac12}$, acts in at least one way on $\E[8,1]$ as well, see Table \ref{tab:E8symmetries}. \\

\noindent It is also possible to obtain similar results for the rest of the holomorphic VOAs with $c\leq 24$ (see Table \ref{tab:E8^2symmetries} and Table \ref{tab:D16psymmetries} for the classification of rank-2 symmetries of the two holomorphic VOAs of central charge $16$), though we do not attempt to be as explicit. We refer readers to \cite{Lin:2019hks,Hegde:2021sdm,Burbano:2021loy} for prior work on the generalized symmetries of $\E[8,1]$ and the monster CFT. See also \cite{Betsumiya:2022avv} for work on invertible symmetries of holomorphic VOAs.

In the remainder of this introduction, we sketch our methods and motivations in slightly more detail. In particular, in \S\ref{subsec:overview:characters}, we describe how to use the Bantay--Gannon theory of vector-valued modular forms to completely solve the modular bootstrap problem for theories with at most four primary operators. We then explain in \S\ref{subsec:overview:bosonic} how to use the results of \S\ref{subsec:overview:characters}, coupled with additional vertex algebraic arguments, to complete the $(8,4)$-Classification, which in turn, we feed into a machine which churns out (most of) the $(24,4)$-Classification. In \S\ref{subsec:overview:fermions}, we explain how the classification of chiral fermionic RCFTs with $c< 23$ follows straightforwardly from the main result above upon using a chiral version of fermionization. We then describe symmetry/subalgebra duality and relate the $(24,4)$-Classification to the classification of non-invertible global symmetries of holomorphic vertex operator algebras  in \S\ref{subsec:overview:symmetries}. We conclude with an outlook in \S\ref{subsec:overview:future}.  \\

\noindent\textbf{Note added}: While this manuscript was being prepared, the author was informed that two other groups, one comprised of physicists and one comprised of mathematicians, were also studying chiral fermionic RCFTs/self-dual vertex operator superalgebras using methods which are complementary to the ones used in this work. Motivated by non-supersymmetric heterotic string theories in ten dimensions, Philip Boyle Smith, Ying-Hsuan Lin, Yuji Tachikawa, and Yunqin Zheng classify all such CFTs with central charge $c\leq 16$ in \cite{Tac23}.  Gerald H\"ohn and Sven M\"oller carry out the classification up to and including central charge $c=24$ in \cite{HohnMoller}. The author sincerely thanks both groups for their flexibility and their willingness to coordinate the submission of our articles to the arXiv on the same date. 

After this work appeared on the arXiv, the author was sent the PhD thesis \cite{nakorn}, which contains overlapping calculations. In particular, op.\ cit.\ studies 54 of the 70 genera $(c,\Cat)$ with $c\leq 16$ and $\mathrm{rank}(\Cat)\leq 4$ using methods similar to the ones employed in Appendix \ref{app:classification:clt8} of this paper. We caution that there appear to be a few chiral algebras which are missing in the classification of \cite{nakorn}, as well as one spurious theory.

\subsection{Characters}\label{subsec:overview:characters}

Many investigations of rational conformal field theories begin with the study of their partition functions and characters, which are highly constrained by modular invariance and covariance, respectively. Indeed, if $\V$ is a bosonic chiral algebra in the genus $(c,\Cat)$ with irreducible representations $\V(i)$, then its character vector $\ch(\tau)$, whose components are defined as
\begin{align}\label{eqn:charactervector}
    \ch_i(\tau) \equiv \mathrm{tr}_{\V(i)}q^{L_0-\sfrac{c}{24}},
\end{align}
transforms like a (weakly-holomorphic) weight zero vector-valued modular form \cite{zhu1996modular}, 
\begin{align}\label{eqn:chartransformation}
    \ch_i(\gamma\cdot\tau) = \sum_j\varrho(\gamma)_{i,j}\ch_j(\tau), \ \ \ \ \text{for all }\gamma\in\SL_2(\ZZ).
\end{align}
Here, $\varrho:\SL_2(\ZZ)\to\GL_p(\BBC)$ is a $p=\mathrm{rank}(\Cat)$ dimensional representation of the modular group which can be extracted from $\Cat$ and the central charge $c$ modulo $24$, and $\SL_2(\ZZ)$ acts on the upper half-plane $\mathbb{H}=\{\tau\in \BBC\mid \Im(\tau)>0\}$ in the standard way by fractional linear transformations, 
\begin{align}
\left(\begin{array}{cc} a & b \\ c & d\end{array}\right) \cdot \tau = \frac{a\tau +b}{c\tau+d}.
\end{align}
The qualifier ``weakly-holomorphic'' refers to the fact that $\ch_i(\tau)$ is holomorphic in the interior of $\mathbb{H}$, but may have poles as $\tau$ tends towards the cusps $\mathbb{Q}\cup\{i\infty\}$. We use the symbol $\mathcal{M}_0^!(\varrho)$ to refer to the space of weakly-holomorphic vector-valued functions which satisfy Equation \eqref{eqn:chartransformation}. 

Often, one's first line of attack is to classify functions in $\mathcal{M}_0^!(\varrho)$ which ``look like'' the characters of a chiral algebra in the genus $(c,\Cat)$, in the sense that they have $q$-expansions with completely non-negative coefficients and leading behavior of the form
\begin{align}\label{eqn:charform}
    \ch_i(\tau) = d_i  q^{h_i-\sfrac{c}{24}} + \cdots 
\end{align}
with $d_0=1$. The $h_i$ in Equation \eqref{eqn:charform} are interpreted as conformal dimensions, and hence are constrained by unitarity to satisfy $h_0=0$ and $h_i> 0$ if $i\neq 0$. We refer to a function which has these properties as being $(c,\Cat)$-admissible, or just admissible when the precise genus does not matter; if we relax positivity of the $q$-expansion, we call such a function quasi-admissible, inspired by the similar (but distinct) notion of quasi-character appearing in \cite{Chandra:2018pjq}.

It is worth emphasizing however that classifying $(c,\Cat)$-admissible functions is not equivalent to classifying chiral algebras in the genus $(c,\Cat)$. There are two reasons for this.
\begin{enumerate}
    \item There exist pairs of chiral algebras which are isospectral. For example, the two theories with $c=16$ and trivial modular category $\Cat\cong \Vec$ (i.e.\ the theories $\E[8,1]^2$ and $\D[16,1]^+$) both have character vector given by $\ch(\tau) = j(\tau)^{\sfrac23}$, where $j(\tau)=q^{-1}+744+196884q+\cdots$ is the Klein  $j$-invariant. Thus, one cannot always distinguish two theories by appealing to their genus one partition functions. This phenomenon is generic and not exclusive to theories with $\Rep(\V)=\Vec$: see e.g.\ \cite{Gaberdiel:2016zke,Mukhi:2022bte} for many examples of isospectral theories with two primaries, and Appendix \ref{app:data} for examples with up to four primaries. 
    \item There exist $(c,\Cat)$-admissible functions which nonetheless are not realized as the character vectors of any chiral algebra. For example, $j(\tau)+\mathcal{N}$ is realized as the character vector of a $c=24$ chiral algebra with trivial modular category for only finitely-many values of $\mathcal{N}$ \cite{Schellekens:1992db}, even though $j(\tau)+\mathcal{N}$ is $(24,\Vec)$-admissible for infinitely many values of $\mathcal{N}$. Again, this phenomenon is generic, and not exclusive to theories with $\Rep(\V)=\Vec$, as both prior work \cite{Harvey:2018rdc,Chandra:2018pjq,Mukhi:2022bte} and this work demonstrates. 
\end{enumerate}

Nonetheless, understanding the set of $(c,\Cat)$-admissible functions often gives one a useful first glimpse at the space of theories in $\Gen(c,\Cat)$. In particular, the constraints coming from modular covariance can sometimes be powerful enough that, coupled with conformal field theoretic arguments, a genuine classification of theories in $\Gen(c,\Cat)$ is possible. Indeed, we will see in \S\ref{subsec:overview:bosonic} that character-theoretic results are crucial input for classifying chiral algebras with $c\leq 8$ and $\mathrm{rank}(\Cat)\leq 4$. Thus, it is of general interest to be able to compute the spaces $\mathcal{M}_0^!(\varrho)$, and identify $(c,\Cat)$-admissible functions within them. 

We study this question in Appendix \ref{app:vvmfs}. In particular, we explain an effective and general theory of vector-valued modular forms due to Bantay and Gannon \cite{Bantay:2005vk,Bantay:2007zz,Gannon:2013jua} which characterizes, for general modular representation $\varrho$, the entire space $\mathcal{M}^!_0(\varrho)$ in terms of two $\varrho$-dependent $p\times p$ matrices $\lambda$ and $\chi$. 

The first matrix $\lambda$ is called a \emph{bijective exponent}. It is a diagonal matrix with the defining property that any function
\begin{align}
X_i(\tau) = q^{\lambda_{i,i}}\sum_{\substack{n\in \ZZ \\ n\gg-\infty}}X_{i,n}q^{n}
\end{align}
in the space $\mathcal{M}^!_0(\varrho)$ is completely determined by its Fourier coefficients $X_{i,n}$ with $n\leq 0$; we refer to these as the $\lambda$-singular coefficients of $X_i(\tau)$. The existence of a bijective exponent can be spiritually thought of as a kind of exact version of the Cardy formula \cite{Cardy:1986ie}. The Cardy formula shows that the asymptotic growth of the Fourier coefficients of a modular form is approximately determined by the leading singular behavior of its Fourier expansion. What a bijective exponent affords is  that the entire $q$-expansion, not just its asymptotics, is exactly determined once one specifies the coefficients  which appear in front of powers of $q$ which are more singular than the entries of $\lambda$. The second matrix $\chi_{i,j}$ encodes the first non-trivial Fourier coefficients of the unique functions $X^{(j,0)}_i(\tau)$ in $\mathcal{M}^!_0(\varrho)$ whose leading behavior takes the form
\begin{align}
X^{(j,0)}_i(\tau)=\delta_{i,j}q^{\lambda_{i,i}}+\chi_{i,j}q^{\lambda_{i,i}+1}+\cdots.
\end{align}
It is a non-trivial result of the Bantay--Gannon theory that the rest of the space can be efficiently and algorithmically computed just from knowledge of this minimal initial data. The laborious part is identifying $\lambda$ and $\chi$ for each modular representation in which one is interested. One of the contributions of the present work is the determination of these matrices for every modular representation $\varrho$ which arises in a theory with at most four primary operators.\\

\noindent \textbf{Main Result.} Appendix \ref{app:data} contains the matrices $\lambda$ and $\chi$ for each of the 105 modular representations obtained from genera $(c,\Cat)$ with $\mathrm{rank}(\Cat)\leq 4$. Thus, each corresponding space $\mathcal{M}_0^!(\varrho)$ may be completely and rapidly computed.\\

It turns out that once one has knowledge of $\lambda$ and $\chi$ for a given $\varrho$, it is nearly trivial to algorithmically extract the $(c,\Cat)$-quasi-admissible functions as a finite linear combination of certain distinguished basis elements of $\mathcal{M}_0^!(\varrho)$; imposing positivity is then typically straightforwardly done case-by-case. As a representative example, we find that any $(c,\Cat)$-quasi-admissible function with $c=\sfrac{38}{3}$ and $\Cat\cong (\GG_2,2):= \Rep(\G[2,2])$ must take the form
\begin{align}\label{eqn:charvectorg22} q^{\sfrac{c}{24}}\ch(\tau)= \left(\begin{array}{l}
1+q \left(22 \alpha+108\right)+q^2 \left(178 \alpha+6469\right)+q^3 \left(915 \alpha+116092\right)+\cdots \\ 
q^{\sfrac{2}{3}}\left(\left(28-2 \alpha\right)+q \left(3850-37 \alpha\right)+q^2 \left(89110-182 \alpha\right)+\cdots \right)  \\
q^{\sfrac{7}{9}}\left(\left(90-9 \alpha\right)+q \left(7641-108 \alpha\right)+q^2 \left(163134-567 \alpha\right)+\cdots \right)  \\
q^{\sfrac{1}{3}}\left(\alpha+q \left(32 \alpha+1610\right)+q^2 \left(209 \alpha+52256\right)+\cdots \right)  \end{array}\right)
\end{align}
where $\alpha$ is an integer, and $\G[2,2]$ is the affine Kac--Moody algebra corresponding to the simple Lie algebra $\G[2]$ at level $k=2$; positivity then requires that $0\leq \alpha \leq 10$. For the convenience of the reader, in Appendix \ref{app:data}, we write down the most general $(c,\Cat)$-quasi-admissible functions (in a form analogous to Equation \eqref{eqn:charvectorg22}) for every $(c,\Cat)$ with $c\leq 24$ and $\mathrm{rank}(\Cat)\leq 4$. We emphasize once more that, if one needed to, one could write down $(c,\Cat)$-quasi-admissible functions \emph{for any desired }$c$ (assuming $\mathrm{rank}(\Cat)\leq 4$) using the data of Appendix \ref{app:data}. It is in this sense that we claim to have complete control over the landscape of functions which could serve as the partition functions of rational conformal field theories with at most four primary operators.

Before moving on to the next subsection, we comment that one of the virtues of incorporating the Bantay--Gannon theory into our modular toolkit is that we do not need to make the kinds of simplifying assumptions on the Wronskian index, 
\begin{align}\label{eqn:Wronskianindex}
    \ell= \frac{p(p-1)}{2}+\frac{pc}{4}-6\sum_ih_i\in \ZZ^{\geq 0},
\end{align}
that one is practically required to make if one exclusively uses modular linear differential equations (MLDEs) to compute vector-valued modular forms ($p$ is the number of primary operators). Indeed, practitioners of MLDEs often assume that $\ell=0$, or at least that $\ell < 6$, in order to control the number of parameters which may appear in the MLDE (though see \cite{Chandra:2018pjq} for a way around this in the context of two-character theories). Thus, we are liberated to explore the full set of theories in a given genus, and are not restricted to a non-generic slice with small Wronskian index. Nonetheless, a researcher interested in the physics of extremal chiral algebras \cite{Grady:2020kls} may still classify admissible functions with small Wronskian index should they wish.

\subsection{Bosonic theories}\label{subsec:overview:bosonic}

With the results on vector-valued modular forms from the previous subsection at our disposal, we can begin to make progress on the $(24,4)$-classification. We split our task into two parts: the classification of theories with $0<c\leq 8$, and the classification of theories with $8<c\leq 24$.

\subsubsection{Theories with $0<c\leq 8$}

The first step in solving any physics problem is to understand the symmetries involved. This reaps especially high rewards in a two-dimensional conformal field theory, where a finite-dimensional continuous symmetry $\mathfrak{g}=\mathfrak{g}_1\oplus\cdots\oplus\mathfrak{g}_n\oplus \mathsf{U}_1^r$ enhances to an infinite-dimensional algebra generated by the corresponding Noether currents,
\begin{align}
    \mathcal{K}=\mathfrak{g}_{1,k_1}\otimes\cdots\otimes\mathfrak{g}_{n,k_n}\otimes \mathsf{U}_1^r.
\end{align}
The current algebra $\mathcal{K}$ is known as an affine Kac--Moody algebra, and the  $k_i$, which are quantized to non-negative integers in unitary theories, are called  levels. This affine Kac--Moody algebra in turn participates as a subalgebra of the chiral algebra $\V$, and one can separate the study of $\V$ into the study of $\mathcal{K}$ (the easy part which is determined by the continuous symmetries), and the complementary subalgebra $\tilde{\mathcal{K}}:= \Com_\V(\mathcal{K})$ of operators in $\V$ which commute with $\mathcal{K}$ (the hard part which does not have any continuous symmetries) \cite{gannon2019reconstruction}. In the nicest situations, most or all of $\V$ is determined by its continuous symmetries in the sense that the central charge of $\mathcal{K}$ nearly or entirely saturates that of $\V$. In particular, if the difference $c(\tilde{\mathcal{K}})=c(\V)-c(\mathcal{K})$ is less than 1, then $\tilde{\mathcal{K}}$ must be the chiral algebra of a minimal model, and $\V$ can then be straightforwardly determined as an extension of $\mathcal{K}\otimes \tilde{\mathcal{K}}$.

Let us now analyze the symmetries of the theories that we want to classify, with the ideas of the previous paragraph in mind. One salient feature of the data in Appendix \ref{app:data} is that, with one exception, there is a unique admissible function within each of the 35 genera $(c,\Cat)$ with $0<c\leq 8$ and $\mathrm{rank}(\Cat)\leq 4$ (see Theorem \ref{theorem:clt8characters}). In the one exceptional genus $(c,\Cat)=\big(\sfrac87,\overline{(\AA_1,5)}_{\sfrac12}\big)$, there is no admissible function, and hence no chiral algebras; one must increment the central charge by $8$ and consider the genus $\big(\sfrac{64}7,\overline{(\AA_1,5)}_{\sfrac12}\big)$ in order to find an admissible function, which turns out to be unique. (See Example \ref{ex:modularcategoriesrankleq4} for an explanation of the notation we use for modular categories.)

One immediate consequence of the preceeding paragraph is that, if a chiral algebra $\V$ exists in one of these genera, then the dimension of the Lie algebra of its global symmetry group is completely fixed by $\Cat$. Indeed, by a version of Noether's theorem, generators of this Lie algebra are in correspondence with dimension-1 operators in $\V_{h=1}$ (Noether currents); the precise number $N_\Cat$ of such Noether currents  can  be read off from the vacuum-component of the (unique) character vector, 
\begin{align}
    \ch_0(\tau) = q^{-\sfrac{c}{24}}(1+N_\Cat q+\cdots)
\end{align} 
which is reported for each $\Cat$ in Appendix \ref{app:data}. 

Thus, within each genus $(c,\Cat)$ with $0<c\leq 8$ and $\mathrm{rank}(\Cat)\leq 4$, we know the value of $\dim\V_1=N_\Cat$. It turns out that leveraging this information, and throwing the full kitchen sink of VOA constraints on $\V$ (see e.g.\ Theorem \ref{theorem:kmsubalgebra} and Lemma \ref{lem:E8subalgebra}) is essentially enough to uniquely determine the Kac--Moody subalgebra $\mathcal{K}_\Cat$ of any chiral algebra $\V$ in the genus $(c,\Cat)$. Furthermore, in each case, we find that the VOA gods are merciful in the sense that the Kac--Moody algebra really does comprise most of $\V$: in particular, the complementary subalgebra $\tilde{\mathcal{K}}$ always has central charge less than $1$, so that it is either trivial or the chiral algebra of a minimal model. It is then smooth sailing to the complete $(8,4)$-classification, in which we find that there is at most one VOA $\V_\Cat$ within each genus: Theorem \ref{theorem:clt8classification} says that the fourth column of Table \ref{tab:generaclassification} is the full list. This result can be thought of as a sort of VOA analog of the classic paper \cite{rowell2009classification}.

\subsubsection{Theories with $8<c\leq 24$}

The methods of the previous subsection generally become less effective the higher one goes in central charge. For example, the set of $(c,\Cat)$-admissible functions tends to grow with $c$; correspondingly, there may be many possibilities a priori for the Kac--Moody subalgebra of a chiral algebra in the genus $(c,\Cat)$, and it becomes quite laborious to analyze them all.

In Appendix \ref{app:classification:clt24}, we therefore leverage a different method to classify theories in the range $8<c\leq 24$. Following \cite{Mukhi:2022bte}, we refer to it as the gluing principle (see also \cite{hohn1996self,Gaberdiel:2016zke} for antecedents). The idea is the following (see Theorem \ref{theorem:gluingprinciple} for a more careful statement). Say that one is interested in the problem of classifying chiral algebras $\V$ in some genus $(c,\Cat)$. In order to achieve this, we fix once and for all a ``seed'' chiral algebra $\tilde{\V}$ in a complementary genus of the form $(\tilde c,\overline{\Cat})$, where $\overline{\Cat}$ is the modular category complex conjugate to $\Cat$ (see Definition \ref{defn:complexconjugate}). One can always convert any $\V\in\Gen(c,\Cat)$ into a holomorphic VOA $\mathcal{A}$ with central charge $C=c+\tilde{c}$ by ``gluing'' the fixed theory $\tilde{\V}$ onto it, 
\begin{align}
    \mathcal{A} = \bigoplus_i \V(i)^\ast\otimes\tilde{\V}(i),
\end{align}
where here, $\V(i)^\ast$ is the module which is charge-conjugate to $\V(i)$. In particular, $\tilde{\V}$ then sits inside of $\mathcal{A}$ as a primitive subalgebra (see Definition \ref{defn:primitivesubalgebra}), and one can then reconstruct $\V$ from $\mathcal{A}$ and the embedding $\iota:\tilde{\V}\hookrightarrow\mathcal{A}$ as 
\begin{align}
    \V\cong \Com_{\mathcal{A}}(\iota\tilde{\V})\equiv \mathcal{A}\big/ \iota \tilde{\V},
\end{align}
where $\Com_{\mathcal{A}}(\iota\tilde{\V})$, or $\mathcal{A}\big/\iota\tilde\V$, is the set of operators in $\mathcal{A}$ which commute with $\iota \tilde{\V}$. This construction is known as a coset to physicists \cite{Goddard:1984vk,Goddard:1986ee}, or as a commutant to mathematicians \cite{frenkel1992vertex}.

\begin{table}
\begin{center}
% \begin{tabular}{c c}
\begin{tabular}[t]{c|c|c|c|Hl}
$p$ & $\Cat$ & $\overline{\Cat}$ & $c$ & Coset & $\V$ \\\toprule
$1$ & $\Vec$ & $\Vec$ & $16$ & & $\hspace{-.07in}\begin{array}{l}\E[8,1]^2 \\ \D[16,1]^+\end{array}$ \\
\midrule
$2$ & $(\AA_1,1)$ & $(\EE_7,1)$ & $9$ & $\mathcal{A}^{(16)}\big/\E[7,1]$& $\A[1,1]\E[8,1]\cong \E[8,1]^2\big/(\E[7,1]\hookrightarrow\E[8,1])$ \\\cmidrule{2-6}
&$(\EE_7,1)$ & $(\AA_1,1)$ & $15$ &$\mathcal{A}^{(16)}\big/\A[1,1]$ & $\hspace{-.07in}\begin{array}{l}\E[7,1]\E[8,1]\cong \E[8,1]^2\big/(\A[1,1]\hookrightarrow\E[8,1])\\\mathrm{Ex}(\A[1,1]\D[14,1])\cong \D[16,1]^+\big/(\A[1,1]\hookrightarrow\D[16,1])\end{array}$ \\\cmidrule{2-6}
& $(\GG_2,1)$ & $(\FF_4,1)$ & $\sfrac{54}{5}$ &$\mathcal{A}^{(16)}\big/\F[4,1]$ & $\E[8,1]\G[2,1]\cong \E[8,1]^2\big/(\F[4,1]\hookrightarrow\E[8,1])$ \\\cmidrule{2-6}
 &$(\FF_4,1)$ & $(\GG_2,1)$ & $\sfrac{66}{5}$ &$\mathcal{A}^{(16)}\big/\G[2,1]$& $\hspace{-.07in}\begin{array}{l}\F[4,1]\E[8,1]\cong \E[8,1]^2\big/(\G[2,1]\hookrightarrow\E[8,1])\\\mathrm{Ex}(\B[12,1]L_{\sfrac{7}{10}})\cong \D[16,1]^+\big/(\G[2,1]\hookrightarrow\D[16,1])\end{array}$\\
\bottomrule
\end{tabular}
\caption{The classification of bosonic chiral algebras with $8<c\leq 16$ and $p=\mathrm{rank}(\Cat)\leq 2$. }\label{tab:clt16classificationpt0}
\end{center}
\end{table}

The upshot is that every $\mathcal{V}\in\Gen(c,\Cat)$ is expressible as a coset/commutant of a holomorphic VOA $\mathcal{A}$ of central charge $c+\tilde{c}$ by the fixed seed theory $\tilde{\V}$. In fact, assuming the widely-believed conjecture that the coset of two rational theories is again rational, it develops that the reverse is true as well: every coset of the form $\Com_{\mathcal{A}}(\iota\tilde{\V})$, with $\tilde{\V}$ primitively embedded into $\mathcal{A}$, defines a theory in $\mathrm{Gen}(c,\Cat)$.\footnote{It turns out that, with the exception of the genus $\big(\sfrac{58}3,(\AA_1,7)_{\sfrac12}\big)$, our classification results do not rely on this conjecture being true because we are able to confirm the rationality/regularity of our theories in every case by showing that they have strongly regular conformal subVOAs.}

Now, if $c+\tilde{c}\leq 24$, then the possibilities for $\mathcal{A}$ are completely known (modulo the uniqueness of the moonshine module) from \cite{Schellekens:1992db}; similarly, natural choices for the seed theory $\tilde{\V}$ when $\mathrm{rank}(\Cat)\leq 4$ are provided by the theories in the fourth column of Table \ref{tab:generaclassification}. Therefore, we are reduced to the problem of studying equivalence classes of primitive embeddings of the chiral algebras in Table \ref{tab:generaclassification} into holomorphic VOAs with central charge less than or equal to $24$. This problem turns out to be solvable for many genera; it is the most tractable when the seed theory $\tilde{\V}$ is an affine Kac--Moody algebra, because in this case, the problem of enumerating embeddings of $\tilde{\V}$ reduces entirely to the problem of calculating embeddings of ordinary Lie algebras, a problem for which many tools are available (we relied strongly on \cite{de2011constructing,Feger:2012bs,Feger:2019tvk} and their corresponding computer packages). Tables \ref{tab:clt16classificationpt0}--\ref{tab:clt16classificationpt3} give the complete list of theories with $8<c\leq 16$ and at most four primary operators, as well as their (non-unique) expressions as cosets of holomorphic VOAs. The tables in Appendix \ref{app:data} provide further details on these theories (like their characters) and extend the classification to many genera $(c,\Cat)$ in the range $16<c\leq 24$.

\begin{table}
\begin{center}
% \begin{tabular}{c c}
\begin{tabular}[t]{H c|c|c|Hl}
$p$ & $\Cat$ & $\overline{\Cat}$ & $c$ & Coset & $\V$ \\\toprule
$3$ & $(\AA_2,1)$ & $(\EE_6,1)$ & $10$&$\mathcal{A}^{(16)}\big/\E[6,1]$ & $\A[2,1]\E[8,1]\cong \E[8,1]^2\big/(\E[6,1]\hookrightarrow\E[8,1])$ \\\cmidrule{2-6}
 & $(\EE_6,1)$ & $(\AA_2,1)$ & $14$ &$\mathcal{A}^{(16)}\big/\A[2,1]$ & \hspace{-.07in}$\begin{array}{l}\E[6,1]\E[8,1] \cong \E[8,1]^2\big/(\A[2,1]\hookrightarrow\E[8,1])\\ \mathrm{Ex}(\D[13,1]\mathsf{U}_{1,12})\cong \D[16,1]^+\big/(\A[2,1]\hookrightarrow\D[16,1])\end{array}$ \\\cmidrule{2-6}
 & $(\BB_8,1)$ & $(\BB_7,1)$ & $\sfrac{17}2$ & $\mathcal{A}^{(16)}\big/\B[7,1]$ &  \hspace{-.07in}$\begin{array}{l}\E[8,1]L_{\sfrac12}\cong \E[8,1]^2\big/(\B[7,1]\hookrightarrow\E[8,1])\\\B[8,1]\cong \D[16,1]^+\big/(\B[7,1]\hookrightarrow\D[16,1])\end{array}$ \\\cmidrule{2-6}
 & $(\AA_1,2)$ & $(\BB_6,1)$ & $\sfrac{19}2$ & $\mathcal{A}^{(16)}\big/\B[6,1]$ & \hspace{-.07in}$\begin{array}{l}\A[1,2]\E[8,1]\cong \E[8,1]^2\big/(\B[6,1]\hookrightarrow\E[8,1])\\ \B[9,1]\cong \D[16,1]^+\big/(\B[6,1]\hookrightarrow\D[16,1])\end{array}$ \\\cmidrule{2-6}
 & $(\BB_2,1)$ & $(\BB_5,1)$ & $\sfrac{21}2$ & $\mathcal{A}^{(16)}\big/\B[5,1]$ & \hspace{-.07in}$\begin{array}{l}\B[2,1]\E[8,1]\cong\E[8,1]^2\big/(\B[5,1]\hookrightarrow\E[8,1])\\ \B[10,1]\cong\D[16,1]^+\big/(\B[5,1]\hookrightarrow\D[16,1])\end{array}$ \\\cmidrule{2-6}
 & $(\BB_3,1)$ & $(\BB_4,1)$ & $\sfrac{23}2$ & $\mathcal{A}^{(16)}\big/\B[4,1]$ & \hspace{-.07in}$\begin{array}{l}\B[3,1]\E[8,1]\cong \E[8,1]^2\big/(\B[4,1]\hookrightarrow\E[8,1])\\ \B[11,1]\cong\D[16,1]^+\big/(\B[4,1]\hookrightarrow\D[16,1])\end{array}$ \\\cmidrule{2-6}
 & $(\BB_4,1)$ & $(\BB_3,1)$ & $\sfrac{25}2$ & $\mathcal{A}^{(16)}\big/\B[3,1]$& \hspace{-.07in}$\begin{array}{l}\B[4,1]\E[8,1]\cong \E[8,1]^2\big/(\B[3,1]\hookrightarrow\E[8,1])\\\B[12,1]\cong \D[16,1]^+\big/(\B[3,1]\xhookrightarrow{1}\D[16,1])\\ \mathrm{Ex}(\D[12,1]L_{\sfrac12})\cong \D[16,1]^+\big/(\B[3,1]\xhookrightarrow{2}\D[16,1])\end{array}$ \\\cmidrule{2-6} 
  & $(\BB_5,1)$ & $(\BB_2,1)$ & $\sfrac{27}2$ & $\mathcal{A}^{(16)}\big/\B[2,1]$&
  \hspace{-.07in}$\begin{array}{l}\B[5,1]\E[8,1]\cong \E[8,1]^2\big/(\B[2,1]\hookrightarrow\E[8,1])\\\B[13,1]\cong \D[16,1]^+\big/(\B[2,1]\xhookrightarrow{1}\D[16,1])\\ \mathrm{Ex}(\A[1,2]\D[12,1])\cong\D[16,1]^+\big/(\B[2,1]\xhookrightarrow{2}\D[16,1])\end{array}$ \\\cmidrule{2-6}
  & $(\BB_6,1)$ & $(\AA_1,2)$ & $\sfrac{29}2$ & $\mathcal{A}^{(16)}\big/\A[1,2]$& \hspace{-.07in}$\begin{array}{l}\B[6,1]\E[8,1]\cong \E[8,1]^2\big/(\A[1,2]\hookrightarrow\E[8,1])\\ \mathrm{Ex}(\E[7,1]^2L_{\sfrac12})\cong \E[8,1]^2\big/(\A[1,2]\hookrightarrow\E[8,1]^2) \\ \B[14,1]\cong \D[16,1]^+\big/(\A[1,2]\xhookrightarrow{1}\D[16,1])\\ \mathrm{Ex}(\B[2,1]\D[12,1])\cong \D[16,1]^+\big/(\A[1,2]\xhookrightarrow{2}\D[16,1])\end{array}$ \\\cmidrule{2-6}
   & $(\BB_7,1)$ & $(\BB_8,1)$ & $\sfrac{31}2$ & $\mathcal{A}^{(24)}\big/\B[8,1]$ & \hspace{-.07in}$\begin{array}{l}\B[7,1]\E[8,1]\cong\D[16,1]^+\E[8,1]\big/(\B[8,1]\hookrightarrow\D[16,1]) \\ \B[15,1]\cong\mathbf{S}(\D[24,1])\big/(\B[8,1]\hookrightarrow\D[24,1]) \\ \mathrm{Ex}(\B[3,1]\D[12,1])\cong\mathbf{S}(\D[12,1]^2)\big/(\B[8,1]\hookrightarrow\D[12,1])\\\E[8,2]\cong \mathbf{S}(\E[8,2]\B[8,1])\big/(\B[8,1]\hookrightarrow\B[8,1])\\\mathrm{Ex}(\A[15,1]L_{\sfrac12})\cong\mathbf{S}(\A[15,1]\D[9,1])\big/(\B[8,1]\hookrightarrow\D[9,1])\\\mathrm{Ex}(\A[1,2]\E[7,1]^2)\cong\mathbf{S}(\D[10,1]\E[7,1]^2)\big/(\B[8,1]\hookrightarrow\D[10,1])\end{array}$ \\\cmidrule{2-6}
 & $(\AA_1,5)_{\sfrac12}$ & $\overline{(\AA_1,5)}_{\sfrac12}$ & $\sfrac{104}{7}$ & $\mathcal{A}^{(24)}\big/\mathrm{Ex}(\A[1,5]\E[7,1])$ & \hspace{-.07in}$\begin{array}{l} \mathrm{Ex}(\E[6,1]L_{\sfrac67})\E[8,1]\cong \E[8,1]^3\big/\mathrm{Ex}(\A[1,5]\E[7,1]) \\
 \ \ \ \ (\A[1,5]\E[7,1]\hookrightarrow[\A[1,1]\E[7,1]][\A[1,4]]'\hookrightarrow[\E[8,1]][\E[8,1]]') \\
 \mathrm{Ex}(\D[13,1]\mathsf{U}_{1,12}L_{\sfrac67})\cong\D[16,1]^+\E[8,1]\big/\mathrm{Ex}(\A[1,5]\E[7,1])\\
 \ \ \ \ 
 (\A[1,5]\E[7,1]\hookrightarrow[\A[1,1]\E[7,1]][\A[1,4]]'\hookrightarrow[\E[8,1]][\D[16,1]]') \\
 \mathrm{Ex}(\B[8,1]\F[4,1]\mathsf{W}^{\A[1]}_{2,3})\cong \mathbf{S}(\D[10,1]\E[7,1]^2)\big/\mathrm{Ex}(\A[1,5]\E[7,1])\\
  \ \ \ \ (\A[1,5]\E[7,1]\hookrightarrow [\A[1,2]][\A[1,3]]'[\E[7,1]]''\hookrightarrow [\D[10,1]][\E[7,1]]'[\E[7,1]]'')
 \end{array}$\\ \cmidrule{2-6}
 & $\overline{(\AA_1,5)}_{\sfrac12}$ & $(\AA_1,5)_{\sfrac12}$ & $\sfrac{64}7$ & $\mathcal{A}^{(16)}\big/\mathrm{Ex}(\E[6,1]L_{\sfrac67})$  & $\mathrm{Ex}(\A[1,5]\E[7,1])$ \\
\bottomrule
\end{tabular}
\caption{The classification of bosonic chiral algebras with $8<c\leq 16$ and $p=\mathrm{rank}(\Cat)= 3$.}\label{tab:clt16classificationpt1}
\end{center}
\end{table}

\begin{table}
\begin{center}
\begin{tabular}[t]{Hc|c|c|Hl}
$p$ & $\Cat$ & $\overline{\Cat}$ & $c$ & Coset & $\V$ \\\toprule
 $4$ & $(\textsl{U}_1,4)$ & $(\DD_7,1)$ & $9$ & $\mathcal{A}^{(16)}\big/\D[7,1]$ & \hspace{-.07in}$\begin{array}{l}\E[8,1]V_{2\ZZ}\cong\E[8,1]^2\big/(\D[7,1]\hookrightarrow\E[8,1])\\ \D[9,1]\cong\D[16,1]^+\big/(\D[7,1]\hookrightarrow\D[16,1])\end{array}$ \\\cmidrule{2-6}
 & $(\AA_1,1)^{\boxtimes 2}$ & $(\DD_6,1)$ & $10$ & $\mathcal{A}^{(16)}\big/\D[6,1]$ & $\hspace{-.07in}\begin{array}{l}\A[1,1]^2\E[8,1]\cong \E[8,1]^2\big/ (\D[6,1]\hookrightarrow\D[6,1])\\ \D[10,1]\cong \D[16,1]^+\big/(\D[6,1]\hookrightarrow\D[16,1])\end{array}$ \\ \cmidrule{2-6}
 & $(\AA_3,1)$ & $(\DD_5,1)$ & $11$ & $\mathcal{A}^{(16)}\big/\D[5,1]$ & \hspace{-.07in}$\begin{array}{l}\A[3,1]\E[8,1]\cong \E[8,1]^2\big/(\D[5,1]\hookrightarrow\E[8,1])\\\D[11,1]\cong\D[16,1]^+\big/(\D[5,1]\hookrightarrow\D[16,1])\end{array}$  \\\cmidrule{2-6}
 & $(\DD_4,1)$ & $(\DD_4,1)$ & $12$ & $\mathcal{A}^{(16)}\big/\D[4,1]$ & \hspace{-.07in}$\begin{array}{l}\D[4,1]\E[8,1]\cong \E[8,1]^2\big/(\D[4,1]\hookrightarrow\E[8,1])\\\D[12,1]\cong\D[16,1]^+\big/(\D[4,1]\hookrightarrow\D[16,1])\end{array}$ \\\cmidrule{2-6}
 & $(\DD_5,1)$ & $(\AA_3,1)$ & $13$ & $\mathcal{A}^{(16)}\big/\A[3,1]$ & \hspace{-.07in}$\begin{array}{l}\D[5,1]\E[8,1]\cong\E[8,1]^2\big/(\A[3,1]\hookrightarrow\E[8,1])\\ \D[13,1]\cong\D[16,1]^+\big/(\A[3,1]\xhookrightarrow{1}\D[16,1])\\ \mathrm{Ex}(\D[12,1]\mathsf{U}_{1,4})\cong\D[16,1]^+\big/(\A[3,1]\xhookrightarrow{2}\D[16,1])\end{array}$ \\\cmidrule{2-6}
 & $(\DD_6,1)$ & $(\AA_1,1)^{\boxtimes 2}$ & $14$ & $\mathcal{A}^{(16)}\big/\A[1,1]^2$ & \hspace{-.07in}$\begin{array}{l}\D[6,1]\E[8,1]\cong\E[8,1]^2\big/(\A[1,1]^2\hookrightarrow\E[8,1])\\ \E[7,1]^2\cong\E[8,1]^2\big/([\A[1,1]][\A[1,1]]'\hookrightarrow[\E[8,1]][\E[8,1]]') \\ \D[14,1]\cong\D[16,1]^+\big/(\A[1,1]^2\xhookrightarrow{1}\D[16,1])\\ \mathrm{Ex}(\D[12,1]\A[1,1]^2)\cong\D[16,1]^+\big/(\A[1,1]^2\xhookrightarrow{2}\D[16,1])\end{array}$ \\\cmidrule{2-6}
 & $(\DD_7,1)$ & $(\textsl{U}_1,4)$ & $15$ & $\mathcal{A}^{(24)}\big/\D[9,1]$ & \hspace{-.07in}$\begin{array}{l}\mathrm{Ex}(\A[15,1])\cong \mathbf{S}(\A[15,1]\D[9,1])\big/(\D[9,1]\hookrightarrow\D[9,1])\\
 \mathrm{Ex}(\E[7,1]^2\mathsf{U}_{1,4})\cong \mathbf{S}(\D[10,1]\E[7,1]^2)\big/(\D[9,1]\hookrightarrow\D[10,1])\\
 \mathrm{Ex}(\A[3,1]\D[12,1])\cong \mathbf{S}(\D[12,1]^2)\big/(\D[9,1]\hookrightarrow\D[12,1])\\
\D[7,1]\E[8,1]\cong\D[16,1]^+\E[8,1]\big/(\D[9,1]\hookrightarrow\D[16,1])\\ \D[15,1]\cong\mathbf{S}(\D[24,1])\big/(\D[9,1]\hookrightarrow\D[24,1])\end{array}$ \\\cmidrule{2-6}
 & $(\DD_8,1)$ & $(\DD_8,1)$ & $16$ & $\mathcal{A}^{(24)}\big/\D[8,1]$ & \hspace{-.07in}$\begin{array}{l}\mathrm{Ex}(\D[8,1]^2)\cong \mathbf{S}(\D[8,1]^3)\big/(\D[8,1]\hookrightarrow\D[8,1])\\
 \mathrm{Ex}(\E[8,2]L_{\sfrac12})\cong \mathbf{S}(\E[8,2]\B[8,1])\big/(\D[8,1]\hookrightarrow\B[8,1])\\
 \mathrm{Ex}(\A[15,1]\mathsf{U}_{1,4})\cong\mathbf{S}(\A[15,1]\D[9,1])\big/(\D[8,1]\hookrightarrow\D[9,1])\\
 \mathrm{Ex}(\A[1,1]^2\E[7,1]^2)\cong \mathbf{S}(\D[10,1]\E[7,1]^2)\big/(\D[8,1]\hookrightarrow\D[10,1])\\
 \mathrm{Ex}(\D[4,1]\D[12,1])\cong \mathbf{S}(\D[12,1]^2)\big/(\D[8,1]\hookrightarrow\D[12,1])\\
 \D[8,1]\E[8,1]\cong \D[16,1]^+\E[8,1]\big/(\D[8,1]\hookrightarrow\D[16,1])\\
 \D[16,1]\cong \mathbf{S}(\D[24,1])\big/(\D[8,1]\hookrightarrow\D[24,1])\end{array}$ \\
\bottomrule
\end{tabular}
\caption{The classification of bosonic chiral algebras with $8<c\leq 16$ and $p=\mathrm{rank}(\Cat)=4$, part 1.}\label{tab:clt16classificationpt2}
\end{center}
\end{table}

\begin{small}
\begin{table}
\begin{center}
\begin{tabular}[t]{Hc|c|c|Hl}
$p$ & $\Cat$ & $\overline{\Cat}$ & $c$ & Coset & $\V$ \\\toprule
$4$ & $(\AA_1,3)$ & $(\AA_1,1)\boxtimes (\FF_4,1)$ & $\sfrac{49}5$ & $\mathcal{A}^{(16)}\big/\A[1,1]\F[4,1]$ & \hspace{-.07in}$\begin{array}{l}\E[7,1]\G[2,1]\cong \E[8,1]^2\big/([\A[1,1]][\F[4,1]]'\hookrightarrow[\E[8,1]][\E[8,1]]')\\ \E[8,1]\A[1,3]\cong \E[8,1]^2\big/(\A[1,1]\F[4,1]\hookrightarrow \E[8,1])\end{array}$ \\\cmidrule{2-6}
 & $(\CC_3,1)$ & $(\AA_1,1)\boxtimes (\GG_2,1)$ & $\sfrac{61}5$  & $\mathcal{A}^{(16)}\big/\A[1,1]\G[2,1]$ & \hspace{-.07in}$\begin{array}{l}\E[7,1]\F[4,1]\cong \E[8,1]^2\big/([\A[1,1]][\G[2,1]]'\hookrightarrow [\E[8,1]][\E[8,1]]') \\ \C[3,1]\E[8,1]\cong \E[8,1]^2\big/(\A[1,1]\G[2,1]\hookrightarrow\E[8,1])\\ \mathrm{Ex}(\A[1,1]\B[10,1]L_{\sfrac{7}{10}})\cong\D[16,1]^+\big/(\A[1,1]\G[2,1]\hookrightarrow\D[16,1])\end{array}$ \\\cmidrule{2-6}
 & $(\GG_2,2)$ & $(\AA_1,7)_{\sfrac12}$ & $\sfrac{38}3$  & $\mathcal{A}^{(16)}\big/\mathrm{Ex}(\A[1,1]\A[1,7])$ &
 \hspace{-.07in}$\begin{array}{l}\G[2,2]\E[8,1]\cong \E[8,1]^2\big/(\mathrm{Ex}(\A[1,1]\A[1,7])\hookrightarrow\E[8,1]) \\
 \mathrm{Ex}(\E[6,1]\F[4,1]\mathsf{W}^{\A[1]}_{3,4})\cong \E[8,1]^2\big/(\mathrm{Ex}(\A[1,1]\A[1,7])\hookrightarrow\E[8,1]^2)\\
  \ \ \ \ \ (\A[1,1]\A[1,7]\hookrightarrow [\A[1,1]\A[1,3]][\A[1,4]]'\hookrightarrow [\E[8,1]][\E[8,1]]') \\
 \mathrm{Ex}(\A[1,6]\B[9,1]L_{\sfrac{11}{12}})\cong \D[16,1]^+\big/(\mathrm{Ex}(\A[1,1]\A[1,7])\hookrightarrow \D[16,1]^+) 
 \end{array}$\\\cmidrule{2-6}
 & $(\AA_1,7)_{\sfrac12}$ & $(\GG_2,2)$ &  $\sfrac{34}3$ & $\mathcal{A}^{(16)}\big/\G[2,2]$  & \hspace{-.07in}$\begin{array}{l}\mathrm{Ex}(\A[1,1]\A[1,7])\E[8,1]\cong \E[8,1]^2\big/(\G[2,2]\hookrightarrow \E[8,1])\\ \mathrm{Ex}(\F[4,1]^2L_{\sfrac{14}{15}})\cong \E[8,1]^2\big/(\G[2,2]\hookrightarrow \E[8,1]^2)\\ {^\ast}\mathrm{Ex}(\D[9,1]L_{\sfrac7{10}}^2L_{\sfrac{14}{15}})\cong \D[16,1]^+\big/(\G[2,2]\hookrightarrow\D[16,1])\end{array}$ \\\cmidrule{2-6}
 & $(\GG_2,1)^{\boxtimes 2}$ & $(\FF_4,1)^{\boxtimes 2}$ & $\sfrac{68}5$  & $\mathcal{A}^{(24)}\big/\F[4,1]^2$ & \hspace{-.07in}$\begin{array}{l}\mathrm{Ex}(\C[8,1])\cong\mathbf{S}(\C[8,1]\F[4,1]^2)\big/(\F[4,1]^2\hookrightarrow\F[4,1]^2)\\ \mathrm{Ex}(\E[6,1]^2L_{\sfrac45}^2)\cong \mathbf{S}(\E[6,1]^4)\big/(\F[4,1]^2\hookrightarrow\E[6,1]^2)\\ \mathrm{Ex}(\A[1,3]^2\D[10,1])\cong \mathbf{S}(\D[10,1]\E[7,1]^2)\big/(\F[4,1]^2\hookrightarrow\E[7,1]^2)\\ \E[8,1]\G[2,1]^2\cong \E[8,1]^3\big/(\F[4,1]^2\hookrightarrow\E[8,1]^2)\end{array}$ \\\cmidrule{2-6}
 & $(\FF_4,1)^{\boxtimes 2}$ & $(\GG_2,1)^{\boxtimes 2}$ & $\sfrac{52}5$   & $\mathcal{A}^{(16)}\big/\G[2,1]^2$ & \hspace{-.07in}$\begin{array}{l}\F[4,1]^2\cong \E[8,1]^2\big/(\G[2,1]^2\hookrightarrow\E[8,1]^2)\\ \E[8,1]\mathrm{Ex}(\A[1,8])\cong \E[8,1]^2\big/(\G[2,1]^2\hookrightarrow \E[8,1])\\ \mathrm{Ex}(\D[9,1]L_{\sfrac7{10}}^2)\cong \D[16,1]^+\big/(\G[2,1]^2\hookrightarrow\D[16,1])\end{array}$ \\\cmidrule{2-6}
 & $(\AA_1,1)\boxtimes (\EE_7,1)$ & $(\AA_1,1)\boxtimes (\EE_7,1)$ & $16$ & $\mathcal{A}^{(24)}\big/\A[1,1]\E[7,1]$  & \hspace{-.07in}$\begin{array}{l}\mathrm{Ex}(\D[6,1]\D[10,1])\cong \mathbf{S}(\D[10,1]\E[7,1]^2)\big/(\A[1,1]\E[7,1]\hookrightarrow \E[7,1]^2)\\ \mathrm{Ex}(\A[1,1]\D[8,1]\E[7,1])\cong \mathbf{S}(\D[10,1]\E[7,1]^2)\big/(\A[1,1]\E[7,1]\hookrightarrow \D[10,1]\E[7,1])\\ \mathrm{Ex}(\A[15,1]\mathsf{U}_{1,144})\cong \mathbf{S}(\A[17,1]\E[7,1])\big/(\A[1,1]\E[7,1]\hookrightarrow\A[17,1]\E[7,1])\\
 \A[1,1]\E[7,1]\E[8,1]\cong \E[8,1]^3\big/(\A[1,1]\E[7,1]\hookrightarrow \E[8,1]^2)\\ \mathrm{Ex}(\A[1,1]^2\D[14,1])\cong \D[16,1]^+\E[8,1]\big/(\A[1,1]\E[7,1]\hookrightarrow\D[16,1]\E[8,1])\end{array}$ \\\cmidrule{2-6}
 & $(\AA_1,1)\boxtimes (\GG_2,1)$ & $(\CC_3,1)$ & $\sfrac{59}5$  & $\mathcal{A}^{(24)}\big/ \E[7,1]\F[4,1]$ & \hspace{-.07in}$\begin{array}{l}\mathrm{Ex}(\A[1,3]\D[10,1])\cong \mathbf{S}(\D[10,1]\E[7,1]^2)\big/(\E[7,1]\F[4,1]\hookrightarrow\E[7,1]^2) \\ \A[1,1]\E[8,1]\G[2,1]\cong \E[8,1]^3\big/(\E[7,1]\F[4,1]\hookrightarrow\E[8,1]^2)\end{array}$ \\\cmidrule{2-6}
 & $(\AA_1,1)\boxtimes (\FF_4,1)$ & $(\AA_1,3)$ & $\sfrac{71}5$  & $\mathcal{A}^{(24)}\big/\E[7,1]\G[2,1]$ & \hspace{-.07in}$\begin{array}{l}\mathrm{Ex}(\C[3,1]\D[10,1])\cong \mathbf{S}(\D[10,1]\E[7,1]^2)\big/(\E[7,1]\G[2,1]\hookrightarrow \E[7,1]^2) \\ \mathrm{Ex}(\B[6,1]\E[7,1]L_{\sfrac{7}{10}})\cong \mathbf{S}(\D[10,1]\E[7,1]^2)\big/(\E[7,1]\G[2,1]\hookrightarrow \E[7,1]\D[10,1]) \\
 \E[8,1]\A[1,1]\F[4,1]\cong \E[8,1]^3\big/(\E[7,1]\G[2,1] \hookrightarrow\E[8,1]^2)\\
 \mathrm{Ex}(\A[1,1]\B[12,1] L_{\sfrac{7}{10}})\cong \D[16,1]^+\E[8,1]\big/(\E[7,1]\G[2,1]\hookrightarrow\E[8,1]\D[16,1])\end{array}$   \\\cmidrule{2-6}
 & $(\GG_2,1)\boxtimes (\FF_4,1)$ & $(\GG_2,1)\boxtimes (\FF_4,1)$ & $16$ & $\mathcal{A}^{(24)}\big/\G[2,1]\F[4,1]$  & \hspace{-.07in}$\begin{array}{l}\mathrm{Ex}(\C[8,1]\A[1,8])\cong \mathbf{S}(\C[8,1]\F[4,1]^2)\big/(\G[2,1]\F[4,1]\hookrightarrow\F[4,1]^2) \\
 \mathrm{Ex}(\E[7,2]\A[1,1]^2L_{\sfrac{7}{10}})\cong\mathbf{S}(\E[7,2]\B[5,1]\F[4,1])\big/(\F[4,1]\G[2,1]\hookrightarrow\F[4,1]\B[5,1]) \\
 \mathrm{Ex}(\E[6,1]^2\A[2,2]L_{\sfrac45})\cong\mathbf{S}(\E[6,1]^4)\big/(\F[4,1]\G[2,1]\hookrightarrow\E[6,1]^2) \\
 \mathrm{Ex}(\A[11,1]\B[3,1]L_{\sfrac7{10}}L_{\sfrac45})\cong\mathbf{S}(\A[11,1]\D[7,1]\E[6,1])\big/(\F[4,1]\G[2,1]\hookrightarrow\E[6,1]\D[7,1]) \\
 \mathrm{Ex}(\A[1,3]\C[3,1]\D[10,1])\cong\mathbf{S}(\D[10,1]\E[7,1]^2)\big/(\F[4,1]\G[2,1]\hookrightarrow \E[7,1]^2) \\
 \mathrm{Ex}(\A[1,3]\B[6,1]\E[7,1]L_{\sfrac{7}{10}})\cong\mathbf{S}(\D[10,1]\E[7,1]^2)\big/(\F[4,1]\G[2,1]\hookrightarrow\E[7,1]\D[10,1]) \\
 \E[8,1]\F[4,1]\G[2,1]\cong \E[8,1]^3\big/(\F[4,1]\G[2,1]\hookrightarrow\E[8,1]^2) \\
 \mathrm{Ex}(\B[12,1]\G[2,1]L_{\sfrac7{10}})\cong\D[16,1]^+\E[8,1]\big/(\G[2,1]\F[4,1]\hookrightarrow \D[16,1]\E[8,1])
 \end{array}$ \\
\bottomrule
\end{tabular}
\caption{The classification of bosonic chiral algebras with $8<c\leq 16$ and $p=\mathrm{rank}(\Cat)=4$, part 2. The theory marked with an asterisk $\ast$ has an extra $\mathsf{U}_1$ Kac--Moody symmetry that is not manifested in the description we have used.}\label{tab:clt16classificationpt3}
\end{center}
\end{table}
\end{small}

\clearpage

\subsection{Chiral fermionic theories}\label{subsec:overview:fermions}

One interesting byproduct of the results summarized in \S\ref{subsec:overview:bosonic} is their implication for chiral fermionic RCFTs. To establish this connection, we make use of a chiral version of fermionization, which we therefore unimaginatively call chiral fermionization.

The difference between fermionization/bosonization and chiral fermionization/chiral bosonization is analogous to the difference between a physicist's orbifold \cite{Dixon:1985jw,Dixon:1986jc} and a mathematician's orbifold. For simplicity, we restrict the present discussion to orbifolds by finite Abelian groups, though more general notions of orbifold are available \cite{Bhardwaj:2017xup}. A physicist defines an orbifold as a procedure which takes a modular-invariant CFT $\mathscr{T}$ with a non-anomalous global symmetry $A$ and produces a new modular-invariant CFT by gauging,\footnote{The notation we use for gauging is similar to the notation used for taking a coset, but one may distinguish between the two because in the former case a group appears after the slash, whereas in the latter case a chiral algebra appears.}
\begin{align}
    (\mathscr{T},A)\mapsto \mathscr{T}\big/ A.
\end{align}
At the level of the Hilbert space, the procedure involves adding in $a$-twisted sectors $\mathcal{H}_a$ for each $a\in A$, and then restricting to states $\mathcal{H}_a^A\subset\mathcal{H}_a$ which are invariant with respect to the action of $A$ to obtain the Hilbert space of the gauged theory,
\begin{align}
    \mathcal{H}\big/ A = \bigoplus_{a\in A}\mathcal{H}_a^A.
\end{align}
A mathematician defines an orbifold as a procedure which is performed on a chiral algebra $\V$. The chiral algebra need not be a full modular-invariant CFT, and its global symmetry group need not be non-anomalous. One simply defines a new chiral algebra by restricting to states which are invariant with respect to the symmetry, but this time without adding in twisted sectors, 
\begin{align}
    (\V,A)\mapsto \V^A.
\end{align}
The two notions of orbifold are related: if $\V$ is the chiral algebra of $\mathscr{T}$, then the chiral algebra of $\mathscr{T}\big/ A$ contains $\V^A$ as a subalgebra.

Likewise, fermionization (see e.g.\ \cite{Gaiotto:2015zta,Kapustin:2017jrc,Karch:2019lnn,Ji:2019ugf} for modern treatments) is a procedure which takes a modular invariant bosonic CFT $\mathscr{T}$ with a non-anomalous $\ZZ_2$ symmetry and produces a fermionic theory by stacking $\mathscr{T}$ with the non-trivial, invertible spin TQFT $(-1)^{\mathrm{Arf}}$ and then gauging the diagonal $\ZZ_2$ symmetry, 
\begin{align}
   ( \mathscr{T},\ZZ_2)\mapsto \mathscr{T}^{\mathrm{Fer}}:= \mathscr{T}\otimes (-1)^{\mathrm{Arf}}\big/ \ZZ_2.
\end{align}
At the level of the Hilbert space, one finds the following relations between the bosonic theory and its fermionization,
\begin{align}
    \mathcal{H}_{\mathrm{NS}}^+=\mathcal{H}^+, \ \ \ \ \mathcal{H}_{\mathrm{NS}}^-=\mathcal{H}^-_a, \ \ \ \  \mathcal{H}_{\mathrm{R}}^+=\mathcal{H}^-, \ \ \ \  \mathcal{H}_{\mathrm{R}}^-= \mathcal{H}^+_{a}
\end{align}
where $\mathcal{H}_{\mathrm{NS}}^\pm$ and $\mathcal{H}_{\mathrm{R}}^\pm$ are the $(-1)^F$-even/odd states in the NS/R sectors of the fermionic theory, and $\mathcal{H}^\pm$ and $\mathcal{H}_a^\pm$ are the even/odd states in the untwisted/$\ZZ_2$-twisted sectors of the bosonic theory.  Chiral fermionization, on the other hand, requires just a bosonic chiral algebra $\V_{\bar 0}$ whose representation category is a spin modular category \cite{Bruillard:2016yio,bruillard2020classification,Dong:2020jhn} (see also Definition \ref{defn:spinMTC}). At the level of the chiral algebra, the requirement that $\Rep(\V_{\bar 0})$ be a spin MTC amounts to demanding that $\V_{\bar 0}$ have a distinguished module $\V_{\bar 1}$ which contains operators of half-integral spin, and whose fusion rule with itself is $\V_{\bar 1}^2 = \V_{\bar 0}$. The output of chiral fermionization is then a fermionic chiral algebra $\V_{\mathrm{NS}}$, defined as 
\begin{align}
    (\V_{\bar 0},\V_{\bar 1}) \mapsto \V_{\mathrm{NS}}:= \V_{\bar 0}\oplus \V_{\bar 1}.
\end{align}
Again, the two versions of fermionization are related: if $\V$ is the chiral algebra of a bosonic CFT $\mathscr{T}$ with a non-anomalous $\ZZ_2$ symmetry, then the fixed-point subalgebra $\V_{\bar 0}\equiv \V^{\ZZ_2}$ has a representation category which is given in a distinguished way by a spin modular category, and its fermionization $\V_{\mathrm{NS}}$ arises as (a subalgebra of) the chiral algebra of $\mathscr{T}^{\mathrm{Fer}}$ in the NS-sector.

Just as orbifolds can be ``undone'' by gauging a quantum/magnetic symmetry \cite{Vafa:1989ih,Bhardwaj:2017xup}, fermionization and chiral fermionization can be undone by bosonization and chiral bosonization, respectively. Ordinary bosonization consists of taking a fermionic CFT, assumed to have a gravitational anomaly $c_L-c_R$ given by a multiple of $8$, and summing over spin structures (equivalently, performing a GSO projection) to obtain a bosonic theory. On the other hand, chiral bosonization takes as input a fermionic chiral algebra $\V_{\mathrm{NS}}$ and outputs the bosonic chiral algebra $\V_{\bar 0}:= \V_{\mathrm{NS}}^{(-1)^F}$, i.e.\ the subalgebra of states which have even fermion parity, or equivalently, the subalgebra of states with integer conformal dimension. Then, the $(-1)^F$-odd states $\V_{\bar 1}$ define a distinguished irreducible module of $\V_{\bar 0}$ which gives $\Rep(\V_{\bar 0})$ the structure of a spin modular category. 

Let us apply these considerations to the special case of chiral fermionic RCFTs, which we label for simplicity by their NS-sector Hilbert space $\V_{\mathrm{NS}}$. Given such a theory, ordinary bosonization is not always available to us because the gravitational anomaly $c_L-c_R=c_L\in \frac12 \ZZ$ may not be a multiple of $8$. However, we are certainly free to apply chiral bosonization. It turns out that when one does this, the bosonic chiral algebra $\V_{\bar 0}$ one obtains has the same representation category as the current algebra $\mathfrak{so}(n)_1$ for $n=2c_L$ (see Proposition \ref{prop:RepCategoryBosonicSubalgebra}); alternatively, using the notation that we will employ in the rest of this paper, the representation category is $(\BB_{c_L-\sfrac12},1)$ if $2c_L$ is odd, and $(\DD_{c_L},1)$ if $2c_L$ is even. Such modular categories have rank 3 or rank 4 for any value of $c_L$, and therefore the bosonized chiral algebra must arise in the Bosonic $(c^\ast,4)$-Classification if $c_L\leq c^\ast$. Conversely, the fermionic theory $\V_{\mathrm{NS}}$ may be reconstructed from the bosonic chiral algebra by chiral fermionization. 

Therefore, we find that the problem of classifying chiral fermionic RCFTs reduces to the problem of enumerating inequivalent chiral fermionizations of bosonic chiral algebras in genera of the form $(c,\Cat)=\big(\sfrac{n}{2},\Rep(\mathfrak{so}(n)_1)\big)$, a problem which we carry out through $c<23$ in Appendix \ref{app:classification:fermionic}. Appendix \ref{app:data:chiralfermionicRCFTs} contains the complete list of theories which have no operators with conformal dimension $h=\sfrac12$. The reason it is sufficient to characterize these theories is that any dimension-$\sfrac12$ operator generates a decoupled free fermion sector \cite{Goddard:1988wv} (see also Theorem \ref{theorem:decoupledfreefermions}), and so one can always obtain the full list of theories from the list of theories with no dimension-$\sfrac12$ operators by tensoring in an arbitrary number of free fermions.

There are a number of satisfying checks one may carry out on this result. For example, any chiral fermionic CFT whose global symmetry Lie algebra has rank equal to its central charge necessarily arises as a VOA associated to an odd, unimodular lattice, and these have been classified through dimension $24$ in \cite{conway2013sphere}; we find perfect agreement with this classic result. Furthermore, all chiral fermionic RCFTs, lattice and otherwise, have been classified rigorously through $c\leq 12$ in \cite{Creutzig:2017fuk} and conjecturally through $c\leq 15\sfrac12$ in \cite{hohn1996self}, where we again find agreement. Also, chiral fermionic RCFTs with $c=16$ arise as world-sheet ingredients for constructing non-supersymmetric heterotic string theories in ten dimensions \cite{Seiberg:1986by,Dixon:1986iz,Alvarez-Gaume:1986ghj}, and the full list of such theories was worked out in the 80s \cite{Kawai:1986vd}.\footnote{We thank Yuji Tachikawa for sharing this reference, and refer readers to \cite{Tac23} for more details on the relationship to non-supersymmetric heterotic string theory.} Happily, we find that our classification recovers all of these previously known theories.

\subsection{Generalized symmetries}\label{subsec:overview:symmetries}

In Appendix \ref{app:classification:symmetries}, we study one more application of our classification result on bosonic rational conformal field theories in small genera. Specifically, we examine its interplay with generalized global symmetries of holomorphic VOAs (see Definition \ref{defn:holomorphicVOA} for the definition of holomorphic VOA).

The starting point for the study of generalized global symmetries is the recasting of ordinary group-like symmetries as extended topological operators of codimension-1 in spacetime. Taking this as the definition, in (1+1)D physics one then encounters, in addition to invertible group-like symmetries, also non-invertible topological line defects, whose abstract structure can be encoded in a mathematical object called a fusion category.

The prototypical example of a non-invertible symmetry in (1+1)D conformal field theory is the category of Verlinde lines \cite{Verlinde:1988sn}. To describe this, assume for simplicity that $\mathscr{T}$ is a diagonal RCFT based on a chiral algebra $\V$, so that the Hilbert space decomposes as 
\begin{align}
    \mathcal{H} = \bigoplus_i \V(i)^\ast \otimes \overline{\V}(i).
\end{align} 
Then the simple objects $\mathcal{L}_k$  of the modular category $\Cat=\Rep(\V)$, thought of as a fusion category by forgetting some of its structure, act as symmetries of $\mathscr{T}$. Concretely, by orienting the corresponding line defects so that they are parallel to space, the $\mathcal{L}_k$ become operators $\hat{\mathcal{L}}_k$ on the Hilbert space which act as
\begin{align}\label{eqn:Verlindelines}
    \hat{\mathcal{L}}_k|\phi_i\rangle = \frac{\mathcal{S}_{ki}}{\mathcal{S}_{0i}}|\phi_i\rangle
\end{align}
where $\mathcal{S}_{ij}$ is the modular S-matrix of the RCFT, and $|\phi_i\rangle \in \V(i)\otimes\overline{\V}(i)$. 

Although Verlinde lines are typically formulated as symmetries of RCFTs with both left and right movers, they have a purely chiral incarnation as well. Indeed, let $\mathcal{A}$ be a holomorphic VOA.\footnote{We anticipate most of what we say has a generalization to VOAs with more than one simple module.} Then, given any subalgebra $\V\subset \mathcal{A}$ (assumed primitive for simplicity), which need not have the same stress tensor as $\mathcal{A}$, we have a decomposition of the form 
\begin{align}\label{eqn:chiralverlindedecomposition}
    \mathcal{A} = \bigoplus_i \V(i)^\ast \otimes\tilde{\V}(\phi i)
\end{align}
where $\tilde{\V}=\Com_{\mathcal{A}}(\V)$, and $\phi:\Rep(\V)\to \Rep(\tilde{\V})$ is a one-to-one map between simple modules of $\V$ and simple modules of $\tilde{\V}$ which mathematically defines what is known as a braid-reversing equivalence \cite{Lin:2016hsa,Creutzig:2019psu}. The claim is that the objects of $\Cat=\Rep(\V)$ act as symmetries of $\mathcal{A}$, in a manner completely analogous to Equation \eqref{eqn:Verlindelines}, but with $\overline{\V}$ replaced with $\tilde{\V}$. One might call these chiral Verlinde lines associated to $\V$.

\begin{figure}
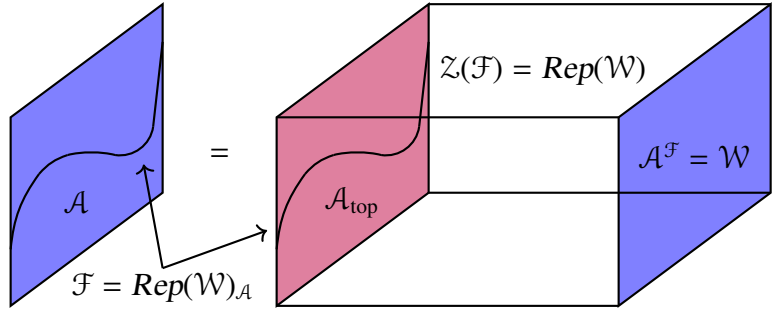

    \ctikzfig{Figures/symmetrysubalgebraduality}
    \caption{A chiral CFT/holomorphic VOA $\mathcal{A}$ with a fusion category symmetry $\mathcal{F}$ may be coupled to a (2+1)D topological $\mathcal{F}$ gauge theory $\mathcal{Z}(\mathcal{F})$ to obtain a conformal subalgebra $\mathcal{A}^{\mathcal{F}}$. Conversely, given a conformal subalgebra $\mathcal{W}$ of a chiral CFT $\mathcal{A}$, there exists a topological boundary condition $\mathcal{A}_{\mathrm{top}}$ such that sandwiching the (2+1)D TQFT $\textsl{Rep}(\mathcal{W})$ with the $\mathcal{A}_{\mathrm{top}}$ and $\mathcal{W}$ boundary conditions recovers $\mathcal{A}$; the fusion category $\Rep(\mathcal{W})_{\mathcal{A}}$ of topological line operators supported on $\mathcal{A}_{\mathrm{top}}$ then acts on $\mathcal{A}$ by symmetries which commute with $\mathcal{W}$.}\label{fig:symsubduality}
\end{figure}

There is actually a more general construction, which we refer to as symmetry/subalgebra duality (see Conjecture \ref{conj:symsubduality}), though it is closely related to ideas which come from the study of ``symmetry TFTs'' \cite{Freed:2012bs,Freed:2018cec,Gaiotto:2020iye,Freed:2022qnc} and ``categorical symmetries'' \cite{Chatterjee:2022jll,Chatterjee:2022tyg,Ji:2021esj,Kong:2020cie,Ji:2019jhk,Ji:2019ebr,Ji:2019eqo}.\footnote{We thank Shu-Heng Shao for valuable discussions related to this construction.}
Let $\mathcal{A}$ be a holomorphic VOA, and let $\mathcal{F}$ be a fusion category which acts on it by symmetries. Then we may associate a conformal subalgebra to $\mathcal{F}$ by considering the operators $\mathcal{A}^{\mathcal{F}}$ in $\mathcal{A}$ which commute with $\mathcal{F}$, %
\begin{align}
    \mathcal{F}\mapsto \mathcal{A}^{\mathcal{F}}.
\end{align}
We can think of the subalgebra $\mathcal{A}^{\mathcal{F}}$ as being obtained by starting with $\mathcal{A}$, and coupling it to a bulk (2+1)D topological gauge theory for $\mathcal{F}$, in which case $\mathcal{A}$ is projected down onto $\mathcal{A}^{\mathcal{F}}$ on the boundary (see Figure \ref{fig:symsubduality}).

Going in the other direction, we may associate a fusion category which acts on $\mathcal{A}$ to any suitably regular conformal subalgebra $\mathcal{W}$. The key idea is that $\mathcal{A}$ defines what is known as a ``Lagrangian algebra''  in the modular category $\Rep(\mathcal{W})$; Equivalently, $\mathcal{A}$ defines a topological boundary condition $\mathcal{A}_{\mathrm{top}}$ in the (2+1)D TQFT corresponding to $\Rep(\mathcal{W})$ which behaves as in Figure \ref{fig:symsubduality}. From this Lagrangian algebra, one can obtain the fusion category $\Rep(\mathcal{W})_{\mathcal{A}}$ of $\mathcal{A}$-modules in $\Rep(\mathcal{W})$, which acts on the holomorphic VOA $\mathcal{A}$; More physically, $\Rep(\mathcal{W})_{\mathcal{A}}$ is the fusion category of topological line operators which live on the topological boundary condition $\mathcal{A}_{\mathrm{top}}$. In total, we can define the subalgebra-to-symmetry map as
\begin{align}
    \mathcal{W}\mapsto \Rep(\mathcal{W})_{\mathcal{A}}.
\end{align}
Because these two maps are inverses to each other, we refer to this as symmetry/subalgebra duality.

Note that one can obtain all the familiar examples of symmetries via this construction. For example, in the special case that $\mathcal{W}=\V\otimes\tilde{\V}$ for some primitively embedded $\V$, the subalgebra-to-symmetry map recovers the chiral Verlinde lines discussed after Equation \eqref{eqn:chiralverlindedecomposition}. Invertible symmetries are also covered. Indeed, for simplicity, consider the special case that $\mathcal{W}$ lives at the boundary of (2+1)D Abelian $G$-gauge theory without Dijkgraaf--Witten twist, so that $\textsl{Rep}(\mathcal{W})\cong \mathcal{D}(G)$ is the quantum double of a finite group, and the simple modules of $\mathcal{W}$ are labeled by pairs $(g,\rho)$ with $g\in G$ and $\rho \in \hat{G}\equiv \mathrm{Hom}(G,\textsl{U}(1))$. If $\mathcal{A}_{\mathrm{top}}$ is further taken to be the standard Dirichlet boundary condition, then it follows that the corresponding CFT $\mathcal{A}$ decomposes as
\begin{align}
\mathcal{A} \cong \bigoplus_{\rho\in \hat{G}}  \mathcal{W}(0,\rho)
\end{align}
whence we obtain an action of $G$ on $\mathcal{A}$ by $|\phi_\rho\rangle \xmapsto{g} \rho(g)|\phi_\rho\rangle$ for each $|\phi_\rho\rangle \in \mathcal{W}(0,\rho)$. Conversely, every action of $G$ on $\mathcal{A}$ can be obtained in this way by taking $\mathcal{W}=\mathcal{A}^G$ to be the subalgebra of $G$-invariant states, which is the subalgebra associated to the $G$-action by the symmetry-to-subalgebra map.

What makes this correspondence particularly useful for us is that the representation category/TQFT of $\mathcal{A}^{\mathcal{F}}$ is completely determined by the fusion category $\mathcal{F}$. Indeed, the three dimensional picture of Figure \ref{fig:symsubduality} suggests that it is precisely a (2+1)D topological gauge theory for $\mathcal{F}$. Mathematically, we expect that
\begin{align}
    \Rep(\mathcal{A}^{\mathcal{F}})\cong \mathcal{Z}(\mathcal{F})
\end{align}
where $\mathcal{Z}(\mathcal{F})$ is known as the Drinfeld center of $\mathcal{F}$ (see Definition \ref{defn:Drinfeldcenter}). Conversely, it is known that every chiral algebra whose representation category is $\mathcal{Z}(\mathcal{F})$ for some $\mathcal{F}$ participates as a conformal subalgebra of at least one holomorphic VOA. Thus, we can determine symmetry actions of fusion categories $\mathcal{F}$ on holomorphic VOAs $\mathcal{A}$ by bootstrapping their corresponding fixed-point subalgebras $\mathcal{A}^{\mathcal{F}}$. In particular, knowing all chiral algebras in a genus of the form $(c,\mathcal{Z}(\mathcal{F}))$ is more or less equivalent to knowing, for each holomorphic VOA $\mathcal{A}$ of central charge $c$, all the ways in which $\mathcal{F}$ may act on $\mathcal{A}$ by symmetries.

\begin{table}[]
    \centering
    \begin{tabular}{c|c|c|c}
        $\mathrm{rank}(\Cat)$ & $\Cat$, $\overline{\Cat}$ & $\mathcal{F}$ & $\E[8,1]^{\mathcal{F}}$ \\ \toprule
        $2$ & & $\Vec_{\ZZ_2}$ & $\D[8,1]$ \\
        & $(\AA_1,1)$, $(\EE_7,1)$ & $\Vec_{\ZZ_2}^\omega$ & $\A[1,1]\E[7,1]$ \\
        & $(\GG_2,1)$, $(\FF_4,1)$ & $\textsl{Fib}$ & $\G[2,1]\F[4,1]$ \\\midrule
       $3$ & $(\AA_2,1)$, $(\EE_6,1)$ & $\Vec_{\ZZ_3}$ & $\A[2,1]\E[6,1]$ \\
        & $(\BB_8,1)$, $(\BB_7,1)$ & $\mathrm{TY}_+(\ZZ_2)$ & $L_{\sfrac12}\B[7,1]$ \\
        & $(\BB_3,1)$, $(\BB_4,1)$ & $\mathrm{TY}_+(\ZZ_2)$ & $\B[3,1]\B[4,1]$ \\
        & $(\AA_1,2)$, $(\BB_6,1)$ & $\mathrm{TY}_-(\ZZ_2)$ & $\A[1,2]\B[6,1]$ \\
        & $(\BB_2,1)$, $(\BB_5,1)$ & $\mathrm{TY}_-(\ZZ_2)$ & $\B[2,1]\B[5,1]$\\\midrule
        $4$ & $(\textsl{U}_1,4)$, $(\DD_7,1)$ & $\Vec_{\ZZ_4}^{\omega^2}$ & $V_{2\ZZ}\D[7,1]$ \\
        & $(\AA_3,1)$, $(\DD_5,1)$ & $\Vec_{\ZZ_4}^{\omega^2}$ & $\A[3,1]\D[5,1]$\\
        & $(\AA_1,1)^{\boxtimes 2}$, $(\EE_7,1)^{\boxtimes 2}$ & $\Vec_{\ZZ_2}^\omega\boxtimes \Vec_{\ZZ_2}^\omega$ & $\A[1,1]^2\D[6,1]$ \\
        & $(\DD_4,1)$ & $\Vec_{\ZZ_2}\boxtimes \Vec_{\ZZ_2}$ & $\D[4,1]^2$ \\
        & $(\AA_1,1)\boxtimes (\GG_2,1)$, $(\CC_3,1)$  & $\Vec_{\ZZ_2}^\omega\boxtimes \textsl{Fib}$ & $\A[1,1]\C[3,1]\G[2,1]$ \\ 
        & $(\AA_1,1)\boxtimes (\FF_4,1)$, $(\AA_1,3)$ & $\Vec_{\ZZ_2}^\omega\boxtimes \textsl{Fib}$ & $\A[1,1]\A[1,3]\F[4,1]$\\
        & $(\GG_2,1)^{\boxtimes 2}$, $(\FF_4,1)^{\boxtimes 2}$ & $\textsl{Fib}\boxtimes \textsl{Fib}$ & $\mathrm{Ex}(\A[1,8])\G[2,1]^2$ \\ 
        & $(\AA_1,7)_{\sfrac12}$, $(\GG_2,2)$ & $\mathrm{PSU}(2)_7$ & $\mathrm{Ex}(\A[1,1]\A[1,7])\G[2,2]$ \\ \bottomrule
    \end{tabular}
    \caption{Examples of fusion categories $\mathcal{F}$ acting on $\E[8,1]$ and their corresponding fixed-point subalgebras $\E[8,1]^{\mathcal{F}}$. The categories $\Cat,\overline{\Cat}$ are complex conjugate pairs of modular tensor categories which correspond to $\mathcal{F}$ when thought of as a fusion categories. Also, $\omega \in H^3(\ZZ_n,\BBC^\times)$ corresponds to one unit of the $\ZZ_n$ anomaly. The classification is complete in the case of rank-2 fusion categories. The classification is partial in the case of higher rank: not all fusion categories of rank 3 and 4 appear in the table, and even for the ones that do, we do not guarantee that we have found all the ways that they act on $\E[8,1]$.}
    \label{tab:E8symmetries}
\end{table}

Of the modular categories which we consider in this paper, three are the Drinfeld centers of some non-trivial fusion category $\mathcal{F}$, and it turns out that these three are precisely the Drinfeld centers of the three unitary fusion categories of rank 2,
\begin{align}
\begin{split}
    (\DD_8,1) &\cong \mathcal{Z}(\Vec_{\ZZ_2}), \ \ \ \ \ \ \text{(non-anomalous }\ZZ_2 \text{ symmetry)} \\
    (\AA_1,1)\boxtimes (\EE_7,1) &\cong \mathcal{Z}(\Vec_{\ZZ_2}^\omega), \ \ \ \ \ \ \text{(anomalous }\ZZ_2 \text{ symmetry)} \\
    (\GG_2,1)\boxtimes (\FF_4,1) &\cong \mathcal{Z}(\textsl{Fib}), \ \ \ \ \ \ \ \ \ \hspace{.025in} \text{(non-invertible Fibonacci symmetry)}.
\end{split}
\end{align}
Because we classify theories with these modular categories through central charge 16, it follows that we have a complete classification of symmetry actions of rank-2 fusion categories on the holomorphic VOAs with $c\leq 16$: namely $\E[8,1]$, $\E[8,1]^2$, and $\D[16,1]^+$. For example, because the genera $(8,(\DD_8,1))$, $(8,(\AA_1,1)\boxtimes (\EE_7,1))$, and $(8,(\GG_2,1)\boxtimes (\FF_4,1))$ each contain a unique theory (cf.\ Table \ref{tab:generaclassification}), it follows that each of the rank-2 fusion categories act in a unique way on $\E[8,1]$, up to conjugation by invertible automorphisms. The corresponding fixed-point subalgebras are $\D[8,1]$, $\A[1,1]\E[7,1]$, and $\G[2,1]\F[4,1]$ respectively. See Proposition \ref{prop:c=8rank2symmetries} and Proposition \ref{prop:c=16rank2symmetries}, as well as Tables \ref{tab:E8symmetries}, \ref{tab:E8^2symmetries}, and \ref{tab:D16psymmetries} for the classification of rank-2 symmetries of holomorphic VOAs with $c\leq 16$.

In fact, we can extract more from Table \ref{tab:generaclassification}. Although we can only obtain a complete classification in the case of rank-2 unitary fusion categories, the existence of the rest of the theories with $c\leq 8$ in our classification imply that every modular category $\Cat$ with $\mathrm{rank}(\Cat)\leq 4$, except for possibly $\Cat\cong (\AA_1,5)_{\sfrac12}$ or $ \overline{(\AA_1,5)}_{\sfrac12}$, acts in at least one way on $\E[8,1]$ by the chiral Verlinde construction (see Proposition \ref{prop:E8chiralverlindelines}). More generally, every pair of theories we find in this paper belonging to conjugate modular categories defines a non-invertible symmetry of some holomorphic VOA. We refer readers to Appendix \ref{app:classification:symmetries} for more details.

\subsection{Future directions}\label{subsec:overview:future}
We conclude this overview by offering a number of open problems for future research.

\subsubsection{The genus $(c,\Cat)=\big(23,(\DD_7,1)\big)$ and chiral fermionic RCFTs with $c=23$}

One of the genera which we did not attempt to classify is $(c,\Cat)=\big(23,(\DD_7,1)\big)$. Chiral fermionization relates bosonic chiral algebras in this genus to chiral fermionic RCFTs with $c=23$.

By the gluing principle, theories in this genus can be enumerated by taking cosets of $c=24$ holomorphic VOAs by the lattice VOA $V_{2\ZZ}$, where $2\ZZ$ is the one-dimensional lattice generated by a vector of length-squared $4$. Just as studying embeddings of the form $\mathsf{X}_{r,k}\hookrightarrow \mathcal{A}$ can be reduced to a problem involving only finite-dimensional Lie algebras (see Appendix \ref{app:liecurrentalgebras}), enumerating embeddings of the form $V_{2\ZZ}\hookrightarrow \mathcal{A}$ can be phrased as a problem in ordinary lattice theory. To explain this, note that every VOA $\mathcal{A}$ has a canonical lattice VOA associated to it, which can be defined through the equation
\begin{align}\label{eqn:orbitlattice}
    V_L = \Com_{\mathcal{A}}(\Com_{\mathcal{A}}(\mathfrak{h}))
\end{align}
where $\mathfrak{h}$ is the $\mathsf{U}_1^r$ Kac--Moody algebra generated by a maximal set of commuting dimension-1 currents in $\mathcal{A}$. Any embedding  $V_{2\ZZ}\hookrightarrow \mathcal{A}$ must factor through this canonical lattice part, 
\begin{align}
    V_{2\ZZ}\hookrightarrow V_L \hookrightarrow \mathcal{A}
\end{align}
and enumerating the embeddings $V_{2\ZZ}\hookrightarrow V_L$ essentially reduces to the problem of computing orbits of the action of $\mathrm{Aut}(L)$ on the vectors of $L$ with length-squared equal to $4$. This problem should be tractable: indeed, the lattice  from Equation \eqref{eqn:orbitlattice} (called an orbit lattice) has been computed for every holomorphic VOA of central charge $c=24$ in \cite{Hohn:2017dsm}, and there is powerful computer software for working with lattices and their automorphism groups \cite{MR1484478}. 

\subsubsection{The genus $(c,\Cat) =\big(23\sfrac12,(\BB_7,1)\big)$ and chiral fermionic RCFTs with $c=23\sfrac12$}

Another genus which deserves attention is $(c,\Cat)=\big(23\sfrac12,(\BB_7,1)\big)$. Again, chiral fermionization relates bosonic theories in $\mathrm{Gen}(c,\Cat)$ to chiral fermionic RCFTs of central charge $c=23\sfrac12$. 

This genus is in principle exhausted by chiral algebras of the form $\Com_{\mathcal{A}}(L_{\sfrac12})$, with $\mathcal{A}$ a holomorphic VOA of central charge $c=24$ and $L_{\sfrac12}$ the chiral algebra of the Ising model. Thus, completion of this genus can be reduced to the  problem of enumerating equivalence classes $[t(z)]$ of ``Ising stress tensors'' (i.e.\ spin-2 operators $t(z)$ with self-OPE $t(z)t(0)\sim \frac{c/2}{z^4}+\frac{2t(0)}{z^2}+\frac{\partial t(0)}{z}$) in holomorphic VOAs. Although more subtle than the problem of computing equivalence classes of affine Kac--Moody subalgebras or lattice subalgebras, a number of existing results in the literature may be brought to bear. Perhaps the most useful one is the fact that each Ising stress tensor  in $\mathcal{A}$ defines a (non-anomalous) $\mathbb{Z}_2<\mathrm{Aut}(\mathcal{A})$, called a Miyamoto involution \cite{Miyamoto1996GriessAA}. This means that Ising stress tensors in $\mathcal{A}$ are heavily constrained by the structure of the automorphism group of $\mathcal{A}$, which has been computed in every case \cite{frenkel1989vertex,dong1999automorphism,Betsumiya:2022avv}\footnote{This statement assumes the conjecture that the monster CFT is the unique holomorphic VOA with $c=24$ and no spin-1 operators.}; in particular, for example, the number of equivalence classes of Ising stress tensors in $\mathcal{A}$ is bounded from above by the number of conjugacy classes of non-anomalous $\ZZ_2$ subgroups of $\mathrm{Aut}(\mathcal{A})$. 

More generally, it would be interesting to develop techniques for classifying embeddings of minimal model chiral algebras. In particular, something like this would be needed to classify theories in the genus $\big(\sfrac{120}{7},\overline{(\AA_1,5)}_{\sfrac12}\big)$ as cosets of the form $\mathcal{A}^{(c=24)}\big/\mathrm{Ex}(\E[6,1]L_{\sfrac67})$ using the gluing principle.

\subsubsection{The genera $(c,\Cat)=\big(24,(\DD_8,1)\big)$, $\big(24,(\AA_1,1)\boxtimes (\EE_7,1)\big)$ and chiral fermionic RCFTs with $c=24$}

Genera with $c=24$ are more difficult to classify. The simplest cosets which the gluing principle leads one to consider involve holomorphic VOAs with central charge $c=32$ in the numerator, and such theories are under poor control. However, $\big(24,(\DD_8,1)\big)$ and $\big(24,(\AA_1,1)\boxtimes (\EE_7,1)\big)$ \emph{can} be treated using symmetry/subalgebra duality. Indeed, chiral algebras in $\big(24,(\DD_8,1)\big)$ and $\big(24,(\AA_1,1)\boxtimes(\EE_7,1)\big)$ arise precisely as the fixed-point subalgebras $\mathcal{A}^{\ZZ_2}$ of $c=24$ holomorphic VOAs $\mathcal{A}$ with respect to non-anomalous and anomalous $\ZZ_2$-symmetries, respectively. Again, one should in principle be able to extract these symmetries from the results of \cite{frenkel1989vertex,dong1999automorphism,Betsumiya:2022avv}. Of course, an additional motivation for treating the genus $\big(24,(\DD_8,1)\big)$ is that chiral fermionization relates it to chiral fermionic RCFTs with $c=24$.

\subsubsection{Genera with $\mathrm{rank}(\Cat)>4$}

Many other small genera should be accessible using the methods described in this paper. We comment on a few interesting cases.

The main reason that this paper topped out at $\mathrm{rank}(\Cat)=4$ is due to the author's fatigue. It is also conceivable that the technology used in this work can be pushed to obtain at least partial results for the classification of bosonic theories in the $10\times 3$ genera $(c,\Cat)$ with $c\leq 24$ and $\mathrm{rank}(\Cat)=5$, which is the current state of the art with respect to the classification of modular tensor categories by rank \cite{ng2022reconstruction,bruillard2016classification,hong2010classification}. For amusement, we have included a few examples of 5-primary theories in Table \ref{tab:rank5UMTCs}. Just from these few examples, it becomes clear that a challenge one must face in the $\mathrm{rank}(\Cat)=5$ classification is the existence of several chiral algebras without any Kac--Moody symmetry, which might resist treatment using the techniques of Appendix \ref{app:classification:clt8}.

Also, fermionic theories with possibly more than one NS-sector primary are related by chiral bosonization to bosonic VOAs whose representation categories are described by spin MTCs; spin MTCs, in turn, have been classified through rank 11 in \cite{bruillard2020classification}, and correspond upon chiral fermionization to fermionic VOAs with up to 3 NS-sector primaries. 

There are  many other infinite families of MTCs which deserve  further study as well: quantum group categories $(\textsl{X}_r,k)$ realized e.g.\ by affine Kac--Moody algebras $\mathsf{X}_{r,k}$, pointed MTCs classified by metric groups $(A,q)$ 
 and realized e.g.\ by lattice VOAs $V_L$ with $L^\ast\big/ L \cong A$, and finite group quantum-doubles $\mathcal{D}(G)$ realized e.g.\ by permutation orbifolds of holomorphic VOAs \cite{bantay2002permutation,evans2022reconstruction,gannon2019vanishing}.

\subsubsection{Genera with $c>24$}

 The previous paragraphs were mainly concerned with generalizations of the results of this paper obtained by varying the modular category $\Cat$. It is natural to also contemplate whether one can classify chiral algebras with central charge larger than $24$. Unfortunately, the gluing principle naively appears to break down once one attempts to exceed $c=24$, because holomorphic VOAs beyond this central charge are poorly understood. So it is worthwhile to try to investigate the simplest case which lies outside the domain of applicability of our methods, but which may still be tractable: the genus $\big(25,(\AA_1,1)\big)$.\footnote{The genus $(32,\Vec)$ is the first genus with $\mathrm{rank}(\Cat)=1$ and $c>24$.  However, due to a refinement \cite{kingmass} of the Smith--Minkowski--Siegel mass formula, this genus has more than a billion chiral algebras (see e.g.\ \S 2.2 of \cite{Mukhi:2022bte}), so one must consider genera with $\mathrm{rank}(\Cat)>1$ to obtain a reasonable classification problem. By the gluing principle, chiral algebras in the genus $\big(25,(\AA_1,1)\big)$ are related to theories in $(32,\Vec)$ which contain an $\E[7,1]$ subalgebra, so there should be far fewer algebras in $\big(25,(\AA_1,1)\big)$ than in $(32,\Vec)$, but precisely how many fewer is unclear.}  One reason one may hope that this genus is within reach is that its \emph{lattice} VOAs have been successfully classified by Borcherds in his PhD thesis \cite{borcherds1999leech}. An additional motivation for considering this genus is that it may house a theory which sheds light on Penumbral Moonshine \cite{harvey2016traces,Duncan:2021ith,Duncan:2022afh,Duncan:2022kpi}, see \cite{Mukhi:2022bte} for an explanation.

It would also be interesting to develop an analog of the Smith--Minkowski--Siegel mass formula which counts general chiral algebras, and not just lattice VOAs. See \cite{Moriwaki:2020ktv} for prior work in this direction. For example, this would lead to estimates for how the number of theories in the genus $(c,\Cat)$ grows with $c$. It would also define an ``ensemsble average'' of (1+1)D CFTs, which one could then contemplate holographically along the lines of \cite{Afkhami-Jeddi:2020ezh,Maloney:2020nni}.

\subsubsection{Supersymmetry, moonshine, and topological modular forms}

Another goal would be to classify $\mathcal{N}=1$ supersymmetry structures on the chiral fermionic theories in Appendix \ref{app:data:chiralfermionicRCFTs}.

One motivation for pursuing this problem is moonshine. The objects which have historically played an important role in the study of moonshine are chiral algebras with finite automorphism groups. For example, in view of the classification of bosonic holomorphic VOAs with $c\leq 24$ \cite{Schellekens:1992db}, the monster CFT (a.k.a.\ the moonshine module) \cite{frenkel1984natural} is the ``first'' theory that one encounters without any continuous symmetries. As far as fermionic theories go, none of the chiral fermionic RCFTs we have classified with $c<23$  have finite automorphism groups.\footnote{However, there is a non-supersymmetric chiral fermionic RCFT with $c=23\sfrac12$ which has finite baby monster symmetry \cite{hohn1996self}.} However, supersymmetry provides a natural loophole:  any chiral fermionic RCFT with $\mathcal{N}=1$ supersymmetry and with no dimension-$\sfrac12$ operators has a finite group of $\mathcal{N}=1$ preserving automorphisms. For example, Appendix \ref{app:data:chiralfermionicRCFTs} shows that $c=12$ is the lowest central charge which supports a theory without any dimension-$\sfrac12$ operators, and consideration of supersymmetry in that context leads one to the Conway module \cite{Harrison:2018joy,Duncan:2014eha}, which has been the subject of a number of mathematically rich developments. It is natural to ask whether any of the other theories in Appendix \ref{app:data:chiralfermionicRCFTs} have supersymetry and are similarly mathematically distinguished. For related work, see e.g.\ \cite{suzukichain} in which the author constructs $\mathcal{N}=1$ supersymmetric fermionic VOAs (generally with more than one NS-sector module) which enjoy interesting exceptional symmetry groups. 

Two-dimensional supersymmetric quantum field theories with one real supercharge are also interesting because of their relationship to topological modular forms (TMF); the subset of chiral fermionic CFTs which admit an $\mathcal{N}=1$ structure may enjoy applications in that setting. See e.g.\ \cite{Albert:2022gcs} for a recent paper which studies the Conway module from the perspective of TMF. The ideas of TMF might even apply to chiral fermionic RCFTs without supersymmetry, in a manner similar to \cite{Lin:2021bcp}.

 \begin{table}[]
     \centering
     \begin{tabular}{c|c}
         $\Cat$ &  $\V$ \\\toprule
         $5_4$ & $\A[4,1]$ \\
         $5_2^a$ & $\A[1,4]$\\ 
         $5_{-2}^a$ & $\C[4,1]$\\
         $5_{\sfrac{16}{11}}$ &  $\F[4,2]$, $\mathcal{W}_{\mathrm{3C}}$\\
         $5_{-\sfrac{16}{11}}$ & $\E[8,3]$, $\E[8,1]\big/\mathcal{W}_{\mathrm{3C}}$, $\textsl{VT}^\natural\cong V^\natural\big/\mathcal{W}_{\mathrm{3C}}$ \\
         $5_{\sfrac{18}{7}}$ & $\mathrm{Ex}(\A[1,12])$\\
         $5_0$, $5_2^b$, $5_{-2}^b$, $5_{-\sfrac{18}7}$ & ? \\\bottomrule
     \end{tabular}
     \caption{Examples of VOAs $\V$ with $\mathrm{rank}(\Rep(\V))=5$. The notation for the modular categories follows \cite{ng2022reconstruction}. Here, $\mathcal{W}_{\mathrm{3C}}$ is the so-called 3C algebra first considered in \cite{miyamoto2003vertex} (see also \cite{dong2021uniqueness} for a more recent study). The 3C algebra embeds into $\E[8,1]$ \cite{lam2005mckay}, and so the corresponding coset $\E[8,1]\big/\mathcal{W}_{\mathrm{3C}}$ has the complex conjugate modular category. It also embeds into the monster CFT $V^\natural$ \cite{sakuma20076}, and the corresponding coset $\textsl{VT}^\natural=\V^\natural\big/\mathcal{W}_{\mathrm{3C}}$ was studied in \cite{Bae:2020pvv} where it was shown to possess an action of the simple sporadic Thompson group.}
     \label{tab:rank5UMTCs}
 \end{table}

 \subsubsection{Dual pairs of VOAs} 

Consider a modular category $\Cat$ with the property that there \emph{do} exist VOAs $\V$ with $\Rep(\V)\cong \Cat$, but there do \emph{not} exist VOAs $\tilde{\V}$ with $\Rep(\tilde\V)\cong \overline{\Cat}$. In this situation, one would not be able to use the gluing principle to classify theories $\V$ in the genus $(c,\Cat)$, because there is no complementary theory $\tilde{\V}$ in any genus $(\tilde{c},\overline{\Cat})$ to glue such $\V$ to. Conjecturally \cite{Schellekens:1990xy}, this situation does not happen.

\begin{conj}\label{conj:ubiquity}
Whenever there exists a VOA $\V$ with $\Rep(\V)\cong \Cat$, then there also exists a VOA $\tilde{\V}$ with $\Rep(\tilde\V)\cong \overline{\Cat}$. Equivalently, any VOA $\V$ can be embedded into a holomorphic VOA $\mathcal{A}$ as one member of a dual pair.
\end{conj}

We remark that this is a special case of the more general and widely-believed conjecture that every modular category $\Cat$ is realized as the representation category of some VOA $\V$, a conjecture which sometimes goes by the name of ``reconstruction''. 

It is possible to prove Conjecture \ref{conj:ubiquity} for certain infinite families of modular categories \cite{Schellekens:1990xy}. For example, let $\Cat\cong (\textsl{X}_r,k)$ be a quantum group category which is completely anisotropic (i.e.\ which cannot be non-trivially condensed to another category, see Definition \ref{defn:condensablealgebra}). This category is of course realized by the affine Kac--Moody algebra $\mathsf{X}_{r,k}$; we show that the dual category is also realized by a chiral algebra. First, note that the diagonal embedding of $\mathsf{X}_r\hookrightarrow \mathsf{X}_r^{\oplus k}$ induces the embedding 
$\mathsf{X}_{r,k} \hookrightarrow \mathsf{X}_{r,1}^{\otimes k}$ of chiral algebras. Next, if $\mathsf{X}$ is $\mathsf{A}$ or $\mathsf{C}$, one can embed $\mathsf{X}^{\otimes k}_{r,1} \hookrightarrow \mathsf{A}_{24n,1}$ for a sufficiently large $n$. Likewise, if $\mathsf{X}$ is $\mathsf{B}$ or $\mathsf{D}$, one can embed $\mathsf{X}_{r,1}^{\otimes k}\hookrightarrow \mathsf{D}_{24n,1}$ for a sufficiently large $n$. Finally, if $\mathsf{X}_r$ is one of the exceptional algebras, then $\mathsf{X}_{r,1}^{\otimes k}$ can be embedded into $\mathsf{E}_{8,1}^{\otimes k}$. In all cases, by composing embeddings, we find that
\begin{align}
\mathsf{X}_{r,k} \hookrightarrow \mathsf{X}_{r,1}^{\otimes k}\hookrightarrow \mathcal{A}
\end{align}
where $\mathcal{A}$ is a suitable holomorphic VOA, either $\mathrm{Ex}(\mathsf{A}_{24n,1})$, $\mathrm{Ex}(\mathsf{D}_{24n,1})$, or $\mathsf{E}_{8,1}^{\otimes k}$. Because $(\textsl{X}_r,k)$ was assumed to be completely anisotropic, it follows that the coset $\Com_{\mathcal{A}}(\mathsf{X}_{r,k})$ defines a theory with representation category given by $\overline{(\textsl{X}_r,k)}$. It would be interesting to try to produce a similar proof for other families of modular categories.

 \subsubsection{The connectedness of orbifolds}

The effectiveness of the classification program implemented in this paper is limited by how much control one has over holomorphic VOAs. Even though there are too many such VOAs beyond $c=24$ to classify explicitly, there may be more structural statements one can make. For example, let $\mathsf{K}^{(c)}$ be the undirected graph (which we call an ``orbifold graph'') whose nodes correspond to isomorphism classes of holomorphic VOAs of central charge $c$, and whose edges correspond to the existence of a generalized symmetry whose gauging connects the nodes. In more detail, we draw an edge between $\mathcal{A}$ and $\mathcal{A}'$ if there is a fusion category $\mathcal{F}$ acting on $\mathcal{A}$ by symmetries, and a module category $\mathcal{M}$ (or an algebra object $A$) of $\mathcal{F}$ such that $\mathcal{A}\big/\mathcal{M}\cong \mathcal{A}'$. See \cite{Bhardwaj:2017xup} for details.\\

\noindent \textbf{Question}: Is $\mathsf{K}^{(c)}$ always a connected graph? \\

In other words, are any two holomorphic VOAs of the same central charge always related by gauging a sequence of generalized symmetries? One piece of evidence in support of the answer being ``yes'' is that it is true for lattice VOAs. Indeed for lattice VOAs, by Kneser's neighbor construction \cite{kneser1957klassenzahlen}, one obtains the stronger statement that any two chiral algebras based on unimodular lattices of the same dimension are related by a sequence of $\ZZ_p$ orbifolds for any prime $p$. Another piece of evidence is that $\mathsf{K}^{(24)}$ is known to be connected by the generalized deep hole construction of \cite{Moller:2019tlx}.\footnote{These authors showed that any $c=24$ holomorphic VOA can be obtained by orbifolding the Leech lattice VOA $V_\Lambda$. However this is sufficient to conclude that $\mathsf{K}^{(24)}$ is connected, because one may always get from one node to another by passing through $V_\Lambda$.}

Analogous lore for chiral fermionic RCFTs holds that all such theories with central charge $c$ can be obtained as discrete orbifolds of  $2c$ free chiral fermions. It would be interesting to work this out explicitly for the chiral fermionic RCFTs classified in this work, in the spirit of \cite{Moller:2019tlx}. \\

\noindent
{\bf Acknowledgments.}
I have greatly benefited from numerous discussions with Jan Albert, Nathan Benjamin, Yichul Choi, Diego Delmastro, Gerald H\"ohn, Zohar Komargodski, Justin Kulp, Ying-Hsuan Lin, Sven M\"oller, Sunil Mukhi, Leonardo Rastelli, Yaman Sanghavi, Sahand Seifnashri, Shu-Heng Shao, and Yuji Tachikawa. I would also like to thank Sunil Mukhi and Shu-Heng Shao for their feedback on a draft, and especially Sunil Mukhi for a fruitful prior collaboration from which this work sprang. I also gratefully acknowledge the hospitality of the Korea Institute for Advanced Study, where part of this work was completed, as well as the Simons Center for Geometry and Physics for providing me with an office. The data that supports the findings of this study are available within the article.

\clearpage 

\appendix 

% \counterwithin*{equation}{section}
% \renewcommand\theequation{\thesection\arabic{equation}}

\section{Notation Guide}

\begin{footnotesize}

\begin{list}{}{
	\itemsep -1pt
	\labelwidth 23ex
	\leftmargin 8ex	
	}

\item[$\mathds{1}$] The identity primary of a chiral algebra.

\item[$\Phi_i$] The non-identity primaries of a chiral algebra, with $i=1,\dots,p-1$.

\item[$\boxtimes$] The Deligne tensor product of Abelian categories, see Definition \ref{defn:Deligneproduct}.

\item[$a_{X,Y,Z}$] The associators in a monoidal category, see Definition \ref{defn:monoidalcategory}.

\item[$A$] An algebra in a category, typically a condensable algebra, see Definition \ref{defn:condensablealgebra}.

\item[$\mathcal{A}$] A holomorphic VOA, see Definition \ref{defn:holomorphicVOA}.

\item[$(\AA_1,k)_{\sfrac12}$] The modular subcategory of $(\AA_1,k)$ which is generated by the $\A[1,k]$-modules whose highest-weight operators transform with integer spin with respect to the $\A[1]$ symmetry, see Example \ref{ex:modularcategoriesrankleq4}.

\item[$c$] The central charge of a chiral algebra.

\item[$c_{X,Y}$] A braiding on a monoidal category, see Definition \ref{defn:braiding}.

\item[$\Cat$] A unitary modular tensor category, see Definition \ref{defn:modularcategory}.

\item[$\overline{\Cat}$] The complex conjugate of a unitary modular tensor category $\Cat$, see Definition \ref{defn:complexconjugate}.

\item[$\Cat_A$] The category of $A$-modules (Definition \ref{defn:Amodules}) in $\Cat$, which is a unitary fusion category if $A$ is a condensable algebra (Definition \ref{defn:condensablealgebra}) and $\Cat$ is a modular category (Definition \ref{defn:modularcategory}), see  Theorem \ref{theorem:Amodulesfusioncategory}.

\item[$\Cat_A^{\mathrm{loc}}$] The category of local $A$-modules (Definition \ref{defn:Amodules}) in $\Cat$, which is a modular category if $A$ is a condensable algebra (Definition \ref{defn:condensablealgebra}) and $\Cat$ is a modular category (Definition \ref{defn:modularcategory}), see Theorem \ref{theorem:anyoncondensation}. More physically, obtaining $\Cat_{A}^{\mathrm{loc}}$ from $\Cat$ is referred to as anyon condensation.

\item[$(c,\Cat)$] The genus of chiral algebras with central charge $c$ and representation category given by $\Cat$.

\item[$\mathrm{Gen}(c,\Cat)$] The set of isomorphism classes of chiral algebras $\V$ with central charge $c$ and $\Rep(\V)=\Cat$.

\item[$\ch(\tau)$] The character vector of a chiral algebra, see Equation \eqref{eqn:charactervector}.

\item[$\ch_i(\tau)$] The components of the character vector of a chiral algebra, where $i=0,\dots,p-1$.

\item[$\Com_\V(\mathcal{W})$] The commutant/coset of $\mathcal{W}$ inside of $\V$, see Definition \ref{defn:commutant}.

\item[$\V\big/\mathcal{W}$] Another notation for the commutant/coset of $\mathcal{W}$ inside of $\V$.

\item[$\mathsf{D}_{16,1}^+$] The unique extension of $\D[16,1]$ with one irreducible module.

\item[$\mathcal{D}_\omega(G)$] The quantum-double modular category, $\mathcal{D}_\omega(G) = \mathcal{Z}(\Vec_G^\omega)$.

\item[$\dim(\Cat)$] The quantum dimension of a category, given by $\dim(\Cat) = \sum_{X} \dim(X)^2$ where the sum runs over isomorphism classes of simple objects in $\Cat$.

\item[$\dim(X)$] The quantum dimension of an object $X$ in a unitary fusion category, see Definition \ref{defn:tracedimension}.

\item[$\Ex(\V)$] A conformal extension of the chiral algebra $\V$, typically clear from context. 

\item[$\mathcal{F}$] A unitary fusion category, see Definition \ref{defn:fusioncategory}.

\item[$\mathrm{Fer}(n)$] The fermionic chiral algebra corresponding to $n$ free, chiral Majorana fermions.

\item[$\textsl{Fib}$] The unique rank-2 unitary fusion category with fusion rules given by $\Phi\times \Phi=1\oplus \Phi$, see Example \ref{ex:rank2fusioncategories}.

\item[$\mathbb{H}$] The upper half-plane, $\mathbb{H}=\{\tau\in\mathbb{C}\mid\Im(\tau)>0\}$.

\item[$h_i$] The conformal dimensions of the primaries of a chiral algebra, where $i=0,\dots,p-1$ and $h_0=0$. 

\item[$\mathcal{K}$] The canonical Kac--Moody subalgebra of a VOA $\V$, see the discussion around Equation \eqref{eqn:KacMoodysubalgebra}.

\item[$\mathcal{K}_\Cat$] The canonical Kac--Moody subalgebra of the VOA $\V_{\Cat}$, see Table \ref{tab:generaclassification}.

\item[$L_c$] The simple Virasoro VOA with central charge $c$.

\item[$\ell$] The Wronskian index of a vector-valued modular form $X(\tau)$, see Equation \eqref{eqn:Wronskianindex}.

\item[$\lambda$] A bijective exponent, see the paragraph around Equation \eqref{eqn:exponent}.

\item[$M$] A module of a bosonic or fermionic VOA. Definition \ref{defn:Vmodules} describes modules in the bosonic case, Definition \ref{defn:NSmodules} describes NS-sector modules in the fermionic case, and Definition \ref{defn:Rmodules} describes R-sector modules in the fermionic case.

\item[$M^\ast$] The contragredient/charge-conjugate module, see Definition \ref{defn:contragredientmodule}.

\item[$\mathcal{M}_0^!(\varrho)$] The space of weakly-holomorphic vector-valued modular forms of weight $0$ with respect to $\varrho$, see the paragraph around Equation \eqref{eqn:vvformtransformation}.

\item[$N_\Cat$] The number of currents (i.e.\ dimension-one operators) in the chiral algebra $\V_\Cat$, see Table \ref{tab:generaclassification}.

\item[$p$] The number of irreducible representations of a chiral algebra $\mathcal{V}$. We mostly consider the cases $p=1,\dots,4$ in this work.

\item[$\Rep(\V)$] The category of admissible $\V$-modules, where $\V$ is a VOA. When $\V$ is strongly regular, $\Rep(\V)$ admits the structure of a modular tensor category, see Theorem \ref{theorem:modularRepV}.

\item[$\Rep(\V_{\bar 0}\vert \V_{\bar 1})$] Given a bosonic chiral algebra $\V_{\bar 0}$ with a fermionic module $\V_{\bar 1}$, we define $\Rep(\V_{\bar 0}\vert \V_{\bar 1})$ to be the super MTC (Definition \ref{defn:superMTC}) which is generated by $\V_{\bar 0}$ modules which occur inside of admissible $\V_{\mathrm{NS}}$-modules in the NS-sector, where $\V_{\mathrm{NS}}=\V_{\bar 0}\oplus \V_{\bar 1}$.

\item 
[$\mathbf{S}(\mathsf{X})$] The unique holomorphic VOA with $c=24$ and $\mathsf{X}$ as its Kac--Moody algebra (see \cite{Schellekens:1992db}). ``S'' is for Schellekens.

\item[$\varrho$] A finite-dimensional representation of $\SL_2(\ZZ)$. 

\item
[$\mathcal{S},\mathcal{T}$] Matrices which generate a representation of $\SL_2(\ZZ)$ through the assignment $\varrho(S)=\mathcal{S}$ and $\varrho(T)=\mathcal{T}$.

\item[$\mathrm{TY}(A,\chi,\tau)$] A Tambara--Yamagami fusion category, see Example \ref{ex:tambarayamagami}.

\item[$\mathrm{TY}_\pm(\ZZ_2)$] The two Tambara--Yamagami categories based on $A=\ZZ_2$, where $\tau=\pm$.

\item[$\tau^\pm(\Cat)$] The Gauss sums of a modular category, see Definition \ref{defn:gausssums}.

\item[$\theta_i$] The twists of a modular category, see Definition \ref{defn:twist}. When the modular category is the representation category of a strongly regular chiral algebra, the twists are given by  $\theta_i=e^{2\pi i h_i}$ with $h_i$ the conformal dimensions of the primary operators.

\item[$\mathsf{U}_{1,k}$] Another name for the VOA $V_L$ associated to the one-dimensional lattice $L=\sqrt{k}\ZZ$.

\item[$\mathcal{V}$] A chiral algebra/VOA (Definition \ref{defn:VOA}), typically strongly regular (Definition \ref{defn:stronglyregular}).

\item[$\V_h$] The space of dimension-$h$ operators in a VOA.

\item[$\V(i)$] The $i$th irreducible module of a VOA $\V$, where $i=0,\dots, p-1$ and $\V(0):=\V$.

\item[$\V(i)_h$] The space of dimension-$h$ operators in the $i$th irreducible module $\V(i)$ of a VOA $\V$.

\item[$\V_{\mathrm{NS}}$] A fermionic VOA, see Definition \ref{defn:fermionicVOA}.

\item[$\V_{\mathrm{NS}}(i)$] The $i$th irreducible module of a fermionic VOA in the NS-sector, where $i=0,\dots,q-1$ and $\V_{\mathrm{NS}}(0):=\V_{\mathrm{NS}}$. See Definition \ref{defn:NSmodules}.

\item[$\V_{\mathrm{R}}(j)',\V_{\mathrm{R}}(k)$] The $\V_{\mathrm{R}}(j)'$ are the irreducible admissible R-sector modules which are not stable, where $j=0,\dots,q'-1$ for some $q'\leq q$. The $\V_{\mathrm{R}}(k)$ are the irreducible admissible R-sector modules which are stable, where $k=q',\dots,q-1$. We also define $\V_{\mathrm{R}}(j)=\V_{\mathrm{R}}(j)\oplus \V_{\mathrm{R}}(j)'\circ (-1)^F$ for $k=0,\dots,q'-1$. See Proposition \ref{prop:modulestability} and Theorem \ref{theorem:modulesbosonicsubalgebra}.

\item[$\V_{\bar 0},\V_{\bar 1}$] The even/odd part of a fermionic VOA $\V_{\mathrm{NS}}=\V_{\bar 0}\oplus \V_{\bar 1}$, where $\V_{\bar 0}$ consists of the operators of $\V_{\mathrm{NS}}$ with integer conformal dimension, and similarly $\V_{\bar 1}$ contains the operators of half-integral conformal dimension.

\item[$\V^G$] The subalgebra of $\V$ which is fixed by $G<\mathrm{Aut}(\V)$.

\item[$\V^{\mathcal{F}}$] The subalgebra of $\V$ which commutes with all of the topological lines in the fusion category $\mathcal{F}$.

\item[$\V_\Cat$] For a modular category $\Cat\neq \overline{(\AA_1,5)}_{\sfrac12}$ with $\mathrm{rank}(\Cat)\leq 4$, we define $\V_\Cat$ to be the unique chiral algebra in the unique genus $(c,\Cat)$ with $0<c\leq 8$. When $\Cat = \overline{(\AA_1,5)}_{\sfrac12}$, we define $\V_\Cat$ to be the unique chiral algebra in the genus $\big(\sfrac{64}7,\overline{(\AA_1,5)}_{\sfrac12}\big)$. See Table \ref{tab:generaclassification}.

\item[$V_L$] The lattice VOA based on an integral, positive-definite lattice $L$. It is bosonic when $L$ is an even lattice, and fermionic otherwise.

\item[$\Vec$] The modular category of complex vector-spaces.

\item[$\textsl{sVec}$] The unitary symmetric fusion category of super vector-spaces. 

\item[$\Vec_G^\omega$] For $G$ a finite group, the unitary fusion category of $G$-graded vector spaces, with associators twisted by $\omega\in H^3(G,\BBC^\times)$. More physically, $\Vec_G^\omega$ describes a $G$-symmetry with anomaly  $\omega$. See Example \ref{ex:pointedfusioncategories}.

\item[$\mathsf{W}^{\mathsf{X}_r}_{p,q}$] The VOA $\mathsf{X}_{r,p}\otimes \mathsf{X}_{r,q} \big/ \mathsf{X}_{r,p+q}$, where $\mathsf{X}_{r,k}$ is the affine Kac--Moody algebra based on the simple Lie algebra $\mathsf{X}_r$ at level $k$.

\item[$\mathsf{X}_{r}^{(x)}$] A notation which is meant to emphasize that the simple Lie algebra $\mathsf{X}_r$ is embedded into some simple Lie algebra $\mathfrak{g}$ with an embedding index $x$. The simple Lie algebra $\mathfrak{g}$ is typically clear from context.

\item[$X(\tau)$] A weakly-holomorphic vector-valued modular form of weight $0$.

\item[$X_i(\tau)$] The components of $X(\tau)$, where $i=0,\dots,p-1$.

\item[$X_{i,n}$] The Fourier coefficients of the components of $X(\tau)$ with respect to an implicit exponent $\lambda$, see Equation \eqref{eqn:Fourierexpansion}.

\item[$X^{(j,m)}(\tau)$] A particular basis of the space $\mathcal{M}_0^!(\varrho)$, defined with respect to a bijective exponent $\lambda$, see the discussion around Equation \eqref{eqn:modularformbasis} and Equation \eqref{eqn:basiscondition}.

\item[$\mathsf{X}_{r,k}$] The current algebra of level $k$ based on the simple Lie algebra $\mathsf{X}_r$ of rank $r$ and type $\mathsf{X}\in\{\mathsf{A},\mathsf{B},\mathsf{C},\mathsf{D},\mathsf{E},\mathsf{F},\mathsf{G}\}$.

\item[$(\textsl{X}_r,k)$] The modular tensor category associated to Chern--Simons theory with gauge group $G$ given by the simply-connected Lie group associated to $\mathsf{X}_r$ at level $k$. 

\item[$\chi$] The characteristic matrix of a modular representation $\varrho$ with respect to a bijective exponent $\lambda$, defined in Equation \eqref{eqn:characteristicmatrix}.

\item[$\Xi(\tau)$] The fundamental matrix of a modular representation $\varrho$ with respect to a bijective exponent $\lambda$, defined in Equation \eqref{eqn:fundamentalmatrix}.

\item[$\mathcal{Z}(\mathcal{F})$] The Drinfeld center of a fusion category $\mathcal{F}$, see Definition \ref{defn:Drinfeldcenter}.

\end{list}

\end{footnotesize}

\clearpage

\section{Preliminaries}\label{app:prelims}
In this appendix, we give a lightning quick description of several relevant concepts from (1+1)D conformal field theory. We start by building some categorical machinery in Appendix \ref{app:prelims:categories}, which we then apply to bosonic chiral algebras in Appendix \ref{app:prelims:chiralalgebras} and fermionic chiral algebras in Appendix \ref{app:prelims:fermionicChiralAlgebras}. Appendix \ref{app:liecurrentalgebras} is devoted to technical details regarding affine Kac--Moody algebras, also known as current algebras. Against better judgement, we try our best to be self-contained by explicitly giving most of the relevant mathematical definitions. More physically-oriented readers should take this as a content warning.

\subsection{Category theory}\label{app:prelims:categories}

This subsection contains a telegraphic review of tensor categories. We refer readers to e.g.\ \cite{etingof2016tensor,bakalov2001lectures,Kong:2022cpy} for more thorough treatments. After establishing some basic setup in Appendix \ref{app:prelims:categorytheory:setup}, the main two categorical notions we build up to are the notion of a unitary fusion category in Appendix \ref{app:prelims:categorytheory:fusioncategories} (relevant for describing the generalized global symmetries of a (1+1)D QFT) and the notion of a unitary modular tensor category in Appendix \ref{app:prelims:categorytheory:MTCs} (relevant for describing the representation theory of a chiral algebra). 

\subsubsection{Some setup}\label{app:prelims:categorytheory:setup}

Let $\mathsf{k}$ be a field. We start by defining what it means for a category to be $\mathsf{k}$-linear. Heuristically, in a linear category, morphisms should form vector spaces and direct sums should exist. In more detail, we have the following definition.

\begin{defn}[Linear category]
    A category $\mathcal{F}$ is $\mathsf{k}$-linear if
    \begin{enumerate}
        \item  $\mathrm{Hom}_{\mathcal{F}}(X,Y)$ is a vector space over $\mathsf{k}$ such that composition of morphisms is linear,
        \item there exists a zero object $0\in\mathcal{F}$ such that $\mathrm{Hom}_{\mathcal{F}}(0,0)=0$, and
        \item direct sums exist. More specifically, for any two objects $X_1$ and $X_2$, there exists an object $X_1\oplus X_2$ (unique up to unique isomorphism), projection morphisms $\pi_i: X_1\oplus X_2\to X_i$, and inclusion morphisms $\iota_i:X_i\to X_1\oplus X_2$ satisfying the universal property that $\pi_i \iota_i = \mathrm{id}_{X_i}$ and $\iota_1 p_1+ \iota_2 p_2 = \mathrm{id}_{X_1\oplus X_2}$.
    \end{enumerate}
\end{defn}

One very useful construction for producing new linear categories from old is the Deligne tensor product.

\begin{defn}[Deligne's tensor product]\label{defn:Deligneproduct}
    Let $\Cat_1$ and $\Cat_2$ be two $\mathsf{k}$-linear categories. Their Deligne tensor product $\Cat_1\boxtimes \Cat_2$ is the $\mathsf{k}$-linear category whose objects are
    \begin{align}
        \mathrm{Objects}(\Cat_1\boxtimes \Cat_2)=\left\{\bigoplus_{i=1}^n X_i\boxtimes Y_i \mid X_i\in \Cat_1, \ Y_i\in \Cat_2\right\}
    \end{align}
    and whose morphism spaces are 
    \begin{align}
        \mathrm{Hom}\left( \bigoplus_{i=1}^n X_i\boxtimes Y_i, \bigoplus_{j=1}^m X_j'\boxtimes Y_j'\right) = \bigoplus_{i=1}^n\bigoplus_{j=1}^m\mathrm{Hom}(X_i,X_j')\otimes \mathrm{Hom}(Y_i,Y_j').
    \end{align}
    See e.g.\ \cite{etingof2016tensor} for a definition of the Deligne tensor product in terms of its universal property.
\end{defn}

Because we restrict to unitary conformal field theories, all of the categories which we consider in this paper correspondingly have a unitary structure.

\begin{defn}[Unitary structure]
    Consider a $\BBC$-linear category $\mathcal{F}$. A dagger structure on $\mathcal{F}$ is a collection of anti-linear maps $\dagger:\mathrm{Hom}(X,Y)\to\mathrm{Hom}(Y,X)$ satisfying the following conditions.
    \begin{enumerate}
        \item The maps are involutive in the sense that $(f^\dagger)^\dagger=f$ for each morphism $f$.
        \item The maps satisfy $(g\circ f)^\dagger = f^\dagger\circ g^\dagger$ for every pair of morphisms $f,g$ which can be composed.
        \item The maps act trivially on the identity morphisms, i.e.\ $\mathrm{id}_X^\dagger=\mathrm{id}_X$ for each object $X\in\mathcal{F}$. 
    \end{enumerate}
    A dagger structure is called a unitary structure if it is the case that $f^\dagger\circ f=0$ if and only if $f=0$. A unitary category is a $\BBC$-linear category equipped with a unitary structure.
\end{defn}

Next, we give the definition of an Abelian category. An Abelian category is essentially one in which kernels and cokernels exist, as they would for linear transformations of vector spaces. Recall that, for a morphism $\phi:X\to Y$ in a linear category $\mathcal{F}$, a kernel (unique up to unique isomorphism if it exists) is a pair $(K,k)$ of an object $K\in\mathcal{F}$ and a morphism $k:K\to X$ such that $\phi k=0$. Furthermore, we demand that if $k':K'\to X$ satisfies $\phi k'=0$, then there exists a morphism $f:K'\to K$ such that $kf=k'$. In short, we ask for the following diagram to commute. 

\begin{center}
\begin{tikzcd}
& K' \arrow[d,"k'"] \arrow[ddr,bend left,"0"] \arrow[dl,"f"] \\
K \arrow[r,"k"]\arrow[drr,bend right,"0"] & X \arrow[dr,"\phi"] \\
& & Y
\end{tikzcd}
\end{center}

Cokernels are defined dually, i.e.\ by reversing arrows in the right places. 

\begin{defn}[Abelian category]
    A $\mathsf{k}$-linear category $\mathcal{F}$ is Abelian if for every morphism $\phi:X\to Y$, there exists a sequence 
    \begin{align}
        K \xrightarrow{k} X \xrightarrow{i} I \xrightarrow{j}Y \xrightarrow{c} C
    \end{align}
    such that $j i =\phi$, $(K,k)=\ker\phi$, $(C,c)=\mathrm{coker}~\phi$, $(I,i)=\mathrm{coker}~k$, and $(I,j)= \ker c$.
\end{defn}
In an Abelian category, we can define what it means for an object to have a subobject, and in turn what it means for an object to be simple.

\begin{defn}
In an Abelian category $\mathcal{F}$, an object $Y$ is said to be a subobject of $X$ if there is a morphism $f:Y\to X$ such that $\ker f=0$. An object is said to be simple if its only subobjects are $0$ and itself. An Abelian category $\mathcal{F}$ is said to be semi-simple if every object is isomorphic to a direct sum of simple objects. Its rank is defined to be the number of isomorphism classes of simple objects.
\end{defn}

The categories we are interested in should also possess a tensor product. In mathematics, these categories are referred to as having a monoidal structure.

\begin{defn}[Monoidal category]\label{defn:monoidalcategory}
A monoidal category is a sextuple $(\mathcal{F},\otimes,a,\mathds 1,\lambda,\rho)$ with $\mathcal{F}$ a category, $\otimes :\mathcal{F}\times\mathcal{F}\to \mathcal{F}$ a bifunctor called the \emph{tensor product}, $a_{X,Y,Z}:(X\otimes Y)\otimes Z \rightarrow X\otimes (Y\otimes Z)$ a family of isomorphisms called \emph{associators}, $\mathds{1}\in\mathcal{F}$ a distinguished object called the \emph{unit}, and $\lambda_X:\mathds{1}\otimes X\rightarrow X$ and $\rho_X:X\otimes \mathds{1}\rightarrow X$ isomorphisms. Together, this data is required to make the pentagon diagram
\begin{center}
\begin{tikzcd}
& ((W\otimes X)\otimes Y)\otimes Z) \arrow[dr,"a_{W\otimes X,Y,Z}"]\arrow[dl,"a_{W,X,Y}\otimes \mathrm{id}_Z"]  \\
(W\otimes (X\otimes Y))\otimes Z \arrow[d,"a_{W,X\otimes Y,Z}" ]& & (W\otimes X) \otimes (Y\otimes Z)\arrow[d,"a_{W,X,Y\otimes Z}"] \\
W\otimes ((X\otimes Y)\otimes Z) \arrow[rr,"\mathrm{id}_W \otimes a_{X,Y,Z}"] & & W\otimes (X\otimes (Y\otimes Z))
\end{tikzcd}
\end{center}
and the triangle diagram
\begin{center}
\begin{tikzcd}
(X\otimes \mathds{1})\otimes Y \arrow[rr,"a_{X,\mathds{1},Y}"] \arrow[dr,"\rho_X\otimes\mathrm{id}_Y"]& & X\otimes (\mathds{1}\otimes Y)\arrow[dl,"\mathrm{id}_X\otimes\lambda_Y"] \\
 & X\otimes Y & 
\end{tikzcd}
\end{center}
both commute.
\end{defn}

All of the categories we consider, in addition to having a monoidal structure, are also required to have duals. Recall that in a monoidal category, we say that an object $X$ admits a left dual $X^\ast$ if there exist morphisms $\mathrm{ev}_X:X^\ast\otimes X\to \mathds{1}$ and $\mathrm{coev}_X:\mathds{1}\rightarrow X\otimes X^\ast$ satisfying that the two maps
\begin{align}
\begin{split}
& \ \ \ \ \ \ \  X\xrightarrow{\lambda_X^{-1}} \mathds{1}\otimes X \xrightarrow{\mathrm{coev}_X\otimes \mathrm{id}_X}(X\otimes X^\ast)\otimes X\xrightarrow{a_{X,X^\ast,X}} X\otimes (X^\ast \otimes X) \xrightarrow{\mathrm{id}_X\otimes\mathrm{ev}_X} X\otimes \mathds{1} \xrightarrow{\rho_X} X \\
&X^\ast \xrightarrow{\rho_{X^\ast}^{-1}} X^\ast\otimes \mathds{1}\xrightarrow{\mathrm{id}_{X^\ast}\otimes \mathrm{coev}_X}X^\ast \otimes (X\otimes X^\ast) \xrightarrow{a^{-1}_{X^\ast,X,X^\ast}} (X^\ast \otimes X)\otimes X^\ast \xrightarrow{\mathrm{ev}_X\otimes \mathrm{id}_{X^\ast}}\mathds{1}\otimes X^\ast \xrightarrow{\lambda_{X^\ast}} X^\ast
\end{split}
\end{align}
are both equal to the identity morphism. Likewise, we say that an object $X$ admits a right dual ${^\ast}X$ if there exist morphisms $\mathrm{ev}'_X:X\otimes {^\ast}X\to \mathds{1}$ and $\mathrm{coev}'_X:\mathds{1}\to{^\ast}X\otimes X$ satisfying that 
\begin{align}
\begin{split}
    & \ \ \ \ \ \ \ X\xrightarrow{\rho_X^{-1}} X\otimes \mathds{1} \xrightarrow{\mathrm{id}_X\otimes\mathrm{coev}'_X} X\otimes ({^\ast}X\otimes X )\xrightarrow{a^{-1}_{X,{^\ast}X,X}} (X\otimes {^\ast}X)\otimes X\xrightarrow{\mathrm{ev}'_X\otimes\mathrm{id}_X} \mathds{1}\otimes X \xrightarrow{\lambda_X}X \\
&{^\ast}X \xrightarrow{\lambda_{{^\ast}X}^{-1}} \mathds{1}\otimes{^\ast}X\xrightarrow{\mathrm{coev}'_X\otimes\mathrm{id}_{{^\ast}X}} ({^\ast}X\otimes X)\otimes {^\ast}X \xrightarrow{a_{{^\ast}X,X,{^\ast}X}} {^\ast}X\otimes(X\otimes{^\ast}X) \xrightarrow{\mathrm{id}_{{^\ast}X}\otimes\mathrm{ev}'_X} {^\ast}X\otimes\mathds{1}\xrightarrow{\rho_{{^\ast}X}}{^\ast}X
\end{split}
\end{align}
are both equal to the identity morphism.

\begin{defn}[Rigidity]
A monoidal category is called rigid if every object has a left dual and a right dual.
\end{defn}

We will also need to extend the definition of taking duals to functions. If $f:X\to Y$ is a morphism, and both $X$ and $Y$ have left duals, then we can define $f^\ast:Y^\ast\to X^\ast$ as the composition
\begin{align}
\begin{split}
    & Y^\ast\xrightarrow{\rho^{-1}_{Y^\ast}}Y^\ast\otimes\mathds{1}\xrightarrow{\mathrm{id}_{Y^\ast}\otimes\mathrm{coev}_X} Y^\ast\otimes (X\otimes X^\ast)\xrightarrow{a^{-1}_{Y^\ast,X,X^\ast}} (Y^\ast\otimes X)\otimes X^\ast \\
    & \hspace{.3in} \xrightarrow{(\mathrm{id}_{Y^\ast}\otimes f)\otimes \mathrm{id}_{X^\ast}} (Y^\ast \otimes Y)\otimes X^\ast\xrightarrow{\mathrm{ev}_Y\otimes \mathrm{id}_{X^\ast}}\mathds{1}\otimes X^\ast\xrightarrow{\lambda_{X^\ast}} X^\ast.
    \end{split}
\end{align}
The right dual of a function can be defined similarly. 

\begin{defn}[Trace and dimension]\label{defn:tracedimension}
    Let $\mathcal{F}$ be a unitary, $\BBC$-linear, rigid, monoidal category in which $\mathrm{Hom}(\mathds{1},\mathds{1}) \cong \BBC$. We define the trace $\mathrm{Tr}(f)$ of a morphism $f:X\to X$ by the following composition, 
    \begin{align}
        \mathds{1} \xrightarrow{\mathrm{coev}_X} X\otimes X^\ast\xrightarrow{f\otimes\mathrm{id}_{X^\ast}} X\otimes X^\ast \xrightarrow{\mathrm{coev}_X^\dagger} \mathds{1}.
    \end{align}
    The dimension of an object is defined to be $\dim(X)=\mathrm{Tr}(\mathrm{id}_X)$.
\end{defn}

The concepts supplied thus far will be enough to define what a fusion category is. We will also need the notion of a modular category, which requires further structure. In particular, we will need the following notions of braiding and twist.

\begin{defn}[Braiding]\label{defn:braiding}
A braiding on a monoidal category is a family of isomorphisms $c_{X,Y}:X\otimes Y \tilde{\rightarrow} Y\otimes X$ which obey the two hexagon identities, 
\begin{center}
\begin{tikzcd}
& X\otimes (Y\otimes Z) \arrow[r,"c_{X,Y\otimes Z}"] & (Y\otimes Z)\otimes X \arrow[dr,"a_{Y,Z,X}"] & \\
(X\otimes Y)\otimes Z \arrow[ur,"a_{X,Y,Z}"]\arrow[dr,"c_{X,Y}\otimes \mathrm{id}_Z"] & & & Y\otimes (Z\otimes X) \\
 & (Y\otimes X)\otimes Z \arrow[r,"a_{Y,X,Z}"] & Y\otimes (X\otimes Z) \arrow[ur,"\mathrm{id}_Y\otimes c_{X,Z}"] &
\end{tikzcd}
\end{center}
\vspace{.2in}
\begin{center}
\begin{tikzcd}
& (X\otimes Y)\otimes Z \arrow[r,"c_{X\otimes Y,Z}"] & Z\otimes (X\otimes Y) \arrow[dr,"a^{-1}_{Z,X,Y}"] & \\
X\otimes (Y\otimes Z)\arrow[ur,"a^{-1}_{X,Y,Z}"]\arrow[dr,"\mathrm{id}_X\otimes c_{Y,Z}"] & & & (Z\otimes X)\otimes Y \\
 & X\otimes(Z\otimes Y) \arrow[r,"a^{-1}_{X,Z,Y}"] & (X\otimes Z)\otimes Y\arrow[ur,"c_{X,Z}\otimes \mathrm{id}_Y"] .&
\end{tikzcd}
\end{center}
A braided monoidal category is a monoidal category with the structure of a braiding $c_{X,Y}$.
\end{defn}

\begin{defn}[Twist]\label{defn:twist}
    A twist on a braided, monoidal category is a family of isomorphisms $\theta_X:X\tilde{\rightarrow}X$ such that 
    \begin{align}
        \theta_{X\otimes Y}=(\theta_X\otimes \theta_Y)\circ c_{Y,X}\circ c_{X,Y}
    \end{align}
    for all objects $X$ and $Y$. A twist in a rigid category is said to be ribbon if it is compatible with taking duals, $(\theta_X)^\ast = \theta_{X^\ast}$.
\end{defn}

\subsubsection{Fusion categories}\label{app:prelims:categorytheory:fusioncategories}

We are finally ready to give the definition of a fusion category. In Appendix \ref{app:classification:symmetries}, we  briefly explain why this is the correct structure to describe the generalized global symmetries of a (1+1)D QFT.
\begin{defn}[Unitary fusion category]\label{defn:fusioncategory}
    A fusion category is a category which is rigid, semi-simple, $\mathsf{k}$-linear, Abelian, and monoidal, with only finitely many isomorphism classes of simple objects, and such that the unit object with respect to the monoidal structure is simple. A fusion category is further said to be unitary if it admits a unitary structure.
\end{defn}

In the rest of this paper, we will always take the base field $\mathsf{k}$ to be $\BBC$, and we will restrict our attention to unitary fusion categories. It is illustrative to consider some examples.

\begin{ex}[Fusion categories associated to finite groups]\label{ex:pointedfusioncategories}
Let $G$ be a finite group, and $\omega\in H^3(G,\BBC^\times)$ a 3-cocycle. We define a fusion category $\Vec_G^\omega$ as follows. The objects are $G$-graded vector spaces $V=\bigoplus_{g\in G}V_g$, and the morphisms are linear maps which preserve the grading. The tensor product is the natural one, 
\begin{align}
(V\otimes W)_g = \bigoplus_{\substack{h,h'\in G\\hh'=g}} V_h\otimes W_{h'}.
\end{align}
The simple objects are the vector spaces $V^{(g)}$ satisfying $V^{(g)}_g=\BBC$ and $V^{(g)}_h=0$ if $h\neq g$; we will simply abbreviate $V^{(g)}$ to $g$, noting that $g\otimes h \cong gh$. The cocycle $\omega$ enters in the definition of the associators, which are defined on the simple objects as
\begin{align}
\begin{split}
    a_{g,h,k}:(g\otimes h)\otimes k &\to g\otimes (h\otimes k) \\
    (v\otimes w)\otimes z&\mapsto \omega(g,h,k) \cdot v\otimes (w\otimes z).
\end{split}
\end{align}
The left and right dual of $g$ are both given by $g^{-1}$.
\end{ex}
Call a fusion category \emph{pointed} if for every simple object $X$, there is an $X^{-1}$ such that $X\otimes X^{-1}\cong\mathds{1}$. That is, pointed fusion categories are fusion categories whose simple objects fuse according to the multiplication law of a finite group. It turns out that any pointed fusion category is equivalent to $\Vec^\omega_G$ for some $(G,\omega)$. The next examples furnish infinite families of fusion categories which are not pointed. 

\begin{ex}[Category of $G$-representations] Let $G$ be a finite group. The category $\Rep(G)$, whose objects are finite-dimensional representations of $G$ over $\BBC$ and whose morphisms are intertwiners, is a fusion category. The tensor product is simply the standard tensor product of representations. The simple objects are the irreducible representations. If $G$ is non-Abelian, then $\Rep(G)$ is not pointed.
\end{ex}
\begin{ex}[Tambara--Yamagami categories \cite{tambara1998tensor}]\label{ex:tambarayamagami}
    Let $A$ be a finite group, and consider the fusion ring with basis $A\cup \{\mathcal{N}\}$ and multiplication given by 
\begin{align}\label{eqn:TYfusionrules}
\begin{split}
    g\star h &= gh, \ \ \ \ (g,h\in A) \\
    g\star \mathcal{N} = \mathcal{N}\star g &= \mathcal{N}, \ \ \ \ (g\in A) \\
    \mathcal{N}\star\mathcal{N} &= \sum_{g\in A}g.
\end{split}
\end{align}
Tambara-Yamagami classified all fusion categories with fusion rules given by Equation \eqref{eqn:TYfusionrules}. They showed that for each \emph{Abelian} group $A$, non-degenerate symmetric bicharacter $\chi:A\times A \to \mathbb{C}^\times$, and sign $\tau=\pm$, there is a fusion category $\mathrm{TY}(A,\chi,\tau)$.  In the case of $A=\ZZ_2$, there is only a single choice for $\chi$, and we write the two corresponding fusion categories as $\mathrm{TY}_{\pm}(\ZZ_2)$.
\end{ex}
We will also be interested in fusion categories of low rank. The following is a theorem due to Ostrik.

\begin{ex}[Classification of rank-2 unitary fusion categories \cite{ostrik2003fusion}]\label{ex:rank2fusioncategories}
    Let $\mathcal{F}$ be a \emph{unitary} fusion category with $\mathrm{rank}(\mathcal{F})=2$. Then $\mathcal{F}$ is either $\Vec_{\ZZ_2}$, $\Vec_{\ZZ_2}^\omega$, or $\textsl{Fib}$. Here, $\omega$ is the unique non-trivial element of $H^3(\ZZ_2,\BBC^\times)$, and $\textsl{Fib}$ is the unitary Fibonacci category, characterized by the fusion rule $\Phi\otimes \Phi\cong \mathds{1}\oplus \Phi$. 
\end{ex}

\subsubsection{Modular categories}\label{app:prelims:categorytheory:MTCs}

We also require the notion of a unitary modular tensor category. As we explain in more detail in Appendix \ref{app:prelims:chiralalgebras}, this is the correct structure to describe the representation category of a suitably regular chiral algebra. 

\begin{defn}[Modular category]\label{defn:modularcategory}
    A pre-modular category is a ribbon fusion category. A pre-modular category is modular if its S-matrix, 
    \begin{align}
        \tilde{\mathcal{S}}_{X,Y} = \mathrm{Tr}(c_{Y,X}\circ c_{X,Y})
    \end{align}
    is non-degenerate, where $X$ and $Y$ run over the simple objects of the category. A unitary modular category is a modular category which is unitary as a fusion category.
\end{defn}
\begin{remark}\label{remark:normalizedS}
The S-matrix which, as we will see, is more relevant to two-dimensional conformal field theory is the normalized S-matrix,
\begin{align}\label{eqn:normalizedSmatrix}
    \mathcal{S} = \tilde{\mathcal{S}}/D
\end{align}
where $D^2=\sum_X \tilde{\mathcal{S}}_{X,\mathds{1}}^2$.
\end{remark}
An important invariant of a unitary modular tensor category is its chiral central charge which, as the name suggests, is closely related to the standard central charge in (1+1)D conformal field theory.
\begin{defn}[Chiral central charge]\label{defn:gausssums}
   The Gauss sums of a unitary modular tensor category $\Cat$ are defined to be 
\begin{align}
    \tau^\pm (\Cat) = \sum_{X} \theta_X^{\pm 1} \dim(X)^2
\end{align}
where the sum runs over the simple objects of $\Cat$, and $\theta_X$ is thought of as a number using the fact that $\mathrm{Hom}(X,X)\cong \BBC$. The chiral central charge $c(\Cat)$ of a unitary modular tensor category is the rational number modulo 8 defined by the equation 
\begin{align}
    e^{2\pi i c(\Cat)/8}= \frac{\tau^+(\Cat)}{\sqrt{\dim(\Cat)}} = \frac{\sqrt{\dim(\Cat)}}{\tau^-(\Cat)}
\end{align}
where here, $\dim(\Cat) = \sum_X \dim(X)^2$ with the sum running again over simple objects.
\end{defn}
Central to our strategy of using the ``gluing principle'' to classify chiral algebras is the notion of the complex conjugate of a modular category.
\begin{defn}[Complex conjugate]\label{defn:complexconjugate}
The complex conjugate of a unitary modular tensor category $\Cat$ is defined to be the category $\overline{\Cat}$ which has the same underlying monoidal structure as $\Cat$, but with braiding given by $\overline{c}_{X,Y} = c_{Y,X}^{-1}$ and twist given by $\overline{\theta}_X=\theta_X^{-1}$.
\end{defn}
We now turn to examples, starting with the generalization of Example \ref{ex:pointedfusioncategories} to the modular setting.

\begin{ex}[Pointed modular categories \cite{drinfeld2010braided}]
Pointed modular tensor categories $\Cat(A,q)$ are classified by pairs $(A,q)$, where $A$ is a finite Abelian group, and $q: A\to \BBC^\times$ is a non-degenerate quadratic form, i.e.\ $q(x)=q(-x)$ and the function $b(x,y)=q(x+y)q(x)^{-1}q(y)^{-1}$ is non-degenerate and satisfies $b(x+x',y)=b(x,y)b(x',y)$. 

The twist of a simple object $X$ whose isomorphism class is represented by $x\in A$ is given precisely by $q(x) \mathrm{id}_X$. The quadratic form also encodes the braiding $c_{X,X}=q(x)\mathrm{id}_{X\otimes X}$, and the double braiding, $c_{Y,X}c_{X,Y}=b(x,y)\mathrm{id}_{X\otimes Y}$. Thus, $\tilde{\mathcal{S}}_{X,Y}=b(x,y)$ and $\mathcal{S}_{X,Y}=b(x,y)/\sqrt{|A|}.$

As a fusion category, $\Cat(A,q)$ is equivalent to $\Vec_A^\omega$ for some $\omega\in H^3(A,\BBC^\times)$. The precise $\omega$ can be obtained as follows. Define the Eilenberg-MacLane cohomology group\footnote{Note, this is not the ordinary group cohomology!} $H^3_{\mathrm{ab}}(A,\BBC^\times)=Z^3_{\mathrm{ab}}(A,\BBC^\times)\big/B^3_{\mathrm{ab}}(A,\BBC^\times)$. Here, $Z^3_{\mathrm{ab}}(A,\BBC^\times)$ consists of pairs $(\omega,\sigma)$, where $\omega\in Z^3(A,\BBC^\times)$ is a 3-cocycle, and $\sigma\in C^2(A,\BBC^\times)$ is a 2-cochain, satisfying 
\begin{align}
    \begin{split}
        \omega(a,b,c)\omega(c,a,b)\sigma(ab,c)&=\omega(a,c,b)\sigma(a,c)\sigma(b,c) \\
        \sigma(a,bc)\omega(b,a,c)&=\sigma(a,b)\sigma(a,c)\omega(a,b,c)\omega(b,c,a).
        \end{split}
\end{align}
We say that $(\omega,\sigma)$ belongs to the subgroup $B^3_{\mathrm{ab}}(A,\BBC^\times)$ of coboundaries if there is an $h\in C^2(A,\BBC^\times)$ such that 
\begin{align}
    \omega(a,b,c)=dh(a,b,c), \ \ \ \ \sigma(a,b)=h(a,b)h(b,a)^{-1}.
\end{align}
Call $\mathrm{Quad}(A,\BBC^\times)$ the group of quadratic forms on $A$. A result of Eilenberg-MacLane says that there is a group isomorphism 
\begin{align}
\begin{split}
    H^3_{\mathrm{ab}}(A,\BBC^\times)&\tilde{\to} \mathrm{Quad}(A,\BBC^\times) \\
    (\omega,\sigma) &\mapsto (a\mapsto \sigma(a,a)).
\end{split}
\end{align}
The pre-image of a quadratic form $q$ under this Eilenberg-MacLane isomorphism defines a 3-cocycle $\omega$ such that $\Cat(A,q)\cong \Vec_A^\omega$ as fusion categories. The corresponding $\sigma$ encodes the braiding.  
\end{ex}

\begin{ex}[Modular categories with $\mathrm{rank}(\Cat)\leq 4$ \cite{rowell2009classification}]\label{ex:modularcategoriesrankleq4}
    We denote the representation category of the affine Kac--Moody algebra $\mathsf{X}_{r,k}$ as $(\textsl{X}_r,k)$, where $\mathsf{X}\in \{\mathsf{A},\mathsf{B},\mathsf{C},\mathsf{D},\mathsf{E},\mathsf{F},\mathsf{G}\}$ and $k$ is a positive integer. Also, we use the notation $(\AA_1,k)_{\sfrac12}$ to denote the modular subcategory of $\A[1,k]$-modules with integer spin (recall that $\A[1,k]$-modules are labeled by a spin $j=0,\sfrac12,\dots,\sfrac{k}{2}$ which describes the $\mathfrak{su}(2)$ representation of the highest-weight states). Then there are 35 unitary modular categories $\Cat$ with $\mathrm{rank}(\Cat)\leq 4$, and they are given in Table \ref{tab:generaclassification}. Some categories arise as Deligne tensor products $\Cat=\Cat'\boxtimes \Cat''$ of modular categories of lower rank (see Definition \ref{defn:Deligneproduct} for the definition of $\boxtimes$). Their complex conjugates $\overline{\Cat}$ can be read off from Tables \ref{tab:clt16classificationpt0}--\ref{tab:clt16classificationpt3}; the tables in Appendix \ref{app:data} give alternative descriptions/other names for each modular category. 
\end{ex}

The following construction is essential for symmetry/subalgebra duality, as it relates fusion categories to modular categories.

\begin{defn}[Drinfeld center]\label{defn:Drinfeldcenter}
    Let $\mathcal{F}$ be a monoidal category. The Drinfeld center $\mathcal{Z}(\mathcal{F})$ is a braided monoidal category defined as follows.
    \begin{enumerate}
        \item The objects of $\mathcal{Z}(\mathcal{F})$ are tuples $(X,\{\gamma_{Y,X}\}_{Y\in\mathcal{F}})$, where the $\gamma_{Y,X}:Y\otimes X\tilde\to X\otimes Y$ define a natural isomorphism such that the diagram 
        \begin{center}
\begin{tikzcd}
& Y\otimes (X\otimes Z) \arrow[r,"\alpha^{-1}_{Y,X,Z}"] &(Y\otimes X)\otimes Z \arrow[dr,"\gamma_{Y,X}\otimes\mathrm{id}_Z "]\\
Y\otimes(Z\otimes X) \arrow[ur,"\mathrm{id}_Y\otimes \gamma_{Z,X}"]\arrow[dr,"\alpha^{-1}_{Y,Z,X}"] & & & (X\otimes Y)\otimes Z \\
& (Y\otimes Z)\otimes X\arrow[r,"\gamma_{Y\otimes Z,X}"] & X\otimes (Y\otimes Z)\arrow[ur,"\alpha^{-1}_{X,Y,Z}"]
\end{tikzcd}
\end{center}
commutes.
        \item The morphisms $(X,\{\gamma_{Y,X}\})$ to $(A,\{\gamma_{B,A}\})$ are defined to be morphisms $f:X\to A$ such that
    \begin{center}
\begin{tikzcd}
Z\otimes X\arrow[r,"\gamma_{Z,X}"]\arrow[d,"\mathrm{id}_Z\otimes f"] & X\otimes Z \arrow[d,"f\otimes \mathrm{id}_Z"] \\
Z\otimes A\arrow[r,"\gamma_{Z,A}"] & A \otimes Z
\end{tikzcd}
\end{center}
commutes.
\item The tensor product $(X,\{\gamma_{Y,X}\}_Y)\otimes (A,\{\gamma_{B,A}\}_B)$ is defined to be $(X\otimes A,\{\gamma_{Z,X\otimes A}\}_{Z})$ where $\gamma:Z\otimes(X\otimes A)\tilde\rightarrow (X\otimes A)\otimes Z$ is given by the composition 
\begin{align}
\begin{split}
   & Z\otimes(X\otimes A)\xrightarrow{\alpha^{-1}_{Z,X,Z}} (Z\otimes X)\otimes Z \xrightarrow{\gamma_{Z,X}\otimes\mathrm{id}_A} (X\otimes Z)\otimes A \\
   & \hspace{.3in} \xrightarrow{\alpha_{X,Z,A}} X\otimes (Z\otimes A)\xrightarrow{\mathrm{id}_X\otimes \gamma_{Z,A}}X\otimes (A\otimes Z) \xrightarrow{\alpha^{-1}_{X,A,Z}}(X\otimes A)\otimes Z.
\end{split}
\end{align}
\item The tensor unit is $(\mathds{1},\{\gamma_{Y,\mathds{1}}\}_{Y})$, with the half braiding defined by composing the following maps,
\begin{align}
    Y\otimes \mathds{1}\xrightarrow{\rho_Y} Y \xrightarrow{\lambda_Y^{-1}} \mathds{1}\otimes Y .
\end{align}
\item The associators are those of $\mathcal{F}$.
\item The braiding of $(X,\{\gamma_{Y,X}\}_Y)$ and $(A,\{\gamma_{B,A}\}_B)$ is $\gamma_{X,A}$.
    \end{enumerate}
\end{defn}

\begin{theorem}[Modularity of the Drinfeld center \cite{muger2003subfactors}]
    The Drinfeld center $\mathcal{Z}(\mathcal{F})$ of a unitary fusion category $\mathcal{F}$ is a unitary modular tensor category.
\end{theorem}
\begin{ex}[Drinfeld center of a modular category \cite{muger2003subfactors}]
    One may think of a unitary modular tensor category $\Cat$ as a unitary fusion category by forgetting some of its structure. Then, the Drinfeld center of $\Cat$ is given by $\mathcal{Z}(\Cat)\cong \Cat\boxtimes \overline{\Cat}$.
\end{ex}

An important mechanism for obtaining new modular categories from old (and, as it will turn out, new chiral algebras from old) is anyon condensation. For this, we need the notion of an algebra in a category. We follow the references \cite{Kong:2013aya,Kong:2022cpy} closely.

\begin{defn}[Algebras in braided tensor categories]
    Let $\Cat$ be a braided tensor category. An algebra in $\Cat$ is a triple $(A,\mu,\iota)$ consisting of an object $A$ in $\Cat$, a  morphism $\mu:A\otimes A\to A$ called the multiplication morphism, and a morphism $\iota:\mathds{1}\to A$. Together, they must satisfy the conditions
    \begin{align}
    \begin{split}
        \mu\circ (\mu\otimes\mathrm{id}_A) \circ \alpha_{A,A,A}=\mu\circ (\mathrm{id}_A\otimes \mu), \\
        \mu\circ (\iota\otimes\mathrm{id}_A) = \mathrm{id}_A = \mu\circ (\mathrm{id}_A\otimes\iota). \ \ 
    \end{split}
    \end{align}
    The first condition says that the multiplication $\mu$ should be associative, and the second says that $\iota$ should give $A$ a unit. 
    The algebra $A$ is called commutative if $\mu=\mu \circ c_{A,A}$. 
\end{defn}

This is a generalization of the ordinary notion of an algebra. Indeed, if $\Cat=\Vec$, the category of vector spaces, then the notion of an algebra in $\Cat$ reduces to the standard definition. Just as one can study modules of ordinary algebras, one can also study modules of algebras in a braided tensor category.
\begin{defn}[Modules over algebras]\label{defn:Amodules}
    A left module over an algebra $(A,\mu,\iota)$ in $\Cat$ is a pair $(M,\mu_M^L)$ of an object $M$ in $\Cat$ and a morphism $\mu_M^L:A\otimes M\to M$ satisfying 
    \begin{align}
    \begin{split}
        \mu_M^L\circ (\mu\otimes\mathrm{id}_M)\circ\alpha_{A,A,M} = \mu_M^L\circ (\mathrm{id}_A\otimes\mu^L_M) \\
        \mu_M^L\circ(\iota \otimes\mathrm{id}_M)=\mathrm{id}_M. \ \ \ \ \ \ \  \ \ \ \ \ 
    \end{split}
    \end{align}
    The first equation demands that it should not matter whether one first multiplies two elements within $A$ and then acts on $M$, or whether one acts with the two elements sequentially on $M$. The second equation says that the unit of $A$ should act trivially on $M$. A left $A$-module is called local if $\mu_M^L=\mu_M^L\circ c_{M,A}\circ c_{A,M}$, where $c_{X,Y}$ is the braiding on $\Cat$. 
    
    One can similarly define right modules $(M,\mu_M^R)$. An $A$-$B$-bimodule is then a triple $(M,\mu_M^L,\mu_M^R)$ such that $(M,\mu_M^L)$ is a left $A$-module, $(M,\mu_M^R)$ is a right $B$-module, and these structures are compatible with one another in the sense that 
\begin{align}
    \mu_M^R\circ (\mu_M^L\otimes\mathrm{id}_B)\circ \alpha_{A,M,B} = \mu_M^L\circ (\mathrm{id}_A\otimes\mu_M^R).
\end{align}
\end{defn}
We then have the following useful adjectives.
\begin{defn}
    An algebra is called separable if there exists an $A$-$A$-bimodule map $e:A\to A\otimes A$ such that $\mu\circ e=\mathrm{id}_A$. A separable algebra is called connected (sometimes also haploid) if $\dim\mathrm{Hom}_{\Cat}(\mathds{1},A)=1$. A connected, commutative, separable algebra $A$ is called Lagrangian if $(\dim A)^2=\dim \Cat$.
\end{defn}
We would like to have a good notion of tensor product of $A$-modules. For this, we must give $A$ more structure.

\begin{defn}
    A coalgebra is an algebra with the arrows reversed. More specifically, a coalgebra is a triple $(A,\Delta,\epsilon)$ with $\Delta:A\to A\otimes A$ and $\epsilon:A\to \mathds{1}$ such that the coassociativity and counit conditions hold,
    \begin{align}
    \begin{split}
        \Delta\circ (\Delta\otimes\mathrm{id}_A)=\alpha_{A,A,A}\circ \Delta\circ (\mathrm{id}_A\otimes\Delta) \\ 
        (\epsilon\otimes\mathrm{id}_A)\circ \Delta=\mathrm{id}_A=(\mathrm{id}_A\otimes\epsilon)\circ \Delta.
    \end{split}
    \end{align}
\end{defn}
\begin{defn}[Frobenius algebra]
    A Frobenius algebra $A=(A,\mu,\iota,\Delta,\epsilon)$ is both an algebra and a coalgebra in such a way that these structures are compatible, 
    \begin{align}
        (\mathrm{id}_A\otimes\mu)\circ (\Delta\otimes\mathrm{id}_A) = \Delta\circ \mu = (\mu\otimes\mathrm{id}_A)\circ (\mathrm{id}_A\otimes\Delta).
    \end{align}
    In the case that $\Cat$ is unitary, a Frobenius algebra is $\dagger$-Frobenius if it is compatible with the unitary structure, $\Delta = \mu^\dagger$.
\end{defn}
\begin{defn}[Symmetric, normalized-special]
Now let $\Cat$ be ribbon as opposed to just braided (so that left duals are the same as right duals). A Frobenius algebra is symmetric if the following morphisms $A\to A^\ast$ are equal,
\begin{align}
    [(\epsilon\circ \mu)\otimes \mathrm{id}_{A^\ast}]\circ (\mathrm{id}_A\otimes \mathrm{coev}_A)=[\mathrm{id}_{A^\ast}\otimes (\epsilon\circ \mu)]\circ (\mathrm{coev}_A'\otimes\mathrm{id}_A).
\end{align}
It is normalized-special if 
\begin{align}
    \mu\circ \Delta=\mathrm{id}_A, \ \ \ \ \epsilon\circ \iota=\dim(A)\mathrm{id}_{\mathds{1}}.
\end{align}
\end{defn}
Finally, we put all of these definitions together.
\begin{defn}[Condensable algebra]\label{defn:condensablealgebra}
    Let $\Cat$ be a unitary modular tensor  category. A condensable algebra $A$ is a connected, commutative, symmetric, normalized-special $\dagger$-Frobenius algebra in $\Cat$. A UMTC $\Cat$ is said to be completely anisotropic if it does not have any non-trivial condensable algebras.
\end{defn}
The importance of condensable algebras is that they allow us to construct new unitary modular categories via anyon condensation.
\begin{theorem}[Anyon condensation]\label{theorem:anyoncondensation}
    If $A$ is a condensable algebra in a unitary modular category $\Cat$, then the category $\Cat_A^{\mathrm{loc}}$ of local $A$-modules in $\Cat$ is a unitary modular category as well. When $A$ is further Lagrangian, $\Cat_A^{\mathrm{loc}}\cong \Vec$.
\end{theorem}
\begin{theorem}\label{theorem:Amodulesfusioncategory}
    The category $\Cat_A$ of $A$-modules is a unitary fusion category. Furthermore, $\mathcal{Z}(\Cat_A)\cong \Cat\boxtimes \overline{\Cat_A^{\mathrm{loc}}}$. When $A$ is Lagrangian, we have that $\mathcal{Z}(\Cat_A)\cong \Cat$.
\end{theorem}
\begin{ex}[See \cite{Creutzig:2019psu} and Lemma 6.19 of \cite{Frohlich:2003hm}]\label{ex:algebrafromequivalence}
    Consider the modular category $\Cat\boxtimes\overline{\Cat}$, and let $\V(i)$ denote the simple objects of $\Cat$. Then for every braid-reversing equivalence $\phi:\Cat\to\overline{\Cat}$, there is a Lagrangian algebra in $\Cat\boxtimes\overline{\Cat}$ which decomposes into simple objects as
\begin{align}\label{eqn:canonicalLagrangianalgebra}
    A = \bigoplus_i \V(i)^\ast \boxtimes \phi(\V(i)).
\end{align}
The fusion category of $A$-modules in $\Cat\boxtimes \overline{\Cat}$ is given by 
\begin{align}\label{eqn:catofVerlindelines}
    (\Cat\boxtimes \overline{\Cat})_A \cong \Cat
\end{align}
and correspondingly, the Drinfeld center is
\begin{align}
    \mathcal{Z}(\Cat)\cong \Cat\boxtimes \overline{\Cat}.
\end{align}
When $\phi$ is the canonical braid-reversing equivalence, Equation \eqref{eqn:canonicalLagrangianalgebra}  essentially describes the charge-conjugation modular invariant of RCFTs, and Equation \eqref{eqn:catofVerlindelines} describes the corresponding Verlinde lines of the theory.
\end{ex}

 \begin{ex}
    Let $\Cat$ be a unitary modular tensor category, and let $S$ be the (Abelian) group of isomorphism classes of simple objects in $\Cat$ which are invertible with respect to the tensor product and which have twist equal to 1. Then for any subgroup $S'\subset S$, the object 
    \begin{align}
        A=\bigoplus_{J\in S'}J
    \end{align}
    admits the structure of a condensable algebra in a unique way. In conformal field theory, these define what are known as simple current extensions.
\end{ex}

Before concluding this subsection, we briefly define a few notions (following \cite{Bruillard:2016yio,bruillard2020classification}) which will come in handy when we discuss fermionic vertex operator algebras. 
\begin{defn}[Spin MTC]\label{defn:spinMTC}
    A spin modular category is a pair $(\mathcal{C},f)$ where $\mathcal{C}$ is a modular category and $f$ is a fermionic object, in the sense that $f^2=\mathds{1}$ and $\theta_f=-1$.
\end{defn}
\begin{defn}[Super MTC]\label{defn:superMTC}
    The M\"{u}ger center $\mathcal{B}'$ of a braided fusion category is the fusion subcategory generated by objects $Y$ in $\mathcal{B}$ such that $c_{Y,X}\circ c_{X,Y}=\mathrm{id}_{X\otimes Y}$ for all objects $X$ in $\mathcal{B}$. A unitary ribbon fusion category $\mathcal{B}$ is called super-modular if $\mathcal{B}'\cong \textsl{sVec}$, where $\textsl{sVec}$ is the unitary symmetric fusion category of super-vector spaces. (Symmetric means that $c_{Y,X}\circ c_{X,Y}=\mathrm{id}_{X\otimes Y}$ for all objects $X$ and $Y$.)
\end{defn}

\begin{defn}\label{defn:minimalmodularextension}
    Let $\mathcal{B}$ be  a super MTC. Then a modular category $\Cat\supset \mathcal{B}$ is said to be a minimal modular extension of $\mathcal{B}$ if $\dim(\Cat)=2\dim(\mathcal{B})$, where $\dim(\Cat)=\sum_{X}\dim(X)^2$ and the sum runs over simple objects of $\Cat$.
\end{defn}

If a super MTC $\mathcal{B}$ admits a minimal modular extension (which is conjectured to always be true), then it is proven that it admits 16 of them; this  phenomenon is known as the 16-fold way. Furthermore, any minimal modular extension of a super MTC is a spin MTC $(\Cat,f)$, with the fermion $f$ corresponding to the non-trivial object of $\mathcal{B}$ which survives in the M\"{u}ger center. Conversely, it is known that any spin MTC is the minimal modular extension of a super MTC. Indeed, any spin MTC admits a $\ZZ_2$ grading $\Cat=\Cat_{+}\oplus \Cat_-$, where $\Cat_\pm$ is generated by the simple objects $X$ satisfying $c_{f,X}\circ c_{X,f}=\pm \mathrm{id}_{X\otimes f}$; the subcategory $\Cat_0$ defines a super MTC of which $\Cat$ is a minimal modular extension.

\begin{ex}
The modular categories $\Rep(\mathfrak{so}(n)_1)$ as $n$ ranges over integers are 16-periodic, and each has the structure of a spin MTC. In turn, they provide the 16 minimal modular extensions of the super MTC $\textsl{sVec}$. 
\end{ex}

\clearpage

\subsection{Bosonic chiral algebras}\label{app:prelims:chiralalgebras}

We would now like to apply the various categorical constructions of Appendix \ref{app:prelims:categories} to the study of chiral algebras. As alluded to at the beginning of \S\ref{sec:overview}, a chiral algebra is the correct algebraic structure to describe the sector of holomorphic local operators in a (1+1)D conformal field theory. Mathematicians refer to this structure as a vertex operator algebra, and we use these terms interchangeably. Let us give the axiomatic definition.

\begin{defn}[Vertex operator algebra]\label{defn:VOA}
    A vertex operator algebra is a quadruple $(\V,Y,|0\rangle,|T\rangle)$ satisfying the following conditions.
    \begin{enumerate}
        \item The Hilbert space of states is $\V$. It is graded by conformal dimension with spectrum bounded from below. That is, $\V=\bigoplus_{h\in\ZZ} \V_h$ is a $\ZZ$-graded vector space such that $\dim\V_h<\infty$ and $\V_h=0$ when $h$ is sufficiently small.
        \item The theory is equipped with a state/operator correspondence. That is, there is a mapping\footnote{We use the notation $\varphi(z)$ and $Y(|\varphi\rangle,z)$ interchangeably. The latter is more common in mathematics, whereas the former may appear more natural to a physicist. We also use a physicist's convention for the mode expansion of $\varphi(z)$.} 
        \begin{align}
        \begin{split}
            Y:\V&\to \mathrm{End}(\V)[[z,z^{-1}]] \\
            |\varphi\rangle &\mapsto Y(|\varphi\rangle,z)=:\varphi(z) = \sum_{n\in\ZZ} \varphi_nz^{-n-h} \ \ \ (\varphi_n\in\mathrm{End}(\V))
        \end{split}
        \end{align}
        from states to operators, where $|\varphi\rangle\in \V_h$, and conversely, any state can be obtained from its corresponding operator as 
\begin{align}
    |\varphi\rangle =  \lim_{z\to 0}\varphi(z)|0\rangle = \varphi_{-h}|0\rangle.
\end{align}
Furthermore, the modes should satisfy the lower truncation condition, which says that 
\begin{align}
    \varphi_n|\psi\rangle = 0
\end{align}
whenever $n$ is taken sufficiently large.
\item The theory has a distinguished choice of vacuum $|0\rangle \in \V_0$ and conformal vector $|T\rangle\in \V_2$. That is, $|0\rangle$ is the vacuum state, whose corresponding operator is the identity on $\V$,
\begin{align}
    Y(|0\rangle,z)=\mathrm{id}_{\V},
\end{align}
        and $|T\rangle$ is the state whose corresponding operator is the stress tensor 
\begin{align}
    Y(|T\rangle,z)=T(z)\equiv \sum_{n\in\ZZ}L_nz^{-n-2}
\end{align}
whose modes we require to satisfy the Virasoro algebra, 
\begin{align}
    [L_m,L_n]=(m-n)L_{m+n} +\frac{c}{12}(m^3-m)\delta_{m+n,0}
\end{align}
for some central charge $c$.
We also require that $L_0$ be compatible with the grading and that $L_{-1}$ generate translations,
\begin{align}
\begin{split}
&L_0|\varphi\rangle = h|\varphi\rangle, \ \ \ |\varphi\rangle\in\V_h \\
& \ \ \     (L_{-1}\varphi)(z)=(\partial_z\varphi)(z).
\end{split}
\end{align}
\item The Jacobi identity should hold, 
\begin{align}
\begin{split}
    &z_0^{-1}\delta\left(\frac{z_1-z_2}{z_0}\right)\varphi(z_1)\tilde\varphi(z_2)-z_0^{-1}\delta\left(\frac{z_2-z_1}{-z_0}\right) \tilde\varphi(z_2)\varphi(z_1)\\
    & \hspace{1in} =z_2^{-1}\delta\left(\frac{z_1-z_0}{z_2}\right)(\varphi(z_0)\tilde\varphi)(z_2).
    \end{split}
\end{align}
\end{enumerate}
\end{defn}
Chiral algebras are essentially infinite-dimensional symmetry algebras which are present in any (1+1)D CFT. Whenever one has a symmetry structure in physics, it is important to understand its representation theory. In this work, we restrict attention to chiral algebras with the ``nicest'' representation theory: namely, we want to introduce as many adjectives as are needed so that $\Rep(\V)$, the category of $\V$-modules, forms a unitary modular tensor category. 

In physics, one would normally use the word ``rational'' to describe this situation. To a physicist, this refers to theories for which the Hilbert space decomposes into a finite number of irreducible representations of $\V\otimes \overline{\V}$. The word ``rational'' also appears in mathematical treatments of vertex operator algebras, but it deviates slightly from the physicists' notion. Since we will be invoking mathematical theorems throughout the text, we will opt to use the mathematicians' conventions instead, following mainly \cite{abe2004rationality}.
\begin{defn}[CFT-type]
    A vertex operator algebra is said to be of CFT-type if $\V=\bigoplus_{h\geq 0} \V_h$ and if $\V_0 = \BBC|0\rangle$.
\end{defn}
These conditions are completely natural to a physicist. They say that the identity operator is the only operator of conformal dimension $h=0$, and that no operator has lower conformal dimension. We emphasize however that there are of course very physical theories which violate these assumptions: e.g.\ the tensor product of a CFT with a (1+1)D TQFT, which has multiple ground states. 
\begin{defn}[$C_n$-cofiniteness]
    A vertex operator algebra is said to be $C_n$-cofinite for $n\geq 2$ if $\V\big/ C_n(\V)$ is finite-dimensional, where $C_n(\V)=\{\varphi_{-n+h-1}|\psi\rangle \mid |\varphi\rangle,|\psi\rangle\in \V\}$.
\end{defn}
It turns out that $C_2$-cofiniteness implies $C_n$-cofiniteness for all $n\geq 2$ \cite{Gaberdiel:2000qn}. $C_2$-cofiniteness of the chiral algebra is typically imposed alongside rationality, which deals more directly with the representation theory of $\V$. We turn to this definition next.

\begin{defn}[Weak, admissible, and ordinary modules]\label{defn:Vmodules}
    A weak $\V$-module is a vector space $M$ with a linear map 
    \begin{align}
    \begin{split}
        Y^{M}:\V&\to \mathrm{End}(M)[[z,z^{-1}]] \\
        |\varphi\rangle &\mapsto Y^M(|\varphi\rangle,z)=:\varphi^M(z)=\sum_{n\in\ZZ} \varphi^M_n z^{-n-h}.
    \end{split}
    \end{align}
    which satisfies the following axioms.
    \begin{enumerate}
        \item The pairings $\varphi_n|\psi\rangle$ vanish for $n\gg 0$, where $|\varphi\rangle\in \V$ and $|\psi\rangle\in M$.
        \item The vacuum of $\V$ is mapped to the identity map on $M$, i.e.\ $Y^M(|0\rangle,z) = \mathrm{id}_M$.
        \item The Jacobi identity holds, 
        \begin{align}
        \begin{split}
           & z_0^{-1}\delta\left(\frac{z_1-z_2}{z_0}\right)\varphi^M(z_1)\tilde{\varphi}^M(z_2)-z_0^{-1}\delta\left(\frac{z_2-z_1}{-z_0}\right)\tilde{\varphi}^M(z_2)\varphi^M(z_1) \\
        &\hspace{1in} =z_2^{-1}\delta\left(\frac{z_1-z_0}{z_2}\right)(\varphi(z_0)\tilde{\varphi})^M(z_2).
        \end{split}
        \end{align} 
        \end{enumerate}
    A weak $\V$-module is admissible if it carries a $\ZZ^+$-grading %
        \begin{align}
            M=\bigoplus_{\substack{n\in \ZZ\\n\geq 0}}M_{n}
        \end{align} 
        such that  $\varphi_m M_{n}\subset M_{n+m}$.\footnote{Note that we are not necessarily requiring that the grading match the action of $L_0$.} 
        
        A weak $\V$-module is ordinary if it carries a $\BBC$-grading 
        \begin{align}
            M=\bigoplus_{\lambda \in \BBC} M_\lambda
        \end{align}
        such that the following conditions are met.
        \begin{enumerate}
            \item The graded components are finite-dimensional, $\dim(M_\lambda)<\infty$.
            \item The spectrum truncates in the sense that $M_{\lambda+n}=0$ for fixed $\lambda$ and $n\ll 0$.
            \item The grading is compatible with the action of $L_0$, i.e.\ $L_0 |\varphi\rangle = \lambda |\varphi\rangle$ for every $|\varphi\rangle\in M_\lambda$. 
        \end{enumerate}
\end{defn}
It turns out that ordinary modules are also admissible. Thus, ordinary modules are perhaps closest to what a physicist imagines a $\V$-module looks like. Rationality is essentially the requirement of complete reducibility of admissible modules.
\begin{defn}[Rationality]
    A vertex operator algebra is rational if every admissible module is a direct sum of simple admissible modules.
\end{defn}
The following result reveals a connection between the mathematicians' notion of rationality formulated above, and what a physicist might conceive of as a rational chiral algebra.

\begin{theorem}[Theorem 8.1 of \cite{dong1998twisted}]
    Rationality  implies that the number of simple admissible modules is finite, and that every simple admissible module is an ordinary module.
\end{theorem}
We require one more ingredient (the notion of a contragredient module) to describe the nice class of VOAs to which we restrict attention in this text. Let $(M,Y_M)$ be an ordinary $\V$-module. For each $\lambda\in \BBC$, let $M_\lambda^\ast$ be the dual vector space of $M_\lambda$, and define $M^\ast = \bigoplus_{\lambda \in \BBC}M^\ast_\lambda$. Let $\langle ,\rangle:M^\ast\otimes M\to \BBC$ be the natural pairing between $M^\ast$ and $M$.

\begin{defn}[Contragredient/charge conjugate module]\label{defn:contragredientmodule}
 The contragredient module to $(M,Y_M)$ is the tuple $(M^\ast,Y_{M^\ast})$  where $Y_{M^\ast}$ is defined through the equation
 \begin{align}
     \langle Y_{M^\ast}(|\varphi\rangle,z)\psi'|\psi\rangle = \langle \psi'|Y_M(e^{zL_1}(-z^{-2})^{L_0}|\varphi\rangle,z^{-1})\psi\rangle
 \end{align}
 for every $|\varphi\rangle \in \V$, $|\psi\rangle \in M$, and $\langle \psi'|\in M^\ast$.
\end{defn}
\begin{defn}[Strongly regular]\label{defn:stronglyregular}
A vertex operator algebra is said to be strongly regular if it is simple (i.e.\ simple as a module over itself), self-contragredient (i.e.\ $\V^\ast$ is isomorphic to $\V$ as an ordinary $\V$-module), $C_2$-cofinite, rational, and of CFT-type.
\end{defn}

As previewed, the main utility of the notion of strong regularity of $\V$ is that it translates to the niceness of the representation category of $\V$.
\begin{theorem}[Main result of \cite{Huang:2005gs}]\label{theorem:modularRepV}
    If $\V$ is strongly regular, then $\Rep(\V)$ naturally carries the structure of a  modular tensor category.
\end{theorem}
\begin{defn}[Genus]
    The genus of $\V$ is the tuple $(c,\Cat)$ consisting of its central charge $c$, and the abstract modular tensor category $\Cat=\Rep(\V)$ formed by its modules.
\end{defn}
\begin{defn}[Holomorphic VOA]\label{defn:holomorphicVOA}
    A chiral algebra $\V$ with $\Rep(\V)\cong \Vec$ (i.e.\ a chiral algebra whose only primary operator is the identity) is said to be a holomorphic VOA.
\end{defn}
While strong regularity is enough to guarantee the modularity of $\Rep(\V)$, we will also assume unitarity, which can be formulated as follows.

\begin{defn}[Unitarity \cite{dong2014unitary}]
    Let $\V$ be a vertex operator algebra of CFT-type, and let $\phi$ be an anti-linear automorphism of order 2. That is, $\phi:\V\to \V$ is an anti-linear isomorphism of the vector space $\V$ which preserves the vacuum ($\phi|0\rangle = |0\rangle$), preserves the stress tensor ($\phi|T\rangle =|T\rangle$) and preserves the OPE/correlation functions ($\phi(\varphi_n|\psi\rangle)=\phi(|\varphi\rangle)_n\phi(|\psi\rangle)$). Then $\V$ is unitary if there exists a positive-definite Hermitian form $\langle,\rangle:\V\times \V\to \BBC$ which is linear on the first argument and anti-linear on the second argument, and is invariant in the sense that 
\begin{align}
    \langle \varphi | Y(\phi|\psi\rangle,z)\varphi'\rangle = \langle Y(e^{zL_1}(-z^{-2})^{L_0}|\psi\rangle,z^{-1})\varphi|\varphi'\rangle
\end{align}
for every $|\varphi\rangle$, $|\varphi'\rangle$, and $|\psi\rangle$ in $\V$.
\end{defn}
The modularity of the representation category of $\V$ has several useful implications for us. The first relates to the modular transformation property of the characters. Let $\V(i)$ run over the simple modules of $\V$, where $i=0,\dots,p-1$ and $\V(0):= \V$. Furthermore, define the character vector $\ch(\tau)$ of $\V$ through its components
\begin{align}
    \ch_i(\tau) = \mathrm{Tr}_{\V(i)}q^{L_0-\sfrac{c}{24}}.
\end{align}

\begin{theorem}[Proven in \cite{zhu1996modular}]
    If $\V$ is a strongly regular vertex operator algebra, then its character vector $\ch(\tau)$ transforms like a weakly-holomorphic vector-valued modular form of weight zero with respect to some modular representation $\varrho:\SL_2(\ZZ)\to \GL_p(\BBC)$ (see Appendix \ref{app:vvmfs} for the relevant definitions).
\end{theorem}
The precise modular representation can be distilled from the structure of $\Rep(\V)$ as a modular category. 
\begin{theorem}[See e.g.\ \cite{dong2015congruence}]\label{theorem:modularRepChiralAlgebra}
    The modular representation $\varrho$ is generated by the assignments 
    \begin{align}
        \varrho\left(\begin{smallmatrix} 0 & -1 \\ 1 & 0 \end{smallmatrix}\right) = \mathcal{S}, \ \ \ \ \ \ \varrho\left(\begin{smallmatrix} 1 & 1 \\ 0 & 1 \end{smallmatrix}\right) = e^{-\pi i c/12}\theta_i
    \end{align}
    where $\theta_i$ are the twists of $\Rep(\V)$, and $\mathcal{S}$ is the normalized S-matrix of Equation \eqref{eqn:normalizedSmatrix}.
\end{theorem}
Another reason that the modularity of the representation category of $\V$ is useful is that the extension theory of $\V$ can be completely recast in terms of condensable algebras in $\Rep(\V)$.
\begin{theorem}[A main result of \cite{Huang:2014ixa}]\label{theorem:VOAextensionfromcondensablealgebra}
    Let $\mathcal{W}$ be a strongly regular vertex operator algebra with $\Cat:=\Rep(\mathcal{W})$, and let $\V$ be a condensable algebra in $\Cat$. Then $\V$ is a conformal extension of $\mathcal{W}$ (i.e.\ $\V$ admits the structure of a VOA which contains $\mathcal{W}$ as a conformal subalgebra), and conversely every conformal extension of $\mathcal{W}$ arises in this way. The representation category of $\V$ is described by anyon condensation: that is, $\Rep(\V)\cong \mathcal{C}_{\V}^{\mathrm{loc}}$, the local $\V$-modules in $\mathcal{C}$. The extension $\V$ is a holomorphic VOA if and only if $\V$ is a Lagrangian algebra in $\Cat$.
\end{theorem}
The following is a useful necessary condition on condensable algebras/conformal extensions.
\begin{prop}\label{prop:modulardataextension}
Let $\mathcal{W}$ be a strongly regular vertex operator algebra with modular data $\mathcal{S}^{\mathcal{W}}$, $\mathcal{T}^{\mathcal{W}}$, and $\V$ a conformal extension with modular data $\mathcal{S}^\V$, $\mathcal{T}^\V$ which decomposes into $\mathcal{W}$-modules as 
\begin{align}
    \V(i)=\bigoplus_j N_{i,j} \mathcal{W}(j), \ \ \ \ (N_{i,j}\in \ZZ^{\geq 0}).
\end{align}
Then the matrix $N$ satisfies
\begin{align}\label{eqn:modularextension}
    N\mathcal{S}^{\mathcal{W}}=\mathcal{S}^{\V}N, \ \ \ \ \ \ N\mathcal{T}^{\mathcal{W}}=\mathcal{T}^\V N.
\end{align}
\end{prop}
\begin{remark}
    We emphasize that the existence of a matrix $N$ satisfying Equation \eqref{eqn:modularextension} is a necessary condition for a conformal extension to exist, but not a sufficient one. Alternatively, if one knows that a conformal extension $\mathcal{W}\subset \V$ must exist by some other argument, then one can use Proposition \ref{prop:modulardataextension} to constrain, or in many cases completely determine, the decomposition of $\V$ into $\mathcal{W}$-modules.
\end{remark}
In addition to anyon condensation, another useful way to construct new vertex operator algebras from old ones is the coset construction. 

\begin{defn}[Coset/commutant \cite{Goddard:1984vk,Goddard:1986ee,frenkel1992vertex}]\label{defn:commutant}
    Let $\V$ be a vertex operator algebra, and $\mathcal{W}\subset \V$ a vertex operator subalgebra which is not necessarily a conformal subalgebra. The commutant, or coset, is defined to be the vertex operator subalgebra of operators in $\V$ which commute with $\mathcal{W}$, i.e. 
    \begin{align}
        \Com_{\V}(\mathcal{W}):=\mathcal{V}\big/\mathcal{W}:=\{|\psi\rangle \in \V \mid \varphi_{n-h+1}|\psi\rangle =0 \text{ for all }|\varphi\rangle \in \mathcal{W} \text{ and } n\geq 0\}.
    \end{align}
\end{defn}
\begin{defn}[Primitive subalgebra]\label{defn:primitivesubalgebra}
Let $\mathcal{W}\subset \V$ be a subalgebra. Noting that $\mathcal{W}$ is always contained in its double commutant, $\mathcal{W}\subset \Com_\V(\Com_\V(\mathcal{W}))$, we say that $\mathcal{W}$ is a primitive subalgebra if $\mathcal{W}=\Com_{\V}(\Com_\V(\mathcal{W}))$. Note that if $\mathcal{W}$ has a completely anisotropic representation category, then it is automatically primitive whenever it occurs as a subalgebra. 
\end{defn}
It is often the case that if $\V$ and $\mathcal{W}\subset \V$ have some nice property, then  $\Com_\V(\mathcal{W})$ has it as well. The two nice properties we have considered thus far are unitarity and strong regularity. One has been settled and the other remains conjectural.

\begin{theorem}[Example 5.27 of \cite{carpi2018vertex}]\label{theorem:unitarycosets}
    If $\V$ and $\mathcal{W}\subset \V$ are both unitary, then so is the coset $\Com_\V(\mathcal{W})$.
\end{theorem}
\begin{conj}\label{conj:rationalcosets}
If $\V$ and $\mathcal{W}\subset \V$ are both strongly regular, then so is the coset $\Com_{\V}(\mathcal{W})$.
\end{conj}
Assuming the strong regularity of $\Com_\V(\mathcal{W})$, it is natural to ask what is its representation category. This is answered by the following. 

\begin{theorem}[Theorem 7.6, or Equation 1.15, of \cite{Frohlich:2003hm}]
    Let $\V$ and $\mathcal{W}\subset \V$ and $\Com_\V(\mathcal{W})$ all be strongly regular. Then the modular category of the commutant $\Com_\V(\mathcal{W})$ is some anyon condensation of $\Rep(\V)\boxtimes \overline{\Rep(\mathcal{W})}$. That is, there is some condensable algebra $A$ in $\Rep(\V)\boxtimes \overline{\Rep(\mathcal{W})}$ such that 
    \begin{align}\label{eqn:holomorphiccosetMTC}
        \Rep(\Com_\V(\mathcal{W})) \cong (\Rep(\V)\boxtimes \overline{\Rep(\mathcal{W})})_A^{\mathrm{loc}}.
    \end{align}
    In particular, if $\V$ is a holomorphic VOA, and if $\mathcal{W}$ is a primitive subalgebra, then 
\begin{align}
\Rep(\Com_\V(\mathcal{W}))\cong \overline{\Rep(\mathcal{W})}.
\end{align}
\end{theorem}

\begin{theorem}[See \cite{Lin:2016hsa,Creutzig:2019psu}]\label{theorem:gluingprelude}
    Let $\mathcal{A}$ be a holomorphic VOA, $\V$ a primitive subalgebra, and set $\tilde{\V}:=\Com_{\mathcal{A}}(\V)$. Then, $\V\otimes\tilde{\V}$ sits as a conformal subalgebra of $\mathcal{A}$, and the decomposition of $\mathcal{A}$ into $\V\otimes\tilde\V$-modules takes the form
\begin{align}\label{eqn:holvoadecomposition}
    \mathcal{A}\cong \bigoplus_i \V(i)^\ast\otimes \tilde{\V}( \phi i).
\end{align}
where $\phi:\Rep(\V)\to \Rep(\tilde{\V})\cong \overline{\Rep(\V)}$ defines a braid-reversing equivalence of modular categories. Moreover, from Example \ref{ex:algebrafromequivalence}, we learn that the converse is true as well. That is, if $\V$ and $\tilde{\V}$ are two strongly regular VOAs and $\phi:\Rep(\V)\to \Rep(\tilde\V)$ is a braid-reversing equivalence, then Equation \eqref{eqn:holvoadecomposition} defines a holomorphic vertex operator algebra. 
\end{theorem}

\clearpage
\subsection{Fermionic chiral algebras}\label{app:prelims:fermionicChiralAlgebras}

We start by recalling some relevant definitions and theorems, following \cite{Dong:2020jhn} closely. A fermionic vertex operator algebra/fermionic chiral algebra is essentially the same as a vertex operator algebra, but allowing for operators with half-integral spin. In physics, a fermionic vertex operator algebra arises as the chiral algebra of a fermionic CFT in the Neveu-Schwarz (NS) sector. For the sake of completeness, we provide the definition in more detail below; it is nearly identical to Definition \ref{defn:VOA}, but with the appropriate modifications made.

    \begin{defn}[Fermionic vertex operator algebra]\label{defn:fermionicVOA}
    A vertex operator algebra is a quadruple $(\V_{\mathrm{NS}},Y,|0\rangle,|T\rangle)$ satisfying the following conditions.
    \begin{enumerate}
        \item The Hilbert space of states is $\V_{\mathrm{NS}}$. It is graded by conformal dimension with spectrum bounded from below. That is, 
        \begin{align}
            \V_{\mathrm{NS}}=\bigoplus_{h\in\frac12\ZZ} (\V_{\mathrm{NS}})_h= \V_{\bar 0}\oplus \V_{\bar 1}
        \end{align} is a $\frac12\ZZ$-graded vector space such that $\dim\V_h<\infty$ and $\V_h=0$ when $h$ is sufficiently small, where 
\begin{align}
    \V_{\bar 0} = \bigoplus_{h\in \ZZ} (V_{\mathrm{NS}})_h, \ \ \ \ \ \V_{\bar 1}=\bigoplus_{h\in \ZZ+\frac12} (V_{\mathrm{NS}})_h.
\end{align}
        \item The theory is equipped with a state/operator correspondence. That is, there is a mapping
        \begin{align}
        \begin{split}
            Y:\V_{\mathrm{NS}}&\to \mathrm{End}(\V_{\mathrm{NS}})[[z,z^{-1}]] \\
            |\varphi\rangle &\mapsto Y(|\varphi\rangle,z)=:\varphi(z) = \sum_{n\in\ZZ+h} \varphi_nz^{-n-h} \ \ \ (\varphi_n\in\mathrm{End}(\V_{\mathrm{NS}}))
        \end{split}
        \end{align}
        from states to operators, where $|\varphi\rangle\in (\V_{\mathrm{NS}})_h$, and conversely, any state can be obtained from its corresponding operator as 
\begin{align}
    |\varphi\rangle =  \lim_{z\to 0}\varphi(z)|0\rangle = \varphi_{-h}|0\rangle.
\end{align}
Furthermore, the modes should satisfy the lower truncation condition, which says that 
\begin{align}
    \varphi_n|\psi\rangle = 0
\end{align}
whenever $n$ is taken sufficiently large.
\item The theory has a distinguished choice of vacuum $|0\rangle \in (\V_{\mathrm{NS}})_0$ and conformal vector $|T\rangle\in (\V_{\mathrm{NS}})_2$. That is, $|0\rangle$ is the vacuum state, whose corresponding operator is the identity on $\V_{\mathrm{NS}}$,
\begin{align}
    Y(|0\rangle,z)=\mathrm{id}_{\V_{\mathrm{NS}}},
\end{align}
        and $|T\rangle$ is the state whose corresponding operator is the stress tensor 
\begin{align}
    Y(|T\rangle,z)=T(z)\equiv \sum_{n\in\ZZ}L_nz^{-n-2}
\end{align}
whose modes we require to satisfy the Virasoro algebra, 
\begin{align}
    [L_m,L_n]=(m-n)L_{m+n} +\frac{c}{12}(m^3-m)\delta_{m+n,0}
\end{align}
for some central charge $c$.
We also require that $L_0$ be compatible with the grading and that $L_{-1}$ generate translations,
\begin{align}
\begin{split}
&L_0|\varphi\rangle = h|\varphi\rangle, \ \ \ |\varphi\rangle\in(\V_{\mathrm{NS}})_h \\
& \ \ \     (L_{-1}\varphi)(z)=(\partial_z\varphi)(z).
\end{split}
\end{align}
\item The Jacobi identity should hold, 
\begin{align}
\begin{split}
    &z_0^{-1}\delta\left(\frac{z_1-z_2}{z_0}\right)\varphi(z_1)\tilde\varphi(z_2)-(-1)^{ab}z_0^{-1}\delta\left(\frac{z_2-z_1}{-z_0}\right) \tilde\varphi(z_2)\varphi(z_1)\\
    & \hspace{1in} =z_2^{-1}\delta\left(\frac{z_1-z_0}{z_2}\right)(\varphi(z_0)\tilde\varphi)(z_2)
    \end{split}
\end{align}
where $\varphi\in \V_{\bar a}$ and $\tilde\varphi \in \V_{\bar b}$. 
\end{enumerate}
\end{defn}
\begin{remark}
    Fermionic VOAs with $\V_{\bar 1}=0$ are ordinary vertex operator algebras. 
\end{remark}
Every fermionic VOA with $\V_{\bar 1}\neq 0$ comes equipped with a $\ZZ_2$ automorphism generated by fermion parity $(-1)^F$, which by definition acts as $+1$ on $\V_{\bar 0}$ and $-1$ on $\V_{\bar 1}$. Importantly, the bosonic/even subalgebra $\V_{\bar 0}=\V_{\mathrm{NS}}^{(-1)^F}\subset \V_{\mathrm{NS}}$, which consists of states with integer conformal dimension, is a bosonic chiral algebra/vertex operator algebra. We define the ``nice'' class of fermionic chiral algebras to which we restrict attention in this work through this bosonic subalgebra. 
\begin{defn}[Strongly regular, unitary]
A fermionic VOA $\V_{\mathrm{NS}}=\V_{\bar 0}\oplus \V_{\bar 1}$ is said to be strongly regular if its bosonic subalgebra $\V_{\bar 0}$ is. Similarly, we restrict attention to fermionic VOAs for which $\V_{\bar 0}$ is unitary.
\end{defn}

\begin{ex}[Free fermions]
The NS sector of $n$ chiral Majorana fermions, which we denote $\mathrm{Fer}(n)$, defines a strongly regular fermionic vertex operator algebra. The bosonic subalgebra of $\mathrm{Fer}(n)$ is essentially $\mathfrak{so}(n)_1$, see Table \ref{tab:freefermions} for the more precise identification.
\end{ex}

\begin{table}[]
    \centering
    \begin{tabular}{c|c}
        $\V_{\mathrm{NS}}$ & $\V_{\bar 0}$  \\ \toprule
       $\mathrm{Fer}(1)$  & $L_{\sfrac12}$ \\
       $\mathrm{Fer}(2)$  & $\mathsf{U}_{1,4}$ \\
       $\mathrm{Fer}(3)$  & $\A[1,2]$ \\ 
       $\mathrm{Fer}(4)$  & $\A[1,1]^2$ \\
       $\mathrm{Fer}(5)$  & $\B[2,1]$ \\ 
       $\mathrm{Fer}(6)$  & $\A[3,1]$ \\ 
       $\mathrm{Fer}(2r+1)$  & $\B[r,1]$ \\
       $\mathrm{Fer}(2r)$  & $\D[r,1]$ \\\bottomrule
    \end{tabular}
    \caption{Bosonic subalgebras of free fermion theories. }
    \label{tab:freefermions}
\end{table}

The following celebrated result explains why the most interesting fermionic VOAs are those with no weight-$\sfrac12$ operators.

\begin{theorem}[Decoupled free fermion sector \cite{Goddard:1988wv}]\label{theorem:decoupledfreefermions}
    Let $\V_{\mathrm{NS}}$ be a fermionic VOA, and set $n = \dim (\V_{\mathrm{NS}})_{\sfrac12}$. The weight-$\sfrac12$ operators of $\V_{\mathrm{NS}}$ generate a decoupled sector of free fermions, 
\begin{align}
\langle (\V_{\mathrm{NS}})_{\sfrac12}\rangle \cong \mathrm{Fer}(n)
\end{align}
and $\V_{\mathrm{NS}}$ splits as 
\begin{align}
    \V_{\mathrm{NS}} \cong \V'_{\mathrm{NS}}\otimes \mathrm{Fer}(n)
\end{align}
where $\V'_{\mathrm{NS}}$ is a theory with no weight-$\sfrac12$ operators. 
\end{theorem}
Thus, for classification purposes, it is sufficient to enumerate the theories with no weight-$\sfrac12$ operators: the full set of fermionic VOAs can be obtained from this smaller set by taking arbitrary tensor products with free fermion theories.

A fermionic vertex operator algebra has two natural kinds of modules: untwisted and $(-1)^F$-twisted modules corresponding to the NS sector and Ramond (R) sector, respectively. We treat these in turn (cf.\ Definition \ref{defn:Vmodules}).

\begin{defn}[Weak, admissible, and ordinary NS-sector modules]\label{defn:NSmodules}
 A weak NS-sector module is a vector space $M$ with a linear map 
    \begin{align}
    \begin{split}
        Y^{M}:\V_{\mathrm{NS}}&\to \mathrm{End}(M)[[z,z^{-1}]] \\
        |\varphi\rangle &\mapsto Y^M(|\varphi\rangle,z)=:\varphi^M(z)=\sum_{n\in\ZZ+h} \varphi^M_n z^{-n-h}.
    \end{split}
    \end{align}
    which satisfies the following axioms.
    \begin{enumerate}
        \item The pairings $\varphi_n|\psi\rangle$ vanish for $n\gg 0$, where $|\varphi\rangle\in \V_{\mathrm{NS}}$ and $|\psi\rangle\in M$.
        \item The vacuum of $\V_{\mathrm{NS}}$ is mapped to the identity map on $M$, i.e.\ $Y^M(|0\rangle,z) = \mathrm{id}_M$.
        \item The Jacobi identity holds, 
        \begin{align}
        \begin{split}
           & z_0^{-1}\delta\left(\frac{z_1-z_2}{z_0}\right)\varphi^M(z_1)\tilde{\varphi}^M(z_2)-(-1)^{ab}z_0^{-1}\delta\left(\frac{z_2-z_1}{-z_0}\right)\tilde{\varphi}^M(z_2)\varphi^M(z_1) \\
        &\hspace{1in} =z_2^{-1}\delta\left(\frac{z_1-z_0}{z_2}\right)(\varphi(z_0)\tilde{\varphi})^M(z_2)
        \end{split}
        \end{align} 
        where $\varphi \in \V_{\bar a}$ and $\tilde{\varphi}\in\V_{\bar b}$. 
        \end{enumerate}
    A weak NS-sector $\V_{\mathrm{NS}}$-module is admissible if it carries a $\frac12\ZZ^+$-grading 
        \begin{align}
            M=\bigoplus_{\substack{n\in \frac12\ZZ\\n\geq 0}}M_{n}
        \end{align} 
        such that  $\varphi_m M_{n}\subset M_{n+m}$. In this case, we can split
            $M=M_{\bar 0}\oplus M_{\bar 1}$ into its even and odd parts, 
            \begin{align}\label{eqn:NSmodulegrading}
                M_{\bar 0} = \bigoplus_{\substack{n\in \ZZ\\ n\geq 0}}M_n, \ \ \ \ M_{\bar 1} = \bigoplus_{\substack{n\in \ZZ+\frac12 \\ n\geq 0}}M_n.
            \end{align}
        
        A weak NS-sector $\V_{\mathrm{NS}}$-module is ordinary if it carries a $\BBC$-grading 
        \begin{align}
            M=\bigoplus_{\lambda \in \BBC} M_\lambda
        \end{align}
        such that the following conditions are met.
        \begin{enumerate}
            \item The graded components are finite-dimensional, $\dim(M_\lambda)<\infty$.
            \item The spectrum truncates in the sense that $M_{\lambda+n}=0$ for fixed $\lambda$ and $n\ll 0$.
            \item The grading is compatible with the action of $L_0$, i.e.\ $L_0 |\varphi\rangle = \lambda |\varphi\rangle$ for every $|\varphi\rangle\in M_\lambda$. 
        \end{enumerate}
\end{defn}
It is also instructive to understand how the Ramond (R) sector arises in the VOA formulation. For this, we need the notion of a $(-1)^F$-twisted module.

\begin{defn}[Weak, admissible, ordinary, and stable R-sector modules]\label{defn:Rmodules}
 A weak R-sector module is a vector space $M$ with a linear map 
    \begin{align}
    \begin{split}
        Y^{M}:\V_{\mathrm{NS}}&\to \mathrm{End}(M)[[z^{\frac12},z^{-\frac12}]] \\
        |\varphi\rangle &\mapsto Y^M(|\varphi\rangle,z)=:\varphi^M(z)=\sum_{n\in\ZZ} \varphi^M_n z^{-n-h}
    \end{split}
    \end{align}
    which satisfies the following axioms.
    \begin{enumerate}
        \item The pairings $\varphi_n|\psi\rangle$ vanish for $n\gg 0$, where $|\varphi\rangle\in \V_{\mathrm{NS}}$ and $|\psi\rangle\in M$.
        \item The vacuum of $\V_{\mathrm{NS}}$ is mapped to the identity map on $M$, i.e.\ $Y^M(|0\rangle,z) = \mathrm{id}_M$.
        \item The Jacobi identity holds, 
        \begin{align}
        \begin{split}
           & z_0^{-1}\delta\left(\frac{z_1-z_2}{z_0}\right)\varphi^M(z_1)\tilde{\varphi}^M(z_2)-(-1)^{ab}z_0^{-1}\delta\left(\frac{z_2-z_1}{-z_0}\right)\tilde{\varphi}^M(z_2)\varphi^M(z_1) \\
        &\hspace{1in} =z_2^{-1}\left(\frac{z_1-z_0}{z_2} \right)^{-\frac12}\delta\left(\frac{z_1-z_0}{z_2}\right)(\varphi(z_0)\tilde{\varphi})^M(z_2)
        \end{split}
        \end{align} 
        where $\varphi \in \V_{\bar a}$ and $\tilde{\varphi}\in\V_{\bar b}$. 
        \end{enumerate}
    A weak R-sector $\V_{\mathrm{NS}}$-module is admissible if it carries a $\ZZ^+$-grading 
        \begin{align}
            M=\bigoplus_{\substack{n\in \ZZ\\n\geq 0}}M_{n}
        \end{align} 
        such that  $\varphi_m M_{n}\subset M_{n+m}$.
        
        A weak R-sector $\V_{\mathrm{NS}}$-module is ordinary if it carries a $\BBC$-grading 
        \begin{align}
            M=\bigoplus_{\lambda \in \BBC} M_\lambda
        \end{align}
        such that the following conditions are met.
        \begin{enumerate}
            \item The graded components are finite-dimensional, $\dim(M_\lambda)<\infty$.
            \item The spectrum truncates in the sense that $M_{\lambda+n}=0$ for fixed $\lambda$ and $n\ll 0$.
            \item The grading is compatible with the action of $L_0$, i.e.\ $L_0 |\varphi\rangle = \lambda |\varphi\rangle$ for every $|\varphi\rangle\in M_\lambda$. 
        \end{enumerate}
\end{defn}
From now on, we decorate modules with NS or R (i.e.\ we write $M_{\mathrm{NS}}$ or $M_{\mathrm{R}}$) when we want to emphasize what sector they belong to.
\begin{defn}
    An admissible NS-sector or R-sector module $M$ is said to be stable if $M$ and $M\circ (-1)^F$ are isomorphic, where $M\circ (-1)^F$ is the  module with the same underlying vector space as $M$, but with 
\begin{align}
    Y^{M\circ (-1)^F}(|\varphi\rangle,z) := Y^M((-1)^F|\varphi\rangle,z).
\end{align}
A map $\sigma:M\to M$ is said to be a lift of $(-1)^F$ if it plays well with respect to the action of $\V_{\mathrm{NS}}$ in the sense that 
\begin{align}
    \sigma Y^{M_{\mathrm{NS}}}(|\varphi\rangle,z)\sigma^{-1} = Y^{M_{\mathrm{NS}}}((-1)^F|\varphi\rangle,z).
\end{align}
We will then abuse notation and write $\sigma = (-1)^F$.
\end{defn}

\begin{remark}\label{remark:ambiguouslifts}
    The action of $(-1)^F$ on the Ramond sector is generally ambiguous: if $\sigma$ defines one lift of $(-1)^F$ on an R-sector module, then $-\sigma$ defines another. In a full CFT, the difference between these two choices amounts to stacking the theory with an Arf invariant, see e.g.\ the discussion around Equation (II.6) of \cite{Ji:2019ugf}.
\end{remark}

The following is a useful result which explains how NS-sector and R-sector modules behave with respect to the action of $(-1)^F$.
\begin{prop}\label{prop:modulestability}
    Admissible NS-sector modules are always stable, and admit involutive lifts of $(-1)^F$ which are  compatible with the grading $M_{\mathrm{NS}}=M_{\bar 0}\oplus M_{\bar 1}$ in Equation \eqref{eqn:NSmodulegrading} in the sense that $(-1)^F\vert_{M_{\bar i}}=(-1)^i$.   
    
    If an irreducible admissible R-sector module $M_{\mathrm{R}}$ is stable, then it admits an involutive lift of $(-1)^F$, whose eigenspaces induce a grading $M_{\mathrm{R}}=M_{\bar 0}\oplus M_{\bar 1}$.

    If an irreducible admissible R-sector module $M_{\mathrm{R}}'$ is not stable, then it can be made stable by considering the module $M_{\mathrm{R}} := M_{\mathrm{R}}'\oplus (M_{\mathrm{R}}'\circ (-1)^F)$. It then admits an involutive lift of $(-1)^F$, whose eigenspaces induce a grading $M_{\mathrm{R}}=M_{\bar 0}\oplus M_{\bar 1}$.

    The number of admissible R-sector modules which are irreducible as stable modules is the same as the number of irreducible NS-sector modules.
\end{prop}
From now on, following the above proposition, our convention for labeling the various modules of a strongly regular fermionic VOA is as follows. We label the admissible NS-sector modules of $\V_{\mathrm{NS}}$ as $\V_{\mathrm{NS}}(i)$, where $i=0,\dots,q-1$ and $\V_{\mathrm{NS}}(0):=\V_{\mathrm{NS}}$. We split the unstable irreducible R-sector modules into pairs as $\V_{\mathrm{R}}(j)'$ and $\V_{\mathrm{R}}(j)'\circ (-1)^F$ for $j=0,\dots,q'-1$, and we label the stable irreducible R-sector modues as $\V_{\mathrm{R}}(k)$ for $k=q',\dots,q-1$. Furthermore, we define the stable modules $\V_{\mathrm{R}}(j) = \V_{\mathrm{R}}(j)'\oplus (\V_{\mathrm{R}}(j)'\circ (-1)^F)$ for $j=0,\dots,q'-1$.

The above proposition can be used to relate the representation theory of the fermionic theory $\V_{\mathrm{NS}}$ to that of its bosonic subalgebra $\V_{\bar 0}$. In particular, we have the following. 
\begin{theorem}[Theorem 5.1 and 5.2 of \cite{Dong:2020jhn}]\label{theorem:modulesbosonicsubalgebra}
    The following are the complete set of irreducible admissible modules of the bosonic subalgebra, 
\begin{align}
\begin{split}
    &\text{Simple objects of }\Rep(\V_{\bar 0})=\{\V_{\mathrm{NS}}(i)_{\bar a}\mid i =0,\dots, q-1 \text{ and } a=0,1 \}\\ 
    & \ \ \ \ \cup \{\V_{\mathrm{R}}(j)' \mid j=0,\dots,q'-1\}\cup \{\V_{\mathrm{R}}(k)_{\bar a}\mid k=q',\dots,q-1 \text{ and }a=0,1\}.
\end{split}
\end{align}
\end{theorem}
We also have the following result, which provides a more abstract/categorical description of the relationship between the representation theory of $\V_{\mathrm{NS}}$ and its bosonic subalgebra $\V_{\bar 0}$.

\begin{theorem}\label{theorem:fermionicCAsuperMTC}
    Let $\V_{\mathrm{NS}}=\V_{\bar 0}\oplus\V_{\bar 1}$, and let $\Rep(\V_{\bar 0}\vert\V_{\bar 1})$ be the category generated by admissible $\V_{\bar 0}$-modules which occur inside of admissible NS-sector $\V_{\mathrm{NS}}$-modules, i.e.\ the category which is generated by the $\V_{\mathrm{NS}}(i)_{\bar a}$. Then $\Rep(\V_{\bar 0}\vert \V_{\bar 1})$ is a super MTC, and $\big(\Rep(\V_{\bar 0}),\V_{\bar 1})$ is a spin MTC which furnishes a minimal modular extension.
\end{theorem}

\begin{ex}\label{ex:oddc}
    Consider $\V_{\mathrm{NS}}=\mathrm{Fer}(1)=L_{\sfrac12}\oplus L_{\sfrac12}(\sfrac12)$, with bosonic subalgebra $\V_{\bar 0} = L_{\sfrac12}$. Then we have the identification 
\begin{align}
    (\V_{\mathrm{NS}})_{\bar 0} = L_{\sfrac12}, \ \ \ (\V_{\mathrm{NS}})_{\bar 1} = L_{\sfrac12}(\sfrac12), \ \ \ \V_{\mathrm{R}}'=L_{\sfrac12}(\sfrac1{16}), \ \ \ \V_{\mathrm{R}}'\circ(-1)^F=L_{\sfrac12}(\sfrac{1}{16}).
\end{align}
In other words, there are two unstable irreducible R-sector modules (which are isomorphic as $\V_{\bar 0}$-modules) and no stable ones (although what a physicist would call the R-sector of the free fermion is actually the ``stabilization'' of $\V_{\mathrm{R}}'$, i.e.\ the stable module $\V_{\mathrm{R}}:=\V_{\mathrm{R}}'\oplus (\V_{\mathrm{R}}'\circ (-1)^F)$). The representation category of the bosonic subalgebra is the Ising MTC (or $(\BB_8,1)$ in our notation), and furthermore $\big( (\BB_8,1),L_{\sfrac12}(\sfrac12)\big)$ is a spin MTC which furnishes a minimal modular extension of $\Rep(L_{\sfrac12}\vert L_{\sfrac12}(\sfrac12))\cong \textsl{sVec}$.
\end{ex}

\begin{ex}\label{ex:evenc}
Consider on the other hand $\V_{\mathrm{NS}}=\mathrm{Fer}(4)$ with bosonic subalgebra $\V_{\bar 0}=\A[1,1]^2$. In this case, for one choice of lift of $(-1)^F$ in the Ramond sector (see Remark \ref{remark:ambiguouslifts}) we have the identifications
\begin{align}
   ( \V_{\mathrm{NS}})_{\bar 0} = \A[1,1]^2, \ \ \ (\V_{\mathrm{NS}})_{\bar 1} = \A[1,1](\sfrac14)^2, \ \ \ (\V_{\mathrm{R}})_{\bar 0} = \A[1,1]\A[1,1](\sfrac14), \ \ \ (\V_{\mathrm{R}})_{\bar 1}=\A[1,1](\sfrac14)\A[1,1].
\end{align}
For the other choice, we find 
\begin{align}
    (\V_{\mathrm{R}})_{\bar 0}=\A[1,1](\sfrac14)\A[1,1], \ \ \ (\V_{\mathrm{R}})_{\bar 1} = \A[1,1]\A[1,1](\sfrac14).
\end{align}
This time, there is one stable irreducible R-sector module $\V_{\mathrm{R}}$, and no unstable ones. The representation category of the bosonic subalgebra is $(\AA_1,1)^{\boxtimes 2}$,  and furthermore $\big((\AA_1,1)^{\boxtimes 2},\A[1,1](\sfrac14)^2\big)$ is a spin MTC which furnishes another minimal modular extension of $\Rep(\A[1,1]^2\vert \A[1,1](\sfrac14)^2)\cong \textsl{sVec}$.
\end{ex}

\clearpage
\subsection{Lie algebras and current algebras}\label{app:liecurrentalgebras}

Given a finite-dimensional complex simple Lie algebra $\mathfrak{g}$, we use the symbol $(\mathfrak{g})_k$ to denote the affine Kac--Moody algebra/current algebra of level $k$, where $k$ is always assumed to be quantized to a positive integer. In a (1+1)D conformal field theory, this structure arises whenever there is a continuous global symmetry group with Lie algebra $\mathfrak{g}$: by Noether's theorem, there will be dimension-1 conserved currents $J^a(z)$ for each generator of $\mathfrak{g}$, and $(\mathfrak{g})_k$ characterizes their OPE. The level $k$ is related to the chiral anomaly of the continuous symmetry, and operationally arises as the leading term in the OPE, whose singular part takes the form 
\begin{align}\label{eqn:currentOPE}
    J^a(z)J^b(w) \sim k\frac{\delta^{ab}}{(z-w)^2}+i\frac{f^{abc}}{z-w}J^c(w)
\end{align}
where $a=1,\dots,\dim\mathfrak{g}$ and $f^{abc}$ are the structure constants of $\mathfrak{g}$.

It is important in this work to understand how to compute embeddings of the form 
\begin{align}\label{eqn:KMembedding}
\mathcal{K}:=(\mathfrak{h}_1)_{k_1}\otimes\cdots \otimes (\mathfrak{h}_n)_{k_n}\hookrightarrow \V
\end{align}
where $\V$ is an arbitrary chiral algebra. 
This subsection is dedicated to providing tools for performing such computations. Our approach is to reduce the problem to one of ordinary Lie algebras, which are easier to work with than their affine versions.

First, recall that the dimension-1 operators $\V_1$ in any VOA $\V$ generate a canonical Kac--Moody subalgebra,
\begin{align}\label{eqn:canonicalKMpart}
    \mathcal{K}_\V := \langle \V_1\rangle = (\mathfrak{g}_1)_{k_1'}\otimes \cdots\otimes (\mathfrak{g}_n)_{k_n'}\otimes \mathsf{U}_1^r.
\end{align}
In particular, a Kac--Moody algebra $\mathcal{K}$ is completely generated by its dimension-1 operators, $\mathcal{K}=\langle \mathcal{K}_1\rangle$. The embedding in Equation \eqref{eqn:KMembedding} must factor through this canonical Kac--Moody subalgebra, i.e. 
\begin{align}
    \mathcal{K}\hookrightarrow\mathcal{K}_\V\hookrightarrow \V.
\end{align}
In general when one studies homomorphisms (say e.g.\ of groups), it is sufficient to specify where the generators are mapped to. Likewise, in the case of Equation \eqref{eqn:KMembedding}, it is sufficient to specify where the dimension-1 operators go, via a map of the form $\mathcal{K}_1\hookrightarrow(\mathcal{K}_\V)_1= \V_1$.

This map $\mathcal{K}_1\hookrightarrow \V_1$ must of course respect the OPE. Recall that in any VOA, the dimension-1 states carry the structure of a Lie algebra, whose bracket is defined through the zero-modes of their corresponding operators, 
\begin{align}\label{eqn:LiealgebraV1}
    [\varphi,\varphi'] = \varphi_0\varphi', \ \ \ \ \varphi,\varphi' \in \V_1.
\end{align}
Thus, a necessary condition for the map $\mathcal{K}_1\hookrightarrow \V_1$ to respect the (second term of the) OPE in Equation \eqref{eqn:currentOPE} is that it define an embedding of finite-dimensional Lie algebras.

The other term in Equation \eqref{eqn:currentOPE} involves the levels, and we must preserve this as well. It turns out that compatibility with the levels  can be stated in terms of a classical notion known as the ``embedding index.'' Before we can define this, we need the notion of a Dynkin index.

\begin{defn}[Dynkin index]
Let $\mathfrak{h}$ be a simple complex finite-dimensional Lie algebra, and let $\mathcal{R}_\mu$ be an irreducible representation with weight $\mu$. Then the Dynkin index of $\mathcal{R}_\mu$ is defined as 
\begin{align}
     x_\mu = \frac{\dim \mathcal{R}_\mu}{2\dim\mathfrak{h}}(\mu,\mu+2\rho)
\end{align}
where $(\cdot,\cdot)$ is the Killing form, $\rho$ is the Weyl-vector, 
\begin{align}
    \rho = \frac12 \sum_{\alpha\in\Delta_+} \alpha,
\end{align}
and $\Delta_+$ are the positive roots. 
\end{defn}

Then, the most straightforward definition of the embedding index is as follows. 

\begin{defn}[Embedding index]\label{defn:embeddingindex}
Let $\mathfrak{h}\hookrightarrow\mathfrak{g}$ be a non-trivial embedding of simple Lie algebras. Pick an irreducible representation $\mathcal{R}$ of $\mathfrak{g}$, and decompose it into irreducible representations $\mathcal{R}_\mu$ of $\mathfrak{h}$ to obtain branching coefficients $b_\mu$,
\begin{align}
    \mathcal{R} = \bigoplus_\mu b_\mu \mathcal{R}_\mu.
\end{align}
Then the embedding index is defined as the positive integer
\begin{align}
    x(\mathfrak{h}\hookrightarrow \mathfrak{g}) = \sum_\mu b_\mu \frac{x_\mu}{x_{\mathcal{R}}}
\end{align}
where  $x_\mu$ and $x_{\mathcal{R}}$ are Dynkin indices. An embedding $\mathfrak{h}_1\oplus \cdots\oplus \mathfrak{h}_m\hookrightarrow \mathfrak{g}$ has $m$ embedding indices associated to it, one for each $\mathfrak{h}_i\hookrightarrow\mathfrak{g}$. A trivial embedding of $\mathfrak{h}$ is defined to have embedding index equal to $0$.
\end{defn}
Now, we can state the following useful result.

\begin{prop}
    A Lie algebra embedding $\mathfrak{h}_1\oplus\cdots\oplus \mathfrak{h}_m\hookrightarrow\mathfrak{g}_1\oplus\cdots\oplus \mathfrak{g}_n$ extends to an embedding of Kac--Moody algebras $(\mathfrak{h}_1)_{k_1}\otimes\cdots\otimes(\mathfrak{h}_m)_{k_m}\hookrightarrow(\mathfrak{g}_1)_{k'_1}\otimes\cdots\otimes(\mathfrak{g}_n)_{k'_n}$ if and only if 
\begin{align}
    k_i = \sum_j x(\mathfrak{h}_i\hookrightarrow\mathfrak{g}_j)k'_j
\end{align}
for $i=1,\dots,m$. Conversely, every Kac--Moody embedding arises in this way.
\end{prop}

\begin{ex}\label{ex:level1embedding}
    Consider an embedding of the form 
    \begin{align}
        (\mathfrak{h})_1\hookrightarrow (\mathfrak{g}_1)_{k'_1}\otimes \cdots\otimes(\mathfrak{g}_n)_{k'_n}.
    \end{align} 
    These precisely correspond to Lie algebra embeddings in which $\mathfrak{h}$ is only embedded non-trivially into a single simple factor $\mathfrak{g}_i$ with $k'_i=1$ and with embedding index $x_i=1$ ($\mathfrak{h}$ is then embedded trivially into all other simple factors).
\end{ex}

\begin{ex}
    Consider an embedding of the form 
    \begin{align}
        (\mathfrak{h})_k \hookrightarrow (\mathfrak{g}_1)_{k'_1}\otimes (\mathfrak{g}_2)_{k'_2}.
    \end{align}
    This is the same as specifying a pair of Lie algebra embeddings $\mathfrak{h}\hookrightarrow \mathfrak{g}_i$ for $i=1,2$ with embedding indices $x_i$ satisfying 
\begin{align}
    k=x_1k'_1+x_2k'_2.
\end{align}
For example, when $\mathfrak{g}_1=\mathfrak{g}_2=\mathfrak{h}$, there are only two possible Lie algebra embeddings of the form $\mathfrak{h}\hookrightarrow\mathfrak{h}\oplus\mathfrak{h}$. The first involves embedding $\mathfrak{h}$ directly into $\mathfrak{g}_1$ with embedding index $x_1=1$, and trivially into $\mathfrak{g}_2$ with $x_2=0$,  in which case we induce Kac--Moody embeddings of the form
\begin{align}
    (\mathfrak{h})_k\hookrightarrow (\mathfrak{h})_k\otimes(\mathfrak{h})_{k'}
\end{align}
for any positive integers $k,k'$. The second is more interesting and involves embedding $\mathfrak{h}$ diagonally into $\mathfrak{g}_1\oplus \mathfrak{g}_2$ with $x_1=x_2=1$, in which case we induce Kac--Moody embeddings of the form 
\begin{align}
    (\mathfrak{h})_{k+k'}\hookrightarrow (\mathfrak{h})_k\otimes (\mathfrak{h})_{k'}
\end{align}
for any positive integers $k,k'$. Famously, when $\mathfrak{h}=\mathfrak{su}(2)$ and $k=1$, the corresponding cosets $\mathfrak{su}(2)_1\otimes\mathfrak{su}(2)_k\big/\mathfrak{su}(2)_{k+1}$ lead to the minimal models. 
\end{ex}
The coset of two VOAs is, on its own, not the most useful description of a chiral algebra. The following proposition is helpful because it allows one to at least compute the Kac--Moody part of any coset. 

\begin{prop}\label{prop:KMofCoset}
    Consider an embedding of chiral algebras of the form 
\begin{align}
    \mathcal{K}=(\mathfrak{h}_1)_{k_1}\otimes \cdots\otimes(\mathfrak{h}_m)_{k_m}\hookrightarrow \V
\end{align}
where $\V$ has a canonical Kac--Moody algebra given by 
\begin{align}
    \mathcal{K}_\V=(\mathfrak{g}_1)_{k_1'}\otimes \cdots \otimes (\mathfrak{g}_n)_{k_n'}.
\end{align}
Furthermore, let 
\begin{align}\label{eqn:inducedembeddingLiealgebras}
    \mathfrak{h}:=\mathfrak{h}_1\oplus\cdots\oplus\mathfrak{h}_m\hookrightarrow \mathfrak{g}_1\oplus\cdots\oplus\mathfrak{g}_n=:\mathfrak{g}
\end{align}
be the induced embedding of semi-simple Lie algebras, and decompose the centralizer as 
\begin{align}
    \mathrm{Cent}_{\mathfrak{g}}(\mathfrak{h}) = \mathfrak{h}'_1\oplus \cdots\oplus \mathfrak{h}'_\ell\oplus\mathsf{U}_1^r.
\end{align}
Then the canonical Kac--Moody subalgebra of the coset $\V\big/ (\mathfrak{h}_1)_{k_1}\otimes\cdots\otimes(\mathfrak{h}_m)_{k_m}$ is given by 
\begin{align}
    (\mathfrak{h}'_1)_{\tilde{k}_1}\otimes\cdots\otimes(\mathfrak{h}'_m)_{\tilde{k}_m}\otimes\mathsf{U}_1^r
\end{align}
where the levels $\tilde{k}_i$ can be computed from the embedding indices, 
\begin{align}
    \tilde{k}_i = \sum_j x(\mathfrak{h}'_i\hookrightarrow \mathfrak{g}_j)k'_j.
\end{align}
\end{prop}
  For most of the cosets considered in this work, the Kac--Moody part is ``large enough'' that the remainder can be straightforwardly determined in terms of more primitive ingredients, like minimal models, using standard arguments. We relied strongly on \cite{Feger:2012bs,Feger:2019tvk}, and especially on \cite{de2011constructing} and its corresponding implementation in Gap \cite{GAP4}, to compute equivalence classes of Lie algebra embeddings, their centralizers, and their corresponding embedding indices. All of the requisite embeddings and their properties are given in tables throughout Appendix \ref{app:data}.

Now, let $\iota_1$ and $\iota_2$ be two embeddings of the form $\mathcal{K}\hookrightarrow \V$. For classification purposes, it is important that we be able to tell when the two corresponding cosets  $\V\big/\iota_i\mathcal{K}$ are isomorphic to one another, and when they are inequivalent. A useful criterion to test for equivalence is the following: if there is an automorphism $\phi:\V\to\V$ which maps $\iota_1\mathcal{K}$ into $\iota_2\mathcal{K}$, then $\phi$ also induces an isomorphism of the corresponding cosets as well. If $\V$ does not have an automorphism which relates $\iota_1\mathcal{K}$ to $\iota_2\mathcal{K}$, then further consideration is required.

Call $\tilde\iota_i$ the induced embeddings of ordinary Lie algebras $\mathfrak{h}\hookrightarrow\mathfrak{g}$, as in Equation \eqref{eqn:inducedembeddingLiealgebras}. Recall that the Lie group $G$ obtained by exponentiation acts on $\mathfrak{g}$ by conjugation. In the simplest cases, the subalgebras $\tilde\iota_1(\mathfrak{h}),\tilde\iota_2(\mathfrak{h})\subset \mathfrak{g}$ are equivalent in the sense that they are mapped into one another by this action of $G$. If this is the case, then it is also true that the $\iota_i\mathcal{K}$ will be mapped to one another by an automorphism of $\V$; indeed, exponentials of the zero-modes of dimension-1 operators always define continuous symmetries of $\V$, which act on the dimension-1 subspace $\V_1$ in the same way that $G$ acts on $\mathfrak{g}$. 

The more subtle possibility is that the two subalgebras $\tilde\iota_1(\mathfrak{h})$ and $\tilde\iota_2(\mathfrak{h})$ are mapped to one another by an \emph{outer} automorphism of $\mathfrak{g}$, such as an automorphism which permutes isomorphic simple factors of $\mathfrak{g}$. Such outer automorphisms sometimes do lift to automorphisms of $\V$, but they also sometimes do not. To this end, Appendix A of \cite{Betsumiya:2022avv} contains invaluable information on how the automorphism groups of the $c=24$ holomorphic VOAs permute the simple factors of their corresponding Kac--Moody algebras.

\begin{ex}
    Let $\mathcal{A}=\mathbf{S}(\A[1,1]^{24})$, the VOA based on the Niemeier lattice with root system $\A[1]^{24}$.  A priori, there are 24 cosets of the form $\mathcal{A}\big/ \A[1,1]$, corresponding to the 24 different factors in the numerator theory into which we can embed the denominator. However, Appendix A of \cite{Betsumiya:2022avv} reveals that the automorphism group of $\mathcal{A}$ induces the permutation action of the sporadic Mathieu group $\textsl{M}_{24}$ on the 24 $\A[1,1]$ factors. Since this permutation action of $\textsl{M}_{24}$ is 5-transitive, we conclude that all 24 $\A[1,1]$ factors are exchanged by automorphisms of $\mathcal{A}$, and therefore the 24 corresponding cosets are all isomorphic.
\end{ex}

\begin{ex}
    Let $\mathcal{A} = \mathbf{S}(\A[7,4]\A[1,1]^3)$, and consider cosets of the form $\mathcal{A}\big/\A[1,2]$. We must embed $\A[1,2]$ diagonally into two of the $\A[1,1]$ factors of the numerator theory, but there are a priori $3={3\choose 2}$ different ways to do this. This time, Appendix A of \cite{Betsumiya:2022avv} reveals that the automorphism group of $\mathcal{A}$ permutes just two of the three $\A[1,1]$ factors, and fixes the third one. Decorating this fixed $\A[1,1]$ factor with a prime, we find that there are two inequivalent cosets, 
\begin{align}
\begin{split}
\mathbf{S}(\A[7,4]\A[1,1]^2\A[1,1]')\big/(\A[1,2]\hookrightarrow \A[1,1]^2) \\ 
\mathbf{S}(\A[7,4]\A[1,1]^2\A[1,1]')\big/(\A[1,2]\hookrightarrow\A[1,1]\A[1,1]').
\end{split}
\end{align}
See the cases labeled as exceptions in \cite{Mukhi:2022bte} for more examples like this one. 
\end{ex}

This exhausts the criteria we use to test for equivalence of two cosets. Checking non-equivalence is much easier: in most cases (with very few exceptions, see e.g.\ \S 3.3.3 of \cite{Mukhi:2022bte}) we can simply check whether the canonical Kac--Moody algebras are equal or not, computed using Proposition \ref{prop:KMofCoset}. 

Finally, we conclude by briefly discussing $\mathsf{U}_1$ factors. In a strongly regular VOA, the Abelian part $\mathsf{U}_1^r$ of the canonical Kac--Moody algebra actually sits inside of a lattice VOA as a conformal subalgebra, 
\begin{align}
    \mathsf{U}_1^r\hookrightarrow V_L\hookrightarrow \V
\end{align}
for some $L$ of dimension $r$. This lattice VOA is by definition 
\begin{align}
    V_L:=\Com_\V(\Com_\V(\mathsf{U}_1^r)).
\end{align}
In a sense, this lattice is to $\mathsf{U}_1^r$ what the level $k$ is to a simple Lie algebra $\mathfrak{g}$. In many cases, a sublattice of the full lattice $L$ can be straightforwardly computed for a coset. We give just one example to serve as illustration.

\begin{ex}
    Consider the embedding $\A[r,1]\hookrightarrow \A[n,1]$ induced by the embedding $\A[r]\hookrightarrow\A[n]$ with embedding index $1$. Then the coset is 
\begin{align}
    \A[n,1]\big/\A[r,1]\cong \A[n-r-1,1]V_{L_r}
\end{align}
where $L_r$ is the one-dimensional lattice generated by a vector $v$ with 
\begin{align}\label{eqn:normsquared}
     v^2 = \frac{(n-r)(r+1)(n+1)}{k^2}, \ \ \ \ k=\mathrm{gcd}(n-r,r+1).
\end{align}

To see this, note that $\A[n,1]$ and $\A[r,1]$ are both themselves lattice VOAs, e.g.\ $\A[n,1]$ is associated to the lattice 
\begin{align}
    L_{\A[n]} = \left\{ (a_0,\dots,a_n) \mid a_i\in \ZZ, ~\sum_i a_i = 0 \right\}.
\end{align}
The lattice VOA associated to the orthogonal complement of $L_{\A[r]}$ inside of $L_{\A[n]}$ participates as a subalgebra of the commutant $\A[n,1]\big/\A[r,1]$. This orthogonal complement takes the form 
\begin{align}
    L_{\A[r]}^\perp = \left\{ (a_0,a_0,\dots,a_0,a_{r+1},a_{r+2},\dots,a_n)\in L_{\A[n]}  \right\}
\end{align}
The sublattice of $L^\perp_{\A[r]}$ consisting of vectors with $a_0=0$ is the lattice $L_{\A[n-r-1]}$. We are interested in the one-dimensional lattice which is orthogonal to $L_{\A[n-r-1]}$ inside of $L^\perp_{\A[r]}$. Its lattice vectors must take the form 
\begin{align}
   V_{L_r}=\{ (a_0,\dots,a_0,a_{r+1},\dots,a_{r+1})\mid (r+1)a_0+a_{r+1}(n-r)=0\}.
\end{align}
One can take a generating vector $v$ of this lattice to have components $a_0=(n-r)/k$ and $a_{r+1}=-(r+1)/k$, where $k=\mathrm{gcd}(n-r,r+1)$, which then has norm-squared given by Equation \eqref{eqn:normsquared}.

\end{ex}

\clearpage 
\section{Vector-Valued Modular Forms}\label{app:vvmfs}

In this appendix, we briefly describe a general theory of vector-valued modular forms due to Bantay and Gannon \cite{Bantay:2005vk,Bantay:2007zz,Gannon:2013jua} (see also \S 6.3.2 of \cite{Cheng:2020srs} for a concise review). In particular, we explain how this theory can be used to algorithmically determine the full set of $(c,\Cat)$-admissible vector-valued modular forms when $\mathrm{rank}(\Cat)\leq 4$ (see \S\ref{subsec:overview:characters} for the definition of $(c,\Cat)$-admissible). Our exposition follows \cite{Gannon:2013jua} closely, and we refer readers to op.\ cit.\ for any details which we have omitted.

\subsection{Generalities}\label{app:vvmfs:generalities}

Let $\mathbb{H}=\{\tau\in\BBC \mid \Im\tau >0\}$ be the upper half-plane, and let $\SL_2(\ZZ)$ be the modular group of $2\times 2$ integer matrices with unit determinant. We write 
\begin{align}
\left(\begin{smallmatrix} a & b \\ c & d\end{smallmatrix}\right) \cdot \tau = \frac{a\tau+b}{c\tau+d}
\end{align}
for the standard group action of $\SL_2(\ZZ)$ on $\mathbb{H}$ by fractional linear transformations. Further consider a $d$-dimensional representation 
\begin{align}
\varrho:\SL_2(\ZZ)\to \GL_d(\BBC).
\end{align}
The modular group is generated by 
\begin{align}
S= \left( \begin{smallmatrix} 0 & -1 \\ 1 & 0 \end{smallmatrix}\right), \ \ \ \ \ T=\left(\begin{smallmatrix} 1 & 1  \\ 0 & 1 \end{smallmatrix}\right)
\end{align}
and so the entire content of the representation $\varrho$ can be summarized by specifying the two $d\times d$ matrices
\begin{align}
\mathcal{S} = \varrho(S), \ \ \ \ \ \mathcal{T}=\varrho(T)
\end{align}
to which $S$ and $T$ are mapped by $\varrho$. We will assume throughout that $\mathcal{T}$ is a diagonal matrix, as is true for any representation $\varrho$ which arises in rational conformal field theory.

As described in Appendix \ref{app:prelims:chiralalgebras} (see e.g.\ Theorem \ref{theorem:modularRepChiralAlgebra}), the characters 
\begin{align}
\ch_i(\tau) = \mathrm{Tr}_{\V_i}q^{L_0-\sfrac{c}{24}}
\end{align}
of the chiral algebra of a rational conformal field theory form the components of a weight-zero, weakly-holomorphic vector-valued modular form with respect to some representation $\varrho$ of the modular group \cite{zhu1996modular}. Recall that a vector-valued function $X:\mathbb{H}\to \mathbb{C}^d$ is said to be a weakly-holomorphic modular form with respect to $\varrho$ if it satisfies the functional equation 
\begin{align}\label{eqn:vvformtransformation}
X(\gamma\cdot\tau) = \varrho(\gamma) X(\tau)
\end{align}
for all $\gamma$ in $\SL_2(\ZZ)$, and if it is holomorphic in the interior of the upper half-plane $\mathbb{H}$ and meromorphic at the cusps $\QQ\cup \{i\infty\}$. We are interested in characterizing the space $\mathcal{M}^!_0(\varrho)$ of all such functions. As we will see, $\mathcal{M}^!_0(\varrho)$ can be efficiently computed from knowledge of just two $d\times d$ matrices: a bijective exponent $\lambda$ and a corresponding characteristic matrix $\chi$. Let us explain how these objects are defined.

Call a $d\times d$ matrix $\lambda$ an \emph{exponent} for $\varrho$ if it satisfies 
\begin{align}\label{eqn:exponent}
\mathcal{T}_{i,i} = \exp \big(2\pi i \lambda_{i,i}\big).
\end{align}
Note that Equation \eqref{eqn:vvformtransformation} applied to $\gamma = T$ implies that $X(\tau)$ admits a Fourier expansion of the form 
\begin{align}\label{eqn:Fourierexpansion}
X_i(\tau) = q^{\lambda_{i,i}} \sum_{\substack{n\in\ZZ \\ n\gg -\infty}} X_{i,n} q^n
\end{align}
for any exponent $\lambda$, where $X_i(\tau)$ are the components of $X$, with $i=0,\dots,d-1$. We will refer to the Fourier coefficients $X_{i,n}$ with $n\leq 0$ as the $\lambda$-singular coefficients of $X$. We call an exponent \emph{bijective} if a vector-valued modular form in $\mathcal{M}_0^!(\varrho)$ is uniquely determined by its $\lambda$-singular coefficients. In other words, $\lambda$ is said to be bijective if it is the case that any two vector-valued modular forms $X(\tau),X'(\tau) \in \mathcal{M}^!_0(\varrho)$ with $X_{i,n}=X_{i,n}'$ for all $i$ and all $n\leq 0$ are completely equal to one another, $X(\tau) = X'(\tau)$.
\begin{theorem}[Theorem 3.2 of \cite{Gannon:2013jua}]\label{theorem:bijectiveexponents}
Any modular representation $\varrho$ for which $\mathcal{T}$ is diagonal and $\varrho(1)=\varrho(-1)=\mathds{1}$ admits a bijective exponent.
\end{theorem}
Naively, this theorem appears not to cover representations $\varrho$ with a non-trivial charge conjugation matrix $\varrho(-1)$, which do arise in the context of conformal field theory. For example, the Kac--Moody algebra $\mathsf{A}_{2,1}$ has $\varrho(-1)$ given by a non-trivial permutation matrix. In such cases, $\varrho:\SL_2(\Z)\to \GL(V)$ is reducible and one has a non-trivial decomposition of the form $V=V_+\oplus V_-$, where 
\begin{align}\label{eqn:pmdecomposition}
V_\pm = \{v\in V \mid \varrho(-1)v=\pm v\}.
\end{align}
The subrepresentation $V_-$ supports no non-trivial vector-valued modular forms,\footnote{Indeed, Equation \eqref{eqn:vvformtransformation} applied to $\gamma=-1$ for a representation $\varrho_-$ which assigns $\varrho_-(-1)=-1$ would imply that $X(\tau)=-X(\tau)$.\label{footnote:trivialspace}} and so the analysis can be carried out entirely with respect to the subrepresentation $V_+$, which does satisfy the hypotheses of Theorem \ref{theorem:bijectiveexponents}. We will see an explicit example in Appendix \ref{app:vvmfs:examples}.

The existence of a bijective exponent implies that $\mathcal{M}_0^!(\varrho)$ admits a convenient basis 
\begin{align}\label{eqn:modularformbasis}
\mathcal{B} = \{ X^{(j,m)}(\tau) \mid j=0,\dots,d-1 \text{ and } \ m \geq 0\}
\end{align}
of vector-valued modular forms, where $X^{(j,m)}(\tau)$ is the unique function with 
\begin{align}\label{eqn:basiscondition}
X^{(j,m)}_{i,-n} = \delta_{i,j} \delta_{m,n} 
\end{align}
for every $i,j = 0,\dots,d-1$ and every $m,n\geq 0$. In other words, $X^{(j,m)}(\tau)$ is the unique function in $\mathcal{M}_0^!(\varrho)$ whose only $\lambda$-singular Fourier coefficient occurs in the $j$th component in front of $q^{\lambda_{j,j}-m}$ and is equal to $1$.

In practice, the following theorem is the main method we use to locate bijective exponents.
\begin{theorem}[Theorem 4.1 of \cite{Gannon:2013jua}]
If $\varrho$ is an irreducible representation of dimension $d\leq 5$ which assigns $\varrho(1)=\varrho(-1)=\mathds{1}$, then an exponent $\lambda$ is bijective if and only if
\begin{align}
\mathrm{Tr}\lambda = -\frac{7d}{12} + \frac14\mathrm{Tr}\mathcal{S} + \frac{2}{3\sqrt{3}}\Re\left( e^{-\pi i/6}\mathrm{Tr}(\mathcal{S}\mathcal{T}^{-1})  \right).
\end{align}
\end{theorem}
Again, the assumption that $\varrho(-1)=\mathds{1}$ is mild, and similarly the assumption of irreducibility is innocuous because one is always free to decompose any reducible representation into its irreducible constituents, and study each constituent independently. 

The following theorem illustrates that the most important basis elements to determine in $\mathcal{B}$ are those with $m=0$. Let 
\begin{align}
J(\tau) = 1728 \frac{E_4(\tau)^3}{E_4(\tau)^3-E_6(\tau)^2}=q^{-1}+744+196884q+\cdots
\end{align}
be the elliptic j-invariant, where $E_k(\tau)$ is the weight $k$ Eisenstein series.

\begin{theorem}[Theorem 3.3(a) of \cite{Gannon:2013jua}]\label{theorem:freemodule}
The space $\mathcal{M}^!_0(\varrho)$ is a free module over $\BBC[J(\tau)]$ with rank $d$. The free generators are $X^{(j,0)}(\tau)$ for $j=0,\dots, d-1$.
\end{theorem}
In other words, any vector-valued modular form $X(\tau)$ in $\mathcal{M}^!_0(\varrho)$ can be written as 
\begin{align}
X(\tau) = \sum_{j=0}^{d-1} P_j(J(\tau)) X^{(j,0)}(\tau)
\end{align}
for some polynomials $P_j(x)$. In particular, when $X(\tau)=X^{(k,m)}(\tau)$, the polynomials can be taken to be of degree $m$. Thus, the entire space $\mathcal{M}_0^!(\varrho)$ can be rapidly computed just from knowledge of $d$ of its basis elements in $\mathcal{B}$. In fact, these $d$ functions can themselves be distilled from just $d^2$ numbers. Indeed, place these forms into a matrix-valued function 
\begin{align}\label{eqn:fundamentalmatrix}
\Xi_{i,j}(\tau) = X^{(j,0)}_i(\tau), \ \ \ \ \ \Xi_{i,j,n}=X^{(j,0)}_{i,n}
\end{align}
and define the $d\times d$ matrix
\begin{align}\label{eqn:characteristicmatrix}
\chi_{i,j}= X^{(j,0)}_{i,1}.
\end{align}
Further define the auxilliary functions 
\begin{align}
\frac{(J(\tau)-984)\Delta(\tau)}{E_{10}(\tau)} =: \sum_{n=0}^\infty f_nq^n, \ \ \ \ \frac{\Delta(\tau)}{E_{10}(\tau)} =: \sum_{n=0}^\infty g_nq^n 
\end{align}
where $\Delta(\tau)=\eta(\tau)^{24}$ is the weight 12 cusp form, with $\eta(\tau)$ the Dedekind eta function. 
\begin{theorem}[Theorem 3.3(b) and Equation (36) of \cite{Gannon:2013jua}]\label{theorem:recursion}
Let $\Xi_{(n)}$ be the $d\times d$ matrix whose $ij$ entry is $\Xi_{i,j,n}$. Then the following recursion relation holds,
\begin{align}\label{eqn:recursion}
[\lambda,\Xi_{(n)}]+n\Xi_{(n)} = \sum_{l=0}^{n-1}\Xi_{(l)}\big(f_{n-l}\lambda+g_{n-l}(\lambda+[\lambda,\chi])\big).
\end{align}
\end{theorem}
The consequence of this theorem is that the full $q$-expansions of the $X^{(j,0)}(\tau)$ can be reconstructed just from $\lambda$ and $\chi$ using the recursion relation in Equation \eqref{eqn:recursion} supplemented with the initial condition $\Xi_{(0)}=\mathds{1}$. (It is useful to note that the $ij$ entry of the left-hand side of the recursion relation is $(\lambda_{i,i}-\lambda_{j,j}+n)\Xi_{i,j,n}$.) 

Thus, we are motivated to compute the function $\Xi(\tau)$ and extract its associated matrix $\lambda$. In practice, this can be achieved once one knows just a single example of a function $X(\tau)$ in $\mathcal{M}^!_0(\varrho)$. Indeed, define the following differential operators, 
\begin{align}
\nabla_{1,w} = \frac{E_4 E_6}{\Delta}D_w, \ \ \ \nabla_{2,w}=\frac{E_4^2}{\Delta}D_{w+2}\circ D_w, \ \ \ \nabla_{3,w}=\frac{E_6}{\Delta}D_{w+4}\circ D_{w+2} \circ D_w
\end{align}
where $D_w$ is the standard modular derivative 
\begin{align}
D_w = \frac{1}{2\pi i } \frac{d}{d\tau}-\frac{w}{12}E_2.
\end{align}
The differential operators $\nabla_i \equiv \nabla_{i,0}$ all map $\mathcal{M}^!_0(\varrho)$ back into itself. In fact, the following theorem shows that they are effective at moving one around anywhere within the space.

\begin{theorem}[Proposition 3.2 of \cite{Gannon:2013jua}]\label{theorem:movearound}
Let $\varrho$ be a representation with $\mathcal{T}$ a diagonal matrix, and with $\varrho(1)=\varrho(-1)=\mathds{1}$. Let $X(\tau)$ be a modular form in $\mathcal{M}^!_0(\varrho)$ with linearly independent components. Then 
\begin{align}
\mathcal{M}_0^!(\varrho) = \BBC[J,\nabla_{1},\nabla_2,\nabla_3]X(\tau).
\end{align}
\end{theorem}
The requirement that $X(\tau)$ have linearly independent components is not a serious restriction because when an $X(\tau)\in\mathcal{M}_0^!(\varrho)$ has linearly-dependent components, it often signals the reducibility of the representation $\varrho$. So one can typically obtain functions which satisfy the assumptions of Theorem \ref{theorem:movearound} by projecting $X(\tau)$ onto the irreducible constituents of $\varrho$.

To summarize, the generic protocol for computing $q$-expansions within the space $\mathcal{M}^!_0(\varrho)$ can be summarized as follows (at least for $\varrho$ of sufficiently small dimension).
\begin{enumerate}
\item Compute a bijective exponent $\lambda$ using Theorem \ref{theorem:bijectiveexponents}.
\item Find an example of a vector-valued modular form $X(\tau)\in\mathcal{M}^!_0(\varrho)$. This can be achieved using a variety of methods: modular differential equations, Rademacher summation, Hecke operators, identification of a known theory with modular data given by $\varrho$, etc.
\item Use Theorem \ref{theorem:movearound} to compute the matrix $\Xi(\tau)$ from $X(\tau)$ and extract the associated matrix $\chi$. One is free to then discard $\Xi(\tau)$ as it can always be reconstructed from $(\lambda,\chi)$ using Theorem \ref{theorem:recursion} should one need it later.
\item Compute whatever other function is needed in $\mathcal{M}^!_0(\varrho)$ from $\Xi(\tau)$ using Theorem \ref{theorem:freemodule}. 
\end{enumerate}
In particular, in the last step one may compute the most general $(c,\Cat)$-admissible function (see \S\ref{subsec:overview:characters} for the definition of admissible), as it must be expressible as a finite linear combination of basis elements of the form
\begin{align}
\ch^{(c)}(\tau) = \sum_{\substack{j,m \\ 24(m-\lambda_{j,j})\leq c}} \alpha_{j,m} X^{(j,m)}(\tau)
\end{align}
where care must be taken to impose positivity, uniqueness of the vacuum, and that $\ch^{(c)}_{i,-n}=0$ whenever $24(n-\lambda_{i,i})>c$.
Additional subtleties may arise if the hypotheses of the theorems are violated in one or more of these steps. See e.g.\ Appendix \ref{app:vvmfs:examples:E6} for an example where such subtleties are addressed.

\subsection{Examples}\label{app:vvmfs:examples}

In this subsection, we illustrate how the generalities of Appendix \ref{app:vvmfs:generalities} can be used to compute spaces of weight zero vector-valued modular forms in several concrete examples. The methods we employ here can (and were) used to generate the modular form data reported in Appendix \ref{app:data}.

\subsubsection{One-dimensional representations}\label{app:vvmfs:examples:1d}

We start by casting well-known facts about ordinary modular forms (with potentially non-trivial multiplier systems) within the Bantay--Ganon framework.

The one-dimensional representations of $\SL_2(\ZZ)$ form the group $\ZZ_{12}=\langle \epsilon^{-2}\rangle$ with generator given by the square of the inverse of the multiplier $\epsilon$ of the Dedekind eta function,
\begin{align}
\epsilon^{-2}(T) = e^{-\pi i/6}, \ \ \ \ \epsilon^{-2}(S) = i.
\end{align}
However the space $\mathcal{M}^!_0(\epsilon^{-2n})$ with $n$ odd is trivial because $\epsilon^{-2n}(-1)=-1$ (see e.g.\ Footnote \ref{footnote:trivialspace} for an explanation). Thus we restrict attention to the $\ZZ_6$ subgroup generated by $\zeta := \epsilon^{-4}$. 

Using Theorem \ref{theorem:bijectiveexponents}, one can compute a bijective exponent for the representation $\zeta^n$, 
\begin{align}
\lambda = (-1)^n \frac{2\sqrt{3}}{9} \sin \left(\tfrac{1}{3} \pi  (n+1)\right)+\frac{(-1)^n}{4}-\frac{7}{12}
\end{align}
which takes the values 
\begin{align}
\lambda = 0,-\sfrac76, -\sfrac13, -\sfrac12, -\sfrac23, -\sfrac56
\end{align}
for $n=0,\dots, 5$. We would like to compute the fundamental matrix $\Xi(\tau)$ for each $\zeta^n$, which in the case of one-dimensional representations is simply the unique modular form in $\mathcal{M}_0^!(\zeta^n)$ whose $q$-expansion starts as  $\Xi(\tau) = q^{\lambda}+O(q^{\lambda+1})$. Using the fact that $\zeta$ is related to the multiplier of the Dedekind eta function, one can immediately write them down,
\begin{align}
\Xi(\tau) = \begin{cases}
1, & n=0  \\
E_{14}(\tau)/\eta(\tau)^{28}, & n=1 \\
E_{2n}(\tau)/\eta(\tau)^{4n},&n=2,\dots,5.
\end{cases}
\end{align}
By examining the $q$-expansion of $\Xi(\tau)$ one finds in each case that 
\begin{align}
\chi= 0,4,248,-492,496,-244.
\end{align}
Theorem \ref{theorem:recursion} then furnishes a recursion relation for the Fourier coefficients of $\Xi(\tau)$ in terms of $(\lambda,\chi)$, and Theorem \ref{theorem:freemodule} obtains that $\mathcal{M}^!_0(\zeta^n)=\BBC[J(\tau)]\Xi(\tau)$, i.e.\ every weight zero modular form can be expressed as $P(J(\tau))\Xi(\tau)$ for some polynomial $P(x)$.

\subsubsection{A two-dimensional example}

We now work out a two-dimensional example. Consider the representation generated by 
\begin{align}
\mathcal{S} = \frac{1}{\sqrt{2}}\left(\begin{array}{rr} 1 & 1 \\ 1 & -1\end{array}\right), \ \ \ \ \mathcal{T} = \left(\begin{array}{rr} e^{-\pi i/12} & 0 \\ 0 & e^{5\pi i/12}  \end{array}\right)
\end{align}
corresponding to genera of the form $\big(1+24m,(\AA_1,1)\big)$, with $m\in\ZZ^{\geq 0}$.
Using Theorem \ref{theorem:bijectiveexponents}, we can construct a bijective exponent, 
\begin{align}
\lambda = \left(\begin{array}{rr} -\sfrac1{24} & 0 \\ 0 & -\sfrac{19}{24}  \end{array}\right).
\end{align}
Since this modular data is realized by the theory $\mathsf{A}_{1,1}$,  we can use its characters as the ``seed'' function $X^{(0,0)}(\tau)$,
\begin{align}
X^{(0,0)}_0(\tau) = \vartheta_3(2\tau)/\eta(\tau), \ \ \ X^{(0,0)}_1(\tau) = \vartheta_2(2\tau)/\eta(\tau)
\end{align}
where $\vartheta_n(\tau)$ are the standard Jacobi-theta functions
\begin{align}
\vartheta_2(\tau) = \sum_{n\in\ZZ}e^{i\pi \tau(n+\frac12)^2}, \ \ \ \ \vartheta_3(\tau) = \sum_{n\in\ZZ}e^{i\pi\tau n^2}.
\end{align}
We have labeled this seed function $X^{(0,0)}(\tau)$ because it is one of the basis elements described in Equation \eqref{eqn:basiscondition}. The other basis element $X^{(1,0)}(\tau)$ which is used to construct the fundamental matrix $\Xi(\tau)$ can be constructed using Theorem \ref{theorem:movearound} as 
\begin{align}
X^{(1,0)}(\tau) = \frac{1}{12}\left(24 \nabla_1+J(\tau) -1056\right)X^{(0,0)}(\tau).
\end{align}
From $X^{(0,0)}(\tau)$ and $X^{(1,0)}(\tau)$ we can form the fundamental matrix $\Xi(\tau)$ and extract 
\begin{align}
\chi = \left(\begin{array}{cc} 3 & 26752 \\ 2 & -247  \end{array}\right).
\end{align}
One can build the rest of the basis elements $X^{(j,m)}(\tau)$ using Theorem \ref{theorem:freemodule}, and use these to construct the most general $\big(1+24m,(\AA_1,1)\big)$-quasi-admissible function. For example, the most general such function with $c=1+24m=25$ must take the shape
\begin{align}
\begin{split}
\ch(\tau) &= X^{(0,1)}(\tau) + \alpha_{0,0}X^{(0,0)}(\tau) + \alpha_{1,0}X^{(1,0)}(\tau) \\
&= q^{-\sfrac{25}{24}}\left(
\begin{array}{l}
 1+\alpha _{0,0}q+q^2 \left(3 \alpha _{0,0}+26752 \alpha _{1,0}+143375\right)+\cdots \\
 q^{\sfrac14}\big(\alpha_{1,0}+q \left(2 \alpha _{0,0}-247 \alpha _{1,0}+490\right)+q^2 \left(2 \alpha _{0,0}-86241 \alpha _{1,0}+566250\right)+\cdots\big) \\
\end{array}
\right)
\end{split}
\end{align}
for some choice of integers $\alpha_{0,0}$ and $\alpha_{1,0}$.

\subsubsection{Non-trivial charge conjugation matrix}\label{app:vvmfs:examples:E6}

Consider the three-dimensional modular representation $\varrho$ generated by the assignments
\begin{align}
\mathcal{S} = \frac{1}{\sqrt{3}}\left(
\begin{array}{ccc}
 1 & 1 & 1 \\
 1 & \omega^2 & \omega \\
 1 & \omega & \omega^2 \\
\end{array}
\right), \ \ \ \ \mathcal{T} = \left(\begin{array}{ccc}  e^{-\pi i/2} & 0 & 0 \\ 0 & e^{-7\pi i /6} & 0 \\ 0 & 0 & e^{-7 \pi i /6} \end{array}\right)
\end{align}
where $\omega=e^{2\pi i /3}$. This representation corresponds to genera of the form $\big(6+24m,(\EE_6,1)\big)$. Naively this representation appears to violate the assumptions of most of the theorems used in previous subsections because 
\begin{align}
\varrho(-1) = \mathcal{S}^2 = \left(\begin{array}{ccc} 1 & 0 & 0 \\ 0 & 0 & 1 \\ 0 & 1 & 0 \end{array}\right)\neq \mathds{1}.
\end{align}
The resolution is to notice that $\varrho$ is reducible. Indeed, let $\mathcal{S}^-$ and $\mathcal{T}^-$ be the modular matrices associated to the one-dimensional representation $\varrho_-:=\epsilon^{-14}$ (defined in Appendix \ref{app:vvmfs:examples:1d}), and define the two-dimensional representation
\begin{align}
\mathcal{S}^+= \frac{1}{\sqrt{3}}\left(\begin{array}{rr} 1 & 2 \\ 1 & -1  \end{array}\right), \ \ \ \ \mathcal{T}^+ = \left(\begin{array}{cc}  e^{-\pi i/2} & 0  \\ 0 & e^{-7\pi i /6}\end{array}\right).
\end{align}
These are the irreducible subrepresentations of $\varrho\cong \varrho_+\oplus\varrho_-$, as can be seen using the intertwiners
\begin{align}
\mathcal{P}^- = \left( \begin{array}{rrr}0 & \sfrac12 & -\sfrac12   \end{array}\right),\ \ \ \mathcal{P}^+ = \left(\begin{array}{rrr} 
1 & 0 & 0 \\ 
0 & \sfrac12 & \sfrac12
\end{array}\right)
\end{align}
which satisfy $\mathcal{S}^\pm\mathcal{P}^\pm=\mathcal{P}^\pm\mathcal{S}$ and $\mathcal{T}^\pm\mathcal{P}^\pm=\mathcal{P}^\pm\mathcal{T}$. We have chosen to label these irreducible constituents $\pm$ because they satisfy $\varrho_\pm(-1)=\pm \mathds{1}$ (cf.\ Equation \eqref{eqn:pmdecomposition}). The representation $\varrho_-=\epsilon^{-14}$ does not admit any non-trivial vector-valued modular forms of weight zero, because $\varrho_-(-1)=-1$ (again, see Footnote \ref{footnote:trivialspace} for the explanation). On the other hand, $\varrho_+$ now satisfies all the hypotheses of the theorems from Appendix \ref{app:vvmfs:generalities}, because it is irreducible with $\varrho_+(-1)=(\mathcal{S}^+)^2=1$. So we may carry out the analysis purely with respect to $\varrho_+$, and lift all of the resulting vector-valued modular forms to ones for $\varrho$ using the surjective map 
\begin{align}
\begin{split}
\mathcal{M}^!_0(\varrho_+) &\rightarrow \mathcal{M}^!_0(\varrho) \\
(X_0(\tau),X_1(\tau)) &\mapsto (X_0(\tau),X_1(\tau),X_1(\tau)).
\end{split}
\end{align}
This modular data thus describes theories with three primaries, but with only two linearly-independent characters.

The calculations then become nearly identical to the ones performed in the previous subsection. Theorem \ref{theorem:bijectiveexponents} can be used to compute a bijective exponent, which we take to be 
\begin{align}
\lambda = \left(\begin{array}{rr} -\sfrac14 & 0 \\ 0  & -\sfrac{7}{12}\end{array}\right).
\end{align}
Because the theory $\mathsf{E}_{6,1}$ realizes the modular data $(\mathcal{S},\mathcal{T})$, its characters can be taken to be the seed function, 
\begin{align}
\begin{split}
X^{(0,0)}_0(\tau)&=J(\tau)^{\sfrac14}{_2}F_1\left( -\sfrac14,\sfrac{1}{12},\sfrac13;  \frac{1728}{J(\tau)}\right) \\
X^{(0,0)}_1(\tau) &= 27J(\tau)^{-\sfrac{5}{12}}{_2}F_1\left(\sfrac{5}{12},\sfrac34,\sfrac53;\frac{1728}{J(\tau)}\right).
\end{split}
\end{align}
Using Theorem \ref{theorem:movearound}, one finds that 
\begin{align}
X^{(1,0)}(\tau) = \left(\frac{1}{18}\nabla_1+\frac{2}{15}\nabla_2-18\right)X^{(0,0)}(\tau).
\end{align}
Computing the $q$-expansions of $X^{(0,0)}(\tau)$ and $X^{(1,0)}(\tau)$, one finds the matrix 
\begin{align}
\chi = \left(
\begin{array}{cc}
 78 & 3402 \\
 27 & -322 \\
\end{array}
\right).
\end{align}
One can then use Theorem \ref{theorem:freemodule} to conclude for example that the most general $\big(30,(\EE_6,1)\big)$-quasi-admissible vector-valued modular form is 
\begin{align}
\begin{split}
\ch(\tau)&=X^{(0,1)}(\tau)+\alpha_{0,0}X^{(0,0)}(\tau)+\alpha_{1,0}X^{(1,0)}(\tau)\\
&=q^{-\sfrac54}\left(
\begin{array}{l}
1+\alpha _{0,0}q+q^2 \left(78 \alpha _{0,0}+3402 \alpha _{1,0}+99675\right)+\cdots \\
 q^{\sfrac23}\left(\alpha _{1,0}+q \left(27 \alpha _{0,0}-322 \alpha _{1,0}+6966\right)+\cdots\right) \\
 q^{\sfrac23}\left(\alpha _{1,0}+q \left(27 \alpha _{0,0}-322 \alpha _{1,0}+6966\right)+\cdots\right) 
\end{array}
\right).
\end{split}
\end{align}

\clearpage

\section{Classification}\label{app:classification}

This appendix contains the technical details behind our various classification results. In Appendix \ref{app:classification:clt8}, we give the proof of the classification of bosonic chiral algebras with $c\leq 8$ and at most four primary operators. In Appendix \ref{app:classification:clt24}, we explain how to use the result of Appendix \ref{app:classification:clt8} to classify bosonic chiral algebras with at most four primary operators and central charge in the range $8<c\leq 24$. In Appendix \ref{app:classification:fermionic}, we explain how the classification of chiral fermionic RCFTs with $c<23$ can be obtained as a corollary of the previous results, using chiral fermionization. Finally, in Appendix \ref{app:classification:symmetries}, we apply symmetry/subalgebra duality to the theories constructed in Appendix \ref{app:classification:clt8} and Appendix \ref{app:classification:clt24} in order to determine certain symmetries of holomorphic VOAs.

\subsection{Bosonic chiral algebras with $\normalfont c\leq 8$}\label{app:classification:clt8}

The main result of this subsection is the classification of bosonic chiral algebras belonging to the 35 genera $(c,\Cat)$ with $0<c\leq 8$ and $\mathrm{rank}(\Cat)\leq 4$. The final product can be seen at a glance in Table \ref{tab:generaclassification}, where one finds that there is at most one chiral algebra $\V_\Cat$ within each genus. We treat this range of central charge separately from $8<c\leq 24$ because the techniques required are different. In fact, Table \ref{tab:generaclassification} will be used as crucial input for the next subsection, Appendix \ref{app:classification:clt24}. 

The basic strategy is as follows. Every chiral algebra $\V$ has a canonically defined subalgebra \cite{Mason:2014gba} which one may call the ``classical part'' of $\V$,
\begin{align}\label{eqn:KacMoodysubalgebra}
\mathcal{K} = (\mathfrak{g}_1)_{k_1}\otimes \cdots\otimes (\mathfrak{g}_n)_{k_n}\otimes \mathsf{U}_1^r\subset  (\mathfrak{g}_1)_{k_1}\otimes \cdots\otimes (\mathfrak{g}_n)_{k_n}\otimes V_L \subset \V
\end{align}
where the $(\mathfrak{g}_i)_{k_i}$ are affine Kac--Moody algebras associated to simple Lie algebras $\mathfrak{g}_i$ at level $k_i$, and $V_L$ is a lattice VOA with $L$ of dimension $r$.
The precise structure of this subalgebra is dictated by the continuous global symmetries of $\V$ and their anomalies; it is essentially the subalgebra generated by the Noether currents, i.e.\ the space of dimension $h=1$ operators, $\mathcal{K}=\langle \V_1\rangle$. The idea is to heavily constrain the possible $\mathcal{K}$ which can arise within each $(c,\Cat)$. It will turn out that in each genus, one can whittle down the possibilities to at most one, which we call $\mathcal{K}_\Cat$ when it exists. From there, one can attempt to determine the full chiral algebra $\V_\Cat$, which turns out to exist if and only if $\mathcal{K}_\Cat$ does. We note that our approach is an elaboration of the methods used in \cite{Mason:2021xfs} to classify certain two-character theories. See also \cite{gannon2019reconstruction} for a nice articulation of the general philosophy.

To this end, we introduce a number of theorems and lemmas to assist with our analysis. Recall that the space of dimension $h=1$ operators $\V_1$ in a chiral algebra $\V$  define a Lie algebra (see Equation \eqref{eqn:LiealgebraV1}).

\begin{theorem}[Mostly in \cite{Dong:2002fs}]\label{theorem:kmsubalgebra}
Let $\V$ be a chiral algebra with central charge $c$. Then its Lie algebra $\V_1$ of Noether currents is reductive, i.e.\ it can be decomposed as 
\begin{align}
\V_1 \cong \mathfrak{g}_1\oplus\cdots\oplus \mathfrak{g}_n \oplus \mathsf{U}_1^r
\end{align}
where the $\mathfrak{g}_i$ are complex simple Lie algebras. Furthermore, the rank of $\V_1$ is bounded from above by the central charge of $\V$, 
\begin{align}\label{eqn:rankcc}
\mathrm{rank}(\V_1)\leq c(\V)
\end{align}
and the Kac--Moody algebra $\mathcal{K}=(\mathfrak{g}_1)_{k_1}\otimes\cdots\otimes (\mathfrak{g}_n)_{k_n} \otimes \mathsf{U}_1^r$ generated by the Noether currents satisfies 
\begin{align}\label{eqn:ccinequality}
c(\mathcal{K})=r+\sum_{i=1}^n \frac{k_i \dim(\mathfrak{g}_i)}{k_i+h_i\text{\v{}}}\leq c(\V)
\end{align}
where $h_i\text{\v{}}$ is the dual Coxeter number of $\mathfrak{g}_i$. When $c(\V)-c(\mathcal{K})< 1$, it must be equal to the central charge of a minimal model, i.e. 
\begin{align}\label{eqn:mmcc}
c(\V)-c(\mathcal{K}) = 1-\frac{6}{(m+2)(m+3)}
\end{align}
for some non-negative integer $m$.
\end{theorem}
\begin{proof}
The reductivity of $\V_1$, as well as Equation \eqref{eqn:rankcc}, are proved in \cite{Dong:2002fs}. Equation \eqref{eqn:ccinequality} is simply the statement that the Kac--Moody subalgebra $\mathcal{K}$ should not have a central charge which exceeds the central charge of the full chiral algebra $\V$; otherwise, the commutant $\V\big/\mathcal{K}$ would have negative central charge, in violation of unitarity which follows from Theorem \ref{theorem:unitarycosets}. Similarly, in the case that $\tilde{c}:=c(\V)-c(\mathcal{K})<1$, the commutant $\V\big/ \mathcal{K}$ must be an extension of the simple Virasoro VOA $L_{\tilde{c}}$; by the classification of unitary minimal models, $\tilde{c}$ must take the form of Equation \eqref{eqn:mmcc}.
\end{proof}

 In conjunction with this theorem, we use modularity to pin down the dimension of $\V_1$ within each genus.

\begin{theorem}\label{theorem:clt8characters}
Let $(c,\Cat)$ be a genus with $0<c\leq 8$ and $\mathrm{rank}(\Cat)\leq 4$. Then any two chiral algebras in $(c,\Cat)$ have the same character vector. In particular, if $\V$ is a chiral algebra in the genus $(c,\Cat)$, then the dimension of its Lie algebra $\V_1$ of Noether currents depends only on $\Cat$.
\end{theorem}

\begin{proof}
The most general character vectors in each such genus are tabulated in Appendix \ref{app:data} using the techniques of Appendix \ref{app:vvmfs}. Inspection of the tables reveals that there is always precisely one character vector in each genus, except when $\Cat=\overline{(\AA_1,5)}_{\sfrac12}$, in which case there is no vector-valued modular form with all the right properties to serve as a character vector. 

An immediate consequence is that the dimension $N_\Cat$ of the Lie algebra $\V_1$ of any VOA in a genus $(c,\Cat)$ with $0<c\leq 8$ and $\mathrm{rank}(\Cat)\leq 4$ is uniquely fixed by $\Cat$, and can be read off from the corresponding vacuum character in Appendix \ref{app:data} as the coefficient of $q^{-\sfrac{c}{24}+1}$, 
\begin{align}
\ch_0(\tau) = q^{-\sfrac{c}{24}}(1+N_\Cat q+\cdots).
\end{align}
The numbers $N_\Cat$ for each category $\Cat$ can be determined from Appendix \ref{app:data}.
\end{proof}

 Before moving on to the classification, we record one more technical lemma.

\begin{lemma}\label{lem:E8subalgebra}
Let $\V$ be a chiral algebra in some genus $(c,\Cat)$ with $c<8$, and let 
\begin{align}
\mathcal{K}=(\mathfrak{g}_1)_{k_1}\otimes\cdots\otimes(\mathfrak{g}_n)_{k_n}\otimes \mathsf{U}_1^r\subset \V
\end{align}
be the Kac--Moody subalgebra of $\V$. Similarly, let $\tilde{\V}$ be a chiral algebra in the dual genus $(8-c,\overline{\Cat})$, and let 
\begin{align}
\tilde{\mathcal{K}} = (\tilde{\mathfrak{g}}_1)_{\tilde{k}_1}\otimes \cdots\otimes (\tilde{\mathfrak{g}}_n)_{\tilde{k}_n} \otimes \mathsf{U}_1^{\tilde{r}}\subset \tilde{\V}
\end{align}
be its Kac--Moody subalgebra. Then the finite-dimensional Lie algebra $\E[8]$ admits a subalgebra of the form 
\begin{align}\label{eqn:E8subalgebra}
\mathfrak{g}_1^{(k_1)}\oplus\cdots \oplus \mathfrak{g}_n^{(k_n)}\oplus \tilde{\mathfrak{g}}_1^{(\tilde{k}_1)}\oplus \cdots\oplus \tilde{\mathfrak{g}}_n^{(\tilde{k}_n)}\oplus \mathsf{U}_1^{r+\tilde{r}}\subset \E[8]
\end{align}
where the levels $k_i$ and $\tilde{k}_j$ are interpreted as embedding indices (see Definition \ref{defn:embeddingindex}).
\begin{proof}
Because $\V$ and $\tilde{\V}$ belong to dual genera, they can be glued together, using Example \ref{ex:algebrafromequivalence} and Theorem \ref{theorem:VOAextensionfromcondensablealgebra}, to produce a holomorphic VOA of central charge $8=c+(8-c)$. The only holomorphic VOA of central charge $8$ is $\E[8,1]$, so 
\begin{align}
\E[8,1] \cong \bigoplus_i \V(i)^\ast \otimes \tilde{\V}(i).
\end{align}
In particular, this implies that both of their Kac--Moody subalgebras embed as (not necessarily conformal) subalgebras of $\E[8,1]$. As a statement about ordinary Lie algebras, this implies Equation \eqref{eqn:E8subalgebra}, where the levels $k_j$ and $\tilde{k}_j$ are to be interpreted as Dynkin embedding indices.
\end{proof}
\end{lemma}
The reason this lemma is useful is that the semi-simple Lie subalgebras of $\E[8]$ are completely classified, and so one can show that a Kac--Moody algebra cannot arise in some particular genus if the corresponding subalgebra of $\E[8]$ does not exist.

With these preliminaries out of the way, we move on to the main result, which constitutes a kind of VOA analog of the classic work of \cite{rowell2009classification}. See \cite{Mason:2021xfs} for a similar theorem about two-character theories.

\begin{theorem}\label{theorem:clt8classification}
Let  $(c,\Cat)$  be a genus with $0<c\leq 8$ and $\mathrm{rank}(\Cat)\leq 4$ which is not equivalent to $\big(\sfrac87,\overline{(\AA_1,5)}_{\sfrac12}\big)$. Then the chiral algebra $\V_\Cat$ in Table \ref{tab:generaclassification} is the unique theory in the genus $(c,\Cat)$. The genus $\big(\sfrac87,\overline{(\AA_1,5)}_{\sfrac12}\big)$ is empty. 
\end{theorem}

\begin{proof}
We prove this theorem case by case. To obtain many of the claims below, we wrote a computer script which generates all Lie algebras of a particular dimension, and with rank less than or equal to some fixed number.\\

\noindent\emph{Case}: 
\begin{align}\label{eqn:case1}
\begin{split}
\Cat\in &\Big\{\Vec,~(\AA_1,1),~(\EE_7,1),~(\GG_2,1),~(\FF_4,1),~(\AA_2,1),~(\EE_6,1), \\
&\hspace{.2in} (\BB_2,1),~ (\BB_3,1), ~(\BB_4,1), ~(\BB_5,1),~ (\BB_7,1),~ (\AA_1,1)^{\boxtimes 2},~(\AA_3,1), \\
&\hspace{.4in}(\DD_4,1),~(\DD_5,1),~ (\DD_6,1),~ (\DD_7,1),~ (\DD_8,1),~ (\AA_1,1)\boxtimes (\EE_7,1),~ (\AA_1,1)\boxtimes (\GG_2,1) \Big\}.
\end{split}
\end{align}
In these cases, a computer calculation reveals that there is a unique Kac--Moody algebra 
\begin{align}
\mathcal{K}_\Cat=(\mathfrak{g}_1)_{k_1}\otimes \cdots\otimes(\mathfrak{g}_n)_{k_n}\otimes \mathsf{U}_1^r
\end{align}
which satisfies Theorem \ref{theorem:kmsubalgebra} and which has the correct dimension predicted by Theorem \ref{theorem:clt8characters}, i.e.
\begin{align}
\begin{split}
&\dim(\mathfrak{g}_1\oplus\cdots\oplus\mathfrak{g}_n\oplus \mathsf{U}_1^r)=N_\Cat.
\end{split}
\end{align}
In particular, $\mathcal{K}_\Cat$ is always the ``obvious'' one which can be read off from the modular category,
\begin{align}
\mathcal{K}_\Cat = \begin{cases}
\E[8,1],& \Cat\cong\Vec \\
\mathsf{X}_{r,1}, & \Cat\cong (\textsl{X}_r,1) \\
\mathsf{X}_{r,1}\mathsf{Y}_{s,1}, & \Cat\cong (\textsl{X}_r,1)\boxtimes (\textsl{Y}_s,1).
\end{cases}
\end{align}
It follows from Theorem \ref{theorem:kmsubalgebra} that $\mathcal{K}_\Cat$ must embed as a subalgebra into any chiral algebra $\mathcal{V}$ in the genus $(c,\Cat)$. In fact, in each case, $\mathcal{K}_\Cat$ is such that the inequality in Equation \eqref{eqn:ccinequality} is saturated. In other words, the central charge of $\mathcal{K}_\Cat$ is the same as that of $\mathcal{V}$, and so $\mathcal{V}$ must be a conformal extension of $\mathcal{K}_\Cat$. However, in each case, the representation category of $\mathcal{K}_\Cat$ is already $\Cat$ on the nose, and a non-trivial extension (if one even exists) would modify this category to a new one. Thus, we conclude that $\V_\Cat= \mathcal{K}_\Cat$ is the unique chiral algebra in the genus $(c,\Cat)$ when $c\leq 8$ and when $\Cat$ is one of the modular categories in Equation \eqref{eqn:case1}.\\

\noindent \emph{Case}: $(c,\Cat)\cong (\sfrac12,(\BB_8,1))$. 

In this case, it follows from the classification of minimal models that the chiral algebra of the Ising model, $\V_\Cat=L_{\sfrac12}$, is the unique theory in this genus. \\

\noindent \emph{Case}: $(c,\Cat)\cong (1,(\textsl{U}_1,4))$. 

In this case, $N_\Cat=1$, which implies that any chiral algebra in this genus contains a $\mathsf{U}_1$ Kac--Moody algebra as a conformal subalgebra. In particular, this means that any theory in this genus must be a lattice VOA $V_L$ associated to a one-dimensional even lattice $L=\sqrt{2m}\ZZ$. The only choice of $L$ which leads to the correct modular tensor category is $L=2\mathbb{Z}$, corresponding to $m=2$, so $\V_\Cat=V_{2\ZZ}$.\\

\noindent \emph{Case}: $(c,\Cat) \in\big\{ (\sfrac32,(\AA_1,2)),~(\sfrac{13}2,(\BB_6,1))\big\}$. 

Let us first consider the genus $(\sfrac32,(\AA_1,2))$. Because $N_\Cat=3$, the only Kac--Moody algebras consistent with Theorem \ref{theorem:kmsubalgebra} are $\A[1,1]$ and $\A[1,2]$. The latter option is certainly a chiral algebra in this genus, so we just need to rule out the former. If the former option were realized in a chiral algebra $\V$, then the commutant of $\A[1,1]$ in $\V$ would have central charge $\sfrac12$, and therefore would necessarily be the chiral algebra of the Ising model, $\V\big/\A[1,1]\cong L_{\sfrac12}$. In particular, $\mathcal{V}$ must then be a conformal extension of $\A[1,1]L_{\sfrac12}$. However, $\A[1,1]L_{\sfrac12}$ does not have any non-trivial conformal extensions (it does not even have a primary operator of integral conformal dimension), and it also does not have the correct modular tensor category (for example, its modular tensor category has rank 6). So $\A[1,1]$ is ruled out as a possible Kac--Moody subalgebra, leaving $\V_\Cat=\A[1,2]$.

The case $(c,\Cat)\cong(\sfrac{13}2,(\BB_6,1))$ proceeds nearly identically. The dimension must be $N_\Cat=78$, and the only consistent Kac--Moody algebras are $\B[6,1]$ and $\E[6,1]$. The first option is certainly a valid chiral algebra in this genus. If the second option were to occur, then $\V$ would need to be a conformal extension of $\E[6,1]L_{\sfrac12}$. However, this VOA again does not admit any non-trivial conformal extensions, and itself has the wrong modular tensor category, so it is ruled out. We are left to identify $\V_\Cat=\B[6,1]$.
\\

\noindent \emph{Case}: $(c,\Cat) \in \big\{ (\sfrac{10}3,(\AA_1,7)_{\sfrac12}),~(\sfrac{14}3,(\GG_2,2))\big\}$. 

For the first of these genera, $N_\Cat=6$, and so Theorem \ref{theorem:kmsubalgebra} implies that the Kac--Moody subalgebra of any theory $\V$ must be either $\A[1,1]\A[1,1]$ or $\A[1,1]\A[1,7]$.

The possibility $\A[1,1]^2$ can be ruled out. Assume there were some $\V$ in this genus with $\A[1,1]^2$ as a subalgebra. Then one could apply Lemma \ref{lem:E8subalgebra} taking $\tilde{\V}=\G[2,2]$ and conclude that $\E[8]$ has a Lie subalgebra of the form $\A[1]^{(1)}\A[1]^{(1)}\G[2]^{(2)}$ which a computer calculation reveals it does not have.

This leaves only $\mathcal{K}_\Cat=\A[1,1]\A[1,7]$. Because the modular tensor category of $\A[1,1]\A[1,7]$ has rank 16, the theory $\V$ would need to be realized as a non-trivial conformal extension. In fact, one can show using Proposition \ref{prop:modulardataextension} that, if $\V$ exists, then it and its modules must decompose into $\A[1,1]\A[1,7]$-representations as
\begin{align}\label{eqn:Ex((A1)1(A1)7)}
\begin{split}
\V &\cong \big(\A[1,1]\otimes \A[1,7]\big) \oplus \big(\A[1,1](\sfrac14)\otimes \A[1,7](\sfrac74)\big) \\
\V_1 &\cong \big(\A[1,1]\otimes \A[1,7](\sfrac43)\big) \oplus \big(\A[1,1](\sfrac14)\otimes \A[1,7](\sfrac1{12})\big) \\
\V_2 &\cong \big(\A[1,1]\otimes \A[1,7](\sfrac29)\big)\oplus \big(\A[1,1](\sfrac14)\otimes \A[1,7](\sfrac{35}{36})\big) \\ 
\V_3&\cong \big(\A[1,1]\otimes \A[1,7](\sfrac23)\big)\oplus \big( \A[1,1](\sfrac14) \otimes \A[1,7](\sfrac5{12})\big).
\end{split}
\end{align}
Happily, this is a simple current extension, so $\V$ does exist, and there is a unique VOA structure which one can place on it. Thus, we conclude that $\V_\Cat\cong \mathrm{Ex}(\A[1,1]\A[1,7])$ is the unique VOA in this genus.

Next, consider the case $(c,\Cat)=(\sfrac{14}3,(\GG_2,2))$, for which $N_\Cat=14$. The possible Kac--Moody algebras consistent with Theorem \ref{theorem:kmsubalgebra} are just $\G[2,1]$ and $\G[2,2]$. The first possibility can be ruled out as follows. Assuming some $\V$ exists in $(\sfrac{14}3,(\GG_2,2))$ with $\G[2,1]$ as a subalgebra, by applying Lemma \ref{lem:E8subalgebra} taking $\tilde{\V} = \mathrm{Ex}(\A[1,1]\A[1,7])$ one would conclude that $\E[8]$ has a subalgebra of the form $\A[1]^{(1)}\A[1]^{(7)}\G[2]^{(1)}$, which again, a computer calculation reveals it does not have. 
This just leaves $\G[2,2]$, which already has central charge saturating $\sfrac{14}{3}$ and the correct representation category, so that $\V_\Cat=\G[2,2]$.\\

\noindent \emph{Case}: $(c,\Cat) \cong (8,(\GG_2,1)\boxtimes(\FF_4,1))$.

The Kac--Moody algebras consistent with Theorem \ref{theorem:kmsubalgebra} are $\D[6,1]$, $\A[1,1]\A[7,1]$, and $\G[2,1]\F[4,1]$. 

The algebra $\A[1,1]\A[7,1]$ is the easiest to rule out. Call $\mathcal{S}',\mathcal{T}'$ the modular data of $\A[1,1]\A[7,1]$ and $\mathcal{S},\mathcal{T}$ the modular data of the category $(\GG_2,1)\boxtimes (\FF_4,1)$. Furthermore, call $\ch'(\tau)$ the character vector of $\A[1,1]\A[7,1]$ and $\ch(\tau)$ the unique admissible character vector in the genus $(8,(\GG_2,1)\boxtimes (\FF_4,1))$. Then, an easy computer calculation shows that there is no matrix $M$ with non-negative integer entries that satisfies
\begin{align}
M\mathcal{S}'=\mathcal{S}M, \ \ M\mathcal{T}'=\mathcal{T}M, \ \ \ch(\tau)=M\ch'(\tau).
\end{align}
Thus, by Proposition \ref{prop:modulardataextension} there is no extension of $\A[1,1]\A[7,1]$ which lives in the genus $(8,(\GG_2,1)\boxtimes (\FF_4,1))$, and it is not an allowed Kac--Moody algebra.

On the other hand, consider the possibility that $\D[6,1]\subset \V$ for some $\V$ in the genus. We will show that there is no VOA which can play the role of the commutant $\tilde\V=\Com_\V(\D[6,1])$. Such a VOA should have central charge equal to 2. Furthermore, by \cite{Frohlich:2003hm}, its representation category should be expressible as an anyon condensation of $(\GG_2,1)\boxtimes (\FF_4,1)\boxtimes \overline{(\DD_6,1)}$. The only two condensations of this category are $(\GG_2,1)\boxtimes (\FF_4,1)\boxtimes \overline{(\DD_6,1)}$ itself and $\overline{(\DD_6,1)}$. The latter possibility cannot be identified with $\Rep(\tilde{\V})$ because then there would be no way to condense $\Rep(\tilde\V)\boxtimes (\DD_6,1)$ to obtain $(\GG_2,1)\boxtimes (\FF_4,1)$. Thus, we must have that $\Rep(\tilde\V)\cong (\GG_2,1)\boxtimes (\FF_4,1)\boxtimes \overline{(\DD_6,1)}.$ We may then study the admissible characters for $\tilde{\V}$. Using the various theorems from Appendix \ref{app:vvmfs}, one can compute a bijective exponent for the requisite space of weakly-holomorphic vector-valued modular forms, 
\begin{align}
\lambda = -\mathrm{diag}\left(\sfrac{1}{12},\sfrac{29}{60},\sfrac{41}{60},\sfrac{13}{12},\sfrac{7}{12},\sfrac{59}{60},
   \sfrac{11}{60},\sfrac{7}{12},\sfrac{5}{6},\sfrac{7}{30},\sfrac{13}{30},\sfrac{5}{6},\sfrac{5}
   {6},\sfrac{7}{30},\sfrac{13}{30},\sfrac{5}{6}\right).
\end{align}
Thus, the only function in the space which could possibly serve as an admissible character vector is $X^{(0,0)}(\tau)$; we must check whether this function is positive and has $X^{(0,0)}_{3,1}=0$ (otherwise, we would have $h_3=0$). It turns out that $X^{(0,0)}_{3,1}\neq 0$, so that there is no admissible modular form, and hence no theory to play the role of $\tilde{\V}$.

This just leaves $\V_\Cat\cong \G[2,1]\F[4,1]$ as the unique theory in this genus.
\\

\noindent \emph{Case}: $(c,\Cat) \in\big\{ (\sfrac{12}5,(\FF_4,1)^{\boxtimes 2}),~(\sfrac{28}5, (\GG_2,1)^{\boxtimes 2})\}$. 

In the first case, $N_\Cat=3$, and the only Kac--Moody algebras consistent with Theorem \ref{theorem:kmsubalgebra} are $\A[1,1]$ and $\A[1,8]$. The first possibility is ruled out because any $\V$ in this genus containing an $\A[1,1]$ subalgebra could be fed into Lemma \ref{lem:E8subalgebra}, taking $\tilde{\V} =\G[2,1]^2$, to conclude that $\E[8]$ has an $\A[1]^{(1)}\G[2]^{(1)}\G[2]^{(1)}$ subalgebra which it does not have. On the other hand, although $\A[1,8]$ does not live in $(\sfrac{12}5,(\FF_4,1)^{\boxtimes 2})$, it admits a unique conformal extension $\Ex(\A[1,8])$ (actually, a simple current extension) which does; this conformal extension decomposes into $\A[1,8]$-modules as
\begin{align}
\begin{split}
\V_\Cat &\cong \A[1,8]\oplus \A[1,8](2) \\
(\V_\Cat)_1 &\cong \A[1,8](\sfrac35) \\
(\V_\Cat)_2 &\cong \A[1,8](\sfrac35) \\
(\V_\Cat)_3 &\cong \A[1,8](\sfrac15)\oplus \A[1,8](\sfrac65).
\end{split}
\end{align}
As an aside, we note that the rational conformal field theory obtained by using the diagonal modular invariant of $\V_\Cat$ is the same as the theory obtained by using the D-type modular invariant of $\A[1,8]$.

Next, consider $(c,\Cat)\cong (\sfrac{28}5, (\GG_2,1)^{\boxtimes 2})$, so that $N_\Cat=28$. The Kac--Moody algebras consistent with Theorem \ref{theorem:kmsubalgebra} are $\D[4,1]$ and $\G[2,1]^2$. The first case can be ruled out by taking $\tilde{\V}= \Ex(\A[1,8])$ and applying Lemma \ref{lem:E8subalgebra} to conclude that $\E[8]$ has a $\D[4]^{(1)}\G[2]^{(1)}\G[2]^{(1)}$ subalgebra which it does not have. This leaves only $\V_\Cat = \G[2,1]^2$. \\

\noindent \emph{Case}: $(c,\Cat) \in\big\{ \big(\sfrac95,(\EE_7,1)\boxtimes (\GG_2,1)\big),~\big( \sfrac{31}5,(\AA_1,1)\boxtimes (\FF_4,1)\big)\big\}$.

Let us first consider the case $(c,\Cat)\cong \big(\sfrac95,(\EE_7,1)\boxtimes (\GG_2,1)\big)$. In this case, $N_\Cat=3$, so the Kac--Moody algebras consistent with Theorem \ref{theorem:kmsubalgebra} are $\A[1,1]$ and $\A[1,3]$.  Assume that $\V$ is a theory in this genus with $\A[1,1]$ as a subalgebra. Taking $\tilde{\V}=\A[1,1] \F[4,1]$ and applying Lemma \ref{lem:E8subalgebra}  implies that $\E[8]$ admits a subalgebra of the form $\A[1]^{(1)}\A[1]^{(1)}\F[4]^{(1)}$, which it does not have. Thus $\A[1,1]$ is ruled out, leaving just $\V_\Cat=\A[1,3]$; it already has the correct central charge and modular category, and therefore it is the unique theory in this genus.

Next, consider $(c,\Cat)\cong \big(\sfrac{31}5,(\AA_1,1)\boxtimes(\FF_4,1)\big)$. Here, $N_\Cat=55$, and the unique Kac--Moody algebras consistent with Theorem \ref{theorem:kmsubalgebra} are $\B[5,1]$ and $\A[1,1]\F[4,1]$. The former can be ruled out by applying Lemma \ref{lem:E8subalgebra} with $\tilde{\V}=\A[1,3]$. This leaves just $\V_\Cat = \A[1,1]\F[4,1]$ as the unique theory in the genus. \\

\noindent \emph{Case}: $(c,\Cat)\cong (\sfrac{21}5,(\CC_3,1))$.

Here, $N_\Cat=21$, and the Kac--Moody algebras consistent with Theorem \ref{theorem:kmsubalgebra} are $\B[3,1]$ and $\C[3,1]$. However, taking $\tilde{\V}=\A[1,1]\G[2,1]$ and applying Lemma \ref{lem:E8subalgebra} rules out $\B[3,1]$ as a possibility, leaving just $\mathcal{K}_\Cat = \C[3,1]$. This Kac--Moody algebra already has central charge $\sfrac{21}5$ and modular category $(\CC_3,1)$, so it follows that $\V_\Cat=\mathcal{K}_\Cat$. \\

\noindent\emph{Case}: $(c,\Cat)\in \left\{\big(\sfrac{48}7,\overline{(\AA_1,5)}_{\sfrac12}\big),~\big(\sfrac87,(\AA_1,5)_{\sfrac12}\big)\right\}$.

From the tables in Appendix \ref{app:data}, one finds that the genus $(\sfrac87,(\AA_1,5)_{\sfrac12})$ does not support an admissible character vector, and therefore does not contain any chiral algebras. 

So we can move on to the genus $\big(\sfrac{48}7,\overline{(\AA_1,5)}_{\sfrac12}\big)$. Actually, before doing this, it will be helpful to establish the existence of a chiral algebra in the genus $(c,\Cat)\cong (8+\sfrac87,(\AA_1,5)_{\sfrac12})$. The unique admissible character vector for this genus has $N_\Cat=136$. The Kac--Moody algebras consistent with Theorem \ref{theorem:kmsubalgebra} are then $\A[1,1]\E[7,1]$ and $\A[1,5]\E[7,1]$. We will show in a moment that $\A[1,1]\E[7,1]$ is not possible, so that $\mathcal{K}_{\Cat}=\A[1,5]\E[7,1]$. In this case, the central charge of $\mathcal{K}_{\Cat}$ already equals the central charge of the genus, so any chiral algebra in $(8+\sfrac87,(\AA_1,5)_{\sfrac12})$ would need to be a conformal extension of $\mathcal{K}_\Cat$. By Proposition \ref{prop:modulardataextension}, if such a conformal extension were to exist, it would need to decompose into $\mathcal{K}_{\mathcal{C}}$-modules as 
\begin{align}
\begin{split}
\V_\Cat&\cong \big(\A[1,5]\otimes \E[7,1]\big)\oplus \big(\A[1,5] (\sfrac54)\otimes \E[7,1](\sfrac34)\big) \\
(\V_\Cat)_1 &\cong \big(\A[1,5](\sfrac{6}{7})\otimes \E[7,1] \big) \oplus \big(\A[1,5](\sfrac{3}{28})\otimes \E[7,1](\sfrac34)  \big) \\
(\V_\Cat)_2 &\cong \big(\A[1,5](\sfrac{2}{7})\otimes \E[7,1] \big) \oplus \big(\A[1,5](\sfrac{15}{28})\otimes \E[7,1](\sfrac34)  \big).
\end{split}
\end{align}
Happily, this is a simple current extension, and so there is a unique VOA structure we can place on $\V_\Cat$, which we will call $\mathrm{Ext}(\A[1,5]\E[7,1])$.

Now, consider the genus $(c,\Cat)\cong \big(\sfrac{48}7,\overline{(\AA_1,5)}_{\sfrac12}\big)$. In this case, $N_\Cat=78$, and the unique Kac--Moody algebra consistent with Theorem \ref{theorem:kmsubalgebra} is  $\E[6,1]$. Therefore, if $\V_\Cat$ exists, it must be a conformal extension of $\E[6,1]L_{\sfrac67}$. In fact, from Proposition \ref{prop:modulardataextension} we have that $\V_\Cat$ would need to decompose into $\E[6,1]L_{\sfrac67}$-modules as 
\begin{align}\label{eqn:Ext((E6)1L(6/7))}
\begin{split}
    \V_\Cat &\cong \big(\E[6,1]\otimes (L_{\sfrac67}\oplus L_{\sfrac67}(5))\big)\oplus \big((\E[6,1](\sfrac23) \oplus \E[6,1](\sfrac23)^\ast)\otimes L_{\sfrac67}(\sfrac43)\big) \\
    (\V_\Cat)_1&\cong \big(\E[6,1]\otimes (L_{\sfrac67}(\sfrac{1}{7})\oplus L_{\sfrac67}(\sfrac{22}7))\big)\oplus \big((\E[6,1](\sfrac23) \oplus \E[6,1](\sfrac23)^\ast)\otimes L_{\sfrac67}(\sfrac{10}{21})\big) \\
    (\V_\Cat)_2 &\cong \big(\E[6,1]\otimes (L_{\sfrac67}(\sfrac{5}{7})\oplus L_{\sfrac67}(\sfrac{12}7))\big)\oplus \big((\E[6,1](\sfrac23) \oplus \E[6,1](\sfrac23)^\ast)\otimes L_{\sfrac67}(\sfrac{1}{21})\big).
\end{split}
\end{align}
Unfortunately, this does not define a simple current extension, so it is not immediately obvious that it is possible to place a VOA structure on $\V_\Cat$. One would need to show that the modular tensor category of $\E[6,1]L_{\sfrac67}$ supports a suitable algebra object. 

We prove the existence and uniqueness of such a $\V_\Cat$ indirectly as follows. Any theory in this genus should be expressible as a coset of the form $\mathcal{A}\big/ \mathrm{Ext}(\A[1,5]\E[7,1])$, where $\mathcal{A}$ is a holomorphic VOA of central charge 16. Since $\E[7,1]$ does not embed into $\D[16,1]^+$, we must take $\mathcal{A}$ to be $\E[8,1]^2$. There are three ways to embed $\A[1]^{(5)}\E[7]^{(1)}$ into $\E[8]^2\cong \E[8]\E[8]'$. The first involves embedding $\E[7]^{(1)}$ into $\E[8]$ and embedding $\A[1]^{(5)}$ into $\E[8]'$; the second two ways involve embedding $\E[7]^{(1)}$ into $\E[8]$, and embedding $\A[1]^{(5)}$ diagonally into $\A[1]^{(4)}\A[1]^{(1)}$, where $\A[1]^{(1)}= \mathrm{Cent}_{\E[8]}(\E[7]^{(1)})$ and there are two inequivalent subalgebras $\A[1]^{(4)}\subset \E[8]'$. Only one of these three embeddings has centralizer given by $\E[6]^{(1)}$;  it is one of the embeddings with $\A[1]^{(5)}$ embedded diagonally. Moreover, one can check that this embedding of $\A[1]^{(5)}\E[7]^{(1)}$ actually induces an embedding of $\mathrm{Ext}(\A[1,5]\E[7,1])\subset \E[8,1]$, and so we can define $\V_\Cat\cong \E[8,1]^2\big/\mathrm{Ext}(\A[1,5]\E[7,1])$. This implicitly shows that the extension in Equation \eqref{eqn:Ext((E6)1L(6/7))} exists and is unique; we call it $\mathrm{Ext}(\E[6,1]L_{\sfrac67})$ henceforth.

Finally, we can revisit our claim that $\mathrm{Ext}(\A[1,5]\E[7,1])$ is the unique chiral algebra in the genus $\big(\sfrac{64}7,(\AA_1,5)_{\sfrac12}\big)$. To do this, we just need to show that the Kac--Moody algebra $\A[1,1]\E[7,1]$ is not realized as the canonical Kac--Moody subalgebra of any $\V$ in this genus. This can be achieved by arguing in the spirit of Lemma \ref{lem:E8subalgebra}. Namely, if it were the case that $\A[1,1]\E[7,1]\subset \V$, we could glue $\V$ to $\mathrm{Ext}(\E[6,1]L_{\sfrac67})$ to obtain a holomorphic VOA of central charge 16, which must necessarily be $\E[8,1]^2$ in order to fit $\E[6,1]$ inside of it. If this were possible, then there should be an $\A[1,1]\E[6,1]\E[7,1]$ subalgebra of $\E[8,1]^2$ with trivial centralizer, but none exists.

\end{proof}

\clearpage 

\subsection{Bosonic chiral algebras with $c>8$}\label{app:classification:clt24}

As we commented earlier, the classification of chiral algebras with $8<c\leq 24$ involves different methods than those used in the previous subsection. One of the reasons for this is that we do not always have an analog of Theorem \ref{theorem:clt8characters}; indeed, inspection of the tables in Appendix \ref{app:data} reveals that a genus $(c,\Cat)$ with $8<c\leq 24$ and $\mathrm{rank}(\Cat)\leq 4$ does not necessarily have a unique $(c,\Cat)$-admissible character vector, and so there may be multiple possibilities for the dimension of the global symmetry group. Second of all, even for genera which do have a unique admissible character vector, Theorem \ref{theorem:kmsubalgebra} and in particular Equation \eqref{eqn:rankcc} become less constraining when $c$ is larger. Thus, we use a different strategy, summarized by the following principle.

\begin{theorem}[The gluing principle]\label{theorem:gluingprinciple}
    Let $\tilde\V$ be a fixed chiral algebra in the genus $(\tilde c,\overline{\Cat})$. There is a gluing map 
\begin{align}
   \mathrm{Glue}_{\tilde \V}(\V,\phi)=(\mathcal{A},\iota).
\end{align} 
Its input is a chiral algebra $\V$ in the genus $(c,\Cat)$, and a one-to-one map $\phi$ from simple $\tilde\V$-modules to simple $\V$-modules which we require to define a braid-reversing equivalence $\Cat\to\overline{\Cat}$. Its output is a pair consisting of the holomorphic VOA
\begin{align}
    \mathcal{A}=\bigoplus_i\V(i)^\ast\otimes \tilde\V( \phi i)
\end{align} 
of central charge $C:=c+\tilde c$, and the induced primitive embedding (see Definition \ref{defn:primitivesubalgebra})
\begin{align}
\begin{split}
    \iota:\tilde\V&\hookrightarrow\mathcal{A} \\
    |\varphi\rangle & \mapsto |0\rangle_{\V}\otimes |\varphi\rangle_{\tilde\V}.
    \end{split}
\end{align}
Conversely, assuming Conjecture \ref{conj:rationalcosets}, there is a coset map
\begin{align}
    \mathrm{Cos}_{\tilde\V}(\mathcal{A},\iota)=(\V,\phi)
\end{align}
which assigns to each holomorphic VOA $\mathcal{A}$ of central charge $C$, and each primitive embedding $\iota:\tilde\V\hookrightarrow\mathcal{A}$, the chiral algebra 
\begin{align}
    \V=\Com_{\mathcal{A}}(\iota \tilde\V)
\end{align}
belonging to the genus $(c,\Cat)$ with $c:=C-\tilde c$, and the braid-reversing equivalence $\phi:\Cat\to\overline{\Cat}$ induced by decomposing $\mathcal{A}$ into $\V\otimes \tilde\V$-modules. Furthermore, these maps are inverses to each other in the sense that 
\begin{align}
   \mathrm{Cos}_{\tilde \V}( \mathrm{Glue}_{\tilde\V}(\V,\phi))\cong (\V,\phi)
\end{align}
and, assuming Conjecture \ref{conj:rationalcosets}, 
\begin{align}
    \mathrm{Glue}_{\tilde\V}(\mathrm{Cos}_{\tilde\V}(\mathcal{A},\iota)) \cong (\mathcal{A},\iota).
\end{align}
\end{theorem}
\begin{proof}
    This essentially follows from Theorem \ref{theorem:gluingprelude}.
\end{proof}
\begin{remark}
    Without assuming Conjecture \ref{conj:rationalcosets}, we still have that every theory $\tilde\V$ is in the image of the map 
    \begin{align}
        \mathrm{Cos}_{\tilde\V}'(\mathcal{A},\iota):= \Com_{\mathcal{A}}(\iota \tilde\V).
    \end{align}
    That is, $\mathrm{Gen}(c,\Cat)\subset\mathrm{Im}(\mathrm{Cos}'_{\tilde\V})$. Conjecture \ref{conj:rationalcosets} strengthens this inclusion to an equality.
\end{remark}
The gluing principle implies that if one has knowledge of even a single theory $\tilde\V$ in the genus $(\tilde c,\overline{\Cat})$, then one can obtain every theory in the complementary genus $(c,\Cat)$ by taking cosets of holomorphic VOAs by $\tilde\V$. Thus, one is reduced to the problem of studying equivalence classes of embeddings of $\tilde\V$ into holomorphic VOAs. Of course we are well-positioned to pursue this program, because Table \ref{tab:generaclassification} gives us a natural choice of seed theory $\tilde\V$ for each modular category $\overline{\Cat}$ with $\mathrm{rank}(\Cat)\leq 4$. Appendix \ref{app:liecurrentalgebras} provides the relevant technical tools to study their embeddings into holomorphic VOAs. It turns out that, with the exception of the classification of the genus $\big(\sfrac{58}3,(\AA_1,7)_{\sfrac12}\big)$, we do not need to rely on Conjecture \ref{conj:rationalcosets}, because we can verify the regularity of the cosets in each case by showing that they are always conformal extensions of strongly regular subVOAs. In the genus $\big(\sfrac{58}3,(\AA_1,7)_{\sfrac12}\big)$, we were not always able to compute a strongly regular conformal subalgebra for each coset, but we expect that this minor gap in our treatment can be filled using standard techniques.

In the rest of this section, we offer two representative sample calculations. All of the rest of our data can be found in Appendix \ref{app:data}.

\begin{ex}[Genus $(c,\Cat)$ with $c=\sfrac{59}5$ and $\Cat=(\AA_1,1)\boxtimes (\GG_2,1)$] The most straightforward genera $(c,\Cat)$ to classify are those for which the seed theory $\tilde\V \in(\tilde c,\overline{\Cat})$ can be taken to be an affine Kac--Moody algebra. The simplest case is when $\tilde{\V}$ is of the form $(\mathfrak{h})_k$, with  $\mathfrak{h}$ a simple Lie algebra at level $k=1$. In this case, if $\mathcal{A}$ is a VOA with Kac--Moody algebra $\mathcal{K}=(\mathfrak{g}_{1})_{k_1}\otimes\cdots\otimes(\mathfrak{g}_{m})_{k_m}\otimes\mathsf{U}_1^r$, then any embeding $(\mathfrak{h})_1\hookrightarrow \mathcal{A}$ must factor through an embedding into one of the simple constituents $(\mathfrak{g}_i)_{k_i}$ with $k_i=1$, i.e.  
\begin{align}
    \tilde{\V}=(\mathfrak{h})_1\hookrightarrow (\mathfrak{g}_i)_{1} \hookrightarrow\mathcal{A},
\end{align}
see Example \ref{ex:level1embedding}.
Furthermore, the embedding $(\mathfrak{h})_1\hookrightarrow(\mathfrak{g}_i)_{1}$ must be induced by an ordinary embedding of Lie algebras $\mathfrak{h}\hookrightarrow \mathfrak{g}_i$ with embedding index equal to 1 (see Appendix \ref{app:liecurrentalgebras} for the relevant background). 

Let us apply these considerations to the genus $(c,\Cat) = \big(\sfrac{59}5,(\AA_1,1)\boxtimes (\GG_2,1)\big)$. In this case, the complex conjugate category is 
\begin{align}
    \overline{\Cat}\cong (\CC_3,1)\cong (\EE_7,1)\boxtimes (\FF_4,1).
\end{align}
We may carry out the classification in two different ways and confirm that they give the same result. We may either enumerate cosets of holomorphic VOAs of central charge $C=16$ by the seed theory $\tilde\V = \C[3,1] \in \mathrm{Gen}(\sfrac{21}5,\overline{\Cat})$, or we may take cosets of holomorphic VOAs of central charge $C=24$ by the seed theory $\tilde \V = \E[7,1]\F[4,1] \in \mathrm{Gen}(\sfrac{61}5,\overline{\Cat})$, i.e.\ 
\begin{align}
    \mathrm{Gen}(c,\Cat) = \{ \mathcal{A}^{(16)}\big/\C[3,1] \} = \{\mathcal{A}^{(24)}\big/\E[7,1]\F[4,1]\}.
\end{align}

Let us first consider $\tilde \V=\C[3,1]$. We must study embeddings of $\tilde \V$ into $\E[8,1]^2$ and $\D[16,1]^+$, the two holomorphic VOAs with $C=16$. In the first case, we must embed $\C[3,1]$ into one of the $\E[8,1]$ factors (it does not matter which of the two, because there is an obvious permutation automorphism which exchanges them). A computer calculation gives that there is a unique Lie algebra embedding of $\C[3]\hookrightarrow\E[8]$ with embedding index 1, and so correspondingly there is a unique way to embed $\C[3,1]\hookrightarrow \E[8,1]^2$ (up to symmetries). The centralizer of this embedding is $\mathrm{Cent}_{\E[8]^2}(\C[3])\cong \A[1]\G[2]\E[8]$, where all three factors have embedding index $1$, and thus we have the inclusion $\A[1,1]\G[2,1]\E[8,1]\subset \E[8,1]^2\big/ \C[3,1]$. In fact, $\A[1,1]\G[2,1]\E[8,1]$ already has the correct central charge and the right representation category, so the inclusion is actually an equality, and we have the first theory in the genus $(c,\Cat)$, 
\begin{align}
    \E[8,1]^2\big/\C[3,1] = \A[1,1]\G[2,1]\E[8,1].
\end{align}

This theory is a somewhat trivial one, because it is a tensor product of $\E[8,1]$ with a theory which was already present in the genus $(c-8,\Cat)$. We can get a non-trivial theory by considering cosets of the form $\D[16,1]^+\big/\C[3,1]$. It turns out that there is a unique embedding $\C[3]\hookrightarrow\D[16]$ with embedding index 1 (up to equivalence), and the centralizer is given by $\mathrm{Cent}_{\D[16]}(\C[3])\cong \A[1]^{(3)}\D[10]^{(1)}$, where the numbers in the superscripts denote the embeding indices into $\D[16]$. Thus, we have the inclusion $\A[1,3]\D[10,1]\subset \D[16,1]^+\big/\C[3,1]$. Since $\A[1,3]\D[10,1]$ already has the correct central charge, the coset must be a conformal extension of $\A[1,3]\D[10,1]$, i.e. 
\begin{align}
    \D[16,1]^+\big/\C[3,1] \cong \mathrm{Ex}(\A[1,3]\D[10,1]).
\end{align}
The decomposition of $\V:=\mathrm{Ex}(\A[1,3]\D[10,1])$ into $\A[1,3]\D[10,1]$-modules can be obtained using Proposition \ref{prop:modulardataextension}. Denoting the spin-$j$ representation of $\A[1,3]$ as $\A[1,3](j)$ for $j=0,\sfrac12,1,\sfrac32$, and denoting the four modules of $\D[10,1]$ as $\D[10,1]$, $\D[10,1](v)$, $\D[10,1](s)$, $\D[10,1](c)$, we have 
    \begin{align}
    \begin{split}
        \V &= \A[1,3]\otimes \D[10,1] \oplus \A[1,3](\sfrac32) \otimes \D[10,1](s) \\
        \V(1) &= \A[1,3](\sfrac12)\otimes \D[10,1](s) \oplus \A[1,3](1) \otimes \D[10,1] \\
        \V(2) &= \A[1,3]\otimes \D[10,1](c)\oplus \A[1,3](\sfrac32)\otimes \D[10,1](v) \\
        \V(3) &= \A[1,3](\sfrac12)\otimes \D[10,1](v) \oplus \A[1,3](1)\otimes \D[10,1](c),
    \end{split}
    \end{align}
which reveals that $\V$ is a simple-current extension.

We can confirm this calculation by enumerating the theories in $\mathrm{Gen}(c,\Cat)$ another way: namely, we take cosets of $C=24$ holomorphic VOAs by $\tilde\V = \E[7,1]\F[4,1]$. It is straightforward to show that the only holomorphic VOAs with $C=24$ which admit a subalgebra isomorphic to $\tilde\V$ are $\E[8,1]^3$ and $\mathbf{S}(\D[10,1]\E[7,1]^2)$, where the latter is the unique $C=24$ holomorphic VOA with Kac--Moody algebra given by $\D[10,1]\E[7,1]^2$. By computing Lie algebra centralizers and embedding indices, one finds that the corresponding cosets take the form
\begin{align}
   \mathrm{Gen}(c,\Cat)=\{ \E[8,1]^3\big/\E[7,1]\F[4,1] \cong \A[1,1]\G[2,1]\E[8,1], \ \mathbf{S}(\D[10,1]\E[7,1]^2)\big/\E[7,1]\F[4,1]\cong \mathrm{Ex}(\A[1,3]\D[10,1])\}.
\end{align}
Thus, we again confirm that the above two theories are the only two theories in the genus $(c,\Cat)=\big(\sfrac{59}5,(\AA_1,1)\boxtimes(\GG_2,1)\big)$.

Finally, we explain how to compute the character vectors of these two theories. In Appendix \ref{app:(A1,1)x(G2,1)}, the most general $(c,\Cat)$-quasi-admissible character is written down; it depends on a single  parameter $\alpha_{2,0}$, which must be a non-negative integer because it appears as the leading coefficient in the $q$-expansion of $\ch_2(\tau)$. The vacuum character has a $q$-expansion which starts as
\begin{align}
\ch_0(\tau) = q^{-\sfrac{c}{24}}(1+(36\alpha_{2,0}+193)q+(278\alpha_{2,0}+7872)q^2+\cdots).
\end{align}
The parameter $\alpha_{2,0}$ can be fixed for each theory by recalling that the coefficient of $q^{-\sfrac{c}{24}+1}$ is the dimension of the global symmetry group. For example, for $\mathrm{Ex}(\A[1,3]\D[10,1])$, we fix $\alpha_{2,0}$ it by demanding that
\begin{align}
    36\alpha_{2,0}+193=\dim(\A[1]\D[10])=193\implies \alpha_{2,0}=0
\end{align}
and for $\A[1,1]\G[2,1]\E[8,1]$ we similarly find that
\begin{align}
36\alpha_{2,0}+193=\dim(\A[1]\G[2]\E[8])=265\implies \alpha_{2,0}=2.
\end{align}
Note that had we misidentified the theories,  the calculation of $\alpha_{2,0}$ above may not have produced a non-negative integer. Thus, this constitutes another non-trivial check on our proposed classification.
\end{ex}

\begin{ex}[Genus $(c,\Cat)$ with $c=\sfrac{104}{7}$ and $\Cat=(\AA_1,5)_{\sfrac12}$]

Other genera are more subtle to classify. Take for example $(c,\Cat)=\big(\sfrac{104}{7},(\AA_1,5)_{\sfrac12}\big)$. Because there is no theory in the genus $(\sfrac{8}{7}, \overline{(\AA_1,5)}_{\sfrac12})$, our only option is to understand cosets of holomorphic VOAs with $C=24$ by the chiral algebra $\tilde\V = \mathrm{Ex}(\A[1,5]\E[7,1])$ in the genus $(\sfrac{64}{7}, \overline{(\AA_1,5)}_{\sfrac12})$, see Appendix \ref{app:classification:clt8} for a description of $\mathrm{Ex}(\A[1,5]\E[7,1])$. 

Calculating embeddings of conformal extensions  of affine Kac--Moody algebras is more involved than calculating embeddings of affine Kac--Moody algebras. We proceed as follows. An embedding of $\mathrm{Ex}(\A[1,5]\E[7,1])$ requires at least an embedding of $\A[1,5]\E[7,1]$, so we may start by listing out the latter. They are 
\begin{align}\label{eqn:possibleA15E71embeddings}
\begin{split}
\A[1,5]\E[7,1] &\hookrightarrow [\A[1,i]\E[7,1] ][\A[1,j]]'[\A[1,k]]'' \hookrightarrow  [\E[8,1]][\E[8,1]]'[\E[8,1]]''
\ \ \ (i+j+k=5, \ \  i=0,1), \\
\A[1,5]\E[7,1] &\hookrightarrow [\A[1,i] ][\A[1,j]\E[7,1]]'\hookrightarrow[\D[16,1]^+][\E[8,1]]' \ \ \ (i+j=5, \ \ j=0,1), \\
[\A[1,5] ][\E[7,1]]' &\hookrightarrow [\A[17,1]][\E[7,1]]' \hookrightarrow \mathbf{S}(\A[17,1]\E[7,1]), \\
\A[1,5]\E[7,1] &\hookrightarrow [\A[1,i]][\A[1,j]]'[\E[7,1]]'' \hookrightarrow [\D[10,1]][\E[7,1]]'[\E[7,1]]'' \hookrightarrow \mathbf{S}(\D[10,1]\E[7,1]^2) \ \ \ (i+j=5)
\end{split}
\end{align}
where we understand $\A[1,i]$ with $i=0$ to be the trivial theory. In the above, we first embed $\A[1,5]\hookrightarrow \A[1,i_1]\cdots \A[1,i_n]$ diagonally, where $i_1+\cdots + i_n=5$. Then, we embed into the Kac--Moody algebra of the $C=24$ holomorphic VOA. Notation like $[X_1][Y_1]'\hookrightarrow [X_2][Y_2]'$ means that $X_1$ is embedded into $X_2$ and $Y_1$ is embedded into $Y_2$.

We can argue that most of these embeddings of $\A[1,5]\E[7,1]$ do not extend to embeddings of $\mathrm{Ex}(\A[1,5]\E[7,1])$ as follows. By way of example, let us consider the unique embedding occurring in the third line above. Assume that $\A[1,5]\E[7,1]\hookrightarrow\mathbf{S}(\A[17,1]\E[7,1])$ does extend to an embedding of $\mathrm{Ex}(\A[1,5]\E[7,1])$. Then the coset $\mathbf{S}(\A[17,1]\E[7,1])\big/\mathrm{Ex}(\A[1,5]\E[7,1])$ would need to have a Kac--Moody algebra of the form $\A[12,1]\mathsf{U}_1^2$, which can be computed from the centralizer of $\A[1,5]\E[7,1]$ in $\A[17,1]\E[7,1]$.  On the other hand, consulting Appendix \ref{app:(A1,5)half}, we find that the most general $(c,\Cat)$-quasi-admissible character has a $q$-expansion of the form 
\begin{align}
    \ch_0(\tau) = q^{-\sfrac{c}{24}}(1+(138\alpha_{1,0}+188)q+\cdots)
\end{align}
where $\alpha_{1,0}$ must be a non-negative integer for admissibility (because it appears as the leading coefficient of one of the other components of the character vector). If there were a theory in the genus $(c,\Cat)$ with Kac--Moody algebra $\A[12,1]\mathsf{U}_1^2$, it would require that 
\begin{align}
    138\alpha_{1,0}+188 = \dim(\A[12]\mathsf{U}_1^2)=170
\end{align}
admit an integer solution. Since it does not we arrive at a contradiction. 

Going through the rest of the embeddings in Equation \eqref{eqn:possibleA15E71embeddings} in a similar manner, we find that only the following four remain as possibilities,
\begin{align}
\begin{split}
    \A[1,5]\E[7,1] &\hookrightarrow [\A[1,1]\E[7,1]][\A[1,4]]'\hookrightarrow[\E[8,1]][\E[8,1]]'\E[8,1] \\ 
    \A[1,5]\E[7,1] &\hookrightarrow [\A[1,1]\E[7,1]][\A[1,1]]'[\A[1,3]]''\hookrightarrow[\E[8,1]][\E[8,1]]'[\E[8,1]]'' \\
    \A[1,5]\E[7,1] &\hookrightarrow [\A[1,4]][\A[1,1]\E[7,1]]'\hookrightarrow[\D[16,1]^+][\E[8,1]]' \\ 
    \A[1,5]\E[7,1] &\hookrightarrow [\A[1,2]][\A[1,3]]'[\E[7,1]]'' \hookrightarrow [\D[10,1]] [\E[7,1]]'[\E[7,1]]''\hookrightarrow\mathbf{S}(\D[10,1]\E[7,1]^2).
\end{split}
\end{align}
It turns out that each of these embeddings, except for the second one, does extend to an embedding of $\mathrm{Ex}(\A[1,5]\E[7,1])$. In each case, one sees this by decomposing the $C=24$ holomorphic VOA $\mathcal{A}$ into modules of the tensor product of $\A[1,5]\E[7,1]$ with some conformal subalgebra $\tilde{\mathcal{W}}$ of $\Com_{\mathcal{A}}(\A[1,5]\E[7,1])$. The decomposition either takes the form 
\begin{align}
\mathcal{A}\cong (\A[1,5]\E[7,1]\oplus \A[1,5](\sfrac54)\E[7,1](\sfrac34)) \otimes \tilde{\mathcal{W}} \oplus (\text{modules})
\end{align}
in which case the embedding \emph{does} extend to one of $\mathrm{Ex}(\A[1,5]\E[7,1])$, or it takes the form
\begin{align}
\mathcal{A}\cong \A[1,5]\E[7,1] \otimes \tilde{\mathcal{W}} \oplus (\text{modules})
\end{align}
in which case it does not. In both cases, we are of course assuming that $\tilde{\mathcal{W}}$ does not appear in the summand that we called ``modules''. 

In total, we conclude that 
\begin{align}
    \mathrm{Gen}(\sfrac{104}{7},(\AA_1,5)_{\sfrac12}) = \left\{ \begin{array}{l}
        \E[8,1]^3\big/\mathrm{Ex}(\A[1,5]\E[7,1]) \cong \E[8,1]\mathrm{Ex}(\E[6,1]L_{\sfrac67}) \\
        \D[16,1]^+\E[8,1]\big/\mathrm{Ex}(\A[1,5]\E[7,1]) \cong \mathrm{Ex}(\D[13,1]\mathsf{U}_{1,12} L_{\sfrac67}) \\
        \mathbf{S}(\D[10,1]\E[7,1]^2)\big/\mathrm{Ex}(\A[1,5]\E[7,1])\cong \mathrm{Ex}(\B[8,1]\F[4,1]\mathsf{W}_{2,3}^{\mathsf{A}_1})
    \end{array}\right\}
\end{align}
where we have defined $\mathsf{W}_{p,q}^{\mathsf{X}_r} = \mathsf{X}_{r,p}\mathsf{X}_{r,q}\big/ \mathsf{X}_{r,p+q}.$ The Kac--Moody part of $\mathcal{A}\big/ \mathrm{Ex}(\A[1,5]\E[7,1])$ can be computed using Proposition \ref{prop:KMofCoset}. In the first two theories, the commutant of the Kac--Moody part inside of $\mathcal{A}$ must be $L_{\sfrac67}$ by the classification of minimal models. In the third case, the appearance of $\mathsf{W}_{2,3}^{\A[1]}$ is due to the fact that the embedding of $\A[1,5]$ into $\D[10,1]\E[7,1]$ factors through an embedding into $\A[1,2]\A[1,3]$.  
\end{ex}

\clearpage 

\subsection{Chiral fermionic theories with $c< 23$}\label{app:classification:fermionic}

Our first application of the results of Appendix \ref{app:classification:clt8} and Appendix \ref{app:classification:clt24} is to the classification of chiral fermionic RCFTs (see Footnote \ref{footnote:superalgebra} for a note on terminology). The connection between the classification of bosonic VOAs and fermionic VOAs is obtained through  chiral bosonization/chiral fermionization.
\begin{theorem}[Chiral bosonization/chiral fermionization]
    Let $\V_{\mathrm{NS}}=\V_{\bar 0}\oplus \V_{\bar 1}$ be a unitary, strongly regular fermionic VOA. Then chiral bosonization is the map 
\begin{align}
    \V_{\mathrm{NS}} \mapsto (\V_{\bar 0},\V_{\bar 1})
\end{align}
which assigns to $\V_{\mathrm{NS}}$ the strongly regular bosonic VOA $\V_{\bar 0}$, and the irreducible $\V_{\bar 0}$-module $\V_{\bar 1}$. The pair $\big(\Rep(\V_{\bar 0}),\V_{\bar 1}\big)$ defines a spin modular category (see Definition \ref{defn:spinMTC}). Conversely, every pair $(\V_{\bar 0},\V_{\bar 1})$, with $\V_0$ a strongly regular bosonic chiral algebra and $\V_{\bar 1}$ an irreducible $\V_0$-module such that $(\Rep(\V_{\bar 0}),\V_{\bar 1})$ is a spin MTC, defines a fermionic VOA through chiral fermionization, 
\begin{align}
    (\V_{\bar 0},\V_{\bar 1})\mapsto \V_{\mathrm{NS}} := \V_{\bar 0}\oplus \V_{\bar 1}.
\end{align}
Chiral bosonization and chiral fermionization are inverse to one another. 
\end{theorem}

\begin{ex}[The Ising model/the free Majorana fermion]
    The representation category $\Rep(L_{\sfrac12})$ of the chiral algebra $L_{\sfrac12}$ of the Ising model is a spin MTC with fermionic object given by $L_{\sfrac12}(\sfrac12)$, the irreducible module built on top of the primary operator of conformal dimension $\sfrac12$. Its chiral fermionization $\mathrm{Fer}(1) = L_{\sfrac12}\oplus L_{\sfrac12}(\sfrac12)$ is the free Majorana fermion. 
\end{ex}
\begin{remark}\label{remark:(D12)1}
    It is important to note that a bosonic chiral algebra $\V_{\bar 0}$ may admit multiple inequivalent chiral fermionizations. Concretely, $\V_{\bar 0}$ may have multiple irreducible modules which could play the role of $\V_{\bar 1}$. This is why we keep track of the choice of $\V_{\bar 1}$ in the chiral fermionization/chiral bosonization procedure: it is required in order for the two maps to be inverses of one another.
    
    The simplest example of this is $\V_{\bar 0} = \D[12,1]$. It has four simple modules, $\D[12,1]$, $\D[12,1](v)$, $\D[12,1](s)$, and $\D[12,1](c)$, and any of the last three can be taken to be $\V_{\bar 1}$ in the chiral fermionization procedure. Taking $\V_{\bar 1}=\D[12,1](s)$ to be the spinor module (or equivalently $\V_{\bar 1}=\D[12,1](c)$ to be the conjugate spinor module, the two are related by the $\ZZ_2$ Dynkin diagram outer automorphism of $\D[12]$) recovers the two Conway theories $V^{f\natural},V^{s\natural}$ \cite{duncan2007super,Duncan:2014eha}, which become the same upon forgetting the $\mathcal{N}=1$ supersymmetry structure of the former. On the other hand, taking $\V_{\bar 1}=\D[12,1](v)$ recovers the theory $\mathrm{Fer}(24)$ of $24$ free fermions. (See \cite{Harrison:2018joy} for a nice discussion.) 
\end{remark}
\begin{remark}
    See \S\ref{subsec:overview:fermions} for remarks on the relationship between chiral fermionization/chiral bosonization and ordinary bosonization/fermionization. 
\end{remark}

Chiral bosonization/chiral fermionization makes precise the idea that we are free to pass back and forth between the bosonic and fermionic description of a theory without losing anything. In particular, any classification of fermionic RCFTs has a bosonic avatar.

We will restrict our attention in this work to chiral fermionic RCFTs, though our methods are certainly more general. These are fermionic VOAs which define a fully consistent chiral CFT in their own right, without having to incorporate right-movers. In a sense, they are fermionic analogs of holomorphic VOAs (see Definition \ref{defn:holomorphicVOA}). Here is one definition.
\begin{defn}[Chiral fermionic RCFT]
A strongly regular fermionic VOA $\V_{\mathrm{NS}}$ is said to be a chiral fermionic CFT if the only irreducible admissible NS-sector module is $\V_{\mathrm{NS}}$ itself.
\end{defn}
In the spirit of working bosonically, we can give an equivalent definition in terms of the representation category of $\V_{\bar 0}$. 

\begin{prop}\label{prop:RepCategoryBosonicSubalgebra}
    A strongly regular fermionic VOA $\V_{\mathrm{NS}}=\V_{\bar 0}\oplus \V_{\bar 1}$ of central charge $c$ is a chiral fermionic RCFT if and only if $\Rep(\V_{\bar 0})\cong \Rep(\mathfrak{so}(n)_1)$, where $n\equiv 2c~\mathrm{mod}~16$. In particular, $c\in \frac12 \ZZ$. 
\end{prop}

\begin{proof}
Let us start with the forward arrow. By Theorem \ref{theorem:fermionicCAsuperMTC}, the $\V_{\bar 0}$-modules which arise inside of $\V_{\mathrm{NS}}$-modules generate a super-modular category $\Rep(\V_{\bar 0}\vert \V_{\bar 1})$, and $\Rep(\V_{\bar 0})$ defines a minimal modular extension of this super-modular category. Because $\V_{\mathrm{NS}}$ is the only $\V_{\mathrm{NS}}$-module, the only generators of this super-modular category are $\V_{\bar 0}$ and $\V_{\bar 1}$, and hence it must be equivalent to $\textsl{sVec}$, the category of super vector spaces. By \cite{bruillard2020classification}, the minimal modular extensions of $\textsl{sVec}$ are the 16 UMTCs given by $\Rep(\mathfrak{so}(n)_1)$ for some $n$. Precisely which minimal modular extension we land on is determined by the central charge, as in the statement of the proposition. 

In the other direction, assume that $\Rep(\V_{\bar 0})\cong \Rep(\mathfrak{so}(n)_1)$ for some $n$, and $\V_{\bar 1}$ is a choice of a fermionic object. In this case, the super MTC $\Rep(\V_{\bar 0}\vert\V_{\bar 1})$ of which $\Rep(\V_{\bar 0})$ is a minimal modular extension is $\textsl{sVec}$;  Theorem \ref{theorem:modulesbosonicsubalgebra} and Theorem \ref{theorem:fermionicCAsuperMTC} illustrate that $\frac12 \mathrm{rank}(\Rep(\V_{\bar 0}\vert\V_{\bar 1}))$ is the number of NS-sector modules, which in this case is $1$. So we conclude that $\V_{\bar 0}\oplus \V_{\bar 1}$ is a chiral fermionic RCFT.
\end{proof}

Because $\Rep(\mathfrak{so}(n)_1)$ has either rank 3 or 4 depending on whether $n$ is odd or even, the bosonic subalgebra $\V_{\bar 0}$ of any chiral fermionic RCFT with sufficiently low central charge must arise in our classification of bosonic chiral algebras with at most four primary operators, and hence must appear somewhere in the tables of Appendix \ref{app:data}. Conversely, $\V_{\mathrm{NS}}$ can be reconstructed from this bosonic subalgebra by chiral fermionization. 

For the sake of being explicit, let us describe how chiral fermionization/chiral bosonization works for chiral fermionic CFTs at the level of characters. When $n=2c$ is odd, the bosonic subalgebra has three characters, 
\begin{align}
    \ch(\tau)=(\ch_0(\tau),\ch_1(\tau),\ch_2(\tau)),
\end{align}
corresponding to modules with $h=(0,\sfrac{n}{16},\sfrac12)~\mathrm{mod}~1$, respectively. Generalizing Example \ref{ex:oddc}, the characters of the fermionic theory are then 
\begin{align}
\begin{split}
Z_{\mathrm{NS}}(\tau)&:=\mathrm{Tr}_{\V_{\mathrm{NS}}}q^{L_0-\sfrac{c}{24}} = \ch_0(\tau)+\ch_2(\tau)\\ 
Z_{\widetilde{\mathrm{NS}}}(\tau) &:= \mathrm{Tr}_{\V_{\mathrm{NS}}}(-1)^Fq^{L_0-\sfrac{c}{24}} = \ch_0(\tau)-\ch_2(\tau)\\
Z_{\mathrm{R}}(\tau) &:= \mathrm{Tr}_{\V_R}q^{L_0-\sfrac{c}{24}}= 2\ch_1(\tau)\\
Z_{\widetilde{\mathrm{R}}}(\tau) &:=\mathrm{Tr}_{\V_R}(-1)^Fq^{L_0-\sfrac{c}{24}}=0,
\end{split}
\end{align}
where $\V_R=\V_R'\oplus\V_R'\circ (-1)^F$ is the unique R-sector module which is irreducible as a stable module (see Proposition \ref{prop:modulestability} and Theorem  \ref{theorem:modulesbosonicsubalgebra}). 

When $n=2c$ is even, the bosonic subalgebra has four characters, 
\begin{align}
    \ch(\tau)=(\ch_0(\tau),\ch_1(\tau),\ch_2(\tau),\ch_3(\tau)),
\end{align}
corresponding to modules  with $h=(0,\sfrac12,\sfrac{n}{16},\sfrac{n}{16})~\mathrm{mod}~1$, respectively. Generalizing Example \ref{ex:evenc}, the characters of the fermionic theory are 
\begin{align}
\begin{split}
Z_{\mathrm{NS}}(\tau)&:=\mathrm{Tr}_{\V_{\mathrm{NS}}}q^{L_0-\sfrac{c}{24}} = \ch_0(\tau)+\ch_1(\tau)\\ 
Z_{\widetilde{\mathrm{NS}}}(\tau) &:= \mathrm{Tr}_{\V_{\mathrm{NS}}}(-1)^Fq^{L_0-\sfrac{c}{24}} = \ch_0(\tau)-\ch_1(\tau)\\
Z_{\mathrm{R}}(\tau) &:= \mathrm{Tr}_{\V_{\mathrm{R}}}q^{L_0-\sfrac{c}{24}}= \ch_2(\tau)+\ch_3(\tau)\\
Z_{\widetilde{\mathrm{R}}}(\tau) &:=\mathrm{Tr}_{\V_{\mathrm{R}}}(-1)^Fq^{L_0-\sfrac{c}{24}}=\ch_2(\tau)-\ch_3(\tau)
\end{split}
\end{align}
where $\V_{\mathrm{R}}$ is the unique irreducible R-sector module, which is stable. Using a different lift of $(-1)^F$ in the Ramond sector (see Remark \ref{remark:ambiguouslifts}) we get instead that 
\begin{align}
Z_{\widetilde{\mathrm{R}}}(\tau)=\ch_3(\tau)-\ch_2(\tau).
\end{align}
When $n=8~\mathrm{mod}~16$, all three non-vacuum modules of $\V_{\bar 0}$ have $h=\sfrac12~\mathrm{mod}~1$, and any of them may play the role of $\ch_1(\tau)$, i.e.\ any of them may be used to chiral fermionize the bosonic theory. 

The above considerations show that when $n\neq 8 ~\mathrm{mod}~16$, the category $\Rep(\mathfrak{so}(n)_1)$ has a unique simple object $\V_{\bar 1}$ (up to isomorphism) which gives $\big(\Rep(\mathfrak{so}(n)_1),\V_{\bar 1}\big)$ the structure of a spin MTC. Thus, we have the following.
\begin{prop}\label{prop:n!=8mod16chiralfermionicRCFTs}
Let $n\neq 8~\mathrm{mod}~16$. The chiral fermionic RCFTs of central charge $\sfrac{n}2$ are  in one-to-one correspondence with strongly regular bosonic chiral algebras in the genus $\big(\sfrac{n}2,\Rep(\mathfrak{so}(n)_1)\big)$.
\end{prop}
The case that $n=8~\mathrm{mod}~16$ requires a bit more care: there are a priori three possible chiral fermionizations, but they may not all be inequivalent. For example, when $n=8$, there is a unique possibility for the bosonic subalgebra $\V_{\bar 0}=\D[4,1]$. Furthermore, all three choices for $\V_{\bar 1}$, corresponding to the vector, spinor, and conjugate spinor modules, are permuted by the $S_3$ triality automorphisms of $\D[4,1]$. Therefore, there is a unique chiral fermionic RCFT with $c=4$. Remark \ref{remark:(D12)1} gives the corresponding analysis for the case of $c=12$.

Naively, we would need to understand the automorphisms of each of the $\V_{\bar 0}$ in the genus $(c,\Cat)=(20,(\DD_4,1))$ as well in order to determine how many chiral fermionic RCFTs there are with $c=20$. However, we will opt to take a shortcut instead.

\begin{prop}
    Let $N_c$ denote the number of (isomorphism classes of) chiral fermionic RCFTs with central charge $c$, and $N_{c}^{\mathrm{ff}}$ denote the number of such theories which have at least one dimension-$\sfrac12$ operator. Then 
\begin{align}\label{eqn:nochiralfermionicRCFTs}
N_c = N_{c+\sfrac12}^{\mathrm{ff}}.
\end{align}
\end{prop}
\begin{proof}
Given any chiral fermionic RCFT $\V_{\mathrm{NS}}$ with central charge $c$, one can obtain another one with central charge $c+\sfrac12$ and at least one dimension-$\sfrac12$ operator by tensoring a free Majorana fermion, $\V_{\mathrm{NS}}\otimes\mathrm{Fer}(1)$. On the other hand, given a chiral fermionic RCFT $\tilde{\V}_{\mathrm{NS}}$ with central charge $c+\sfrac12$ and at least one dimension-$\sfrac12$ operator, one may use Theorem \ref{theorem:decoupledfreefermions} to conclude that 
\begin{align}
\tilde{\V}_{\mathrm{NS}} \cong \tilde{\V}_{\mathrm{NS}}'\otimes \mathrm{Fer}(n)
\end{align}
for some $n>0$, and hence $\tilde{\V}_{\mathrm{NS}}'\otimes \mathrm{Fer}(n-1)$ defines a chiral fermionic RCFT with central charge $c$. These maps are inverse to one another, and hence define bijections which lead to Equation \eqref{eqn:nochiralfermionicRCFTs}.
\end{proof}
We can apply this to $c=20$, in which case we find by Proposition \ref{prop:n!=8mod16chiralfermionicRCFTs} and the tables in Appendix \ref{app:(B4,1)} that
\begin{align}\label{eqn:notheoriesc=20}
N_{20} = N^{\mathrm{ff}}_{20+\sfrac12}=40.
\end{align}
This turns out to be enough to conclude the following.

\begin{prop}
    Let $\V_{\bar 0}$ be a bosonic chiral algebra in the genus $\big(20,(\DD_4,1)\big)$, and let $\V_{\bar 1}$ and $\V_{\bar 1}'$ be two irreducible $\V_{\bar 0}$-modules. Then the corresponding chiral fermionizations $\V_{\bar 0}\oplus \V_{\bar 1}$ and $\V_{\bar 0}\oplus \V_{\bar 1}'$ are isomorphic if and only if $\V_{\bar 1}$ and $\V_{\bar 1}'$ have the same character. 
\end{prop}

\begin{proof}
    If $\V_{\bar 1}$ and $\V_{\bar 1}'$ have different characters, then they will certainly define inequivalent chiral fermionizations (because the fermionic theories they define will have different NS-sector partition functions). By using the tables in Appendix \ref{app:(D4,1)}, one finds that 
\begin{align}
\sum_{\V_{\bar 0}\in \mathrm{Gen}(20,(\DD_4,1))}(\text{\# of distinct non-vacuum characters of }\V_{\bar 0}) = 40
\end{align}
which implies that $N_{20}\geq 40$, with the inequality saturated precisely when it is the case that $\V_{\bar 0}\oplus \V_{\bar 1}\cong \V_{\bar 0}\oplus \V_{\bar 1}'$ whenever $\V_{\bar 1}$ and $\V_{\bar 1}'$ have the same character.  On the other hand, Equation \eqref{eqn:notheoriesc=20} says that the inequality is indeed saturated, and so we obtain the statement of the proposition. 
\end{proof}

\begin{remark}
    A quick way to infer the number of distinct non-vacuum characters of a bosonic chiral algebra $\V_{\bar 0}$ in the genus $\big(20,(\DD_4,1)\big)$, and hence the number of inequivalent chiral fermionizations of $\V_{\bar 0}$, is that it is equal to the number of distinct $d_i$ in the second-to-last column of the table of theories in Appendix \ref{app:(D4,1)}. For example, the chiral algebra $\mathbf{S}(\D[4,1]^6)\big/\D[4,1]$ has just one chiral fermionization, whereas the chiral algebra $\mathbf{S}(\B[4,1]^2\D[8,2])\big/ \D[4,1]$ has two. 
\end{remark}
Putting all of these results together, we obtain the following. 

\begin{theorem}[Classification of chiral fermionic RCFTs with $c<23$]
    The tables in Appendix \ref{app:data:chiralfermionicRCFTs} give the complete list of strongly-regular chiral fermionic RCFTs with $c<23$ and no dimension-$\sfrac12$ operators. The full list of theories, i.e.\ including those with at least one dimension-$\sfrac12$ operator, can be obtained by tensoring arbitrary numbers of decoupled free fermions with the theories in this list.
\end{theorem}

\clearpage

\subsection{Generalized symmetries of holomorphic VOAs}\label{app:classification:symmetries}

Our second application of the results of Appendix \ref{app:classification:clt8} and Appendix \ref{app:classification:clt24} is to the generalized global symmetries of holomorphic VOAs. The connection between the two is obtained through the following conjecture.

\begin{conj}[Symmetry/Subalgebra Duality]\label{conj:symsubduality}
Let $\mathcal{A}$ be a holomorphic VOA, and $\mathcal{F}$ a fusion category which acts on $\mathcal{A}$ by symmetries. Then we may define a map 
\begin{align}
    \mathrm{Sub}_{\mathcal{A}}(\mathcal{F})=\mathcal{A}^{\mathcal{F}}
\end{align}
which associates to $\mathcal{F}$ the conformal subalgebra $\mathcal{A}^{\mathcal{F}}$ of operators in $\mathcal{A}$ which commute with $\mathcal{F}$. Conjecturally, this subalgebra is strongly regular, and its modular category is given by the Drinfeld center construction, 
\begin{align}
    \Rep(\mathcal{A}^{\mathcal{F}}) \cong \mathcal{Z}(\mathcal{F}).
\end{align}

Conversely, given a strongly regular, conformal subalgebra $\mathcal{W}\subset \mathcal{A}$, we may consider the map 
\begin{align}
    \mathrm{Sym}_{\mathcal{A}}(\mathcal{W}) = \Rep(\mathcal{W})_{\mathcal{A}}
\end{align}
which assigns to $\mathcal{W}$ the fusion category $\Rep(\mathcal{W})_{\mathcal{A}}$ of $\mathcal{A}$-modules in $\Rep(\mathcal{W})$, where $\mathcal{A}$ is thought of as a Lagrangian algebra in the modular category $\Rep(\mathcal{W})$. The Drinfeld center of this fusion category is equivalent to the representation category of $\mathcal{W}$,
\begin{align}
    \Rep(\mathcal{W}) \cong \mathcal{Z}(\Rep(\mathcal{W})_{\mathcal{A}}),
\end{align}
and $\Rep(\mathcal{W})_{\mathcal{A}}$ acts as symmetries of $\mathcal{A}$.

Furthermore, these maps are inverses to one another in the sense that 
\begin{align}
    \mathrm{Sym}_{\mathcal{A}}(\mathrm{Sub}_{\mathcal{A}}(\mathcal{F}))\cong \mathcal{F}, \ \ \ \ \mathrm{Sub}_{\mathcal{A}}(\mathrm{Sym}_{\mathcal{A}}(\mathcal{W}))\cong\mathcal{W}
\end{align}
so that the finite symmetries and strongly regular, conformal subalgebras of any holomorphic VOA $\mathcal{A}$ are in correspondence with one another.
\end{conj}

\begin{remark}
Various parts of this conjecture are under better control when $\mathcal{F}=\Vec_G^\omega$ for some $(G,\omega)$, i.e.\ when the fusion category describes a finite group global symmetry (in this case, the Drinfeld center is known as the quantum double, and is often written $\mathcal{D}_\omega(G)$, see e.g.\ \cite{Coste:2000tq,Evans:2018qgz} for nice discussions). For example, in \cite{Carnahan:2016guf} it is proved that if $\V$ is regular, then so is the fixed point subalgebra $\V^G$, at least when $G$ is a \emph{solvable} finite group of automorphisms of $\V$. Under the assumption that fixed point subalgebras are regular for general finite groups, \cite{kirillov2002modular} proves that the modular category of $\mathcal{A}^G$, with $\mathcal{A}$ a holomorphic VOA and $G$ a non-anomalous group of symmetries, is $\mathcal{D}(G)$; the analogous result is proved in \cite{dong2017orbifold} for cyclic groups with possibly non-vanishing anomaly. See also \cite{Dijkgraaf:1989hb} for an early paper from physics which articulates similar concepts. The subalgebra-to-symmetry map is a kind of chiral version of ideas which appear in \cite{Fuchs:2002cm}; see \cite{Burbano:2021loy} for a more recent paper which contains similar ideas. Finally, we comment that a version of symmetry/subalgebra duality for conformal nets appears in \cite{Bischoff:2016jmy,Bischoff:2022fxf}.
\end{remark}
\begin{ex}[Chiral Verlinde lines]\label{ex:chiralVerlindelines}
Let $\mathcal{V}$ be a primitively embedded, non-conformal subalgebra of a holomorphic VOA $\mathcal{A}$, and let $\tilde{\V}:=\Com_{\mathcal{A}}(\V)$ and $\mathcal{W}:=\V\otimes\tilde\V$. Then there is a braid-reversing equivalence $\phi:\Cat\to\overline{\Cat}$ such that $\mathcal{A}$ decomposes into $\V\otimes\tilde\V$-modules as
\begin{align}
    \mathcal{A} = \bigoplus_i \V(i)^\ast\otimes \tilde{\V}(\phi i).
\end{align}
The VOA $\mathcal{A}$ defines a Lagrangian algebra in $\Rep(\V\otimes\tilde{\V})\cong \Cat\boxtimes\overline{\Cat}$, where $\Cat:=\Rep(\V)$. Furthermore, the fusion category of $\mathcal{A}$ modules in $\Cat\boxtimes \overline{\Cat}$ is simply $\Cat$, which acts as symmetries of $\mathcal{A}$ by the subalgebra-to-symmetry map.
\end{ex}

Symmetry/subalgebra duality essentially converts the problem of studying how any fusion category $\mathcal{F}$ with $\mathcal{Z}(\mathcal{F})\cong \Cat$ acts on the holomorphic VOAs with central charge $c$, to the problem of classifying chiral algebras in the genus $(c,\Cat)$. Indeed, for every theory $\mathcal{W}$ in the genus $(c,\Cat)$, and for every Lagrangian algebra $\mathcal{A}$ in $\Cat$, symmetry/subalgebra duality obtains a symmetry action of $\mathcal{F}\cong \Cat_{\mathcal{A}}$ on the holomorphic VOA $\mathcal{A}$, and conversely every symmetry action of an $\mathcal{F}$ satisfying $\mathcal{Z}(\mathcal{F})\cong \mathcal{C}$ on a holomorphic VOA arises in this way. In a sense, we are able to study symmetries by bootstrapping the subalgebras which they preserve. Let us see this in action in an example.

\begin{prop}[Classification of rank-2 symmetries of $\mathsf{E}_{8,1}$]\label{prop:c=8rank2symmetries}
All three unitary, rank-2 fusion categories --- $\Vec_{\ZZ_2}$, $\Vec_{\ZZ_2}^\omega$, and $\textsl{Fib}$ --- act in a unique way on $\E[8,1]$ up to conjugation by invertible automorphisms. Their fixed point subalgebras are $\D[8,1]$, $\A[1,1]\E[7,1]$ and $\G[2,1]\F[4,1]$, respectively.
\end{prop}
\begin{remark}
    The invertible cyclic group symmetries of $\E[8,1]$ can be determined by combining the results of \cite{kac1969automorphisms,kac1990infinite,dong1999automorphism}, as summarized in \cite{Burbano:2021loy}. Our methods give an independent proof in the case of $\ZZ_2$ symmetries, which should scale to other finite groups of low order. Perhaps more interesting is the fact that we are able to prove the uniqueness of the Fibonacci symmetry (up to conjugacy), which to our knowledge has not been done before. We also expect that it is possible to probe the structure of other non-invertible symmetry categories acting on $\E[8,1]$, and in particular, ones not obtained from Kramers--Wannier duality defects or chiral Verlinde lines. 
\end{remark}
\begin{proof}
    The Drinfeld centers of the fusion categories $\mathcal{F}=\Vec_{\ZZ_2}$, $\Vec^\omega_{\ZZ_2}$, and $\textsl{Fib}$, are $\mathcal{Z}(\mathcal{F}) = (\DD_8,1),$ $(\AA_1,1)\boxtimes (\EE_7,1)$, and $(\GG_2,1)\boxtimes (\FF_4,1)$, respectively. If $\mathcal{F}$ acts on $\E[8,1]$ by symmetries, then $\mathrm{Sub}_{\E[8,1]}(\mathcal{F})$, the subalgebra of $\E[8,1]$ commuting with $\mathcal{F}$, belongs to the genus $(8,\mathcal{Z}(\mathcal{F}))$. By Theorem \ref{theorem:clt8classification}, we must have that $\E[8,1]^{\mathcal{F}}\cong \D[8,1]$, $\A[1,1]\E[7,1]$, or $ \G[2,1]\F[4,1]$, respectively. Since $\D[8,1]$, $\A[1,1]\E[7,1]$, and $\G[2,1]\F[4,1]$ embed uniquely into $\E[8,1]$ up to conjugation by automorphisms of $\E[8,1]$, it follows that the symmetries they induce are also unique up to conjugacy, and so we conclude that there are unique equivalence classes of symmetry actions of $\Vec_{\ZZ_2}$, $\Vec^\omega_{\ZZ_2}$, and $\textsl{Fib}$ on the holomorphic  VOA $\E[8,1]$.
\end{proof}

We are able to obtain complete results on the rank-2 categories $\mathcal{F}$ acting on $\E[8,1]$ because we classified all chiral algebras in the genera $(8,\mathcal{Z}(\mathcal{F}))$. We can also obtain partial results on higher-rank fusion categories. While we do not have the complete list of chiral algebras in $(8,\mathcal{Z}(\mathcal{F}))$ when $\mathrm{rank}(\mathcal{F})>2$, tensoring complementary theories from the fourth column of Table \ref{tab:generaclassification} does give us examples of theories in genera of the form $(8,\Cat\boxtimes \overline{\Cat})$, and therefore examples of actions of modular categories $\Cat$ with $\mathrm{rank}(\Cat)\leq 4$.

\begin{table}[]
    \centering
    \begin{tabular}{c|c|c}
        $\mathcal{F}$ & $(\E[8,1]^2)^{\mathcal{F}}$ & Description \\ \toprule
        $\Vec_{\ZZ_2}$ & $\mathrm{Ex}(\E[8,2]L_{\sfrac12})$ & Permutation \\ 
        & $\mathrm{Ex}(\A[1,1]^2\E[7,1]^2)$ & $g_1'\cdot g_2'$ \\
        & $\D[8,1]\E[8,1]$ & $g_1$ \\ 
        & $\mathrm{Ex}(\D[8,1]^2)$ & $g_1\cdot g_2$ \\ \midrule
        $\Vec_{\ZZ_2}^\omega$ & $\mathrm{Ex}(\A[1,1]\D[8,1]\E[7,1])$ & $g_1'\cdot g_2$\\ 
        & $\A[1,1]\E[7,1]\E[8,1]$ & $g_1'$ \\  \midrule 
        $\textsl{Fib}$ & $\mathrm{Ex}(\E[7,2]\A[1,1]^2 L_{\sfrac7{10}})$ \\
        & $\mathrm{Ex}(\E[6,1]^2\A[2,2]L_{\sfrac45})$ \\ 
        & $\mathrm{Ex}(\A[1,3]\B[6,1]\E[7,1]L_{\sfrac{7}{10}})$ \\
        & $\E[8,1]\F[4,1]\G[2,1]$ \\\bottomrule
    \end{tabular}
    \caption{The classification of rank-2 unitary fusion categories acting on $\E[8,1]^2$ and their corresponding fixed-point subalgebras $(\E[8,1]^2)^{\mathcal{F}}$. Here, $g_i$ is the non-anomalous $\ZZ_2$ symmetry acting on the $i$th $\E[8,1]$ factor, and $g_i'$ is the anomalous $\ZZ_2$ symmetry acting on the $i$th $\E[8,1]$ factor, where $i=1,2$.}
    \label{tab:E8^2symmetries}
\end{table}

\begin{prop}[Chiral Verlinde lines of $\mathsf{E}_{8,1}$]\label{prop:E8chiralverlindelines}
Let $\Cat$ be any unitary modular tensor category with $c\neq 0~\mathrm{mod}~8$ and $\mathrm{rank}(\Cat)\leq 4$, except for $(\AA_1,5)_{\sfrac12}$ or $\overline{(\AA_1,5)}_{\sfrac12}$. Then $\Cat$, thought of as a fusion category, acts via chiral Verlinde lines on $\E[8,1]$ in at least one way. See Table \ref{tab:E8symmetries} for a partial list.
\end{prop}

\begin{remark}
    A few of these examples were also studied in \cite{Hegde:2021sdm}.
\end{remark}

\begin{proof}
This follows from the chiral Verlinde line construction of Example \ref{ex:chiralVerlindelines}. If $\Cat$ satisfies the hypotheses of the corollary, then $\V_\Cat\otimes \V_{\overline{\Cat}}$ is a chiral algebra in the genus $(8,\Cat\boxtimes\overline{\Cat})$. Any fusion category $\mathcal{F}$ with $\mathcal{C}\boxtimes \overline{\mathcal{C}}=\mathcal{Z}(\mathcal{F})$ defines a Lagrangian algebra in $\Cat\boxtimes\overline{\Cat}$, which in turn defines an extension of $\V_\Cat\otimes \V_{\overline{\Cat}}$ to a holomorphic VOA of central charge $8$ which admits $\mathcal{F}$ as a fusion category symmetry. The only option for this holomorphic VOA is $\E[8,1]$, and taking $\mathcal{F}=\mathcal{C}$ recovers first statement of the corollary.
\end{proof}

\begin{table}[]
    \centering
    \begin{tabular}{c|c}
        $\mathcal{F}$ & $(\D[16,1]^+)^{\mathcal{F}}$ \\ \toprule
        $\Vec_{\ZZ_2}$ & $\mathrm{Ex}(\D[8,1]^2)$ \\ 
        & $\mathrm{Ex}(\A[15,1]\mathsf{U}_1)$ \\
        & $\mathrm{Ex}(\D[4,1]\D[12,1])$ \\
        & $\D[16,1]$ \\ \midrule 
        $\Vec_{\ZZ_2}^\omega$ & $\mathrm{Ex}(\D[6,1]\D[10,1])$ \\ 
        & $\mathrm{Ex}(\A[15,1]\mathsf{U}_1)$ \\
        & $\mathrm{Ex}(\A[1,1]^2\D[14,1])$ \\ \midrule 
        $\textsl{Fib}$ & $\mathrm{Ex}(\C[8,1]\A[1,8])$ \\ 
        & $\mathrm{Ex}(\A[11,1]\B[3,1]L_{\sfrac{7}{10}}L_{\sfrac45})$ \\ 
        & $\mathrm{Ex}(\A[1,3]\C[3,1]\D[10,1])$ \\ 
        & $\mathrm{Ex}(\B[12,1]\G[2,1]L_{\sfrac7{10}})$ \\ \bottomrule
    \end{tabular}
    \caption{The classification of rank-2 unitary fusion categories acting on $\D[16,1]^+$ and their corresponding fixed-point subalgebras $(\D[16,1]^+)^{\mathcal{F}}$.}
    \label{tab:D16psymmetries}
\end{table}

It is possible to obtain similar results for holomorphic VOAs of central charge $16$ and $24$. Already at $c=16$, the list of chiral Verlinde lines which follow from the theories in Tables \ref{tab:generaclassification}--\ref{tab:clt16classificationpt3} is quite long, and is anyways straightforward to determine. Therefore we content ourselves here with tabulating just the rank-2 symmetries of $c=16$ holomorphic VOAs.

\begin{prop}[Classification of rank-2 symmetries of $\mathsf{E}_{8,1}^2$ and $\mathsf{D}_{16,1}^+$]\label{prop:c=16rank2symmetries}
    Table \ref{tab:E8^2symmetries} and Table \ref{tab:D16psymmetries} give the full classification of unitary rank-2 fusion categories acting on $\E[8,1]^2$ and $\D[16,1]^+$, respectively.
\end{prop}

\begin{remark}
    Notice that the chiral algebra $\V:=\mathrm{Ex}(\D[8,1]^2)$ appears as a $\ZZ_2$ fixed point subalgebra inside of both $\E[8,1]^2$ and $\D[16,1]^+$. Operationally, this is because $\Rep(\V)=(\DD_8,1)$ admits two Lagrangian algebras (corresponding to the rough and smooth boundary conditions of the toric code): extending $\V$ using one of these Lagrangian algebras leads to $\E[8,1]^2$, and extending with respect to the other leads to $\D[16,1]^+$. Relatedly, gauging the $\ZZ_2$-symmetry of $\E[8,1]^2$ which fixes $\V$ leads to $\D[16,1]^+$, and the $\ZZ_2$-symmetry of $\D[16,1]^+$ which fixes $\V$ is the emergent ``quantum'' symmetry (in the sense of \cite{Vafa:1989ih}) obtained from the gauging procedure. The same comments apply in the other direction.

The other non-anomalous $\ZZ_2$ symmetries of $\E[8,1]^2$ and $\D[16,1]^+$ are different. For example, consider the $\ZZ_2$ symmetry of $\D[16,1]^+$ whose fixed-point subalgebra is $\D[16,1]$. Gauging this $\ZZ_2$ symmetry of $\D[16,1]^+$ obtains $\D[16,1]^+$ again; equivalently, extending $\D[16,1]$ by its two Lagrangian algebras happens to lead to $\D[16,1]^+$ in both cases.
\end{remark}

\clearpage

\section{Data}\label{app:data}

In Appendix \ref{app:data:chiralfermionicRCFTs}, the complete list of chiral fermionic RCFTs with $c<23$ and no dimension-$\sfrac12$ operators is given. We only report theories with $\V_{\bar 1}\neq 0$, i.e.\ theories which are not purely bosonic. The entries with a $\star$ are lattice VOAs. The first column numbers these theories. The second column indicates the central charge. The third column provides the bosonic subalgebra $\V_{\bar 0}$, and the fourth column gives a conformal subalgebra (CSA) of $\V_{\bar 0}$. This CSA always contains the full Kac--Moody subalgebra, though we do not guarantee to have found the maximal lattice VOA which contains the $\mathsf{U}_1^r$ part.  The full list of chiral fermionic RCFTs with $c<23$, i.e.\ including those with dimension-$\sfrac12$ operators, can be obtained from this list by tensoring in arbitrary numbers of decoupled free chiral Majorana fermions.

In Appendix \ref{app:data:2primaries}, Appendix \ref{app:data:3primaries}, and Appendix \ref{app:data:4primaries}, we tabulate the classification of bosonic chiral algebras with two, three, and four primary operators, respectively. Each subsection within these appendices is based on a modular category $\Cat$. Table \ref{tab:generaclassification} is useful for efficiently navigating the many subsections.

The first table within each such subsection gives basic information about the category $\Cat$: its various names/realizations, the fusion category it corresponds to upon forgetting its braiding and twist, its normalized modular S-matrix (see Remark \ref{remark:normalizedS}), its twists (i.e.\ its conformal dimensions $h_i$ modulo $1$), the quantum dimensions of its simple objects (see Definition \ref{defn:tracedimension}), and its fusion rules. If the charge-conjugation matrix $\mathcal{S}^2$ is non-trivial, then we also provide the matrix $\mathcal{S}^+$ corresponding to the subrepresentation $V_+$ in Equation \eqref{eqn:pmdecomposition}, as well as the matrix $\mathcal{P}^+$ which intertwines between these matrices. We use the notation $\Vec_{\ZZ_n}^{\omega^m}$ to denote the fusion category corresponding to $\ZZ_n$ symmetry with $m$ units of anomaly. 

The second table within each subsection gives the data necessary to compute vector-valued modular forms. Each modular category corresponds to three modular representations. Each of these three representations has the same S-matrix (given in the previous table), but they differ in their T-matrix, which takes the form
\begin{align}
    \mathcal{T}_{i,j} =\delta_{i,j} e^{2\pi i (h_i-\sfrac{c}{24})}
\end{align}
where the three values of $c$ are reported in the first column of the second table, and the values of $h_i$ are reported in the first table. The three choices for $c$ are precisely the three integers modulo $24$ which are congruent to the chiral central charge $c(\Cat)$ of the category modulo $8$ (see Definition \ref{defn:gausssums}). For example, $\Cat=(\AA_1,1)$ has chiral central charge $c(\Cat)=1$ modulo 8, and there are three corresponding modular representations which describe chiral algebras $\V$ with $\Rep(\V)=\Cat$ and with central charge $c=1$, $9$, and $17$ modulo 24, respectively. If the modular representations have a non-trivial charge conjugation matrix, we use $\mathcal{P}^+$ to intertwine to subrepresentations with trivial charge-conjugation matrices. For each of these three modular representations, we provide a bijective exponent $\lambda$ and a characteristic matrix $\chi$ which can be used to algorithmically compute the entire space of weight-zero weakly-holomorphic vector-valued modular forms using the techniques of Appendix \ref{app:vvmfs}.

The third table within each subsection provides the low-order $q$-expansions for the most general $(c,\Cat)$-quasi-admissible vector-valued modular forms corresponding to the three values of $c$ in the range $0<c\leq 24$ which are congruent to $c(\Cat)$ modulo $8$ (see \S\ref{subsec:overview:characters} for the definition of $(c,\Cat)$-quasi-admissible). The most general $(c,\Cat)$-quasi-admissible functions for genera with $c>24$ can of course be reconstructed from the data in the second table using the techniques of Appendix \ref{app:vvmfs}.

The fourth table generally is a list of embeddings of Kac--Moody algebras \begin{align}
    \mathcal{K}:=(\mathfrak{h}_1)_{k_1}\otimes\cdots \otimes (\mathfrak{h}_n)_{k_n}\hookrightarrow(\mathfrak{g}_1)_{k_1'}\otimes\cdots\otimes (\mathfrak{g}_m)_{k_m'}=:\mathcal{K}'.
\end{align}
This list is complete in the sense that, for any $\mathcal{K}'$ which arises as a simple factor in a holomorphic VOA of central charge $c=24$, every possible embedding of $\mathcal{K}\hookrightarrow\mathcal{K}'$ is reported. For each embedding, a conformal subalgebra (CSA) of the corresponding commutant $\Com_{\mathcal{K}'}(\mathcal{K})$ is given. We make use of the following identifications, 
\begin{align}
    \B[1,1]\cong \A[1,2], \ \C[2,1]\cong \B[2,1], \  \C[1,1]\cong \A[1,1],   \ \D[3,1]\cong \A[3,1], \  \D[2,1] \cong \A[1,1]^2, \ \D[1,1]\cong \mathsf{U}_{1,4}, \  \mathsf{X}_{0,k}=0.
\end{align}
Sometimes this table is absent. 

The last tables contain lists of chiral algebras with representation category given by $\Cat$ and with $c\leq 24$.  The entries with a $\star$ are lattice VOAs. The first column is a number which labels the chiral algebra. The second column is the central charge $c$. The third column is a description of the chiral algebra, often as a coset of a holomorphic VOA. The fourth column is a conformal subalgebra (CSA); we often use the chiral algebras 
\begin{align}
    \mathsf{W}^{\mathsf{X}_r}_{p,q} := \mathsf{X}_{r,p}\otimes \mathsf{X}_{r,q}\big/ \mathsf{X}_{r,p+q}
\end{align}
as ingredients in defining this CSA. The CSA always includes the full Kac--Moody subalgebra, although we do not guarantee that we have found the maximal lattice VOA containing the $\mathsf{U}_1^r$ part. The fifth column contains the conformal dimensions of the non-identity primary operators, and the sixth column contains the degeneracies of the non-identity primary operators. The last column gives a concise specification of the character-vector of the theory, by specifying the coefficients $\alpha_{j,m}$ which are to be plugged in to the $q$-expansions recorded in the third table.

\clearpage 

\newcounter{theoryno}

\newcounter{fermionicno}

\subsection{Chiral fermionic theories}\label{app:data:chiralfermionicRCFTs}{\ }
\begin{table}[ht!]
   \begin{center}
\begin{tabular}{r|c|c|cHHH}
    No. & $c$ & $\V_{\bar 0}$ & CSA & $h_i$ & $d_i$ & $\alpha_{j,m}$ \\\toprule
    $\star$ \stepcounter{fermionicno}\thefermionicno & $12$ & $\D[12,1]$ & $\D[12,1]$ & $\sfrac{1}{2}, \sfrac{3}{2}, \sfrac{3}{2}$ & $24,2048,2048$ & $0,0$ \\\midrule 
    $\star$ \stepcounter{fermionicno}\thefermionicno & $14$ & $\E[7,1]^2\cong \E[8,1]^2\big/(\A[1,1]^2\hookrightarrow\E[8,1]^2)$ & $\E[7,1]^2$ & $\sfrac{3}{2}, \sfrac{3}{4}, \sfrac{3}{4}$ & $3136,56,56$ & $56$ \\\midrule
        $\star$ \stepcounter{fermionicno}\thefermionicno & $15$ & $\mathbf{S}(\A[15,1]\D[9,1])\big/(\D[9,1]\hookrightarrow\D[9,1])$ & $\A[15,1]$ & $\sfrac{3}{2}, \sfrac{7}{8}, \sfrac{7}{8}$ & $3640,120,120$ & $0$ \\\midrule
    \stepcounter{fermionicno}\thefermionicno & $\sfrac{31}{2}$ & $\E[8,2]\cong \mathbf{S}(\E[8,2]\B[8,1])\big/(\B[8,1]\hookrightarrow\B[8,1])$ & $\E[8,2]$ & $\sfrac{15}{16}, \sfrac{3}{2}$ & $248,3875$ & $0$ \\\midrule
    $\star$ \stepcounter{fermionicno}\thefermionicno & $16$ & $\mathbf{S}(\D[8,1]^3)\big/(\D[8,1]\hookrightarrow\D[8,1])$ & $\D[8,1]^2$ & $\sfrac{3}{2}, 1, 1$ & $4096,256,256$ & $256$  \\\midrule
    $\star$ \stepcounter{fermionicno}\thefermionicno & $17$ & $\mathbf{S}(\A[11,1] \D[7,1] \E[6,1])\big/(\D[7,1] \hookrightarrow \D[7,1])$ & $\A[11,1] \E[6,1]$ & $\sfrac{3}{2}, \sfrac{9}{8}, \sfrac{9}{8}$ & $4488, 544, 2097152$ & $0, 0$ \\\midrule
     \stepcounter{fermionicno}\thefermionicno & $\sfrac{35}{2}$ & $\mathbf{S}(\B[6,1] \C[10,1])\big/(\B[6,1] \hookrightarrow \B[6,1])$ & $\C[10,1]$ & $\sfrac{19}{16}, \sfrac{3}{2}$ & $1120, 4655$ & $0, 0$ \\
\end{tabular}
\end{center}
\end{table}

\clearpage

\begin{table}[H]
    \centering
    \begin{tabular}{r|c|c|cHHH}
    No. & $c$ & $\V_{\bar 0}$ & CSA & $h_i$ & $d_i$ & $\alpha_{j,m}$ \\\toprule
\stepcounter{fermionicno}\thefermionicno & $\sfrac{41}{2}$ & $\mathbf{S}(\B[3,1]^{2} \C[4,1] \D[6,2])\big/(\B[3,1] \hookrightarrow \B[3,1])$ & $\B[3,1] \C[4,1] \D[6,2]$ & $\sfrac{25}{16}, \sfrac{3}{2}$ & $10496, 5125$ & $0, 0$ \\
\stepcounter{fermionicno}\thefermionicno & $\sfrac{41}{2}$ & $\mathbf{S}(\A[4,1] \A[9,2] \B[3,1])\big/(\B[3,1] \hookrightarrow \B[3,1])$ & $\A[4,1] \A[9,2]$ & $\sfrac{25}{16}, \sfrac{3}{2}$ & $10496, 5125$ & $0, 0$ \\\midrule 
$\star$ \stepcounter{fermionicno}\thefermionicno & $21$ & $\mathbf{S}(\A[3,1]^{8})\big/(\A[3,1] \hookrightarrow \A[3,1])$ & $\A[3,1]^{7}$ & $\sfrac{3}{2}, \sfrac{13}{8}, \sfrac{13}{8}$ & $5096, 10752, 10752$ & $0, 0$ \\
\stepcounter{fermionicno}\thefermionicno & $21$ & $\mathbf{S}(\A[3,1]^{2} \D[5,2]^{2})\big/(\A[3,1] \hookrightarrow \A[3,1])$ & $\A[3,1] \D[5,2]^{2}$ & $\sfrac{3}{2}, \sfrac{13}{8}, \sfrac{13}{8}$ & $5096, 10752, 10752$ & $0, 0$ \\
\stepcounter{fermionicno}\thefermionicno & $21$ & $\mathbf{S}(\A[3,1] \A[7,2] \C[3,1]^{2})\big/(\A[3,1] \hookrightarrow \A[3,1])$ & $\A[7,2] \C[3,1]^{2}$ & $\sfrac{3}{2}, \sfrac{13}{8}, \sfrac{13}{8}$ & $5096, 10752, 10752$ & $0, 0$ \\
\stepcounter{fermionicno}\thefermionicno & $21$ & $\mathbf{S}(\A[3,1] \D[7,3] \G[2,1])\big/(\A[3,1] \hookrightarrow \A[3,1])$ & $\D[7,3] \G[2,1]$ & $\sfrac{3}{2}, \sfrac{13}{8}, \sfrac{13}{8}$ & $5096, 10752, 10752$ & $0, 0$ \\
\stepcounter{fermionicno}\thefermionicno & $21$ & $\mathbf{S}(\A[3,1] \C[7,2])\big/(\A[3,1] \hookrightarrow \A[3,1])$ & $\C[7,2]$ & $\sfrac{3}{2}, \sfrac{13}{8}, \sfrac{13}{8}$ & $5096, 10752, 10752$ & $0, 0$ \\
\end{tabular}
\end{table}

\begin{scriptsize}
\begin{table}[H]
    \centering
    \begin{tabular}{r|c|c|cHHH}
    No. & $c$ & $\V_{\bar 0}$ & CSA & $h_i$ & $d_i$ & $\alpha_{j,m}$ \\\toprule
$\star$ \stepcounter{fermionicno}\thefermionicno & $22$ & $\mathbf{S}(\A[1,1]^{24})\big/(\A[1,1]^{2} \hookrightarrow \A[1,1]^{2})$ & $\A[1,1]^{22}$ & $\sfrac{3}{2}, \sfrac{7}{4}, \sfrac{7}{4}$ & $4928, 22528, 22528$ & $0, 0$ \\
\stepcounter{fermionicno}\thefermionicno & $22$ & $\mathbf{S}(\A[1,1]^{4} \A[3,2]^{4})\big/(\A[1,1]^{2} \hookrightarrow \A[1,1]^{2})$ & $\A[1,1]^{2} \A[3,2]^{4}$ & $\sfrac{3}{2}, \sfrac{7}{4}, \sfrac{7}{4}$ & $4928, 22528, 22528$ & $0, 0$ \\
\stepcounter{fermionicno}\thefermionicno & $22$ & $\mathbf{S}(\A[1,1]^{3} \A[5,3] \D[4,3])\big/(\A[1,1]^{2} \hookrightarrow \A[1,1]^{2})$ & $\A[1,1] \A[5,3] \D[4,3]$ & $\sfrac{3}{2}, \sfrac{7}{4}, \sfrac{7}{4}$ & $4928, 22528, 22528$ & $0, 0$ \\
\stepcounter{fermionicno}\thefermionicno & $22$ & $\mathbf{S}(\A[1,1]'\A[1,1]^2 \A[7,4])\big/(\A[1,1]^{2} \hookrightarrow \A[1,1]^{2})$ & $\A[1,1] \A[7,4]$ & $\sfrac{3}{2}, \sfrac{7}{4}, \sfrac{7}{4}$ & $4928, 22528, 22528$ & $0, 0$ \\
\stepcounter{fermionicno}\thefermionicno & $22$ & $\mathbf{S}(\A[1,1]'\A[1,1]^2 \A[7,4])\big/(\A[1,1]^{2} \hookrightarrow \A[1,1]\A[1,1]')$ & $\A[1,1] \A[7,4]$ & $\sfrac{3}{2}, \sfrac{7}{4}, \sfrac{7}{4}$ & $4928, 22528, 22528$ & $0, 0$ \\
\stepcounter{fermionicno}\thefermionicno & $22$ & $\mathbf{S}(\A[1,1]^{2} \C[3,2] \D[5,4])\big/(\A[1,1]^{2} \hookrightarrow \A[1,1]^{2})$ & $\C[3,2] \D[5,4]$ & $\sfrac{3}{2}, \sfrac{7}{4}, \sfrac{7}{4}$ & $4928, 22528, 22528$ & $0, 0$ \\
\stepcounter{fermionicno}\thefermionicno & $22$ & $\mathbf{S}(\A[1,1]^{2} \D[6,5])\big/(\A[1,1]^{2} \hookrightarrow \A[1,1]^{2})$ & $\D[6,5]$ & $\sfrac{3}{2}, \sfrac{7}{4}, \sfrac{7}{4}$ & $4928, 22528, 22528$ & $0, 0$ \\
\end{tabular}
    \end{table}
\end{scriptsize}

    \clearpage

    \begin{table}[H]
    \centering
    \begin{tabular}{r|c|c|cHHH}
    No. & $c$ & $\V_{\bar 0}$ & CSA & $h_i$ & $d_i$ & $\alpha_{j,m}$ \\\toprule
\stepcounter{fermionicno}\thefermionicno & $\sfrac{45}{2}$ & $\mathbf{S}(\A[1,2]^{16})\big/(\A[1,2] \hookrightarrow \A[1,2])$ & $\A[1,2]^{15}$ & $\sfrac{29}{16}, \sfrac{3}{2}$ & $46080, 4785$ & $0, 0$ \\
\stepcounter{fermionicno}\thefermionicno & $\sfrac{45}{2}$ & $\mathbf{S}(\A[1,2] \A[3,4]^{3})\big/(\A[1,2] \hookrightarrow \A[1,2])$ & $\A[3,4]^{3}$ & $\sfrac{29}{16}, \sfrac{3}{2}$ & $46080, 4785$ & $0, 0$ \\
\stepcounter{fermionicno}\thefermionicno & $\sfrac{45}{2}$ & $\mathbf{S}(\A[1,2] \A[5,6] \B[2,3])\big/(\A[1,2] \hookrightarrow \A[1,2])$ & $\A[5,6] \B[2,3]$ & $\sfrac{29}{16}, \sfrac{3}{2}$ & $46080, 4785$ & $0, 0$ \\
\stepcounter{fermionicno}\thefermionicno & $\sfrac{45}{2}$ & $\mathbf{S}(\A[1,2] \D[5,8])\big/(\A[1,2] \hookrightarrow \A[1,2])$ & $\D[5,8]$ & $\sfrac{29}{16}, \sfrac{3}{2}$ & $46080, 4785$ & $0, 0$ \\
\stepcounter{fermionicno}\thefermionicno & $\sfrac{45}{2}$ & $\mathbf{S}(\B[2,2]^{6})\big/(\A[1,2] \hookrightarrow \B[2,2])$ & $\A[1,2] \B[2,2]^{5} \mathsf{W}^{\mathsf{A}_1}_{2,2}$ & $\sfrac{13}{16}, \sfrac{3}{2}$ & $2, 4761$ & $2, 0$ \\
\stepcounter{fermionicno}\thefermionicno & $\sfrac{45}{2}$ & $\mathbf{S}(\A[2,2]^{4} \D[4,4])\big/(\A[1,2] \hookrightarrow \A[2,2])$ & $\A[2,2]^{3} \D[4,4] \mathsf{U}_{1,12} L_{\sfrac{7}{10}}$ & $\sfrac{13}{16}, \sfrac{3}{2}$ & $2, 4761$ & $2, 0$ \\
\stepcounter{fermionicno}\thefermionicno & $\sfrac{45}{2}$ & $\mathbf{S}(\A[2,2] \F[4,6])\big/(\A[1,2] \hookrightarrow \A[2,2])$ & $\F[4,6] \mathsf{U}_{1,12} L_{\sfrac{7}{10}}$ & $\sfrac{13}{16}, \sfrac{3}{2}$ & $2, 4761$ & $2, 0$ \\
\bottomrule
    \end{tabular}
    \end{table}

\clearpage

\subsection{Two primaries}\label{app:data:2primaries}

\subsubsection{$\normalfont (\AA_1,1)$}\label{app:(A1,1)}{ \ }

\begin{table}[H]
\begin{tabular}{c|c}
$\mathcal{C}$ & $(\AA_1,1)$\\\midrule 
Other names & $\overline{(\EE_7,1)}$, Semion, $\mathbb{Z}_2$ MTC \\
Fusion category & $\Vec_{\ZZ_2}^\omega$ \\
$\mathcal{S}$ & $\frac{1}{\sqrt{2}}\left(\begin{array}{rr} 1 & 1 \\ 1 & -1 \end{array}\right)$ \\
$h_i~\mathrm{mod}~1$ & $0$, $\sfrac14$ \\
QDims &  $1$, $1$ \\
Fusion rules & $\Phi\times \Phi=1$
\end{tabular}
\end{table}

\begin{table}[H]
\begin{center}
\begin{tabular}{c|c|c}
$c~\mathrm{mod}~24$ & $\lambda$ & $\chi$ \\ \midrule
$1$ & $\left(
\begin{array}{rr}
 -\sfrac{1}{24} & 0 \\
 0 & -\sfrac{19}{24} \\
\end{array}
\right)$ & $\left(
\begin{array}{rr}
 3 & 26752 \\
 2 & -247 \\
\end{array}
\right)$ \\
$9$ & $\left(
\begin{array}{rr}
 -\sfrac{3}{8} & 0 \\
 0 & -\sfrac{9}{8} \\
\end{array}
\right)$ & $\left(
\begin{array}{rr}
 251 & 26752 \\
 2 & 1 \\
\end{array}
\right)$ \\
$17$ & $\left(
\begin{array}{rr}
 -\sfrac{17}{24} & 0 \\
 0 & -\sfrac{11}{24} \\
\end{array}
\right)$ & $\left(
\begin{array}{rr}
 323 & 88 \\
 1632 & -319 \\
\end{array}
\right)$ \end{tabular}
\end{center}
\end{table}

\begin{small}\begin{table}[H]
\begin{center}
\begin{tabular}{r|l}
$c$ & $q^{\sfrac{c}{24}}\mathrm{ch}(\tau)=q^{\sfrac{c}{24}}\sum_{j,m}\alpha_{j,m}X^{(j,m)}(\tau)$ \\ \toprule
$1$ & $ \begin{array}{l}
1+3 q+4 q^2+7 q^3+13 q^4+19 q^5+29 q^6+43 q^7+62 q^8+90 q^9+126 q^{10}+\cdots \\ 
q^{\sfrac{1}{4}}\left(2+2 q+6 q^2+8 q^3+14 q^4+20 q^5+34 q^6+46 q^7+70 q^8+96 q^9+\cdots \right)  \end{array}$ \\\midrule
$9$ & $ \begin{array}{l}
1+251 q+4872 q^2+48123 q^3+335627 q^4+1868001 q^5+8857377 q^6+\cdots \\ 
q^{\sfrac{1}{4}}\left(2+498 q+8750 q^2+79248 q^3+522498 q^4+2786256 q^5+12795998 q^6+\cdots \right)  \end{array}$ \\\midrule
$17$ & $ \begin{array}{l}
1+q \left(88 \alpha _{1,0}+323\right)+q^2 \left(5192 \alpha _{1,0}+60860\right)+q^3 \left(83840 \alpha _{1,0}+2158575\right)+\cdots \\ 
q^{\sfrac{1}{4}}\left(\alpha _{1,0}+q \left(1632-319 \alpha _{1,0}\right)+q^2 \left(162656-11077 \alpha _{1,0}\right)+\cdots \right)  \end{array}$\\  \end{tabular}
\end{center}
\end{table}\end{small}

\begin{table}[H]
    \centering
    \begin{tabular}{c|c c}
    \toprule
       $\mathfrak{g}_k$  &  $\E[7,1]$ & $\E[8,1]$\\
       $\Com_{\mathfrak{g}_k}(\E[7,1])$ CSA  & $0$ & $\A[1,1]$ \\\bottomrule
    \end{tabular}
\end{table}

\begin{table}[H]
    \centering
    \begin{tabular}{c|c|c|c|c|c|c}
    No. & $c$ & Theory & CSA & $h$ & $d$ & $\alpha_{1,0}$ \\\toprule
    $\star$ \stepcounter{theoryno}\thetheoryno & $1$ & $\A[1,1]$ & $\A[1,1]$ & $\sfrac{1}{4}$ & $2$ & \\\midrule
    $\star$ \stepcounter{theoryno}\thetheoryno & $9$ & $\A[1,1]\E[8,1]\cong\E[8,1]^2\big/(\E[7,1]\hookrightarrow\E[8,1])$ & $\A[1,1]\E[8,1]$ & $\sfrac{1}{4}$ & $2$ &  \\\midrule 
$\star$ \stepcounter{theoryno}\thetheoryno & $17$ & $\mathbf{S}(\D[10,1] \E[7,1]^{2})\big/(\E[7,1] \hookrightarrow \E[7,1])$ & $\D[10,1] \E[7,1]$ & $\sfrac{5}{4}$ & $1632$ & $0$ \\
$\star$ \stepcounter{theoryno}\thetheoryno & $17$ & $\mathbf{S}(\A[17,1] \E[7,1])\big/(\E[7,1] \hookrightarrow \E[7,1])$ & $\A[17,1]$ & $\sfrac{5}{4}$ & $1632$ & $0$ \\
$\star$ \stepcounter{theoryno}\thetheoryno & $17$ & $\A[1,1]\E[8,1]^2\cong\E[8,1]^{3}\big/(\E[7,1] \hookrightarrow \E[8,1])$ & $\A[1,1] \E[8,1]^{2}$ & $\sfrac{1}{4}$ & $2$ & $2$ \\
$\star$ \stepcounter{theoryno}\thetheoryno & $17$ & $\A[1,1]\D[16,1]^+\cong\D[16,1]^+ \E[8,1]\big/(\E[7,1] \hookrightarrow \E[8,1])$ & $\A[1,1] \D[16,1]$ & $\sfrac{1}{4}$ & $2$ & $2$ \\\bottomrule
\end{tabular}
\end{table}

\clearpage

\subsubsection{$\normalfont (\EE_7,1)$}\label{app:(E7,1)}{\ }

\begin{table}[H]
\begin{tabular}{c|c}
$\mathcal{C}$ & $(\EE_7,1)$\\\midrule 
Other names & $\overline{(\AA_1,1)}$ \\
Fusion category & $\Vec_{\ZZ_2}^\omega$ \\
$\mathcal{S}$ & $\frac{1}{\sqrt{2}}\left(\begin{array}{rr} 1 & 1 \\ 1 & -1 \end{array}\right)$ \\
$h_i~\mathrm{mod}~1$ & $0$, $\sfrac34$ \\
QDims &  $1$, $1$ \\
Fusion rules & $\Phi\times \Phi=1$
\end{tabular}
\end{table}

\begin{table}[H]
\begin{center}
\begin{tabular}{c|c|c}
$c~\mathrm{mod}~24$ & $\lambda$ & $\chi$ \\ \midrule
$7$ & $\left(
\begin{array}{rr}
 -\sfrac{7}{24} & 0 \\
 0 & -\sfrac{13}{24} \\
\end{array}
\right)$ & $\left(
\begin{array}{rr}
 133 & 1248 \\
 56 & -377 \\
\end{array}
\right)$ \\
$15$ & $\left(
\begin{array}{rr}
 -\sfrac{5}{8} & 0 \\
 0 & -\sfrac{7}{8} \\
\end{array}
\right)$ & $\left(
\begin{array}{rr}
 381 & 1248 \\
 56 & -129 \\
\end{array}
\right)$ \\
$23$ & $\left(
\begin{array}{rr}
 -\sfrac{23}{24} & 0 \\
 0 & -\sfrac{5}{24} \\
\end{array}
\right)$ & $\left(
\begin{array}{rr}
 69 & 10 \\
 32384 & -65 \\
\end{array}
\right)$ \end{tabular}
\end{center}
\end{table}

\begin{small}\begin{table}[H]
\begin{center}
\begin{tabular}{r|l}
$c$ & $q^{\sfrac{c}{24}}\mathrm{ch}(\tau)=q^{\sfrac{c}{24}}\sum_{j,m}\alpha_{j,m}X^{(j,m)}(\tau)$ \\ \toprule
$7$ & $ \begin{array}{l}
1+133 q+1673 q^2+11914 q^3+63252 q^4+278313 q^5+1070006 q^6+3711557 q^7+\cdots \\ 
q^{\sfrac{3}{4}}\left(56+968 q+7504 q^2+42616 q^3+194768 q^4+772576 q^5+2742040 q^6+\cdots \right)  \end{array}$ \\\midrule
$15$ & $ \begin{array}{l}
1+381 q+38781 q^2+1010062 q^3+14752518 q^4+152641503 q^5+\cdots \\ 
q^{\sfrac{3}{4}}\left(56+14856 q+478512 q^2+7841752 q^3+87285024 q^4+751128624 q^5+\cdots \right)  \end{array}$ \\\midrule
$23$ & $ \begin{array}{l}
1+q \left(10 \alpha _{1,0}+69\right)+q^2 \left(98 \alpha _{1,0}+131905\right)+q^3 \left(690 \alpha _{1,0}+12195106\right)+\cdots \\ 
q^{\sfrac{3}{4}}\left(\alpha _{1,0}+q \left(32384-65 \alpha _{1,0}\right)+q^2 \left(4418944-450 \alpha _{1,0}\right)+q^3 \left(189846784-2175 \alpha _{1,0}\right)+\cdots \right)  \end{array}$ \end{tabular}
\end{center}
\end{table}\end{small}

\newpage

\begin{table}[H]
    \centering
    \begin{tabular}{c|c c c c| c c c c c}
    \toprule
       $\mathfrak{g}_k$  &  $\A[r,1]$ $(r\geq 1)$ & $\B[r,1]$ $(r\geq 2)$ & $\C[r,1]$ $(r\geq 3)$ & $\D[r,1]$ $(r\geq 4)$ & \multirow{2}{*}{$m_r=\frac{2(r^2-1)}{\mathrm{gcd}(2,r-1)^2}$}\\
       $\Com_{\mathfrak{g}_k}(\A[1,1])$ CSA & $\A[r-2,1]\mathsf{U}_{1,m_r}$ & $\A[1,1]\B[r-2,1]$ & 
 $\C[r-1,1]L_{c_{r-1}}$ & $\A[1,1]\D[r-2,1]$ &\\\bottomrule
    \end{tabular}\\\vspace{.1in}
     \begin{tabular}{c|c c c c c c c c c}
    \toprule
       $\mathfrak{g}_k$  &   $\E[6,1]$ & $\E[7,1]$ & $\E[8,1]$ & $\F[4,1]$ & $\G[2,1]$\\
       $\Com_{\mathfrak{g}_k}(\A[1,1])$ CSA & $\A[5,1]$ & $\D[6,1]$ & $\E[7,1]$ & $\C[3,1]$ & $\A[1,3]$ \\\bottomrule
    \end{tabular} \\\vspace{.3in}
    \begin{tabular}{r|c|c|c|c|c|c}
    No. & $c$ & Theory & CSA & $h$ & $d$ & $\alpha_{1,0}$ \\\toprule
    $\star$ \stepcounter{theoryno}\thetheoryno & $7$ & $\E[7,1]$ & $\E[7,1]$ & $\sfrac{3}{4}$ & $56$ &  \\\midrule 
    $\star$ \stepcounter{theoryno}\thetheoryno & $15$ & $\E[7,1]\E[8,1]\cong\E[8,1]^2\big/(\A[1,1]\hookrightarrow\E[8,1])$ & $\E[7,1]\E[8,1]$ & $\sfrac{3}{4}$ & $56$ &  \\
        $\star$ \stepcounter{theoryno}\thetheoryno & $15$ & $\D[16,1]^+\big/(\A[1,1]\hookrightarrow \D[16,1])$ & $\A[1,1]\D[14,1]$ & $\sfrac{3}{4}$ & $56$ &  \\\midrule 
$\star$ \stepcounter{theoryno}\thetheoryno & $23$ & $\mathbf{S}(\A[1,1]^{24})\big/(\A[1,1] \hookrightarrow \A[1,1])$ & $\A[1,1]^{23}$ & $\sfrac{7}{4}$ & $32384$ & $0$ \\
\stepcounter{theoryno}\thetheoryno & $23$ & $\mathbf{S}(\A[1,1]^{4} \A[3,2]^{4})\big/(\A[1,1] \hookrightarrow \A[1,1])$ & $\A[1,1]^{3} \A[3,2]^{4}$ & $\sfrac{7}{4}$ & $32384$ & $0$ \\
\stepcounter{theoryno}\thetheoryno & $23$ & $\mathbf{S}(\A[1,1]^{3} \A[5,3] \D[4,3])\big/(\A[1,1] \hookrightarrow \A[1,1])$ & $\A[1,1]^{2} \A[5,3] \D[4,3]$ & $\sfrac{7}{4}$ & $32384$ & $0$ \\
\stepcounter{theoryno}\thetheoryno & $23$ & $\mathbf{S}(\A[1,1]'\A[1,1]^2\A[7,4])\big/(\A[1,1] \hookrightarrow \A[1,1])$ & $\A[1,1]^{2} \A[7,4]$ & $\sfrac{7}{4}$ & $32384$ & $0$ \\
\end{tabular}
\end{table}

\clearpage

\begin{table}[H]
    \centering
    \begin{tabular}{r|c|c|c|c|c|c}
    No. & $c$ & Theory & CSA & $h$ & $d$ & $\alpha_{1,0}$ \\\toprule
\stepcounter{theoryno}\thetheoryno & $23$ & $\mathbf{S}(\A[5,1] \C[5,1] \E[6,2])\big/(\A[1,1] \hookrightarrow \C[5,1])$ & $\A[5,1] \C[4,1] \E[6,2] L_{\sfrac{6}{7}}$ & $\sfrac{3}{4}$ & $8$ & $8$\\
\stepcounter{theoryno}\thetheoryno & $23$ & $\mathbf{S}(\A[5,1] \E[7,3])\big/(\A[1,1] \hookrightarrow \A[5,1])$ & $\A[3,1] \E[7,3] \mathsf{U}_{1,12}$ & $\sfrac{3}{4}$ & $8$ & $8$ \\
$\star$ \stepcounter{theoryno}\thetheoryno & $23$ & $\mathbf{S}(\A[6,1]^{4})\big/(\A[1,1] \hookrightarrow \A[6,1])$ & $\A[4,1] \A[6,1]^{3} \mathsf{U}_{1,70}$ & $\sfrac{3}{4}$ & $10$ & $10$ \\
\stepcounter{theoryno}\thetheoryno & $23$ & $\mathbf{S}(\B[4,1]^{2} \D[8,2])\big/(\A[1,1] \hookrightarrow \B[4,1])$ & $\A[1,1] \B[2,1] \B[4,1] \D[8,2]$ & $\sfrac{3}{4}$ & $10$ & $10$ \\
\stepcounter{theoryno}\thetheoryno & $23$ & $\mathbf{S}(\B[4,1] \C[6,1]^{2})\big/(\A[1,1] \hookrightarrow \C[6,1])$ & $\B[4,1] \C[5,1] \C[6,1] L_{\sfrac{25}{28}}$ & $\sfrac{3}{4}$ & $10$ & $10$ \\
\stepcounter{theoryno}\thetheoryno & $23$ & $\mathbf{S}(\B[4,1] \C[6,1]^{2})\big/(\A[1,1] \hookrightarrow \B[4,1])$ & $\A[1,1] \B[2,1] \C[6,1]^{2}$ & $\sfrac{3}{4}$ & $10$ & $10$ \\
$\star$ \stepcounter{theoryno}\thetheoryno & $23$ & $\mathbf{S}(\A[7,1]^{2} \D[5,1]^{2})\big/(\A[1,1] \hookrightarrow \D[5,1])$ & $\A[1,1] \A[3,1] \A[7,1]^{2} \D[5,1]$ & $\sfrac{3}{4}$ & $12$ & $12$ \\
\bottomrule
\end{tabular}
\end{table}

\clearpage

\subsubsection{$\normalfont (\GG_2,1)$}\label{app:(G2,1)}{\ }

\begin{table}[H]
\begin{tabular}{c|c}
$\mathcal{C}$ & $(\GG_2,1)$\\\midrule 
Other names & $\overline{(\FF_4,1)}$, Fibonacci MTC \\
Fusion category & $\textsl{Fib}$ \\
$\mathcal{S}$ & $\frac{1}{\sqrt{2+\varphi}}\left(\begin{array}{rr} 1 & \varphi \\ \varphi & -1 \end{array}\right)$, \ \ \ 
 \ $(\varphi=\frac{1+\sqrt{5}}{2})$ \\
$h_i~\mathrm{mod}~1$ & $0$, $\sfrac25$ \\
QDims &  $1$, $\varphi$ \\
Fusion rules & $\Phi\times \Phi=1+\Phi$
\end{tabular}
\end{table}

\begin{table}[H]
\begin{center}
\begin{tabular}{c|c|c}
$c~\mathrm{mod}~24$ & $\lambda$ & $\chi$ \\ \midrule
$\sfrac{14}{5}$ & $\left(
\begin{array}{rr}
 -\sfrac{7}{60} & 0 \\
 0 & -\sfrac{43}{60} \\
\end{array}
\right)$ & $\left(
\begin{array}{rr}
 14 & 12857 \\
 7 & -258 \\
\end{array}
\right)$ \\
$\sfrac{54}{5}$ & $\left(
\begin{array}{rr}
 -\sfrac{9}{20} & 0 \\
 0 & -\sfrac{21}{20} \\
\end{array}
\right)$ & $\left(
\begin{array}{rr}
 262 & 12857 \\
 7 & -10 \\
\end{array}
\right)$ \\
$\sfrac{94}{5}$ & $\left(
\begin{array}{rr}
 -\sfrac{47}{60} & 0 \\
 0 & -\sfrac{23}{60} \\
\end{array}
\right)$ & $\left(
\begin{array}{rr}
 188 & 46 \\
 4794 & -184 \\
\end{array}
\right)$ \end{tabular}
\end{center}
\end{table}

\begin{small}\begin{table}[H]
\begin{center}
\begin{tabular}{r|l}
$c$ & $q^{\sfrac{c}{24}}\mathrm{ch}(\tau)=q^{\sfrac{c}{24}}\sum_{j,m}\alpha_{j,m}X^{(j,m)}(\tau)$ \\ \toprule
$\sfrac{14}{5}$ & $ \begin{array}{l}
1+14 q+42 q^2+140 q^3+350 q^4+840 q^5+1827 q^6+3858 q^7+7637 q^8+14756 q^9+\cdots \\ 
q^{\sfrac{2}{5}}\left(7+34 q+119 q^2+322 q^3+819 q^4+1862 q^5+4025 q^6+8218 q^7+16177 q^8+\cdots \right)  \end{array}$ \\\midrule
$\sfrac{54}{5}$ & $ \begin{array}{l}
1+262 q+7638 q^2+103044 q^3+907932 q^4+6165852 q^5+34805919 q^6+\cdots \\ 
q^{\sfrac{2}{5}}\left(7+1770 q+37419 q^2+413314 q^3+3244881 q^4+20317202 q^5+108065839 q^6+\cdots \right)  \end{array}$ \\\midrule
$\sfrac{94}{5}$ & $ \begin{array}{l}
1+q \left(46 \alpha _{1,0}+188\right)+q^2 \left(2093 \alpha _{1,0}+62087\right)+q^3 \left(27002 \alpha _{1,0}+2923494\right)+\cdots \\ 
q^{\sfrac{2}{5}}\left(\alpha _{1,0}+q \left(4794-184 \alpha _{1,0}\right)+q^2 \left(532134-3841 \alpha _{1,0}\right)+q^3 \left(17518686-40388 \alpha _{1,0}\right)+\cdots \right)  \end{array}$ \end{tabular}
\end{center}
\end{table}\end{small}

\clearpage

\begin{table}[H]
    \centering
    \begin{tabular}{c|c c c c}
    \toprule
       $\mathfrak{g}_k$  & $\E[6,1]$ &  $\E[7,1]$ & $\E[8,1]$ & $\F[4,1]$\\
       $\Com_{\mathfrak{g}_k}(\F[4,1])$ CSA  & $L_{\sfrac45}$ & $\A[1,3]$ & $\G[2,1]$ & $0$ \\\bottomrule
    \end{tabular}
\end{table}

\begin{table}[H]
    \centering
    \begin{tabular}{r|c|c|c|c|c|c}
    No. & $c$ & Theory & CSA & $h$ & $d$ & $\alpha_{1,0}$ \\\toprule
    \stepcounter{theoryno}\thetheoryno & $\sfrac{14}{5}$ & $\G[2,1]$ & $\G[2,1]$ & $\sfrac{2}{5}$ & $7$ &  \\\midrule 
     \stepcounter{theoryno}\thetheoryno & $\sfrac{54}{5}$ & $\E[8,1]\G[2,1]\cong\E[8,1]^2\big/(\F[4,1]\hookrightarrow\E[8,1])$ & $\E[8,1]\G[2,1]$ & $\sfrac{2}{5}$ & $7$ &  \\\midrule 
\stepcounter{theoryno}\thetheoryno & $\sfrac{94}{5}$ & $\mathbf{S}(\C[8,1] \F[4,1]^{2})\big/(\F[4,1] \hookrightarrow \F[4,1])$ & $\C[8,1] \F[4,1]$ & $\sfrac{7}{5}$ & $4794$ & $0$ \\
\stepcounter{theoryno}\thetheoryno & $\sfrac{94}{5}$ & $\mathbf{S}(\B[5,1] \E[7,2] \F[4,1])\big/(\F[4,1] \hookrightarrow \F[4,1])$ & $\B[5,1] \E[7,2]$ & $\sfrac{7}{5}$ & $4794$ & $0$ \\
\stepcounter{theoryno}\thetheoryno & $\sfrac{94}{5}$ & $\mathbf{S}(\E[6,1]^{4})\big/(\F[4,1] \hookrightarrow \E[6,1])$ & $\E[6,1]^{3} L_{\sfrac{4}{5}}$ & $\sfrac{2}{5}$ & $1$ & $1$ \\
\stepcounter{theoryno}\thetheoryno & $\sfrac{94}{5}$ & $\mathbf{S}(\A[11,1] \D[7,1] \E[6,1])\big/(\F[4,1] \hookrightarrow \E[6,1])$ & $\A[11,1] \D[7,1] L_{\sfrac{4}{5}}$ & $\sfrac{2}{5}$ & $1$ & $1$ \\
\stepcounter{theoryno}\thetheoryno & $\sfrac{94}{5}$ & $\mathbf{S}(\D[10,1] \E[7,1]^{2})\big/(\F[4,1] \hookrightarrow \E[7,1])$ & $\A[1,3] \D[10,1] \E[7,1]$ & $\sfrac{2}{5}$ & $3$ & $3$ \\
\stepcounter{theoryno}\thetheoryno & $\sfrac{94}{5}$ & $\mathbf{S}(\A[17,1] \E[7,1])\big/(\F[4,1] \hookrightarrow \E[7,1])$ & $\A[1,3] \A[17,1]$ & $\sfrac{2}{5}$ & $3$ & $3$ \\
\stepcounter{theoryno}\thetheoryno & $\sfrac{94}{5}$ & $\E[8,1]^2\G[2,1]\cong\E[8,1]^{3}\big/(\F[4,1] \hookrightarrow \E[8,1])$ & $\E[8,1]^{2} \G[2,1]$ & $\sfrac{2}{5}$ & $7$ & $7$ \\
\stepcounter{theoryno}\thetheoryno & $\sfrac{94}{5}$ & $\D[16,1]^+\G[2,1]\cong \D[16,1]^+ \E[8,1]\big/(\F[4,1] \hookrightarrow \E[8,1])$ & $\D[16,1] \G[2,1]$ & $\sfrac{2}{5}$ & $7$ & $7$ \\\bottomrule
\end{tabular}
\end{table}

\clearpage

\subsubsection{$\normalfont (\FF_4,1)$}\label{app:(F4,1)}{\ }

\begin{table}[H]
\begin{tabular}{c|c}
$\mathcal{C}$ & $(\FF_4,1)$\\\midrule 
Other names & $\overline{(\GG_2,1)}$, Conjugate Fibonacci \\
Fusion category & $\textsl{Fib}$ \\
$\mathcal{S}$ & $\frac{1}{\sqrt{2+\varphi}}\left(\begin{array}{rr} 1 & \varphi \\ \varphi & -1 \end{array}\right)$, \ \ \ 
 \ $(\varphi=\frac{1+\sqrt{5}}{2})$ \\
$h_i~\mathrm{mod}~1$ & $0$, $\sfrac35$ \\
QDims &  $1$, $\varphi$ \\
Fusion rules & $\Phi\times \Phi=1+\Phi$
\end{tabular}
\end{table}

\begin{table}[H]
\begin{center}
\begin{tabular}{c|c|c}
$c~\mathrm{mod}~24$ & $\lambda$ & $\chi$ \\ \midrule
$\sfrac{26}{5}$ & $\left(
\begin{array}{rr}
 -\sfrac{13}{60} & 0 \\
 0 & -\sfrac{37}{60} \\
\end{array}
\right)$ & $\left(
\begin{array}{rr}
 52 & 3774 \\
 26 & -296 \\
\end{array}
\right)$ \\
$\sfrac{66}{5}$ & $\left(
\begin{array}{rr}
 -\sfrac{11}{20} & 0 \\
 0 & -\sfrac{19}{20} \\
\end{array}
\right)$ & $\left(
\begin{array}{rr}
 300 & 3774 \\
 26 & -48 \\
\end{array}
\right)$ \\
$\sfrac{106}{5}$ & $\left(
\begin{array}{rr}
 -\sfrac{53}{60} & 0 \\
 0 & -\sfrac{17}{60} \\
\end{array}
\right)$ & $\left(
\begin{array}{rr}
 106 & 17 \\
 15847 & -102 \\
\end{array}
\right)$ \end{tabular}
\end{center}
\end{table}

\begin{small}\begin{table}[H]
\begin{center}
\begin{tabular}{r|l}
$c$ & $q^{\sfrac{c}{24}}\mathrm{ch}(\tau)=q^{\sfrac{c}{24}}\sum_{j,m}\alpha_{j,m}X^{(j,m)}(\tau)$ \\ \toprule
$\sfrac{26}{5}$ & $ \begin{array}{l}
1+52 q+377 q^2+1976 q^3+7852 q^4+27404 q^5+84981 q^6+243230 q^7+\cdots \\ 
q^{\sfrac{3}{5}}\left(26+299 q+1702 q^2+7475 q^3+27300 q^4+88452 q^5+260650 q^6+\cdots \right)  \end{array}$ \\\midrule
$\sfrac{66}{5}$ & $ \begin{array}{l}
1+300 q+17397 q^2+344672 q^3+4072878 q^4+35365284 q^5+247802227 q^6+\cdots \\ 
q^{\sfrac{3}{5}}\left(26+6747 q+183078 q^2+2566199 q^3+24832272 q^4+188053434 q^5+\cdots \right)  \end{array}$ \\\midrule
$\sfrac{106}{5}$ & $ \begin{array}{l}
1+q \left(17 \alpha _{1,0}+106\right)+q^2 \left(442 \alpha _{1,0}+84429\right)+q^3 \left(4063 \alpha _{1,0}+5825442\right)+\cdots \\ 
q^{\sfrac{3}{5}}\left(\alpha _{1,0}+q \left(15847-102 \alpha _{1,0}\right)+q^2 \left(1991846-1088 \alpha _{1,0}\right)+\cdots \right)  \end{array}$ \end{tabular}
\end{center}
\end{table}\end{small}

\clearpage

\begin{table}[H]
    \centering
    \begin{tabular}{c|c c c c c c c}
    \toprule
       $\mathfrak{g}_k$  & $\B[r,1]$ $(r\geq 3)$ & $\D[r,1]$ $(r\geq 4)$ & $\E[6,1]$ & $\E[7,1]$ & $\E[8,1]$ & $\F[4,1]$ & $\G[2,1]$ \\
       $\Com_{\mathfrak{g}_k}(\G[2,1])$ CSA  & $\D[r-3,1]L_{\sfrac7{10}}$ & $\B[r-4,1]L_{\sfrac{7}{10}}$ & $\A[2,2]$ & $\C[3,1]$ & $\F[4,1]$ & $\A[1,8]$ & $0$ \\\bottomrule
    \end{tabular}
\end{table}

\begin{table}[H]
    \centering
    \begin{tabular}{r|c|c|c|c|c|c}
    No. & $c$ & Theory & CSA & $h$ & $d$ & $\alpha_{1,0}$ \\\toprule
    \stepcounter{theoryno}\thetheoryno & $\sfrac{26}{5}$ & $\F[4,1]$ & $\F[4,1]$ & $\sfrac35$ & $26$ &  \\\midrule 
    \stepcounter{theoryno}\thetheoryno & $\sfrac{66}{5}$ & $\E[8,1]\F[4,1]\cong\E[8,1]^2\big/(\G[2,1]\hookrightarrow\E[8,1])$ & $\E[8,1]\F[4,1]$ & $\sfrac35$ & $26$ &  \\
    \stepcounter{theoryno}\thetheoryno & $\sfrac{66}{5}$ & $\D[16,1]^+\big/(\G[2,1]\hookrightarrow\D[16,1])$ & $\B[12,1]L_{\sfrac{7}{10}}$ & $\sfrac35$ & $26$ &  \\\midrule 
\stepcounter{theoryno}\thetheoryno & $\sfrac{106}{5}$ & $\mathbf{S}(\E[6,3] \G[2,1]^{3})\big/(\G[2,1] \hookrightarrow \G[2,1])$ & $\E[6,3] \G[2,1]^{2}$ & $\sfrac{8}{5}$ & $15847$ & $0$ \\
\stepcounter{theoryno}\thetheoryno & $\sfrac{106}{5}$ & $\mathbf{S}(\A[3,1] \D[7,3] \G[2,1])\big/(\G[2,1] \hookrightarrow \G[2,1])$ & $\A[3,1] \D[7,3]$ & $\sfrac{8}{5}$ & $15847$ & $0$ \\
\stepcounter{theoryno}\thetheoryno & $\sfrac{106}{5}$ & $\mathbf{S}(\B[3,1]^{2} \C[4,1] \D[6,2])\big/(\G[2,1] \hookrightarrow \B[3,1])$ & $\B[3,1] \C[4,1] \D[6,2] L_{\sfrac{7}{10}}$ & $\sfrac{3}{5}$ & $1$ & $1$ \\
\stepcounter{theoryno}\thetheoryno & $\sfrac{106}{5}$ & $\mathbf{S}(\A[4,1] \A[9,2] \B[3,1])\big/(\G[2,1] \hookrightarrow \B[3,1])$ & $\A[4,1] \A[9,2] L_{\sfrac{7}{10}}$ & $\sfrac{3}{5}$ & $1$ & $1$ \\
\bottomrule
\end{tabular}
\end{table}

\clearpage

\subsection{Three primaries}\label{app:data:3primaries}

\subsubsection{$\normalfont (\AA_2,1)$}\label{app:(A2,1)}{\ }

\begin{table}[H]
\begin{tabular}{c|c}
$\mathcal{C}$ & $(\AA_2,1)$\\\midrule 
Other names & $\overline{(\EE_6,1)}$, $\mathbb{Z}_3$ MTC \\
Fusion category & $\Vec_{\mathbb{Z}_3}$ \\
$\mathcal{S}$ & $\frac{1}{\sqrt{3}}\left(\begin{array}{ccc} 1 &1 & 1  \\ 1 & \omega & \omega^2 \\ 1 & \omega^2 & \omega \end{array}\right)$, \ \ \ 
 \ $(\omega = e^{2\pi i/3})$ \\
 $\mathcal{P}^+,\mathcal{S}^+$ & $\left(\begin{array}{c c c} 1 & 0 & 0 \\ 0 & \sfrac12 & \sfrac12 \end{array}\right)$, $\frac{1}{\sqrt{3}}\left(
\begin{array}{cc}
 1 & 2\\
 1 & -1 \\
\end{array}
\right)$ \\
$h_i~\mathrm{mod}~1$ & $0$, $\sfrac13$, $\sfrac13$ \\
QDims &  $1$, $1$, $1$ \\
\multirow{2}{*}{Fusion rules} & $\Phi_1\times \Phi_1 = \Phi_2$, \  $\Phi_2\times \Phi_2=\Phi_1$ \\ 
&  $\Phi_1\times \Phi_2=1$
\end{tabular}
\end{table}

\begin{table}[H]
\begin{center}
\begin{tabular}{c|c|c}
$c~\mathrm{mod}~24$ & $\lambda$ & $\chi$ \\ \midrule
$2$ & $\left(
\begin{array}{rr}
 -\sfrac{1}{12} & 0 \\
 0 & -\sfrac{3}{4} \\
\end{array}
\right)$ & $\left(
\begin{array}{rr}
 8 & 26244 \\
 3 & -252 \\
\end{array}
\right)$ \\
$10$ & $\left(
\begin{array}{rr}
 -\sfrac{5}{12} & 0 \\
 0 & -\sfrac{13}{12} \\
\end{array}
\right)$ & $\left(
\begin{array}{rr}
 256 & 26244 \\
 3 & -4 \\
\end{array}
\right)$ \\
$18$ & $\left(
\begin{array}{rr}
 -\sfrac{3}{4} & 0 \\
 0 & -\sfrac{5}{12} \\
\end{array}
\right)$ & $\left(
\begin{array}{rr}
 234 & 90 \\
 2187 & -230 \\
\end{array}
\right)$ \end{tabular}
\end{center}
\end{table}

\begin{small}\begin{table}[H]
\begin{center}
\begin{tabular}{r|l}
$c$ & $q^{\sfrac{c}{24}}\mathrm{ch}(\tau)=q^{\sfrac{c}{24}}\sum_{j,m}\alpha_{j,m}X^{(j,m)}(\tau)$ \\ \toprule
$2$ & $ \begin{array}{l}
1+8 q+17 q^2+46 q^3+98 q^4+198 q^5+371 q^6+692 q^7+1205 q^8+2082 q^9+\cdots \\ 
q^{\sfrac{1}{3}}\left(3+9 q+27 q^2+57 q^3+126 q^4+243 q^5+465 q^6+828 q^7+1458 q^8+\cdots \right)\\
q^{\sfrac{1}{3}}\left(3+9 q+27 q^2+57 q^3+126 q^4+243 q^5+465 q^6+828 q^7+1458 q^8+\cdots \right)\end{array}$ \\\midrule
$10$ & $ \begin{array}{l}
1+256 q+6125 q^2+72006 q^3+572756 q^4+3567502 q^5+18666137 q^6+\cdots \\ 
q^{\sfrac{1}{3}}\left(3+753 q+14631 q^2+148125 q^3+1077756 q^4+6295509 q^5+\cdots \right)\\
q^{\sfrac{1}{3}}\left(3+753 q+14631 q^2+148125 q^3+1077756 q^4+6295509 q^5+\cdots \right) 
\end{array}$ \\\midrule
$18$ & $ \begin{array}{l}
1+q \left(90 \alpha _{1,0}+234\right)+q^2 \left(4644 \alpha _{1,0}+59805\right)+q^3 \left(66420 \alpha _{1,0}+2482242\right)+\cdots \\ 
q^{\sfrac{1}{3}}\left(\alpha _{1,0}+q \left(2187-230 \alpha _{1,0}\right)+q^2 \left(231822-6025 \alpha _{1,0}\right)+\cdots \right) \\
q^{\sfrac{1}{3}}\left(\alpha _{1,0}+q \left(2187-230 \alpha _{1,0}\right)+q^2 \left(231822-6025 \alpha _{1,0}\right)+\cdots \right) \end{array}$ \end{tabular}
\end{center}
\end{table}\end{small}

\clearpage

\begin{table}[H]
    \centering
    \begin{tabular}{c|c c c}
    \toprule
       $\mathfrak{g}_k$  & $\E[6,1]$ & $\E[7,1]$ & $\E[8,1]$\\
       $\Com_{\mathfrak{g}_k}(\E[6,1])$ CSA  & $0$ & $\mathsf{U}_{1,6}$ & $\A[2,1]$ \\\bottomrule
    \end{tabular}
\end{table}

\begin{table}[H]
    \centering
    \begin{tabular}{r|c|c|c|c|c|c}
    No. & $c$ & Theory & CSA & $h_i$ & $d_i$ & $\alpha_{1,0}$ \\\toprule
    $\star$ \stepcounter{theoryno}\thetheoryno & $2$ & $\A[2,1]$ & $\A[2,1]$ & $\sfrac{1}{3}, \sfrac{1}{3}$ & $3,3$ &  \\\midrule 
    $\star$ \stepcounter{theoryno}\thetheoryno & $10$ & $\A[2,1]\E[8,1]\cong\E[8,1]^2\big/(\E[6,1]\hookrightarrow\E[8,1])$ & $\A[2,1]\E[8,1]$ & $\sfrac{1}{3}, \sfrac{1}{3}$ & $3,3$ &  \\\midrule 
$\star$ \stepcounter{theoryno}\thetheoryno & $18$ & $\mathbf{S}(\E[6,1]^{4})\big/(\E[6,1] \hookrightarrow \E[6,1])$ & $\E[6,1]^{3}$ & $\sfrac{4}{3}, \sfrac{4}{3}$ & $2187, 2187$ & $0$ \\
$\star$ \stepcounter{theoryno}\thetheoryno & $18$ & $\mathbf{S}(\A[11,1] \D[7,1] \E[6,1])\big/(\E[6,1] \hookrightarrow \E[6,1])$ & $\A[11,1] \D[7,1]$ & $\sfrac{4}{3}, \sfrac{4}{3}$ & $2187, 2187$ & $0$ \\
$\star$ \stepcounter{theoryno}\thetheoryno & $18$ & $\mathbf{S}(\D[10,1] \E[7,1]^{2})\big/(\E[6,1] \hookrightarrow \E[7,1])$ & $\D[10,1] \E[7,1] \mathsf{U}_{1,6}$ & $\sfrac{1}{3}, \sfrac{1}{3}$ & $1, 1$ & $1$ \\
$\star$ \stepcounter{theoryno}\thetheoryno & $18$ & $\mathbf{S}(\A[17,1] \E[7,1])\big/(\E[6,1] \hookrightarrow \E[7,1])$ & $\A[17,1] \mathsf{U}_{1,6}$ & $\sfrac{1}{3}, \sfrac{1}{3}$ & $1, 1$ & $1$ \\
$\star$ \stepcounter{theoryno}\thetheoryno & $18$ & $\A[2,1]\E[8,1]^2\cong \E[8,1]^{3}\big/(\E[6,1] \hookrightarrow \E[8,1])$ & $\A[2,1] \E[8,1]^{2}$ & $\sfrac{1}{3}, \sfrac{1}{3}$ & $3, 3$ & $3$ \\
$\star$ \stepcounter{theoryno}\thetheoryno & $18$ & $\A[2,1]\D[16,1]^+\cong\D[16,1]^+ \E[8,1]\big/(\E[6,1] \hookrightarrow \E[8,1])$ & $\A[2,1] \D[16,1]$ & $\sfrac{1}{3}, \sfrac{1}{3}$ & $3, 3$ & $3$ \\\bottomrule
\end{tabular}
\end{table}

\clearpage

\subsubsection{$\normalfont (\EE_6,1)$}\label{app:(E6,1)}{\ }

\begin{table}[H]
\begin{tabular}{c|c}
$\mathcal{C}$ & $(\EE_6,1)$\\\midrule 
Other names & $\overline{(\AA_2,1)}$ \\
Fusion category & $\Vec_{\ZZ_3}$ \\
$\mathcal{S}$ & $\frac{1}{\sqrt{3}}\left(\begin{array}{ccc} 1 &1 & 1  \\ 1 & \omega^2 & \omega \\ 1 & \omega & \omega^2 \end{array}\right)$, \ \ \ 
 \ $(\omega = e^{2\pi i/3})$ \\
 $\mathcal{P}^+,\mathcal{S}^+$ & $\left(\begin{array}{c c c} 1 & 0 & 0 \\ 0 & \sfrac12 & \sfrac12 \end{array}\right)$, $\frac{1}{\sqrt{3}}\left(
\begin{array}{cc}
 1 & 2\\
 1 & -1 \\
\end{array}
\right)$ \\
$h_i~\mathrm{mod}~1$ & $0$, $\sfrac23$, $\sfrac23$ \\
QDims &  $1$, $1$, $1$ \\
\multirow{2}{*}{Fusion rules} & $\Phi_1\times \Phi_1 = \Phi_2$, \  $\Phi_2\times \Phi_2=\Phi_1$ \\ 
&  $\Phi_1\times \Phi_2=1$
\end{tabular}
\end{table}

\begin{table}[H]
\begin{center}
\begin{tabular}{c|c|c}
$c~\mathrm{mod}~24$ & $\lambda$ & $\chi$ \\ \midrule
$6$ & $\left(
\begin{array}{rr}
 -\sfrac{1}{4} & 0 \\
 0 & -\sfrac{7}{12} \\
\end{array}
\right)$ & $\left(
\begin{array}{rr}
 78 & 3402 \\
 27 & -322 \\
\end{array}
\right)$ \\
$14$ & $\left(
\begin{array}{rr}
 -\sfrac{7}{12} & 0 \\
 0 & -\sfrac{11}{12} \\
\end{array}
\right)$ & $\left(
\begin{array}{rr}
 326 & 3402 \\
 27 & -74 \\
\end{array}
\right)$ \\
$22$ & $\left(
\begin{array}{rr}
 -\sfrac{11}{12} & 0 \\
 0 & -\sfrac{1}{4} \\
\end{array}
\right)$ & $\left(
\begin{array}{rr}
 88 & 18 \\
 16038 & -84 \\
\end{array}
\right)$ \end{tabular}
\end{center}
\end{table}

\begin{small}\begin{table}[H]
\begin{center}
\begin{tabular}{r|l}
$c$ & $q^{\sfrac{c}{24}}\mathrm{ch}(\tau)=q^{\sfrac{c}{24}}\sum_{j,m}\alpha_{j,m}X^{(j,m)}(\tau)$ \\ \toprule
$6$ & $ \begin{array}{l}
1+78 q+729 q^2+4382 q^3+19917 q^4+77274 q^5+264644 q^6+827388 q^7+\cdots \\ 
q^{\sfrac{2}{3}}\left(27+378 q+2484 q^2+12312 q^3+49896 q^4+177660 q^5+570807 q^6+\cdots \right) \\
q^{\sfrac{2}{3}}\left(27+378 q+2484 q^2+12312 q^3+49896 q^4+177660 q^5+570807 q^6+\cdots \right) \end{array}$ \\\midrule
$14$ & $ \begin{array}{l}
1+326 q+24197 q^2+541598 q^3+7036831 q^4+66103598 q^5+496234478 q^6+\cdots \\ 
q^{\sfrac{2}{3}}\left(27+7074 q+207576 q^2+3125520 q^3+32237946 q^4+258764760 q^5+\cdots \right)  \\
q^{\sfrac{2}{3}}\left(27+7074 q+207576 q^2+3125520 q^3+32237946 q^4+258764760 q^5+\cdots \right)  \end{array}$ \\\midrule
$22$ & $ \begin{array}{l}
1+q \left(18 \alpha _{1,0}+88\right)+q^2 \left(342 \alpha _{1,0}+99935\right)+q^3 \left(2790 \alpha _{1,0}+7846300\right)+\cdots \\ 
q^{\sfrac{2}{3}}\left(\alpha _{1,0}+q \left(16038-84 \alpha _{1,0}\right)+q^2 \left(2092959-729 \alpha _{1,0}\right)+\cdots \right) \\
q^{\sfrac{2}{3}}\left(\alpha _{1,0}+q \left(16038-84 \alpha _{1,0}\right)+q^2 \left(2092959-729 \alpha _{1,0}\right)+\cdots \right) \end{array}$ \end{tabular}
\end{center}
\end{table}\end{small}

\clearpage

\begin{table}[H]
    \centering
    \begin{tabular}{c|c c c c |c}
    \toprule
       $\mathfrak{g}_k$  & $\A[r,1]$ $(r\geq 2)$ & $\B[3,1]$ & $\B[r,1]$ $(r\geq 4)$ & $\D[r,1]$ $(r\geq 4)$ & \multirow{2}{*}{$m_r=\frac{3(r+1)(r-2)}{\mathrm{gcd}(3,r-2)^2}$} \\
       $\Com_{\mathfrak{g}_k}(\A[2,1])$ CSA  &  $\A[r-3,1]\mathsf{U}_{1,m_r}$ & $L_{\sfrac12}\mathsf{U}_{1,12}$ & $\B[r-3,1]\mathsf{U}_{1,12}$ & $\D[r-3,1]\mathsf{U}_{1,12}$ \\\bottomrule
    \end{tabular}\\\vspace{.1in}
    \begin{tabular}{c|c c c c c c c c}
    \toprule
       $\mathfrak{g}_k$  & $\E[6,1]$ & $\E[7,1]$ & $\E[8,1]$ & $\F[4,1]$ & $\G[2,1]$ \\
       $\Com_{\mathfrak{g}_k}(\A[2,1])$ CSA   & $\A[2,1]^2$ & $\A[5,1]$ & $\E[6,1]$ & $\A[2,2]$ & $L_{\sfrac45}$\\\bottomrule
    \end{tabular}
\end{table}

\begin{table}[H]
    \centering
    \begin{tabular}{r|c|c|c|c|c|c}
    No. & $c$ & Theory & CSA & $h_i$ & $d_i$ & $\alpha_{1,0}$ \\\toprule
$\star$ \stepcounter{theoryno}\thetheoryno & $6$ & $\E[6,1]$ & $\E[6,1]$ & $\sfrac{2}{3}, \sfrac{2}{3}$ & $27,27$ &  \\\midrule 
$\star$ \stepcounter{theoryno}\thetheoryno & $14$ & $\E[6,1]\E[8,1]\cong \E[8,1]^2\big/(\A[2,1]\hookrightarrow\E[8,1])$ & $\E[6,1]\E[8,1]$ & $\sfrac{2}{3}, \sfrac{2}{3}$ & $27,27$ &  \\
$\star$ \stepcounter{theoryno}\thetheoryno & $14$ & $\D[16,1]^+\big/(\A[2,1]\hookrightarrow\D[16,1])$ & $\D[13,1]\mathsf{U}_{1,12}$ & $\sfrac{2}{3}, \sfrac{2}{3}$ & $27,27$ &  \\\midrule 
$\star$ \stepcounter{theoryno}\thetheoryno & $22$ & $\mathbf{S}(\A[2,1]^{12})\big/(\A[2,1] \hookrightarrow \A[2,1])$ & $\A[2,1]^{11}$ & $\sfrac{5}{3}, \sfrac{5}{3}$ & $16038, 16038$ & $0$ \\
\stepcounter{theoryno}\thetheoryno & $22$ & $\mathbf{S}(\A[2,1]^{2} \A[5,2]^{2} \B[2,1])\big/(\A[2,1] \hookrightarrow \A[2,1])$ & $\A[2,1] \A[5,2]^{2} \B[2,1]$ & $\sfrac{5}{3}, \sfrac{5}{3}$ & $16038, 16038$ & $0$ \\
\stepcounter{theoryno}\thetheoryno & $22$ & $\mathbf{S}(\A[2,1]^{2} \A[8,3])\big/(\A[2,1] \hookrightarrow \A[2,1])$ & $\A[2,1] \A[8,3]$ & $\sfrac{5}{3}, \sfrac{5}{3}$ & $16038, 16038$ & $0$ \\
\stepcounter{theoryno}\thetheoryno & $22$ & $\mathbf{S}(\A[2,1] \B[2,1] \E[6,4])\big/(\A[2,1] \hookrightarrow \A[2,1])$ & $\B[2,1] \E[6,4]$ & $\sfrac{5}{3}, \sfrac{5}{3}$ & $16038, 16038$ & $0$ \\
\end{tabular}
\end{table}

\clearpage

\begin{table}[H]
    \centering
    \begin{tabular}{r|c|c|c|c|c|c}
    No. & $c$ & Theory & CSA & $h_i$ & $d_i$ & $\alpha_{1,0}$ \\\toprule
\stepcounter{theoryno}\thetheoryno & $22$ & $\mathbf{S}(\B[6,1] \C[10,1])\big/(\A[2,1] \hookrightarrow \B[6,1])$ & $\B[3,1] \C[10,1] \mathsf{U}_{1,12}$ & $\sfrac{2}{3}, \sfrac{2}{3}$ & $8, 8$ & $8$ \\
$\star$ \stepcounter{theoryno}\thetheoryno & $22$ & $\mathbf{S}(\E[6,1]^{4})\big/(\A[2,1] \hookrightarrow \E[6,1])$ & $\A[2,1]^{2} \E[6,1]^{3}$ & $\sfrac{2}{3}, \sfrac{2}{3}$ & $9, 9$ & $9$ \\
$\star$ \stepcounter{theoryno}\thetheoryno & $22$ & $\mathbf{S}(\A[11,1] \D[7,1] \E[6,1])\big/(\A[2,1] \hookrightarrow \E[6,1])$ & $\A[2,1]^{2} \A[11,1] \D[7,1]$ & $\sfrac{2}{3}, \sfrac{2}{3}$ & $9, 9$ & $9$ \\
$\star$ \stepcounter{theoryno}\thetheoryno & $22$ & $\mathbf{S}(\A[11,1] \D[7,1] \E[6,1])\big/(\A[2,1] \hookrightarrow \A[11,1])$ & $\A[8,1] \D[7,1] \E[6,1] \mathsf{U}_{1,36}$ & $\sfrac{2}{3}, \sfrac{2}{3}$ & $9, 9$ & $9$ \\
$\star$ \stepcounter{theoryno}\thetheoryno & $22$ & $\mathbf{S}(\A[11,1] \D[7,1] \E[6,1])\big/(\A[2,1] \hookrightarrow \D[7,1])$ & $\A[11,1] \D[4,1] \E[6,1] \mathsf{U}_{1,12}$ & $\sfrac{2}{3}, \sfrac{2}{3}$ & $9, 9$ & $9$ \\
$\star$ \stepcounter{theoryno}\thetheoryno & $22$ & $\mathbf{S}(\A[12,1]^{2})\big/(\A[2,1] \hookrightarrow \A[12,1])$ & $\A[9,1] \A[12,1] \mathsf{U}_{1,390}$ & $\sfrac{2}{3}, \sfrac{2}{3}$ & $10, 10$ & $10$ \\
$\star$ \stepcounter{theoryno}\thetheoryno & $22$ & $\mathbf{S}(\D[8,1]^{3})\big/(\A[2,1] \hookrightarrow \D[8,1])$ & $\D[5,1] \D[8,1]^{2} \mathsf{U}_{1,12}$ & $\sfrac{2}{3}, \sfrac{2}{3}$ & $11, 11$ & $11$ \\
\bottomrule
\end{tabular}
\end{table}

\clearpage

\subsubsection{$\normalfont (\BB_8,1)$}\label{app:(B8,1)}{\ }

\begin{table}[H]
\begin{tabular}{c|c}
$\mathcal{C}$ & $(\BB_8,1)$\\\midrule 
Other names & $(\BB_{8+8n},1)$, $\overline{(\BB_{7+8n},1)}$, $\overline{(\EE_8,2)}$, Ising \\
Fusion category & $\mathrm{TY}_+(\ZZ_2)$ \\
$\mathcal{S}$ & $\frac12 \left(
\begin{array}{ccc}
 1 & \sqrt{2} & 1 \\
 \sqrt{2} & 0 & -\sqrt{2} \\
 1 & -\sqrt{2} & 1 \\
\end{array}
\right)$ \\
$h_i~\mathrm{mod}~1$ & $0$, $\sfrac{1}{16}$, $\sfrac12$ \\
QDims &  $1$, $\sqrt{2}$, $1$ \\
\multirow{2}{*}{Fusion rules} & $\Phi_1^2=1+\Phi_2$, \ $\Phi_2^2=1$,   \\ 
&  $\Phi_1\times \Phi_2=\Phi_2\times \Phi_1=\Phi_1$
\end{tabular}
\end{table}

\begin{table}[H]
\begin{center}
\begin{tabular}{c|c|c}
$c~\mathrm{mod}~24$ & $\lambda$ & $\chi$ \\ \midrule
$\sfrac{1}{2}$ & $\left(
\begin{array}{ccc}
 -\sfrac{1}{48} & 0 & 0 \\
 0 & -\sfrac{23}{24} & 0 \\
 0 & 0 & -\sfrac{25}{48} \\
\end{array}
\right)$ & $\left(
\begin{array}{ccc}
 0 & 94208 & 2325 \\
 1 & -23 & -25 \\
 1 & -4096 & 275 \\
\end{array}
\right)$ \\
$\sfrac{17}{2}$ & $\left(
\begin{array}{ccc}
 -\sfrac{17}{48} & 0 & 0 \\
 0 & -\sfrac{7}{24} & 0 \\
 0 & 0 & -\sfrac{41}{48} \\
\end{array}
\right)$ & $\left(
\begin{array}{ccc}
 136 & 112 & 5125 \\
 256 & -7 & -10496 \\
 17 & -16 & 123 \\
\end{array}
\right)$ \\
$\sfrac{33}{2}$ & $\left(
\begin{array}{ccc}
 -\sfrac{11}{16} & 0 & 0 \\
 0 & -\sfrac{5}{8} & 0 \\
 0 & 0 & -\sfrac{3}{16} \\
\end{array}
\right)$ & $\left(
\begin{array}{ccc}
 231 & 256 & 9 \\
 528 & -15 & -16 \\
 4301 & -3840 & 36 \\
\end{array}
\right)$ \end{tabular}
\end{center}
\end{table}

\begin{small}\begin{table}[H]
\begin{center}
\begin{tabular}{r|l}
$c$ & $q^{\sfrac{c}{24}}\mathrm{ch}(\tau)=q^{\sfrac{c}{24}}\sum_{j,m}\alpha_{j,m}X^{(j,m)}(\tau)$ \\ \midrule
$\sfrac{1}{2}$ & $ \begin{array}{l}
1+q^2+q^3+2 q^4+2 q^5+3 q^6+3 q^7+5 q^8+5 q^9+7 q^{10}+8 q^{11}+11 q^{12}+12 q^{13}+\cdots \\ 
q^{\sfrac{1}{16}}\left(1+q+q^2+2 q^3+2 q^4+3 q^5+4 q^6+5 q^7+6 q^8+8 q^9+10 q^{10}+12 q^{11}+15 q^{12}+\cdots \right)  \\
q^{\sfrac{1}{2}}\left(1+q+q^2+q^3+2 q^4+2 q^5+3 q^6+4 q^7+5 q^8+6 q^9+8 q^{10}+9 q^{11}+12 q^{12}+\cdots \right)  \end{array}$ \\\midrule
$\sfrac{17}{2}$ & $ \begin{array}{l}
1+q \left(112 \alpha _{1,0}+136\right)+q^2 \left(1456 \alpha _{1,0}+2669\right)+q^3 \left(10640 \alpha _{1,0}+24361\right)+\cdots \\ 
q^{\sfrac{1}{16}}\left(\alpha _{1,0}+q \left(256-7 \alpha _{1,0}\right)+q^2 \left(21 \alpha _{1,0}+4352\right)+q^3 \left(39168-42 \alpha _{1,0}\right)+\cdots \right)  \\
q^{\sfrac{1}{2}}\left(\left(17-16 \alpha _{1,0}\right)+q \left(697-448 \alpha _{1,0}\right)+q^2 \left(8517-4144 \alpha _{1,0}\right)+\cdots \right)  \end{array}$ \\\midrule
$\sfrac{33}{2}$ & $ \begin{array}{l}
1+q \left(256 \alpha _{1,0}+9 \alpha _{2,0}+231\right)+q^2 \left(30720 \alpha _{1,0}+93 \alpha _{2,0}+38940\right)+\cdots \\ 
q^{\sfrac{1}{16}}\left(\alpha _{1,0}+q \left(-15 \alpha _{1,0}-16 \alpha _{2,0}+528\right)+q^2 \left(105 \alpha _{1,0}-144 \alpha _{2,0}+70288\right)+\cdots \right)  \\
q^{\sfrac{1}{2}}\left(\alpha _{2,0}+q \left(-3840 \alpha _{1,0}+36 \alpha _{2,0}+4301\right)+q^2 \left(-177920 \alpha _{1,0}+207 \alpha _{2,0}+247962\right)+\cdots \right)  \end{array}$ \end{tabular}
\end{center}
\end{table}\end{small}

\clearpage

\begin{table}[H]
    \centering
    \begin{tabular}{c|c c c c}
    \toprule
       $\mathfrak{g}_k$  &  $\B[r,1]$ $(r\geq 7)$ & $\D[8,1]$ & $\D[r,1]$ $(r\geq 9)$ & $\E[8,1]$\\
       $\Com_{\mathfrak{g}_k}(\B[7,1])$ CSA  & $\D[r-7,1]$ & $L_{\sfrac12}$ & $\B[r-8,1]$ & $L_{\sfrac12}$\\\bottomrule
    \end{tabular}
\end{table}

\begin{table}[H]
    \centering
    \begin{tabular}{r|c|c|c|c|c|c}
    No. & $c$ & Theory & CSA & $h_i$ & $d_i$ & $\alpha_{j,m}$ \\\toprule
   \stepcounter{theoryno}\thetheoryno & $\sfrac{1}{2}$ & $L_{\sfrac12}$ & $L_{\sfrac12}$ & $\sfrac{1}{16}, \sfrac{1}{2}$ & $1,1$ &  \\ \midrule 
       \stepcounter{theoryno}\thetheoryno & $\sfrac{17}{2}$ & $\B[8,1]\cong\D[16,1]^+\big/(\B[7,1]\hookrightarrow\D[16,1])$ & $\B[8,1]$ & $\sfrac{17}{16}, \sfrac{1}{2}$ & $256,17$ & $0$ \\ 
    \stepcounter{theoryno}\thetheoryno & $\sfrac{17}{2}$ & $\E[8,1]L_{\sfrac12}\cong \E[8,1]^2\big/(\B[7,1]\hookrightarrow\E[8,1])$ & $\E[8,1]L_{\sfrac12}$ & $\sfrac{1}{16}, \sfrac{1}{2}$ & $1,1$ & $1$ \\ \midrule 
\stepcounter{theoryno}\thetheoryno & $\sfrac{33}{2}$ & $\mathbf{S}(\D[8,1]^{3})\big/(\B[7,1] \hookrightarrow \D[8,1])$ & $\D[8,1]^{2} L_{\sfrac{1}{2}}$ & $\sfrac{17}{16}, \sfrac{1}{2}$ & $512, 1$ & $0, 1$ \\
\stepcounter{theoryno}\thetheoryno & $\sfrac{33}{2}$ & $\mathbf{S}(\B[8,1] \E[8,2])\big/(\B[7,1] \hookrightarrow \B[8,1])$ & $\E[8,2] \mathsf{U}_{1,4}$ & $\sfrac{17}{16}, \sfrac{1}{2}$ & $496, 2$ & $0, 2$ \\
\stepcounter{theoryno}\thetheoryno & $\sfrac{33}{2}$ & $\mathbf{S}(\A[15,1] \D[9,1])\big/(\B[7,1] \hookrightarrow \D[9,1])$ & $\A[1,2] \A[15,1]$ & $\sfrac{17}{16}, \sfrac{1}{2}$ & $480, 3$ & $0, 3$ \\
\stepcounter{theoryno}\thetheoryno & $\sfrac{33}{2}$ & $\mathbf{S}(\D[10,1] \E[7,1]^{2})\big/(\B[7,1] \hookrightarrow \D[10,1])$ & $\B[2,1] \E[7,1]^{2}$ & $\sfrac{17}{16}, \sfrac{1}{2}$ & $448, 5$ & $0, 5$ \\
\bottomrule
\end{tabular}
\end{table}

\clearpage 

\subsubsection{$\normalfont (\AA_1,2)$}\label{app:(A1,2)}{\ }

\begin{table}[H]
\begin{tabular}{c|c}
$\mathcal{C}$ & $(\AA_1,2)$\\\midrule 
Other names & $(\BB_{9+8n},1)$, $\overline{(\BB_{6+8n},1)}$ \\
Fusion category & $\mathrm{TY}_-(\ZZ_2)$ \\
$\mathcal{S}$ & $\frac12 \left(
\begin{array}{ccc}
 1 & \sqrt{2} & 1 \\
 \sqrt{2} & 0 & -\sqrt{2} \\
 1 & -\sqrt{2} & 1 \\
\end{array}
\right)$ \\
$h_i~\mathrm{mod}~1$ & $0$, $\sfrac3{16}$, $\sfrac12$ \\
QDims &  $1$, $\sqrt{2}$, $1$ \\
\multirow{2}{*}{Fusion rules} & $\Phi_1^2=1+\Phi_2$, \ $\Phi_2^2=1$,   \\ 
&  $\Phi_1\times \Phi_2=\Phi_2\times \Phi_1=\Phi_1$
\end{tabular}
\end{table}

\begin{table}[H]
\begin{center}
\begin{tabular}{c|c|c}
$c~\mathrm{mod}~24$ & $\lambda$ & $\chi$ \\ \midrule
$\sfrac{3}{2}$ & $\left(
\begin{array}{ccc}
 -\sfrac{1}{16} & 0 & 0 \\
 0 & -\sfrac{7}{8} & 0 \\
 0 & 0 & -\sfrac{9}{16} \\
\end{array}
\right)$ & $\left(
\begin{array}{ccc}
 3 & 43008 & 2871 \\
 2 & -21 & -54 \\
 3 & -2048 & 270 \\
\end{array}
\right)$ \\
$\sfrac{19}{2}$ & $\left(
\begin{array}{ccc}
 -\sfrac{19}{48} & 0 & 0 \\
 0 & -\sfrac{5}{24} & 0 \\
 0 & 0 & -\sfrac{43}{48} \\
\end{array}
\right)$ & $\left(
\begin{array}{ccc}
 171 & 40 & 5031 \\
 512 & -5 & -22016 \\
 19 & -8 & 86 \\
\end{array}
\right)$ \\
$\sfrac{35}{2}$ & $\left(
\begin{array}{ccc}
 -\sfrac{35}{48} & 0 & 0 \\
 0 & -\sfrac{13}{24} & 0 \\
 0 & 0 & -\sfrac{11}{48} \\
\end{array}
\right)$ & $\left(
\begin{array}{ccc}
 210 & 128 & 11 \\
 1120 & -13 & -32 \\
 4655 & -1664 & 55 \\
\end{array}
\right)$ \end{tabular}
\end{center}
\end{table}

\begin{small}\begin{table}[H]
\begin{center}
\begin{tabular}{r|l}
$c$ & $q^{\sfrac{c}{24}}\mathrm{ch}(\tau)=q^{\sfrac{c}{24}}\sum_{j,m}\alpha_{j,m}X^{(j,m)}(\tau)$ \\ \toprule
$\sfrac{3}{2}$ & $ \begin{array}{l}
1+3 q+9 q^2+15 q^3+30 q^4+54 q^5+94 q^6+153 q^7+252 q^8+391 q^9+\cdots \\ 
q^{\sfrac{3}{16}}\left(2+6 q+12 q^2+26 q^3+48 q^4+84 q^5+146 q^6+240 q^7+384 q^8+\cdots \right)  \\
q^{\sfrac{1}{2}}\left(3+4 q+12 q^2+21 q^3+43 q^4+69 q^5+123 q^6+193 q^7+318 q^8+\cdots \right)  \end{array}$ \\\midrule
$\sfrac{19}{2}$ & $ \begin{array}{l}
1+q \left(40 \alpha _{1,0}+171\right)+q^2 \left(320 \alpha _{1,0}+4237\right)+q^3 \left(1648 \alpha _{1,0}+46075\right)+\cdots \\ 
q^{\sfrac{3}{16}}\left(\alpha _{1,0}+q \left(512-5 \alpha _{1,0}\right)+q^2 \left(10 \alpha _{1,0}+9728\right)+q^3 \left(97280-15 \alpha _{1,0}\right)+\cdots \right)  \\
q^{\sfrac{1}{2}}\left(\left(19-8 \alpha _{1,0}\right)+q \left(988-120 \alpha _{1,0}\right)+q^2 \left(14896-760 \alpha _{1,0}\right)+\cdots \right)  \end{array}$ \\\midrule
$\sfrac{35}{2}$ & $ \begin{array}{l}
1+q \left(128 \alpha _{1,0}+11 \alpha _{2,0}+210\right)+q^2 \left(11648 \alpha _{1,0}+176 \alpha _{2,0}+47425\right)+\cdots \\ 
q^{\sfrac{3}{16}}\left(\alpha _{1,0}+q \left(-13 \alpha _{1,0}-32 \alpha _{2,0}+1120\right)+q^2 \left(78 \alpha _{1,0}-352 \alpha _{2,0}+143392\right)+\cdots \right)  \\
q^{\sfrac{1}{2}}\left(\alpha _{2,0}+q \left(-1664 \alpha _{1,0}+55 \alpha _{2,0}+4655\right)+q^2 \left(-59904 \alpha _{1,0}+451 \alpha _{2,0}+329707\right)+\cdots \right)  \end{array}$ \end{tabular}
\end{center}
\end{table}\end{small}

\clearpage

\begin{table}[H]
    \centering
    \begin{tabular}{c|c c c c}
    \toprule
       $\mathfrak{g}_k$  &  $\B[r,1]$ $(r\geq 6)$ & $\D[7,1]$ & $\D[r,1]$ $(r\geq 8)$ & $\E[8,1]$\\
       $\Com_{\mathfrak{g}_k}(\B[6,1])$ CSA  & $\D[r-6,1]$ & $L_{\sfrac12}$ & $\B[r-7,1]$ & $\A[1,2]$\\\bottomrule
    \end{tabular}
\end{table}

\begin{table}[H]
    \centering
    \begin{tabular}{r|c|c|c|c|c|c}
    No. & $c$ & Theory & CSA & $h_i$ & $d_i$ & $\alpha_{j,m}$ \\\toprule
\stepcounter{theoryno}\thetheoryno & $\sfrac{3}{2}$ & $\A[1,2]$ & $\A[1,2]$ & $\sfrac{3}{16}, \sfrac{1}{2}$ & $2,3$ &  \\\midrule 
\stepcounter{theoryno}\thetheoryno & $\sfrac{19}{2}$ & $\B[9,1]\cong\D[16,1]^+\big/(\B[6,1]\hookrightarrow\D[16,1])$ & $\B[9,1]$ & $\sfrac{19}{16}, \sfrac{1}{2}$ & $512,19$ & $0$ \\
\stepcounter{theoryno}\thetheoryno & $\sfrac{19}{2}$ & $\A[1,2]\E[8,1]\cong \E[8,1]^2\big/(\B[6,1]\hookrightarrow \E[8,1])$ & $\A[1,2]\E[8,1]$ & $\sfrac{3}{16}, \sfrac{1}{2}$ & $2,3$ & $2$ \\\midrule 
\stepcounter{theoryno}\thetheoryno & $\sfrac{35}{2}$ & $\mathbf{S}(\B[6,1] \C[10,1])\big/(\B[6,1] \hookrightarrow \B[6,1])$ & $\C[10,1]$ & $\sfrac{19}{16}, \sfrac{3}{2}$ & $1120, 4655$ & $0, 0$ \\
\stepcounter{theoryno}\thetheoryno & $\sfrac{35}{2}$ & $\mathbf{S}(\A[11,1] \D[7,1] \E[6,1])\big/(\B[6,1] \hookrightarrow \D[7,1])$ & $\A[11,1] \E[6,1] L_{\sfrac{1}{2}}$ & $\sfrac{19}{16}, \sfrac{1}{2}$ & $1088, 1$ & $0, 1$ \\
\stepcounter{theoryno}\thetheoryno & $\sfrac{35}{2}$ & $\mathbf{S}(\D[8,1]^{3})\big/(\B[6,1] \hookrightarrow \D[8,1])$ & $\A[1,2] \D[8,1]^{2}$ & $\sfrac{19}{16}, \sfrac{1}{2}$ & $1024, 3$ & $0, 3$ \\
\stepcounter{theoryno}\thetheoryno & $\sfrac{35}{2}$ & $\mathbf{S}(\B[8,1] \E[8,2])\big/(\B[6,1] \hookrightarrow \B[8,1])$ & $\A[1,1]^{2} \E[8,2]$ & $\sfrac{19}{16}, \sfrac{1}{2}$ & $992, 4$ & $0, 4$ \\
\bottomrule
\end{tabular}
\end{table}

\clearpage

\subsubsection{$\normalfont (\BB_2,1)$}\label{app:(B2,1)}{\ }

\begin{table}[H]
\begin{tabular}{c|c}
$\mathcal{C}$ & $(\BB_2,1)$\\\midrule 
Other names & $(\BB_{2+8n},1)$, $\overline{(\BB_{5+8n},1)}$ \\
Fusion category & $\mathrm{TY}_-(\ZZ_2)$ \\
$\mathcal{S}$ & $\frac12 \left(
\begin{array}{ccc}
 1 & \sqrt{2} & 1 \\
 \sqrt{2} & 0 & -\sqrt{2} \\
 1 & -\sqrt{2} & 1 \\
\end{array}
\right)$ \\
$h_i~\mathrm{mod}~1$ & $0$, $\sfrac5{16}$, $\sfrac12$ \\
QDims &  $1$, $\sqrt{2}$, $1$ \\
\multirow{2}{*}{Fusion rules} & $\Phi_1^2=1+\Phi_2$, \ $\Phi_2^2=1$,   \\ 
&  $\Phi_1\times \Phi_2=\Phi_2\times \Phi_1=\Phi_1$
\end{tabular}
\end{table}

\begin{table}[H]
\begin{center}
\begin{tabular}{c|c|c}
$c~\mathrm{mod}~24$ & $\lambda$ & $\chi$ \\ \midrule
$\sfrac{5}{2}$ & $\left(
\begin{array}{ccc}
 -\sfrac{5}{48} & 0 & 0 \\
 0 & -\sfrac{19}{24} & 0 \\
 0 & 0 & -\sfrac{29}{48} \\
\end{array}
\right)$ & $\left(
\begin{array}{ccc}
 10 & 19456 & 3393 \\
 4 & -19 & -116 \\
 5 & -1024 & 261 \\
\end{array}
\right)$ \\
$\sfrac{21}{2}$ & $\left(
\begin{array}{ccc}
 -\sfrac{7}{16} & 0 & 0 \\
 0 & -\sfrac{1}{8} & 0 \\
 0 & 0 & -\sfrac{15}{16} \\
\end{array}
\right)$ & $\left(
\begin{array}{ccc}
 210 & 12 & 4785 \\
 1024 & -3 & -46080 \\
 21 & -4 & 45 \\
\end{array}
\right)$ \\
$\sfrac{37}{2}$ & $\left(
\begin{array}{ccc}
 -\sfrac{37}{48} & 0 & 0 \\
 0 & -\sfrac{11}{24} & 0 \\
 0 & 0 & -\sfrac{13}{48} \\
\end{array}
\right)$ & $\left(
\begin{array}{ccc}
 185 & 64 & 13 \\
 2368 & -11 & -64 \\
 4921 & -704 & 78 \\
\end{array}
\right)$ \end{tabular}
\end{center}
\end{table}

\begin{small}\begin{table}[H]
\begin{center}
\begin{tabular}{r|l}
$c$ & $q^{\sfrac{c}{24}}\mathrm{ch}(\tau)=q^{\sfrac{c}{24}}\sum_{j,m}\alpha_{j,m}X^{(j,m)}(\tau)$ \\ \midrule
$\sfrac{5}{2}$ & $ \begin{array}{l}
1+10 q+30 q^2+85 q^3+205 q^4+465 q^5+960 q^6+1925 q^7+3660 q^8+6775 q^9+\cdots \\ 
q^{\sfrac{5}{16}}\left(4+20 q+60 q^2+160 q^3+380 q^4+824 q^5+1700 q^6+3340 q^7+6300 q^8+\cdots \right)  \\
q^{\sfrac{1}{2}}\left(5+15 q+56 q^2+130 q^3+315 q^4+665 q^5+1380 q^6+2651 q^7+5020 q^8+\cdots \right)  \end{array}$ \\\midrule
$\sfrac{21}{2}$ & $ \begin{array}{l}
1+q \left(12 \alpha _{1,0}+210\right)+q^2 \left(52 \alpha _{1,0}+6426\right)+q^3 \left(168 \alpha _{1,0}+82845\right)+\cdots \\ 
q^{\sfrac{5}{16}}\left(\alpha _{1,0}+q \left(1024-3 \alpha _{1,0}\right)+q^2 \left(3 \alpha _{1,0}+21504\right)+q^3 \left(236544-4 \alpha _{1,0}\right)+\cdots \right)  \\
q^{\sfrac{1}{2}}\left(\left(21-4 \alpha _{1,0}\right)+q \left(1351-24 \alpha _{1,0}\right)+q^2 \left(24780-96 \alpha _{1,0}\right)+\cdots \right)  \end{array}$ \\\midrule
$\sfrac{37}{2}$ & $ \begin{array}{l}
1+q \left(64 \alpha _{1,0}+13 \alpha _{2,0}+185\right)+q^2 \left(4224 \alpha _{1,0}+299 \alpha _{2,0}+56351\right)+\cdots \\ 
q^{\sfrac{5}{16}}\left(\alpha _{1,0}+q \left(-11 \alpha _{1,0}-64 \alpha _{2,0}+2368\right)+q^2 \left(55 \alpha _{1,0}-832 \alpha _{2,0}+292928\right)+\cdots \right)  \\
q^{\sfrac{1}{2}}\left(\alpha _{2,0}+q \left(-704 \alpha _{1,0}+78 \alpha _{2,0}+4921\right)+q^2 \left(-19008 \alpha _{1,0}+884 \alpha _{2,0}+427868\right)+\cdots \right)  \end{array}$ \end{tabular}
\end{center}
\end{table}\end{small}

\clearpage

\begin{table}[H]
    \centering
    \begin{tabular}{c|c c c c c}
    \toprule
       $\mathfrak{g}_k$  &  $\B[r,1]$ $(r\geq 5)$ & $\D[6,1]$ & $\D[r,1]$ $(r\geq 7)$ & $\E[7,1]$ & $\E[8,1]$\\
       $\Com_{\mathfrak{g}_k}(\B[5,1])$ CSA  & $\D[r-5,1]$ & $L_{\sfrac12}$ & $\B[r-6,1]$ & $\A[1,1]L_{\sfrac12}$ & $\B[2,1]$\\\bottomrule
    \end{tabular}%
\end{table}

\clearpage

\subsubsection{$\normalfont (\BB_3,1)$}\label{app:(B3,1)}{\ }

\begin{table}[H]
%
\end{table}

\clearpage

\subsubsection{$\normalfont (\BB_4,1)$}\label{app:(B4,1)}{\ }

\begin{table}[H]
%
\end{table}

\clearpage

\subsubsection{$\normalfont (\BB_5,1)$}\label{app:(B5,1)}{\ }

\begin{table}[H]
%
\end{table}\end{small}

\clearpage

\subsubsection{$\normalfont (\BB_6,1)$}\label{app:(B6,1)}{\ }

\begin{table}[H]
%
\\\vspace{.1in}
    See also Appendix \ref{app:(E7,1)} for information on commutants of the form $\Com_{\mathfrak{g}_1}(\A[1,1])$, from which one can determine commutants of the form $\Com_{\mathfrak{g}_1\otimes \mathfrak{g}'_1}(\A[1,2])$.
\end{table}

\begin{footnotesize}\begin{table}[H]
    \centering
    %
\end{table}\end{footnotesize}

\clearpage

\subsubsection{$\normalfont (\BB_7,1)$}\label{app:(B7,1)}{\ }

\begin{table}[H]
%
\end{table}

\clearpage

\subsubsection{$\normalfont (\AA_1,5)_{\sfrac12}$}\label{app:(A1,5)half}{\ }

\begin{table}[H]
%
\end{table}
\end{footnotesize}

\clearpage

\subsubsection{$\normalfont \overline{(\AA_1,5)}_{\sfrac12}$}\label{app:C(A1,5)half}{\ }

\begin{table}[H]
%
\end{table}

\clearpage

\subsection{Four primaries}\label{app:data:4primaries}

\subsubsection{$\normalfont (\textsl{U}_1,4)$}{\ }\label{app:(U1,4)}

\begin{table}[H]
%
\end{table}
\end{small}

\clearpage

\subsubsection{$\normalfont (\AA_1,1)^{\boxtimes 2}$}{\ }\label{app:(A1,1)^2}

\begin{table}[H]
%
\end{table}\end{footnotesize}

\clearpage

 \subsubsection{$\normalfont (\AA_3,1)$}{\ }\label{app:(A3,1)}

\begin{table}[H]
%
\end{table}\end{footnotesize}

 \clearpage

 \subsubsection{$\normalfont (\DD_4,1)$}{\ }\label{app:(D4,1)}

\begin{table}[H]
%
\end{table}\end{footnotesize}

 \clearpage

 \subsubsection{$\normalfont (\DD_5,1)$}{\ }\label{app:(D5,1)}

\begin{table}[H]
%
\end{table}\end{footnotesize}

\clearpage

\subsubsection{$\normalfont (\DD_6,1)$}{\ }\label{app:(D6,1)}

\begin{table}[H]
%
\\\vspace{.1in}
    See also Appendix \ref{app:(E7,1)} for information on commutants of the form $\Com_{\mathfrak{g}_k}(\A[1,1])$.
\end{table}

\begin{scriptsize}\begin{table}[H]
    \centering
    %
\end{table}\end{scriptsize}

\clearpage

\subsubsection{$\normalfont (\DD_7,1)$}{\ }\label{app:(D7,1)}

\begin{table}[H]
%
\end{small}
\end{table}

\clearpage

\subsubsection{$\normalfont (\DD_8,1)$}{\ }\label{app:(D8,1)}

\begin{table}[H]
%
\end{table}
Note that the embedding $\D[8,1]\hookrightarrow \E[8,1]$ is not a primitive embedding. Thus, one should not count cosets which arise from this embedding, as they have the wrong modular category: they effectively just delete an $\E[8,1]$ factor, and therefore lead to another holomorphic VOA.

\begin{table}[H]
    \centering
    %
\end{table}

\clearpage

\subsubsection{$\normalfont (\AA_1,3)$}{\ }\label{app:(A1,3)}

\begin{table}[H]
%
\end{center}
\end{table}\end{small}

\clearpage

All relevant commutants of the form $\Com_{\mathcal{K}}(\A[1,1]\F[4,1])$ can be derived from the tables in Appendix \ref{app:(E7,1)} and Appendix \ref{app:(G2,1)}.

\begin{footnotesize}\begin{table}[H]
    \centering
    %
\end{table}\end{footnotesize}

\clearpage

\subsubsection{$\normalfont (\CC_3,1)$}{\ }\label{app:(C3,1)}

\begin{table}[H]
%
\end{center}
\end{table}\end{small}

\clearpage

All relevant commutants of the form $\Com_{\mathcal{K}}(\A[1,1]\G[2,1])$ can be derived from the tables in Appendix \ref{app:(E7,1)} and Appendix \ref{app:(F4,1)}.
\begin{scriptsize}\begin{table}[H]
\begin{center}
    %
\end{center}
\end{table}\end{scriptsize}

\clearpage

\subsubsection{$\normalfont (\GG_2,2)$}{\ }\label{app:(G2,2)}

\begin{table}[H]
%
\end{table}

\clearpage

\subsubsection{$\normalfont (\AA_1,7)_{\sfrac12}$}{\ }\label{app:(A1,7)half}

\begin{table}[H]
%
\\\vspace{.1in}
    See also Appendix \ref{app:(F4,1)} for information on commutants of the form $\Com_{\mathfrak{g}_1}(\G[2,1])$, from which one can determine commutants of the form $\Com_{\mathfrak{g}_1\otimes \mathfrak{g}'_1}(\G[2,2])$. We also define the chiral algebras
\begin{align}
\begin{split}
    & \hspace{.25in} \mathsf{V}^{(1)}_{\sfrac43}:= \A[6,1]\big/\G[2,2], \ \  \mathsf{V}^{(2)}_{\sfrac43}:=\D[7,1]\big/(\G[2,2]\mathsf{U}_1), \ \  \mathsf{V}^{(3')}_{\sfrac43} := \B[3,2]\big/\G[2,2], \ \ \mathsf{V}^{(3)}_{\sfrac73}:=\D[4,2]\big/\G[2,2]\\ 
    & \mathsf{V}^{(4)}_{\sfrac{40}{21}}:=\E[6,2]\big/(\A[2,4]\G[2,2]), \ \  \mathsf{V}^{(5)}_{\sfrac{49}{30}}:=\E[7,2]\big/(\C[3,2]\G[2,2]), \ \ \mathsf{V}^{(6)}_{\sfrac{91}{66}}:=\E[8,2]\big/(\G[2,2]\F[4,2]), \ \ \mathsf{V}^{(7)}_{\sfrac{70}{33}}:= \F[4,2]\big/(\A[1,16]\G[2,2])
\end{split}
\end{align}
It is desirable to have a more explicit description of these cosets, in particular to prove that they are strongly regular. \\ \vspace{.2in}

\begin{scriptsize}
    \begin{center}

\end{table}\end{scriptsize}

\clearpage 

\subsubsection{$\normalfont (\GG_2,1)^{\boxtimes 2}$}{\ }\label{app:(G2,1)^2}

\begin{table}[H]
%
\end{center}
\end{table}\end{small}

\clearpage

All relevant commutants of the form $\Com_{\mathcal{K}}(\F[4,1]^2)$ can be derived from the tables in Appendix \ref{app:(G2,1)}.

\begin{table}[H]
    \centering
    %
\end{table}

\clearpage

\subsubsection{$\normalfont (\FF_4,1)^{\boxtimes 2}$}{\ }\label{app:(F4,1)^2}

\begin{table}[H]
%
\end{center}
\end{table}\end{small}

\clearpage

All relevant commutants of the form $\Com_{\mathcal{K}}(\G[2,1]^2)$ can be derived from the tables in Appendix \ref{app:(F4,1)}.

\begin{footnotesize}\begin{table}[H]
    \centering
    %
\end{table}\end{footnotesize}

\clearpage

\subsubsection{$\normalfont (\AA_1,1)\boxtimes (\EE_7,1)$}{\ }\label{app:(A1,1)x(E7,1)}

\begin{table}[H]
%
\end{center}
\end{table}\end{small}

\clearpage

All relevant commutants of the form $\Com_{\mathcal{K}}(\A[1,1]\E[7,1])$ can be derived from the tables in Appendix \ref{app:(A1,1)} and Appendix \ref{app:(E7,1)}. Note that one should not count cosets obtained from the embedding $\A[1,1]\E[7,1]\hookrightarrow \E[8,1]$ as it is not a primitive embedding. A coset by this embedding has the effect of deleting the $\E[8,1]$ factor and leading to another holomorphic VOA.

\begin{table}[H]
    \centering
    %
\end{table}

\clearpage

\subsubsection{$\normalfont (\AA_1,1)\boxtimes (\GG_2,1)$}{\ }\label{app:(A1,1)x(G2,1)}

\begin{table}[H]
%
\end{table}\end{scriptsize}

\clearpage

\subsubsection{$\normalfont (\AA_1,1)\boxtimes (\FF_4,1)$}{\ }\label{app:(A1,1)x(F4,1)}

\begin{table}[H]
%
\end{center}
\end{table}\end{small}

\clearpage 

All relevant commutants of the form $\Com_{\mathcal{K}}(\E[7,1]\G[2,1])$ can be derived from the tables in Appendix \ref{app:(A1,1)} and Appendix \ref{app:(F4,1)}.

\begin{table}[H]
    \centering
    %
\end{table}

\clearpage

\subsubsection{$\normalfont (\GG_2,1)\boxtimes (\FF_4,1)$}{\ }\label{app:(G2,1)x(F4,1)}

\begin{table}[H]
%
\end{center}
\end{table}\end{footnotesize}

\clearpage 

All relevant commutants of the form $\Com_{\mathcal{K}}(\G[2,1]\F[4,1])$ can be derived from the tables in Appendix \ref{app:(G2,1)} and Appendix \ref{app:(F4,1)}. Note that one should not count cosets obtained from the embedding $\G[2,1]\F[4,1]\hookrightarrow \E[8,1]$ as it is not a primitive embedding. A coset by this embedding has the effect of deleting the $\E[8,1]$ factor and leading to another holomorphic VOA.

\begin{table}[H]
    \centering
    %
\end{table}

\clearpage

\bibliographystyle{ytphys}
\bibliography{main}
\end{document}